\documentclass[reprint,superscriptaddress,showpacs,amssymb,amsmath,aps,prx,longbibliography]{revtex4-2}
\usepackage[utf8]{inputenc}
\usepackage{graphicx}
\usepackage{xcolor}
\usepackage{amsmath}
\usepackage{amsthm}
\usepackage{bm}
\usepackage{bbm}
\usepackage{comment}
\usepackage{simpler-wick}
\usepackage{tikz}
\usepackage[compat=1.0.0]{tikz-feynman}

\usepackage{hyperref}
\hypersetup{hidelinks}
\usepackage{float}

\usepackage{mathdots}
\usepackage{lipsum}
\usepackage{verbatim}
\usepackage{natbib}
\usepackage{nccmath}
\usepackage{amsfonts} 
\usepackage[caption=false]{subfig}

\usepackage{amsfonts} 
\DeclareMathSymbol{\shortminus}{\mathbin}{AMSa}{"39}

\newtheorem{theorem}{Theorem}
\newtheorem{lemma}{Lemma}
\newtheorem{suplemma}{Supplemental Lemma}

\newtheorem{corollary}{Corollary}

\newtheorem{definition}{Definition}

\newcommand{\smallprod}{\mathop{\raisebox{0.25ex}{\scalebox{0.7}{$\displaystyle \prod$}}}}

\begin{document}

\title{Unifying spacetime approaches to quantum mechanics}

\author{N. L. Diaz}
\affiliation{Information Sciences, Los Alamos National Laboratory, Los Alamos, New Mexico 87545, USA}
\thanks{nldiaz.unlp@gmail.com}
\affiliation{Center for Non-Linear Studies, Los Alamos National Laboratory, Los Alamos, NM 87545, USA}

\author{M. Cerezo}
\affiliation{Information Sciences, Los Alamos National Laboratory, Los Alamos, New Mexico 87545, USA}

\author{Paolo Braccia}
\affiliation{Theoretical Division, Los Alamos National Laboratory, Los Alamos, New Mexico 87545, USA}

\begin{abstract}
    Recent efforts to formulate quantum mechanics in a way that treats space and time on a more equal footing have led to a large variety of spacetime-oriented approaches. In this work we present a detailed study of \emph{spacetime states}, the objects that play the role of quantum states in the recently introduced framework of ``spacetime quantum mechanics'', and show that the main proposals in the literature are different manifestations of the same underlying object. Path integrals, quantum states over time, pseudo-density matrices, the Page and Wootters mechanism, superdensity operators, and timelike-entanglement proposals all arise from spacetime states through particular evaluations, reduced information, linear maps, or quantum channels. This unification provides explicit mathematical representations of these formalisms, reveals relations among them, and clarifies the spacetime information each one captures. 
    We also study the broader relevance of the spacetime-state point of view for Leggett-Garg inequalities, OTOCs, temporal tensor networks, fermionic systems, relativistic QFTs, quantum reference frames, and classical physics, together with additional insights and perspectives revealed by the common unifying framework.
\end{abstract}

\maketitle

\section{Introduction}

The relativistic revolution profoundly changed the meaning of space and time assumed by Newtonian physics. In Minkowski's famous formulation, \textit{space by itself, and time by itself, have vanished into the merest shadows}, leaving only their union with independent meaning \cite{minkowski1909raum}. Yet quantum mechanics (QM), the second great revolution of twentieth-century physics, appears to take a step back from this perspective: Dirac's statement that the practical task of physics consists in relating events on one section of spacetime to events on another section at a later time becomes almost unavoidable from the canonical quantum point of view \cite{dirac1963evolution}. In fact, the central object of QM---the quantum state, encoding not a mere configuration of the system but the information from which all compatible physical predictions are extracted---is assigned to a spacelike slice. 
States on different slices are then related by Hamiltonian evolution with respect to a time parameter external to the states themselves. Thus, while relativity suggests a self-contained geometrical picture of physics in which, to use Eddington's phrase, \textit{events do not happen; they are just there} \cite{eddington1920space}, canonical QM remains at its core a dynamical theory whose formulation depends on a classical separation between space and time.

This conceptual tension becomes a concrete mathematical obstruction in canonical attempts to quantize gravity, where the Hamiltonian formulation of a generally covariant theory leads to constraints rather than to an ordinary evolution equation \cite{dirac1958theory,isham1993canonical,anderson2012problem}. At the origin of these difficulties, often gathered under the name of the ``problem of time'', lies the fact that the quantum formalism presupposes the very temporal structure that gravity teaches us should be dynamical.

A more recent motivation to revisit this issue comes from quantum information. Over the last decades, the idea that information is physical, together with the study of quantum correlations, has greatly refined our understanding of the quantum world, of how it differs from classical physics, and of how classical behavior emerges from quantum systems \cite{zurek2003decoherence}. 
 In particular, decoherence \cite{zurek2003decoherence}, thermalization \cite{abanin2019colloquium} and quantum Darwinism  \cite{zurek2009quantum} have shown that classical behavior can often be understood in terms of how quantum information is redistributed, hidden, or redundantly encoded across many degrees of freedom \cite{touil2024branching}.
In this sense, quantum information 
provided the tools needed to make the emergence of classical behavior a well-posed physical problem.
At the same time, 
these developments take the temporal aspects of the theory for granted: standard quantum information considerations inherit from the axioms of QM the same 
classically 
 assumed separation between space and time. 
In particular, while quantum states encode quantum correlations across space (or subsystems), there is no equally canonical object encoding quantum correlations across time.

At the intersection of quantum information and high-energy physics, and in particular in holography, this point becomes especially striking: geometric quantities are related to entanglement entropies \cite{ryu2006holographic}, and spacetime itself has been argued to emerge from quantum entanglement \cite{van2010building}. However, while spatial geometry can be recovered from standard Hilbert-space entanglement \cite{cao2017space}, the emergence of time remains an additional challenge. Recent work on timelike entanglement \cite{harper2023timelike, heller2025temporal, milekhin2025observable} aims to tackle the temporal side of this problem by means of pseudo-entropies, namely complex extensions of standard entanglement entropies, not given by standard quantum states. These considerations suggest that 
the need to separate space and time at the classical level represents a genuine obstruction to our understanding of the emergence of spacetime from correlations.

Motivated by different aspects of the problems previously described,  several recent approaches have sought to formulate QM in a more spacetime-oriented Hilbert-space language. Notable proposals include quantum states over time \cite{horsman2017can}, operators defined to  contain standard quantum states as marginals; pseudo-density matrices \cite{fitzsimons2015quantum, fullwood2024operator}, which generalize expectation values to sequential-in-time measurements; superdensity operators \cite{cotler2018superdensity}, which encode spacetime correlations in an extended operator space; the Page and Wootters mechanism \cite{page1983evolution}, together with its quantum-information-inspired revival \cite{giovannetti2015quantum, boette2016system, giovannetti2023geometric}, where external evolution is replaced by correlations with a quantum clock system; timelike-entanglement proposals \cite{milekhin2025observable, guo2025spacetime}, which assign entropic meaning to regions separated in time; and  work on quantum reference frames \cite{giacomini2019quantum} and indefinite causal order \cite{castro2018dynamics}, which questions the classical status of reference systems and causal structure themselves. 
While all of these proposals allow one to recover standard QM under reasonable assumptions, they start from rather different points of view and emphasize different aspects of the same underlying problem. Despite their common scope, they remain notably distinct from one another. Thus, it remains an open question whether QM truly admits a natural spacetime formulation in a canonical language.

Recently, the framework of spacetime quantum mechanics (SQM) has emerged as an alternative answer to this problem \cite{diaz2023spacetime,diaz2025spacetime}. Its starting point is to associate independent Hilbert-space degrees of freedom to each spacetime point, or classical ``event''. Thus, while SQM still assumes an underlying classical spacetime structure, the quantum object is no longer tied to a spacelike slice. In this geometrical picture, standard evolution is not imposed externally, but recovered from correlators defined through a quantum action \cite{diaz2021spacetime}, the operator counterpart of the action of classical mechanics. One of the strengths of this proposal is that it applies to both non-relativistic and relativistic systems, providing, in particular, a canonical Hilbert-space setting in which Lorentz covariance is manifest for special relativistic QFTs \cite{diaz2023spacetime,diaz2026quantum}. Moreover, the path-integral formulation \cite{feynman1948space}, the most successful spacetime approach to QM but one that traditionally replaces the canonical structure by a sum over histories, is naturally embedded in a Hilbert space within the SQM framework \cite{diaz2021path}.

In the present work we develop the SQM framework substantially further, with particular emphasis on what we call \emph{spacetime states}. In SQM, these are the objects that play the role played by ordinary quantum states in standard QM. They encode the physical information of the system, including their possible ``histories'', but they are assigned to spacetime rather than to a spacelike surface. 
In this way, the dynamical point of view emphasized by Dirac is recovered from fixed-time marginals of spacetime states, but it is no longer the only possible description. A full spacetime description of quantum theories also becomes feasible.

We then show that, once the structure and properties of spacetime states are properly grasped, all the spacetime approaches listed above 
can be understood as manifestations of the same underlying formalism: they arise from spacetime states by particular evaluations, reductions or linear maps. We represent this unification in Figure \ref{fig:schemeintro}. 
This result provides explicit mathematical representations and natural generalizations of these formalisms, while also revealing relations among them that are hidden when they are developed separately. In turn, these results give additional insight into spacetime states themselves.
With this unification in place, SQM can arguably be regarded as the natural extension of QM to spacetime. In particular, other spacetime approaches emerge from a single constrained framework 
that preserves, or naturally generalizes, many of the features of standard QM. \\

The manuscript is organized as follows. In Section \ref{sec:formalism} we review the SQM formalism and study new properties of spacetime states. In Section \ref{sec:unifying} we prove the unifying results summarized in Figure \ref{fig:schemeintro}, showing how the main spacetime approaches to QM arise from spacetime states.
After establishing this unification, we use the spacetime-state point of view to revisit several topics where the temporal structure of QM is directly involved. This is presented in Section \ref{sec:additional}, where we discuss Leggett-Garg inequalities, folded spacetime states and out-of-time-order correlators (OTOCs), temporal tensor networks, fermions, relativistic QFTs, quantum reference frames, and the classical counterpart of the formalism. These results are not all part of the unification result itself, but they situate SQM in a broader physical landscape and illustrate the range of questions that become accessible once quantum states are assigned to spacetime.
We conclude in Section \ref{sec:conclusions} where we add additional conceptual considerations and perspectives.

\begin{figure}[t!]
    \centering
    \includegraphics[width=\linewidth]{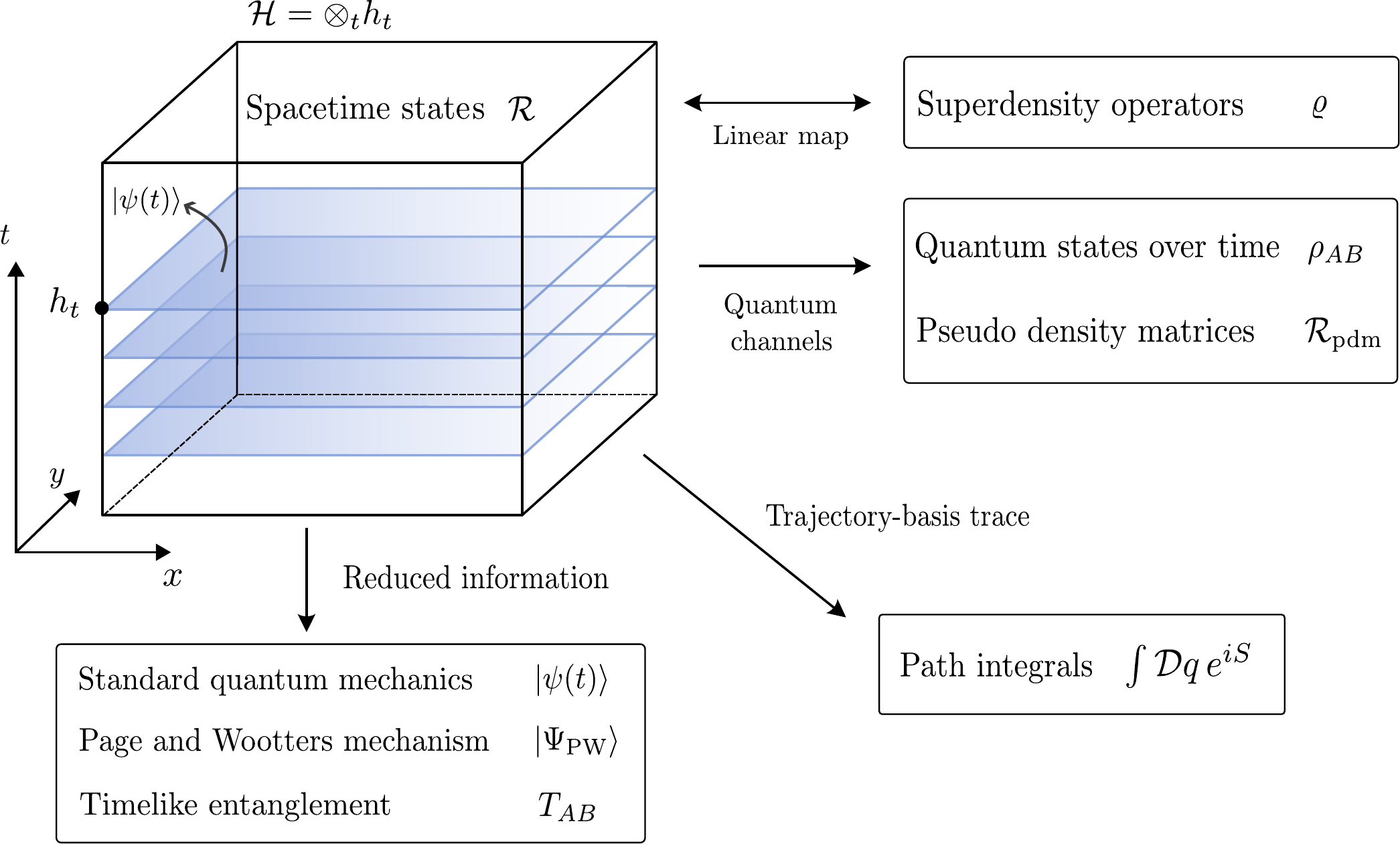}
    \caption{Schematic depiction of how standard QM and its main extensions to a spacetime setting arise from a common object.}
    \label{fig:schemeintro}
\end{figure}

\section{Spacetime Quantum mechanics}\label{sec:formalism}

\subsection{Basic Formalism}\label{sec:basicformalism}
In this section we introduce the formalism developed in \cite{diaz2023spacetime,diaz2025spacetime}, which we will refer to as \emph{spacetime quantum mechanics} (SQM). The presentation is self-contained up to technical details already treated in previous works, and differs slightly from the original formulation in order to better emphasize the role of spacetime states. \\

The first step toward formulating QM directly in spacetime, rather than on a single spacelike slice, is to enlarge the kinematical space of the theory. In standard QM, a system is described by a Hilbert space $h$ associated with its degrees of freedom at a given time. If we want to describe the same system as distributed over a finite spacetime region, the natural analogue is to assign independent copies of these degrees of freedom to the different time slices. This leads to the following definition.

\begin{definition} \textbf{Spacetime Hilbert space}. 
Consider a finite time window $T=\epsilon N$. Given a quantum system which in standard QM is described by a Hilbert space $h$, the corresponding spacetime Hilbert space is
\begin{equation}
    \mathcal{H}=\mathop{{\textstyle\bigotimes}}\limits_{t=0}^{N-1} h_t \simeq h^{\otimes N}\,,
\end{equation}
where each factor $h_t$ is associated with the system at the time slice $t$.
\end{definition}

Let us make a few comments. First, we assume a finite time window and a discrete set of time slices in order to avoid unnecessary mathematical subtleties. The formalism can be extended to continuum and unbounded time limits, as discussed in \cite{diaz2023spacetime,diaz2021path,diaz2025spacetime}. Second, although the definition above is written in terms of copies of the standard Hilbert space $h$, the more fundamental idea is algebraic: time is treated as an additional label for independent quantum degrees of freedom. For instance, for a single particle one may postulate
\begin{equation}
    [q_t,p_{t'}]=i\delta_{tt'}\,,
\end{equation}
which leads to $\mathcal{H}\simeq\otimes_t h$, with $h$ defined by the usual canonical algebra $[q,p]=i$ (we set $\hbar\equiv 1$) \footnote{For a finite number of time labels, $[q_t,p_{t'}]=i\delta_{tt'}$ defines a direct sum of independent canonical algebras. The corresponding Weyl algebra factorizes in time, and Stone-von Neumann gives $\mathcal{H}\simeq\otimes_t h$ up to unitary equivalence.}. If the system is instead a quantum field, with equal-time algebra $[q_x,p_{x'}]=i\delta_{xx'}$ in $h$, the spacetime construction gives
\begin{equation}\label{eq:spacetimealg}
    [q_{tx},p_{t'x'}]=i \delta_{tt'}\delta_{xx'}\,.
\end{equation}
Thus, at the level of the kinematical algebra, space and time appear on the same footing. In this way, the formalism associates quantum operators to spacetime points, or events. The same idea extends to other algebras, including those leading to finite dimensional systems \footnote{The case of fermions is however not equivalent to a tensor-product structure; for this reason we study it in the separate  Section \ref{sec:fermions}.}.

The spacetime Hilbert space is only the kinematical part of the construction. 
We still need to connect this scheme with conventional quantum evolution. To do so, we introduce a few natural definitions. 
\begin{definition}\label{def:eip}\textbf{Time translations across slices}. 
    We define the time translation operator across slices as
    \begin{equation}
        e^{i\epsilon \mathcal{P}}=\prod_{t=0}^{N-2}\text{\upshape SWAP}_{t+1,t}\,.
    \end{equation}
    where the product follows an increasing order to the left $ e^{i\epsilon \mathcal{P}}=\text{\upshape SWAP}_{1,0}\text{\upshape SWAP}_{2,1}\dots \text{\upshape SWAP}_{N-1,N-2}$. 
\end{definition}
Let us remark that this important operator is not related to evolution by itself. Instead it simply connects different slices in a geometrical way. It satisfies
$e^{i\epsilon\mathcal{P}}|i_0i_1\dots i_{N-1}\rangle=|i_{N-1}i_0i_1\dots i_{N-2}\rangle$. Equivalently, 
\begin{equation}
e^{i\epsilon \mathcal{P}}O_te^{-i\epsilon \mathcal{P}}=O_{t+1}    
\end{equation}
with cyclic conditions ($O_N=O_0$) for any operator  $O_t$ acting on $h_t$ as $O$ and trivially for the other slices 
 (to clarify the notation consider $N=3$, we have $O_0= O\otimes \mathbbm{1} \otimes \mathbbm{1}$, $O_1=\mathbbm{1}\otimes O \otimes \mathbbm{1}$, $O_2=\mathbbm{1}\otimes  \mathbbm{1}\otimes O$). Let us also stress that $e^{i\epsilon\mathcal{P}}$ is not separable-in-time.

To make further progress we introduce another useful definition.

\begin{definition} \textbf{Time translations within slices}. \label{def:translationwithin}
    We define the time translation operator within slices as
    \begin{equation}
        e^{i\epsilon \mathcal{K}}=\mathop{{\textstyle\bigotimes}}\limits_{t=0}^{N-1}e^{i\epsilon H_t}=e^{i\epsilon \sum_{t=0}^{N-1}H_t}\,,
    \end{equation}
    with $H$ the Hamiltonian of the system.
\end{definition}
These time translations are the ones more directly related to standard unitary evolution. Their adjoint action on an operator $O_t$ is given by 
\begin{equation}
    e^{i\epsilon \mathcal{K}}O_t e^{-i\epsilon \mathcal{K}}=O_t(\epsilon)\,,
\end{equation}
with $O(\epsilon)=e^{i\epsilon H}O e^{-i\epsilon H}$  where we write conventional Heisenberg evolution within parenthesis. Similarly, a tensor product of operators is also evolved a single step (in the Heisenberg sense). 
For simplicity we  will mostly focus on the time independent case so that the evolution operator of a single time step is given by $U(\epsilon)=e^{-i\epsilon H}$. The time-dependent case can be obtained by direct generalization as discussed in Appendix \ref{app:spacetimest}. 
Let us also recall that the notation $H_t$ indicates that the operator acts trivially in the other slices so that, e.g., for $N=2$ we have $\mathcal{K}=H_0+H_1=H\otimes \mathbbm{1}+\mathbbm{1}\otimes H$.  \\

If we want to relate the Hilbert space time slices to evolution we need to relate translations across time with translations within slices in such a way that 
``moving'' through time, in the geometrical sense, corresponds to standard unitary evolution. This insight leads naturally to the following definition.
\begin{definition}
    \textbf{Quantum action operator (QA)}. The composition of time translations across time with the inverse of time translations within slices defines the quantum action operator $\mathcal{S}$:
    \begin{equation}
        e^{i\mathcal{S}}=e^{i\epsilon \mathcal{P}}e^{-i\epsilon \mathcal{K}}=e^{i\epsilon (\mathcal{P}-\mathcal{K})}\,,
    \end{equation}
    namely $\mathcal{S}=\epsilon(\mathcal{P}-\mathcal{K})$. 
\end{definition}
We have thus defined $e^{i\mathcal{S}}$ from the mismatch between two notions of time translation: geometrical translations across slices and dynamical translations within slices. The reason for calling $\mathcal{S}$ a quantum action is not apparent from the definition alone. However, as we show in Section \ref{sec:PIs}, for systems with a classical analogue, $\mathcal{S}$ takes the same form as the action of classical mechanics, with $\mathcal{P}$ related to the Legendre-transform term and $\mathcal{K}$ to the sum over time of the Hamiltonian.

Before presenting the mechanisms from which evolution can be recovered let us introduce yet another important definition.

\begin{definition}
\textbf{State quantum action operator (SQA)}. The state quantum action operator is defined by the relation
\begin{equation}\label{eq:qainin}
    e^{i\tilde{\mathcal{S}}}=U^\dag_0 (T) \,e^{i\mathcal{S}}= \mathcal{V}^\dag e^{i\epsilon \mathcal{P}}\mathcal{V}\,,
\end{equation}
for $\mathcal{V}=\otimes_{t=0}^{N-1}e^{i\epsilon t H_t}$ and $U_0^\dag(T)=e^{iT H_0}$.
\end{definition}
The difference with the previous QA is a boundary term that becomes important when discussing evolution within a finite time window: while the QA is directly related to transition amplitudes, the SQA is better suited to describing expectation values (see Theorem \ref{th:theorem1} and the discussion below). This is the version that most naturally enters the definition of spacetime states, hence the name.
The second equality in \eqref{eq:qainin} has been proven in \cite{diaz2021path} and can be easily seen through a diagrammatic representation.

With these basic definitions, we can resume our previous discussion about identifying time translations across and within. Mathematically, this identification corresponds to $e^{i\mathcal{S}}\mathcal{O}e^{-i\mathcal{S}}\equiv\mathcal{O}$, or equivalently $e^{i\epsilon\mathcal{P}}\mathcal{O}e^{-i\epsilon\mathcal{P}}\equiv e^{i\epsilon\mathcal{K}}\mathcal{O}e^{-i\epsilon\mathcal{K}}$. Clearly this is not satisfied as an operator equation. What we require instead is 
\begin{equation}\label{eq:constraint}
    \langle e^{i\mathcal{S}}\mathcal{O}e^{-i\mathcal{S}}\rangle_\Gamma=\langle\mathcal{O}\rangle_\Gamma
\end{equation} with $\langle \dots \rangle_\Gamma={\rm Tr}[\Gamma \dots]$. Namely, we want the equality to hold within ``expectation values'', a condition that defines a subspace of possible $\Gamma$ and will lead us to the concept of spacetime state.  The solution is far from unique, and any function of the QA leads to a possible choice of $\Gamma$. In what follows, we will focus only in the case $\Gamma=e^{i\mathcal{S}}$ which leads to many remarkable properties  (we also consider $\Gamma=e^{i\tilde{\mathcal{S}}}$, satisfying the constraint with the SQA). 
The ``expectation values'' 
${\rm Tr}[e^{i\mathcal{S}}\mathcal{O}]$
 will be one of our main objects of study.

Having identified the structure needed to extract physical information from the formalism one can explore its consequences. One immediate but powerful result follows. 
\begin{theorem}\label{th:theorem1}
The following traces in the spacetime Hilbert space $\mathcal{H}$ are connected with standard traces in $h$ as follows:
\begin{subequations}
    \begin{align}
      {\rm Tr}\big[e^{i\mathcal{S}}\otimes_{t=0}^{N-1}O^{(t)}_t\,\big]&={\rm tr}\big[U(T)\,\hat{T}\smallprod\nolimits_{t=0}^{N-1} O^{(t)}(\epsilon t)\big]\\
       {\rm Tr}\big[e^{-i\mathcal{S}}\otimes_{t=0}^{N-1}O^{(t)}_t\,\big]&={\rm tr}[\bar{T}\smallprod\nolimits_{t=0}^{N-1} O^{(t)}(\epsilon t)U^\dag(T)]\\
          {\rm Tr}\big[e^{i\tilde{\mathcal{S}}}\otimes_{t=0}^{N-1}O^{(t)}_t\, \big]&={\rm tr}\big[\hat{T}\smallprod\nolimits_{t=0}^{N-1} O^{(t)}(\epsilon t)\big]\\
       {\rm Tr}\big[e^{-i\tilde{\mathcal{S}}}\otimes_{t=0}^{N-1}O^{(t)}_t\, \big]&={\rm tr}\big[\bar{T}\smallprod\nolimits_{t=0}^{N-1}O^{(t)}(\epsilon t)\big]
\end{align}
\end{subequations}
where $\hat{T}, \bar{T}$ denoting time ordering and anti-time ordering respectively. 
\end{theorem}
Notice that we use the symbol ``${\rm Tr}$'' for traces in $\mathcal{H}$ and ``${\rm tr}$'' for traces in $h$ in order to distinguish quantities in SQM from quantities in ordinary QM. 
Notably, the correlators in time of the exponential of the QA (SQA) are equal to traditional traces of the composition of operators evolved in the Heisenberg picture at different times. The amount of evolution on the r.h.s. is determined by the position (in time) of the operator on the l.h.s. In particular, if a single operator $O_t$ is inserted, the condition \eqref{eq:constraint}, which might be rewritten as $\langle O_{t+1}-O_t\rangle=i\epsilon\langle [H,O_t] \rangle +\mathcal{O}(\epsilon^2)$, is clearly satisfied. 
The proof of these relations can be found in \cite{diaz2023spacetime,diaz2021path,diaz2025spacetime}. Here we provide a Tensor Network (TN) representation of the map in Figures \ref{fig:tntheorem}.

\begin{figure}[t!]
    \centering
\includegraphics[width=\linewidth]{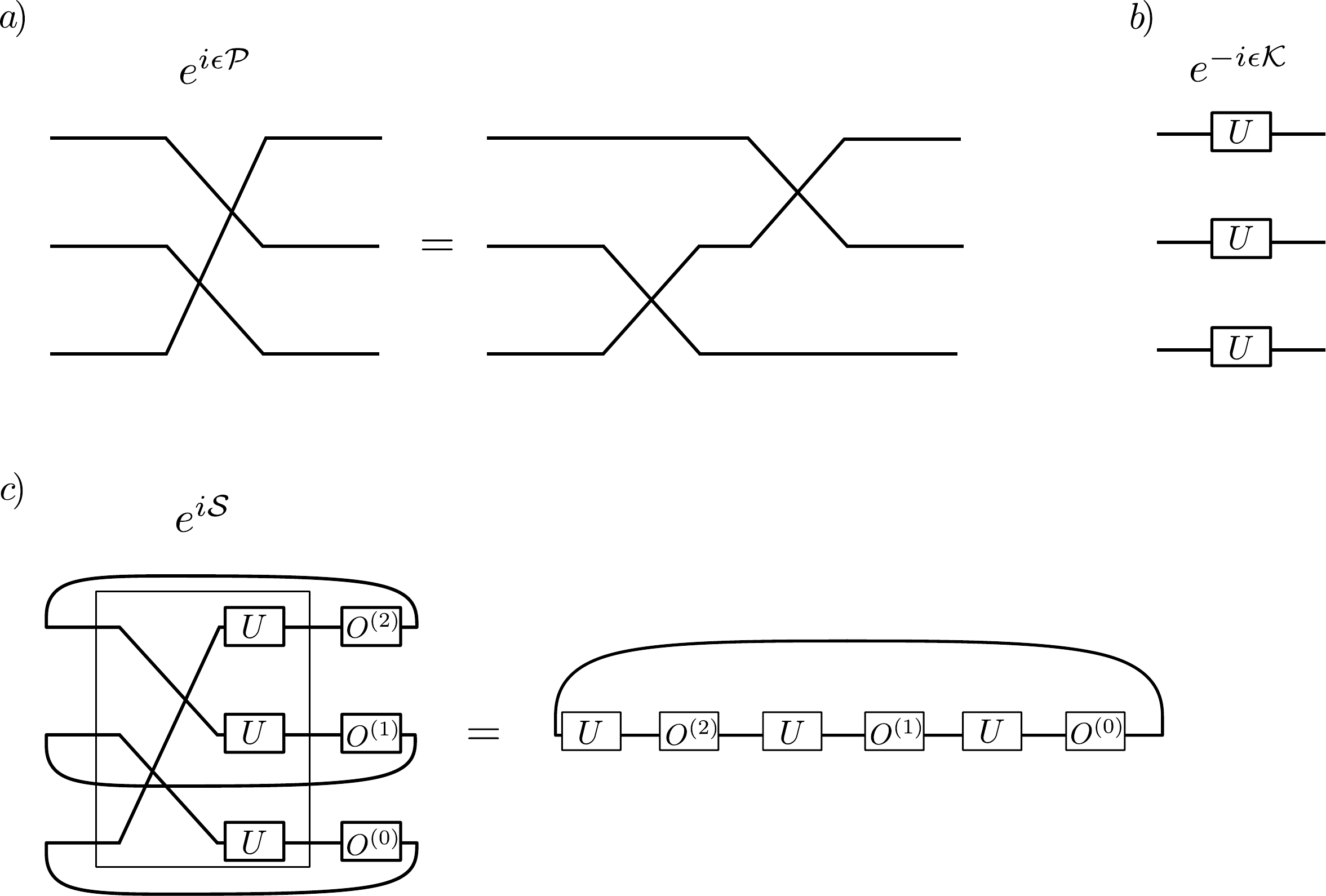}
    \caption{Tensor network representation of the main objects of the formalism for $N=3$.
   Panel a) Time translations across slices and the corresponding SWAP decomposition. Panel b) Time translations within slices. Panel c) the quantum action operator and a particular case of Theorem \ref{th:theorem1}.     }
    \label{fig:tntheorem}
\end{figure}

Let us now frame the usual concept of quantum state within this extended framework. Specifying a quantum state $\rho$ corresponds to specifying all the expectation values of the system at a given time, namely ${\rm tr}[\rho O]$ for all $O\in \mathcal{L}(h)=h\otimes h^\ast$. Thus, considering Theorem \ref{th:theorem1}, we need an object that for operators at a given time, say the first slice, yields the same mean values as $\rho$. At the same time, we want to preserve as much as possible the previous structure. The following definition of spacetime states is what we require.

\begin{definition}\label{def:spacetimest}
    \textbf{Spacetime states}.  We define spacetime states $\mathcal{R}$ as
    \begin{equation}
        \mathcal{R}=\rho_0\, e^{i\tilde{\mathcal{S}}}\,.
    \end{equation}
\end{definition}

The spacetime states $\mathcal{R}$ satisfy ${\rm Tr}[\mathcal{R}]=1$ but are clearly not-Hermitian. On the other hand, considering Theorem \ref{th:theorem1} one can easily prove that the partial trace ${\rm Tr}_{t\neq 0}[\mathcal{R}]=\rho$. The Hermiticity condition of standard QM is thus recovered for equal-time correlators, where the (anti) time-ordering acts trivially. For operators inserted at different times, $\mathcal{R}$ and $\mathcal{R}^\dag$ lead to different time orderings and Hermiticity breaks.  
Let us also notice that, as before, the approach will be to consider quantities such as ${\rm Tr}[\mathcal{R}\mathcal{O}]$ which will now be interpreted as a (generalized or quasi) expectation value of $\mathcal{R}$.
Notice also that $\mathcal{R}=\mathcal{V}^\dag\rho_0 e^{i\epsilon\mathcal{P}}\mathcal{V}$ meaning that for a common initial condition, spacetime states with different histories are unitarily related.

Next, we state a few basic important properties of spacetime states that motivate their name.

\begin{corollary}\label{cor:schrodheis}\textbf{Schr\"{o}dinger and Heisenberg evolution}. 
The partial trace of a spacetime states over all time slices, except a chosen $t$, gives Schr\"{o}dinger evolution
\begin{equation}
    {\rm Tr}_{t'\neq t}[\mathcal{R}]=\rho(t)=e^{-i\epsilon t H}\rho e^{i\epsilon tH}\,.
\end{equation}
Similarly, 
\begin{equation}
    {\rm Tr}[\mathcal{R} O_t]={\rm tr}[\rho O(\epsilon t)]
\end{equation}
for $O(\epsilon t)=e^{i\epsilon tH}O e^{-i\epsilon t H}$ the Heisenberg evolved operator. 
\end{corollary}
These results are once again a direct consequence of Theorem \ref{th:theorem1} and for this reason we state them as a corollary.  Notably, we have recovered the basic pictures of quantum evolution from a framework where there is no evolution in the traditional sense but only correlations. In particular, this confirms that $\mathcal{R}$ contains the information of a standard quantum state and all of its evolution.
Further information is codified within spacetime states, as we now state in the next corollary.

\begin{corollary}\label{cor:wigthman}\textbf{Time-ordered Wightman functions}. The correlators of spacetime states are Wightman functions:
\begin{subequations}
     \begin{align}
       {\rm Tr}\big[\mathcal{R} \otimes_{t=0}^{N-1}O_t^{(t)}\,\big]&={\rm tr}\big[\rho\,\hat{T}\,\smallprod\nolimits_{t=0}^{N-1} O^{(t)}(\epsilon t)\big]\\
        {\rm Tr}\big[\mathcal{R}^\dag \otimes_{t=0}^{N-1}O_t^{(t)}\,\big]&={\rm tr}\big[\,\bar{T}\,\smallprod\nolimits_{t=0}^{N-1} O^{(t)}(\epsilon t)\rho\big]\,.
   \end{align}
\end{subequations}
   \end{corollary}
This result makes use of $\tilde{\mathcal{S}}$ and relates extended mean values with standard mean values of (anti) time- ordered operators. We can state a similar result for the QA regarding transitions. 
   \begin{corollary}\label{cor:wigthmaninout}\textbf{In-Out Time-ordered Wightman functions}. Replacing the SQA in spacetime states with the QA leads to in-out Wightman functions:
   \small
   \begin{subequations}
        \begin{align}
       {\rm Tr}\big[|\psi\rangle_0 \langle \varphi|e^{i\mathcal{S}} \otimes_{t=0}^{N-1}O_t^{(t)}\,\big]&=\langle \varphi,T|\,\hat{T}\smallprod\nolimits_{t=0}^{N-1} O^{(t)}(\epsilon t) |\psi\rangle\\
        {\rm Tr}\big[(|\psi\rangle_0 \langle \varphi| e^{i\mathcal{S}})^\dag \otimes_{t=0}^{N-1}O_t^{(t)}\,\big]&=\langle \psi|\,\bar{T}\smallprod\nolimits_{t=0}^{N-1} O^{(t)}(\epsilon t) |\varphi,T\rangle\,,
   \end{align}
   \end{subequations}
   \normalsize
   with $|\psi,T\rangle=e^{iTH}|\psi\rangle$. 
\end{corollary}

Hence, we have recovered Wightman functions by considering expectation values of operators inserted at different time slices. 
It is now interesting to draw a parallel with traditional quantum states: take the example where Alice and Bob have a qubit each. We have $\rho_{AB}=\sum_{i,j=0}^3\rho_{ij}P^{(i)}_A\otimes P^{(j)}_B$ for $P^{(i)}$ normalized Pauli operators. Thus the global quantum state is completely determined by all the correlators $\rho_{ij}={\rm tr}\big[\rho\, P^{(i)}_A\otimes P^{(j)}_B\,\big]$, with the tensor product indicating a spacelike separation. Similarly, consider two qubits and two times. The corresponding spacetime state can be written as
\begin{equation}
\mathcal{R}=\!\!\!\sum_{i_0,j_0,i_1,j_2}R_{i_0j_0i_1j_1}\,P^{(i_0)}_{A,0}\otimes P^{(j_0)}_{B,0}\otimes P^{(i_1)}_{A,1}\otimes P^{(j_1)}_{B,1}
\end{equation}
 with $R_{i_0j_0i_1j_1}={\rm Tr}\big[\mathcal{R}\, P^{(i_0)}_{A,0}\otimes P^{(j_0)}_{B,0}\otimes P^{(i_1)}_{A,1}\otimes P^{(j_1)}_{B,1}\,\big]$ a tensor playing a similar role to $\rho_{ij}$ in the standard formalism. These correlators are completely determined by Corollary \ref{cor:wigthman} yielding  $$R_{i_0j_0i_1j_1}={\rm tr}\big[\rho_{AB}\, \big(P_A^{(i_1)}\otimes P_B^{(j_1)}\big)\,(\epsilon)\,P_A^{(i_0)}\otimes P_B^{(j_0)}\,\big]\,,$$ 
 which is fixed by the quantum state $\rho_{AB}$ and the Hamiltonian. Conversely, one can consider the problem of extracting the state and Hamiltonian from $\mathcal{R}$, which can be solved by means of Corollary \ref{cor:schrodheis}. 
 These considerations show that 
the Definition \ref{def:spacetimest} of spacetime state provides the only generalization of quantum state whose timelike correlators are time-ordered timelike Wightman functions. In other words, spacetime states
 satisfying Corollary \ref{cor:wigthman} for a given Hilbert space $\mathcal{H}$, initial state and Hamiltonian are \emph{unique}.

On the other hand, Corollary \ref{cor:wigthmaninout} becomes particularly useful when considering transitions. In particular, if no operator is inserted, one recovers \emph{propagators}. \\

Another interesting comment concerns correlators which correspond to projectors:
\begin{equation}\label{eq:KDquasiprobability}
    {\rm Tr}[\mathcal{R}\, |n\rangle \langle n| \otimes |m,\epsilon\rangle \langle m,\epsilon|]
=\langle \psi|m\rangle\langle m|n\rangle\langle n|\psi\rangle\,,
\end{equation}
where we have considered a projector $|n\rangle \langle n|$ at time $0$ and $|m\rangle\langle m|$ at the next time step, with $\rho=|\psi\rangle \langle \psi|$. Here the indices $n,m$ label different families of projectors resolving the identity ($\sum_{n}|n\rangle\langle n|=\sum_m |m\rangle \langle m|=\mathbbm{1}$ with in general $\langle n|m\rangle\neq \delta_{nm}$). We have thus recovered the \emph{Kirkwood-Dirac (KD) quasiprobability} $Q_{KD}[n,m]=\langle \psi|m\rangle\langle m|n\rangle\langle n|\psi\rangle$ \cite{kirkwood1933quantum,dirac1945analogy}.
Notice that $Q_{KD}$ is not a genuine probability distribution precisely because  the projectors may not commute. On the other hand,
$\sum_n Q_{KD}[n,m]=|\langle m|\psi\rangle|^2$ with $\sum_{n,m} Q_{KD}[n,m] \lambda_m=\langle \psi|O|\psi\rangle$ for $O=\sum_m \lambda_m |m\rangle \langle m|$,
so the marginals of $Q_{KD}$ can still define genuine probabilities and reproduce standard expectation values following Born rule. This mirrors the behavior of $\mathcal{R}$, which is not by itself a standard quantum state, but whose marginals can be. Moreover, Eq.\ \eqref{eq:KDquasiprobability} shows that $\mathcal{R}$
 contains $Q_{KD}$ as the expectation value of a Hermitian observable in $\mathcal{H}$. Therefore if $\mathcal{R}$ were a standard quantum state, then $Q_{KD}$ would be a genuine probability distribution.\\

Let us also add a few comments regarding the Euclidean time case. Here, one can define a Euclidean quantum action $e^{-\mathcal{S}_E}=e^{i\epsilon\mathcal{P}} e^{-\mathcal{K}}$ such that thermal correlators are recovered, instead of temporal ones \cite{diaz2025spacetime}. In other words, operators evolve under imaginary time.
In particular, if we take definition \ref{def:spacetimest}, perform a Wick rotation $H\to -iH$ and choose $\rho=e^{-\beta H}/Z$, for $Z={\rm tr}[e^{-\beta H}]$ and $\beta=\epsilon N$, one finds that the spacetime state is given directly by
\begin{equation}\label{eq:Reuclid}
   \mathcal{R}_E= \frac{1}{Z}e^{-\mathcal{S}_E}\,.
\end{equation}
 Just as before we recover the standard thermal state from $\rho={\rm Tr}_{t\neq 0}[\mathcal{R}_E]$ while ${\rm Tr}[e^{-\mathcal{S}_E}]=Z$. Similarly, if a small imaginary component is added to time, and the limit of large $T$ is considered, one can formally replace $\mathcal{R}\to e^{i\mathcal{S}}/Z$ to recover zero temperature Wightman functions, i.e., associated with the ground state of the  Hamiltonian.

\subsection{Pure spacetime states and  subsystems}\label{sec:pureststates}

Given the central role of spacetime states in the formalism, and in order to compare SQM with other proposals, we now study some of their fundamental properties.
While most of the results reviewed in the previous subsection were already presented in \cite{diaz2023spacetime,diaz2025spacetime}, this subsection develops new results showing that spacetime states satisfy several properties that generalize basic features of quantum states, while also providing a clear interpretation of the regimes in which these features break down.

Let us begin by defining ``pseudo-pure spacetime states'' as those corresponding to closed systems throughout their complete \emph{history}, namely pure initial states and unitary evolution. In the notation of Definition \ref{def:spacetimest}, these are spacetime states with $\rho^2=\rho$. In what follows, we simply denote as pure spacetime state any $\mathcal{R}$
generated by a pure initial state and unitary closed evolution (not a rank-one  projector on $\mathcal{H}$).
We first characterize this class of spacetime states, and then explain how subsystems and arbitrary spacetime subregions can be described by reduced spacetime states obtained through partial traces over $\mathcal{R}$. In this way, reduced spacetime states define the spacetime analogue of mixed states. We will also discuss how coherences can arise. \\

\emph{Pure spacetime states}. Spacetime states of the form $\mathcal{R}=|\psi\rangle _0\langle \psi| e^{i\tilde{\mathcal{S}}}$ correspond to closed systems and admit the Jordan decomposition
    \begin{equation}\label{eq:decompositionpix}
\mathcal{R}=\mathop{{\textstyle\bigotimes}}\limits_{t=0}^{N-1} |\psi(\epsilon t)\rangle \langle \psi(\epsilon t)|+X\,,
    \end{equation}
    with $\Pi={\textstyle\bigotimes}_{t=0}^{N-1} |\psi(\epsilon t)\rangle \langle \psi(\epsilon t)|$ a projector and for $X$ satisfying the idempotent condition $X^N=0$ and $\Pi X=X\Pi=0$. 
   This decomposition shows that $\mathcal{R}$ is generally non-diagonalizable. Moreover, we remark that while $\Pi$ has the form of a tensor product in time of the standard evolved states, the additional orthogonal term $X$ makes the spacetime state richer than a simple product history. We derive Eq.\ \eqref{eq:decompositionpix} in the Appendix \ref{app:spacetimest}.

    As a direct consequence, the spectrum of $\mathcal{R}$ corresponds to a non-degenerate eigenvalue $1$ and an eigenvalues $0$ with degeneracy $d^{N}-1$. In particular, all pseudo-entropies, the direct generalization of entropies to  non-Hermitian operators \cite{harper2025non}, vanish:
    \begin{equation}\label{eq:vanishent}
        S_\alpha(\mathcal{R})=0\,,
    \end{equation}
for $ S_\alpha(\mathcal{R})=\frac{1}{1-\alpha}\log({\rm Tr}[\mathcal{R}^\alpha])$ the pseudo-Rényi entropies. Equivalently, one finds for pseudo-purities ${\rm Tr}[\mathcal{R}^k]=1$ for all positive integers $k$. 
As we explain below, this definition of entropy works as a witness of openness of a quantum system, detecting  (i) if the system is open in the sense of standard statistical mixture of quantum states, with e.g., ${\rm Tr}[\mathcal{R}_\rho^{N+1}]={\rm tr}[\rho^2]$ (ii) if the system is not evolving unitarily due to a global entangling unitary evolution.

It is also worth noting that $\mathcal{R}^\dag \mathcal{R}=\rho_0^2=\rho^2\otimes \mathbbm{1}... \otimes \mathbbm{1}$ showing that for pure states the singular values of $\mathcal{R}$ are either $1$, with degeneracy or $d^{N-1}$ or $0$. We thus see that entropies defined via singular value decomposition (SVD) are not vanishing for pure spacetime states. Interestingly,  
this is a manifestation of $\mathcal{R}$ containing causal information among time-slices, independently if the system is closed. 
Considering Corollary \ref{cor:wigthman} we can write an unequal time commutator as
\begin{equation}
    {\rm tr}[\rho [B(t_B),A(t_A)]]={\rm Tr}[(\mathcal{R}-\mathcal{R}^\dag)A_{t_A}B_{t_B}]
\end{equation}
where on the r.h.s. $A_{t_A}$ and $B_{t_B}$  commute for $t_A\neq t_B$ and on the l.h.s. the operators are in the Heisenberg picture (with $t\equiv \epsilon t$ i.e., we measure time in units of $\epsilon$ for ease of notation). Let us recall that in general for unequal time commutators $[B(t_B),A(t_A)]\neq 0$ is a signature of causality, as if $A$ and $B$ initially commute (e.g., correspond to two separated systems) after some time $[A,B(t)]\neq 0$ indicates signaling between the two parts. In our current considerations $A$ and $B$ are general so we must have ${\rm Tr}[\mathcal{R}^\dag\mathcal{R}]>1$ as it follows from the 
 Cauchy-Schwarz  inequality 
 \begin{equation}\label{eq:imagitivboundweak}
     \begin{split}
         |  {\rm tr}[\rho [B(t_B),A(t_A)]]|&\leq ||A||\, ||B||\, ||\mathcal{R}-\mathcal{R}^\dag||\\&=||A||\, ||B||\, \sqrt{2\big({\rm Tr}[\mathcal{R}^\dag \mathcal{R}]-1\big)}\,.
     \end{split}
 \end{equation}
 Here $||. ||$ denotes the Hilbert-Schmidt norm. 
 Below we discuss a tighter version of this bound that corresponds to restricting $\mathcal{R}$ to the spacetime region of interest. 
 Let us also recall that for $N=2$ the quantity $||\mathcal{R}-\mathcal{R}^\dag||$ has been recently introduced in the related context of the entanglement-in-time approach \cite{milekhin2025observable} and denoted as \emph{imagitivity} (see Section \ref{sec:timelikeent}). Therein the authors also applied the H\"older’s inequality leading to the concept of $p$-imagitivity, which corresponds to take different Schatten $p$-norms of $\mathcal{R}-\mathcal{R}^\dag$. The authors also proposed the use of pseudo-entropies as timelike entanglement quantifiers. \\

In summary, pure spacetime states correspond to trace one operator with null pseudo entropy ($S_\alpha(\mathcal{R})=0$) but non-trivial singular values ($S_\alpha(\frac{1}{d^N-1}\sqrt{\mathcal{R}^\dag\mathcal{R}})\neq 0$). The first is a statement about \emph{closeness of the history} of the quantum system, the second tells us that $\mathcal{R}$ contains causal information.

Notice also that in contrast to the pure case, the euclidean spacetime state $\mathcal{R}_E$ of Eq.\ \eqref{eq:Reuclid} is a normal operator, namely $[\mathcal{R}_E,\mathcal{R}_E^\dag]=0$. This means that it can be diagonalized in an orthonormal basis, as discussed in \cite{diaz2021path}. Considering their exponential form, it is straightforward to show that these Euclidean spacetime states satisfy a variational principle of stationary quantum action  based on the von-Neumann pseudo-entropy \cite{diaz2025spacetime}.  \\

\begin{figure}[t!]
    \centering
    \includegraphics[width=\linewidth]{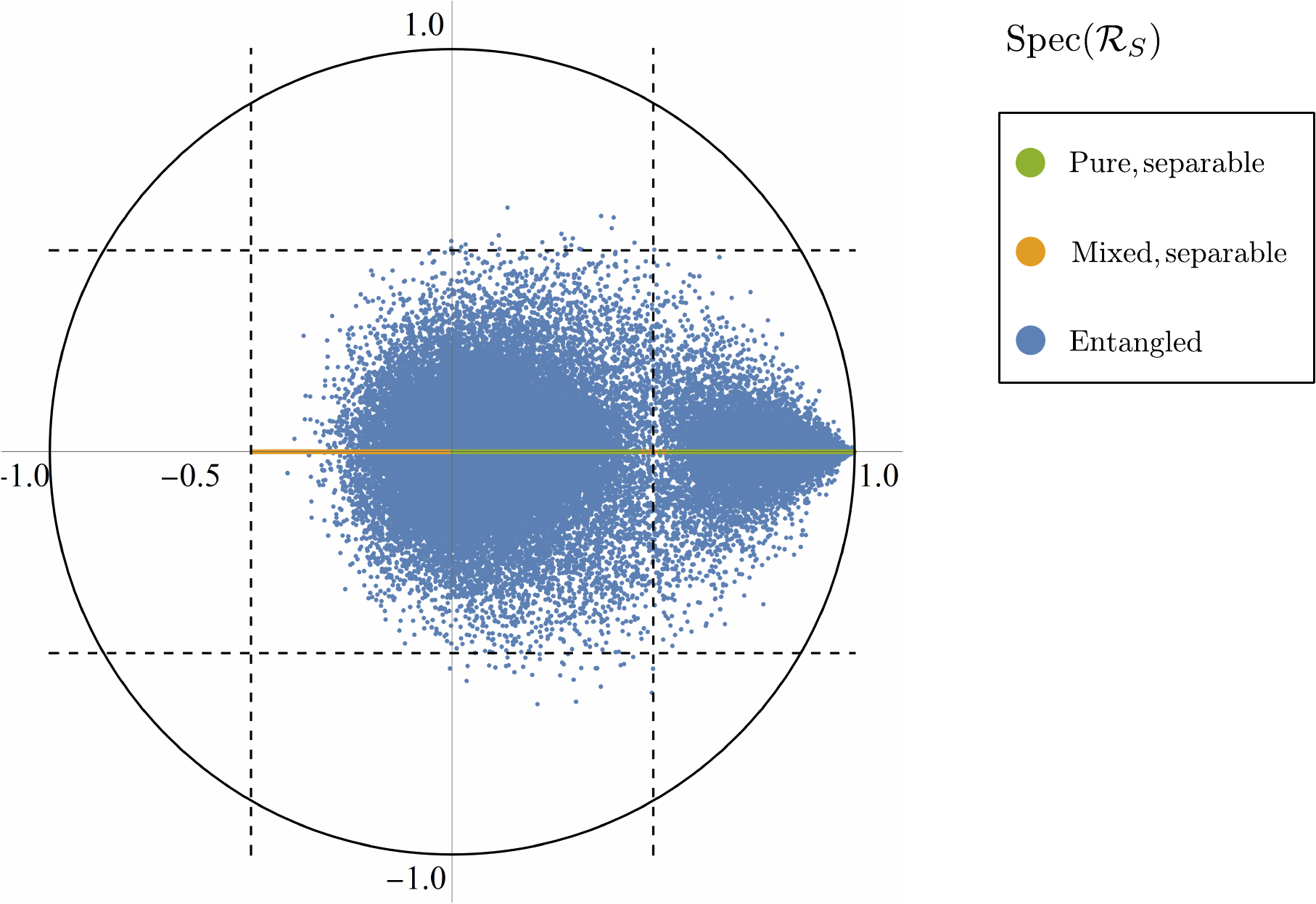}
    \caption{Spectrum in the complex plane of $\mathcal{R}_S={\rm Tr}_E[\mathcal{R}]$ for  $S,E$ qubits and $N=2$. We sampled Haar randomly both global unitaries and  (i) Separable system-environment and pure state (green points)
  (ii) Separable system-environment and mixed state (orange points) (iii) General (entangled) system-environment state (blue points). For each scenario a total of $10^4$ points have been sampled. In case (i) the spectrum is restricted to the real interval $(0,1)$ while in case (ii) to $(-1/2,1)$. In the general case (iii) the spectrum is complex. }
    \label{fig:qubitspectrum}
\end{figure}

\emph{Reduced spacetime states}. Given that Wightman functions are obtained via the traces of 
Corollary \ref{cor:wigthman}, if one is only interested in operators of a  particular ``subsystem'' a reduced spacetime state can be defined as $\mathcal{R}_A={\rm Tr}_B[\mathcal{R}_{AB}]$. Just as in standard QM, the latter satisfies   
${\rm Tr}[\mathcal{R}_A O_A]={\rm Tr}[\mathcal{R}_{AB}O_A]$ for $\mathcal{H}=\mathcal{H}_A\otimes \mathcal{H}_B$. Let us remark that here the partition $A-B$ can be anything, from a traditional system-environment partition to arbitrary regions in spacetime, including a single system separated in two time intervals and the system at a given time with the rest, as in Corollary \ref{cor:schrodheis}. In particular, the latter example shows that marginals of $\mathcal{R}$ can be standard quantum states.

 Once partial traces are considered, the spectral properties of spacetime states become non-trivial with negative and complex eigenvalues naturally arising as we show in Figure \ref{fig:qubitspectrum}. In Appendix \ref{app:spacetimest} we  develop a simple analytical representation of pseudo purities allowing us to prove that
\begin{equation}
    |{\rm Tr}[\mathcal{R}^k]|\leq 1\,.
\end{equation}
This condition also implies $|\lambda|\leq 1$ for every eigenvalue $\lambda$ of
$\mathcal{R}$: any eigenvalue with modulus greater than one would force the purities to become unbounded for arbitrarily large $k$. 
Moreover, using the result of Appendix \ref{app:spacetimest} we show that under a bipartition $A-B$ the isospectrality condition holds, namely, for the nonzero spectra
  \begin{equation}
        \text{spec}(\mathcal{R}_{A})= \text{spec}(\mathcal{R}_{B})\,,
    \end{equation}
which implies 
$S_\alpha(\mathcal{R}_{A})=S_\alpha(\mathcal{R}_{B})$.
In ordinary QM, such an isospectrality property is tied to the fact that a pure bipartite state is a rank-one projector. That is, the Schmidt coefficients determine the spectra of both reduced density matrices. It is therefore notable that an analogous property holds for pure spacetime states across arbitrary bipartitions, even though these objects are generally non-Hermitian and need not be projectors in the usual sense.

Moreover, when $\mathcal{R}_{A}$ (or equivalently $\mathcal{R}_{B}$) is diagonalizable, this isospectrality can be used to decompose the pure spacetime state as
\begin{equation}\label{eq:jointschmidt}
    \mathcal{R}=|\Psi\rangle \langle \Phi|+X_{AB}\,,
\end{equation}
with ${\rm Tr}_{A}[X_{AB}]={\rm Tr}_B[X_{AB}]=0$ and with the pair of states $|\Psi\rangle, |\Phi\rangle$ admitting a \emph{joint Schmidt-like decomposition}:
\begin{equation}
     |\Psi\rangle=\sum_\nu \sqrt{\lambda_\nu}|\nu\rangle_A |\nu\rangle_B\,,\quad  |\Phi\rangle=\sum_\nu \sqrt{\lambda^\ast_\nu}|\tilde{\nu}\rangle_A |\tilde{\nu}\rangle_B
\end{equation}
where the bases are biorthogonal, \begin{equation} _A\langle \tilde{\nu}'|\nu\rangle_A=_B\langle\tilde\nu'|\nu\rangle_B=\delta_{\nu\nu'}\,,
 \end{equation} 
 and $\langle \Phi|\Psi\rangle=\sum_\nu \lambda_\nu={\rm Tr}[\mathcal{R}_{A(B)}]=1$. We see that the eigenvalues $\lambda_\nu$ of the marginals, which determine the corresponding pseudo entropies, are also linked to a Schmidt-like decomposition thus generalizing the standard bipartite case. Furthermore, the decomposition implies 
 \begin{equation}
    {\rm Tr}[\mathcal{R}O_{A(B)}]=\langle \Phi|O_{A(B)}|\Psi\rangle\,.
\end{equation}
This shows that if the global spacetime state is pure we can interpret the expectation values of local operators as \emph{weak values} \cite{aharonov1988result}. Equivalently, the non-orthogonal projector $|\Psi\rangle \langle \Phi|$
(projecting in the direction of $|\Psi\rangle$ vectors expanded as $\alpha |\Psi\rangle+\beta |\Phi_\perp\rangle$, for $\langle \Phi| \Phi_\perp\rangle=0$) can be interpreted as a \emph{biorthogonal purification} of $\mathcal{R}_{A(B)}$. As a matter of fact, when the marginals collapse to standard quantum states (see examples below) we can take $|\Phi\rangle\to |\Psi\rangle$ and, with partitions of equal size, $|\Psi\rangle=|\sqrt{\mathcal{R}}_A\rangle \rangle$ the standard purification. 
In Appendix \ref{app:spacetimest} we also discuss how the decomposition \eqref{eq:jointschmidt} holds under the weaker condition of $\mathcal{R}_{A(B)}$ admitting a compatible Jordan decomposition. We conjecture that a decomposition of the form \eqref{eq:jointschmidt} holds in general. \\

Let us now discuss a simple but particularly relevant case that will allow us 
to make direct contact with other formalisms:  quantum channels for two times ($N=2$).

\begin{theorem}\label{th:quantumchannel}
    The spacetime state corresponding to an arbitrary quantum channel $\mathcal{E}$ for two times is given by \begin{equation}
\begin{split}
     \mathcal{R}_S&={\rm Tr}_E[(\rho\otimes |0\rangle_E\langle 0|)_0 e^{i\tilde{\mathcal{S}}_{SE}}]= \rho_0 J(\mathcal{E})
\end{split}
\end{equation} for $J(\mathcal{E})=\sum_{i,j}|i\rangle \langle j|\otimes \mathcal{E}(|j\rangle \langle i|)$ the Jamiolkowski matrix.\\
\end{theorem}

Let us clarify the statement of the theorem. If we consider a composite system $S+E$ in an initial state $\rho\,\otimes\, |0\rangle_E\langle 0|$ undergoing global unitary evolution, the global spacetime state is given by $(\rho\otimes |0\rangle_E\langle 0|)_0\, e^{i\tilde{\mathcal{S}}_{SE}}$. The reduced spacetime state of the system is then the one in Theorem \ref{th:quantumchannel}. 
Notice that $J(\mathcal{E})$ is replacing the role of $e^{i\tilde{\mathcal{S}}}$ when the channel is not unitary. In fact,  we can also write 
\begin{equation}
   J(\mathcal{E})=\text{SWAP}\sum_k E_k \otimes E_k^\dag\,, 
\end{equation}
 for $E_k=\langle k|U|0\rangle$ the Kraus operators. In particular, one recovers $J(\text{Ad}_U)=e^{i\tilde{\mathcal{S}}}$ for a unitary channel with $E_k\equiv U$  and $\text{Ad}_U(.)=U(.)U^\dag$. 
The fixed time marginals of $\mathcal{R}_S$ yield
\begin{equation}\label{eq:marginalschannel}
    {\rm Tr}_{1}[\mathcal{R}_S]=\rho\,,\quad {\rm Tr}_{0}[\mathcal{R}_S]=\mathcal{E}(\rho)\,,
\end{equation}
for $\mathcal{E}(.)=\sum_k E_k (.) E_k^\dag$ the quantum channel. These relations can be easily verified and generalize Corollary \ref{cor:schrodheis}. Moreover, for pure initial states we can state the following result.

\begin{lemma}\label{lemma:specquantumchannels}
For a pure initial state the non-vanishing spectrum of $\mathcal{R}_S$ satisfies
    \begin{equation}
        \text{Spec}(\mathcal{R}_S)=\text{Spec}(\mathcal{E}(|\psi\rangle \langle \psi|))\,,
    \end{equation}
    implying $S_\alpha(\mathcal{R}_S)=S_\alpha[\mathcal{E}(|\psi\rangle \langle \psi|)]$.
\end{lemma}

Interestingly, this lemma confirms the intuition provided before, whereby the entropy of the spacetime state is a direct quantifier on how much the initial pure state got mixed by the quantum channel, with the pseudo entropy reducing to the standard entropy $S_\alpha[\mathcal{E}(|\psi\rangle \langle \psi|)]$.

Notice that the previous lemma does not imply that $\mathcal{R}_S$ is a standard quantum state. As a matter of fact the corresponding spacetime state takes the form $\mathcal{R}_S=\sum_i |\psi\rangle \langle i| \otimes \mathcal{E}(|i\rangle \langle \psi|)$ (see Appendix \ref{app:spacetimest}) which is not Hermitian in general. Only in the limit where the channel destroys all coherences, i.e., $\mathcal{E}(|i\rangle \langle \psi|)\simeq \delta_{i,\psi} \mathcal{E}(|\psi\rangle \langle \psi|)$ we find 
\begin{equation}\label{eq:collapsingdecoherence}
    \mathcal{R}_S\simeq |\psi\rangle \langle \psi| \otimes \mathcal{E}(|\psi\rangle \langle \psi|)\,.
\end{equation}
This is an important remark: a spacetime state can collapse to a standard quantum state when the decoherence across time is maximal (${\rm tr}_E[U^\dag|\psi\rangle \langle \phi|\otimes |0\rangle_E\langle 0| U]=0$ for $|\phi\rangle \neq |\psi\rangle$).  
Conceptually, this suggests that temporal coherences are precisely what make time different from space at the level of spacetime quantum states. When decoherence across time is sufficiently strong, spacetime states can instead behave as ordinary quantum states on a tensor-product Hilbert space, in this sense capturing a classical-like picture of spacetime. Departures from this fully decohered scenario then provide a mechanism to separate space and time. In Section \ref{sec:legget} we relate this observation to the violation of  Leggett-Garg inequalities \cite{leggett1985quantum}.

\begin{figure}[t!]
    \centering
    \includegraphics[width=\linewidth]{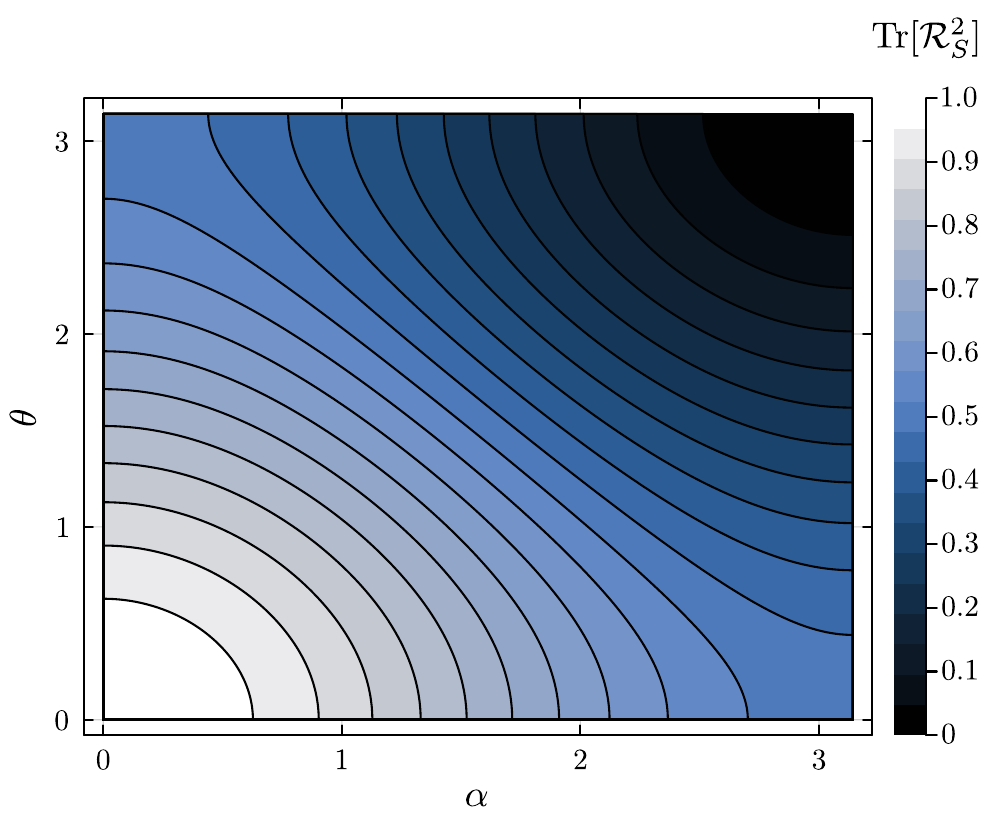}
    \caption{Contour plot of the pseudo-purity
   ${\rm Tr}[\mathcal{R}_S^2]$ with initial state  given by $|\psi(\theta)\rangle=\cos(\theta/4)|00\rangle+\sin(\theta/4)|11\rangle$ and  global unitary $U=(\mathbbm{1}_S\otimes |0\rangle_E\langle 0|+X_S\otimes |1\rangle_E\langle 1|)e^{-i\alpha Y_E/4}$. Here,  $\theta$ parameterizes the initial state entanglement while $\alpha$ parameterizes the entanglement generated by the evolution (and the final state entanglement). In particular, the line $\theta=0$ corresponds to the standard purity of the state of the system after the evolution. In this scenario the pseudo-purity is a real quantity for all values of the parameters since the complex eigenvalues of $\mathcal{R}_S$ come in conjugate pairs.}
    \label{fig:pseudopurity}
\end{figure}

Another  extreme case where a collapse to standard states could arise corresponds to the maximally mixed initial state $\rho_S=\mathbbm{1}/d$, so that 
\begin{equation}\label{eq:collapsingchannel}
    \mathcal{R}_S=J(\mathcal{E})/d\,.
\end{equation}
We thus have $\text{spec}(\mathcal{R}_S)=\text{spec}(J(\mathcal{E}))/d$ leading to real, in general non-positive, eigenvalues. Moreover, since $J$ is the partial transpose of the (unnormalized) Choi state corresponding to the channel, the normalized entanglement-negativity of the channel is the  negativity of $\mathcal{R}_S$. The condition for $\mathcal{R}_S$ becoming a quantum state corresponds to the channel having a  Positive Partial Transpose (PPT), i.e., $\mathcal{R}_S$ is a standard quantum state when the channel is PPT \cite{singh2022detecting}.

In general, if initial entanglement between the system and the environment is present, the simple form of Theorem \ref{th:quantumchannel} does not hold and the spectrum of $\mathcal{R}_S$ becomes complex. Thus, let us analyze a simple example in which the entropic properties of $\mathcal{R}_S$ are still completely intuitive. 
In Figure \ref{fig:pseudopurity} we describe a two qubit two times example for  a unitary $U=(\mathbbm{1}_S\otimes |0\rangle_E\langle 0|+X_S\otimes |1\rangle_E\langle 1|)e^{-i\alpha Y_E/4}$ which provides a Stinespring dilation of the bit flip channel, and initial state $|\psi(\theta)\rangle=\cos(\theta/4)|00\rangle+\sin(\theta/4)|11\rangle$.
In this simple setting one can straightforwardly show that the pseudo-purity is real and given by ${\rm Tr}[\mathcal{R}_S^2]={\rm tr}[\rho_S^2]+{\rm tr}[(\rho'_S)^2]-1=\frac{1}{4}(2+\cos(\theta)+\cos(\alpha))$  where $\rho'_S$ denotes the state of the system after the evolution. Thus the more mixed the reduced initial and final system states are. Equivalently, the more entangled the corresponding global pure states are, the smaller the pseudo-purity of the system spacetime state becomes. In particular, for maximally mixed initial and final states ($\theta=\alpha=\pi$) one obtains ${\rm Tr}[\mathcal{R}_S^2]=0$. \\

Let us now notice that Theorem \ref{th:quantumchannel} gives a straightforward interpretation of convex mixtures of spacetime states (we focus on the case $N=2$). First, if the initial state is mixed, namely $\rho=\sum_l q_l |\psi_l\rangle \langle \psi_l|$  with $q_l$ forming a probability distribution, we can write the corresponding spacetime state as $\mathcal{R}=\sum_l q_l |\psi_l\rangle_0\langle \psi_l| \text{SWAP} U \otimes U^\dag$. This is a convex mixture of closed spacetime states. There is, however, a new possible source of mixedness. Consider now a family of closed spacetime states $\mathcal{R}_k=\rho_0 \text{SWAP}\, U_k \otimes U^\dag_k$ where we consider different unitary evolutions $U_k$ but same initial state. We can define an associated convex mixture 
\begin{equation}\label{eq:convexmix}
   \mathcal{R}=\sum_k p_k \mathcal{R}_k\,, 
\end{equation}
where the $p_k$ form a probability distribution. For a pure initial state the convex mixture involves only closed spacetime states.
More generally, if $\rho_0$ is mixed, one can expand it in a pure-state ensemble and obtain a corresponding double-sum decomposition, so that $\mathcal{R}$ is again a convex mixture of closed spacetime states. 
Interestingly, we can rewrite 
 \begin{equation}\label{eq:convextoJ}
      \mathcal{R}=\rho_0\,\text{SWAP}\sum_k \sqrt{p_k\!}\,U_k \otimes  \sqrt{p_k\!}\,U^\dag_k=\rho_0 J(\mathcal{E})\,,
 \end{equation}
for $\mathcal{E}(\cdot)=\sum_k p_k U_k (\cdot) U_k^\dag$ the quantum channel corresponding to the Kraus operators $E_k=\sqrt{p_k\!}\,U_k$. This means that $\mathcal{R}={\rm Tr}_E[(\rho \otimes |e_0\rangle_E \langle e_0|)_0e^{i\tilde{S}_{SE}}]$ for $|e_0\rangle$ a reference state of the environment chosen to provide a Stinespring dilation of the channel $\mathcal{E}$. The corresponding quantum action is just the action of the closed system-environment system undergoing a global entangling unitary evolution. This in particular implies (see Eq.\ \eqref{eq:marginalschannel}) ${\rm Tr}_{t=1}[\mathcal{R}]=\sum_k p_k U_k\rho U_k^\dag$.
 These simple results show that convex mixtures of closed spacetime states with common initial state are themselves meaningful spacetime states, and that they  can be interpreted as reduced states of a larger closed and ``entangled history''. Their mixedness can thus be understood as arising from entangling evolution in a larger closed description. This closely parallels the standard density matrix formalism where we can think of convex mixtures as a consequence of quantum correlations present on a larger system. We have seen that in our current scenario there are two possible types of sources for mixedness, the standard one coming from the initial state, and the one coming from correlating evolution. One can easily verify that if both these sources are further combined by allowing different initial states we can still understand $\mathcal{R}$ as a marginal of a larger closed spacetime state. We explain this in detail in Appendix \ref{app:spacetimest}. \\

Let us briefly comment on the generalization of the spacetime state of an open system $\mathcal{R}_S$ to multiple time steps. If one assumes that the state is evolving under sequential (memoryless) quantum channels, then the simplest approach is to consider a different ancilla for each time step. This leads to a simple form for $\mathcal{R}_S$ which can be obtained in complete analogy with our $N=2$ derivation and that only involves Kraus operators. If instead the same ancilla is re-used more complicated objects could arise, reflecting the presence of a quantum memory. This would indicate that the formalism could be used to characterize quantum non-Markovianity. \\

Having discussed convex mixtures of spacetime states we briefly consider examples of spacetime coherences. In the density matrix formalism, coherences appear as non-diagonal terms in $\rho$, which for pure states correspond to quantum superpositions such as $|\psi\rangle=\alpha |0\rangle +\beta |1\rangle$. This type of coherence is directly reflected in the corresponding spacetime state, with e.g., $\mathcal{R}=\rho_0 e^{i\tilde{\mathcal{S}}}= |\alpha|^2 |0\rangle_0\langle 0| e^{i\tilde{\mathcal{S}}}+|\beta|^2 |1\rangle_0\langle 1| e^{i\tilde{\mathcal{S}}}+\alpha \beta^\ast|0\rangle_0\langle 1| e^{i\tilde{\mathcal{S}}}+\alpha^\ast \beta |1\rangle_0\langle 0| e^{i\tilde{\mathcal{S}}}$. The coherence is thus reflected in the appearance of  orthogonal spacetime transitions, such as $|0\rangle_0\langle 1| e^{i\tilde{\mathcal{S}}}$, multiplied by complex numbers. This is just the usual case seen from the spacetime perspective. On the other hand, in the current formalism we can consider  more general scenarios. As a concrete example, consider an ancilla $C$ so that $h=h_C\otimes h_S$ and a system, undergoing a global controlled operation $U=\sum_a |a\rangle_C \langle a| \otimes W_a$. The corresponding spacetime state is $\mathcal{R}=\sum_{a,b}(c_{ab}|a\rangle_C \langle b|\otimes \rho_S)_0 e^{i\tilde{\mathcal{S}}}$ where we assumed a separable initial state. Considering  $N=2$ for simplicity, we can write
$\mathcal{R}=\sum_{a,b}(c_{ab}|a\rangle_C \langle b|\otimes \mathbbm{1})_C\text{SWAP}_C (\rho \otimes \mathbbm{1})_S\text{SWAP}_S \,U \otimes U^\dag$. If we now expand the controlled unitary, and acting with $U^\dag$ on the initial bra of the control system we obtain 
\begin{equation}\label{eq:coherent}
    \mathcal{R}=\sum_{a,b,c}c_{ab}|ac\rangle_C \langle cb| \otimes \mathcal{R}_{S,cb}\,,
\end{equation}
 where we have defined $\mathcal{R}_{S, cb}=(\rho \otimes \mathbbm{1})_S\text{SWAP}_S W_b \otimes W^\dag_c$ acting on the system. Equation~ \eqref{eq:coherent} already has a coherent-like expansion with the non-diagonal $\mathcal{R}_{S, cb}$ (for $b\neq c$) indicating coherences across evolutions. The equation also contains all the information of the ancilla, but if we only care about its final state we can ``close the first ancilla leg'' corresponding to $(h_{C})_0$ to write ${\rm Tr}_{C_0} [\mathcal{R}]=\sum_{a,b,c}c_{ab}|a\rangle_C \langle b| \otimes \mathcal{R}_{S,ba}$. This now reflects a standard control-induced coherence but where the system terms are either spacetime states or their non-diagonal generalization. To be even more concrete, consider a standard qubit ancilla in the superposition $\alpha|0\rangle+\beta |1\rangle$. We obtain
\begin{equation}
\begin{split}
     {\rm Tr}_{C_0}[ \mathcal{R}]&=|\alpha|^2 |0\rangle_C\langle 0| \otimes \mathcal{R}_{S,00}+|\beta|^2 |1\rangle_C\langle 1| \otimes \mathcal{R}_{S,11}\\&+\alpha \beta^\ast|0\rangle_C\langle 1| \otimes \mathcal{R}_{S,01}+\alpha^\ast \beta |1\rangle_C\langle 0|\otimes \mathcal{R}_{S,10}\,.
\end{split}
\end{equation}
Let us finally notice that if we also project the ancilla onto the coherent basis $|+\rangle=\frac{|0\rangle+|1\rangle}{\sqrt{2}}$ (i.e., we are post-selecting the $|+\rangle$ reading), we can write $
        2{\rm Tr}_{C} [(\mathbbm{1}\otimes |+\rangle \langle +|)\mathcal{R}]=|\alpha|^2 \mathcal{R}_{S,00}+|\beta|^2 \mathcal{R}_{S,11}+\alpha \beta^\ast \mathcal{R}_{S,01}+\alpha^\ast \beta \mathcal{R}_{S,10}\,,
$
which is a coherent spacetime state without the ancilla ``bookkeeping''. 
In summary, evolution controlled by quantum states is reflected in global spacetime states that have a natural  coherent-like expansion.  Tracing over the full ancilla gives instead a statistical mixture of spacetime states as the ones we described above.

 It is also interesting to remark that as we considered arbitrary controlled operations, the previous discussion encompasses the \emph{quantum switch} \cite{chiribella2013quantum} $(\alpha|0\rangle+\beta |1\rangle) \otimes |\psi\rangle_S\to (\alpha|0\rangle \otimes U_1U_2|\psi\rangle_S+\beta |1\rangle \otimes U_2U_1|\psi\rangle_S)$ which corresponds to $W_1=U_1U_2$, $W_2=U_2U_1$. The quantum switch, a demonstrated resource for quantum computation, is a basic example of indefinite causal order \cite{castro2018dynamics}.
 We see that this scenario is naturally embedded in the formalism and induces a coherent behavior of spacetime states. Notice, however, that in principle one could explore the possibility of more general coherences involving $\mathcal{R}^\dag$ or,  even more generally, arbitrary foldings of time, corresponding to other indefinite causal structures. Let us also recall that the quantum switch falls into a particular class of indefinite causal orders and is unable to break causal inequalities \cite{araujo2015witnessing}. We make a few additional comments on this topic in Section \ref{sec:otocs} where we discuss general OTOCs and in the Conclusions \ref{sec:conclusions}. \\

\begin{figure}[t!]
    \centering
    \includegraphics[width=\linewidth]{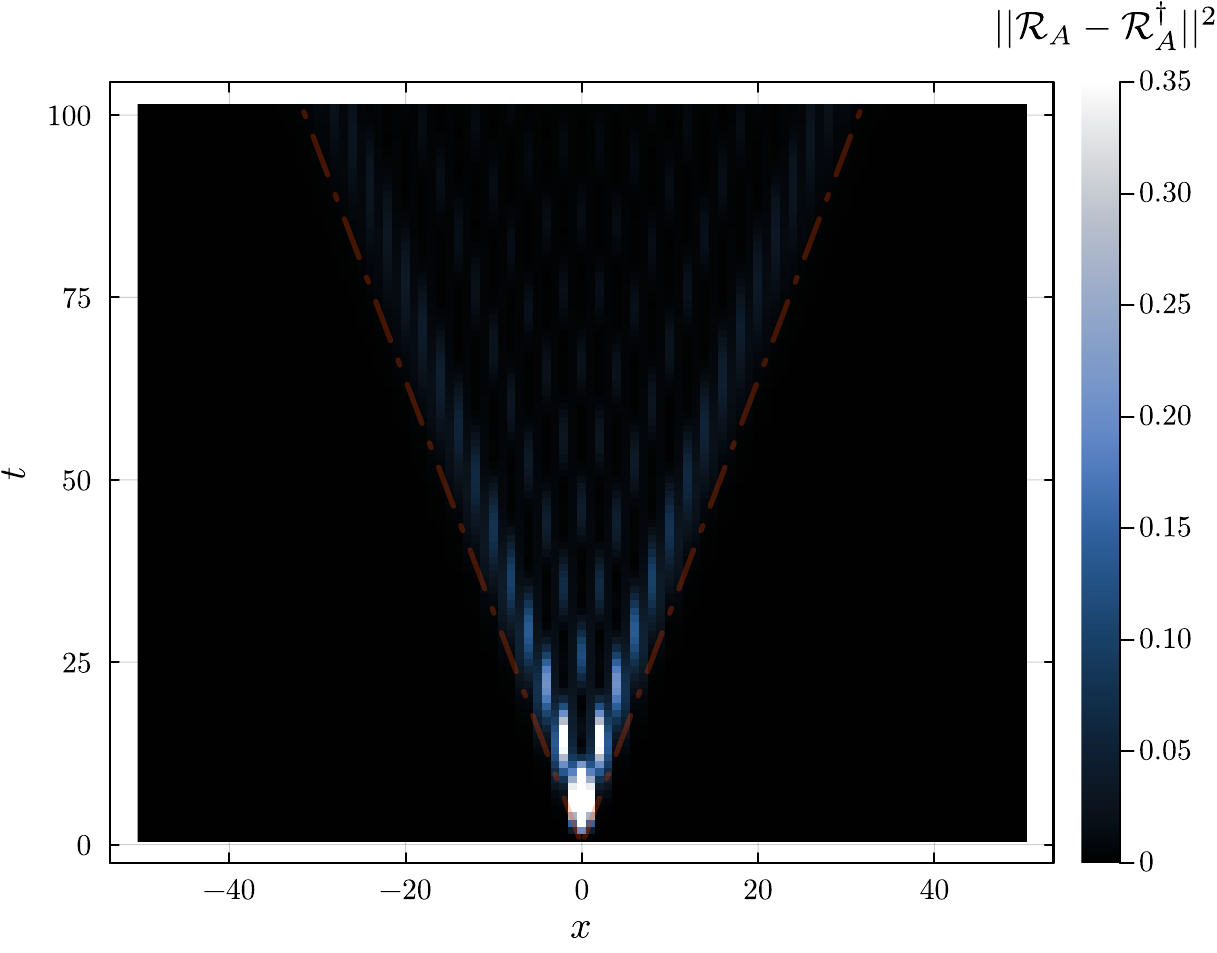}
    \caption{Contour plot of the imagitivity of the spacetime state $\mathcal{R}_A$ corresponding to the ground state of a critical antiferromagnetic Heisenberg model under bipartition $A=(0,0)\cup (t,x)$ $B=\text{rest}$. Here we are considering the same number of temporal and spatial slices $N=L+1=101$ with $\epsilon=0.05$. The red dash-doted lines correspond to the speed of low-energy modes associated with an effective infrared CFT description.}
    \label{fig:lightcone}
\end{figure}

Let us finally discuss a more general scenario beyond the simple system-environment bipartition. As we mentioned before, and as we used in Eq.\ \eqref{eq:coherent}, the bipartitions of spacetime we can take are completely arbitrary. It is interesting to leverage this in order to expand our previous comments about causality. In particular, when we write ${\rm tr}\big[\rho\, [O^{(2)}(t_2),O^{(1)}(t_1)]\big]={\rm Tr}\big[(\mathcal{R}-\mathcal{R}^\dag)O^{(1)}_{t_1}O^{(2)}_{t_2}\big]$ it is sufficient to restrict the spacetime state to the shared support (in spacetime) of the operators $O^{(1)}_{t_1}, O^{(2)}_{t_2}$, which we denote as $A$. So, using the Cauchy-Schwarz inequality as in Eq.\ \eqref{eq:imagitivboundweak} we obtain
 \begin{equation}\label{eq:imagitivbound}
     \begin{split}
         |{\rm tr}\big[\rho\, [O^{(2)}(t_2),O^{(1)}(t_1)]\big]|&\leq ||O^{(1)}|| \,||O^{(2)}||\, ||\mathcal{R}_A-\mathcal{R}_A^\dag||\,.
     \end{split}
 \end{equation}
As an example, consider a quantum system extended in space, e.g., a spin chain. One can then choose two arbitrary points in spacetime and quantify their causal relation through $||\mathcal{R}_{p,q}-\mathcal{R}^\dag_{p,q}||$, for $\mathcal{R}_{p,q}={\rm Tr}_{\overline{p \cup q}}[\mathcal{R}]$ with $p,q$ spacetime points. To illustrate a typical imagitivity behavior we consider the example of 
a one-dimensional chain under the critical antiferromagnetic Heisenberg Hamiltonian~\cite{cerezo2017factorization} 
\begin{equation}
    H=\sum_{i=-L/2}^{L/2}(X_iX_{i+1}+Y_iY_{i+1}+Z_iZ_{i+1})\,,
\end{equation} with open boundary conditions and $L+1$ sites. We then consider 
$p=(0,0)$, i.e., initial time slice and middle of the chain, and numerically compute the imagitivity of the spacetime state reduced to the points $p$ and $q$ for a generic $q$. For the initial state we consider the ground state of the Hamiltonian. The results are shown in Figure \ref{fig:lightcone} as a contour plot over all points $q=(t,x)$ within a time window $T=5$ for $\epsilon=0.05$ ($N=101$) and within the $L+1=101$ spatial sites. 
The imagitivity exhibits a light-cone-like causal structure, showing that $\mathcal{R}_{p,q}$ is essentially a standard density matrix for ``spacelike'' separations. 
The light-cone edge is
approximately concentrated within the effective speed determined by the low-energy regime: the energies of the first excited states at a given ``momentum'' $q$ and in the infinite chain limit are given by $\omega(q)=2\pi |\sin(q)|$ \cite{des1962spin}. Here $e^{iq}$ are the eigenvalues of the site translation operator along the chain. In the continuum limit one can expand around the gapless momentum $q=\pi+k$ to find $\omega(k)\approx v |k|$ with $v=2\pi$, consistently with a massless dispersion relation and with an effective CFT description. Let us stress, however, that the imagitivity is not determined solely by the low-energy sector of the spectrum, so it is not a priori clear that this should be the velocity controlling the causal edge, even if we find a good numerical agreement. 
Let us also notice that this causal behavior is in perfect agreement with results in the literature regarding Green functions in spin chains and their spacetime profile, usually described in relation to dynamical factors and neutron scattering experiments \cite{lee2026benchmarking}. In fact, a vanishing imagitivity between a pair of spacetime points, as the one we find outside the lightcone, also implies a vanishing Green function for that pair. 
We explain the details of our computation and numerical approximations in the Appendix \ref{app:numerics}.

\subsection{Forward and backward evolution}\label{sec:fbevolution}
In the previous, we introduced two types of QAs: the standard one $\mathcal{S}$ whose definitions only involves single step evolutions and $\tilde{\mathcal{S}}$ which defined spacetime states and also includes an evolution operator $U^\dag(T)$. Here we show that we can think of 
the extra factor $U^\dag(T)$ as arising from ignoring the backward in time evolution of an action involving both forward and backward in time evolutions. We discuss in Section \ref{sec:unifying} how this is related to a Schwinger–Keldysh contour and other formalisms.

Consider the Hilbert space $\mathcal{H}_{\text{ext}}\equiv h^{\otimes 2N}\simeq \mathcal{H}\otimes \mathcal{H}$, where for convenience we have defined a spacetime Hilbert space corresponding to $2N$ copies so that we may split it in two spacetime Hilbert spaces $\mathcal{H}$.  
We define the following time translation operator
\begin{equation}\label{eq:extendedP}
    e^{i\epsilon \mathcal{P}_{\text{ext}}}=e^{i\epsilon (\mathcal{P}\otimes \mathbbm{1}-\mathbbm{1}\otimes \mathcal{P})}\,\text{SWAP}_{N-1,N}\,.
\end{equation}
The operator is defined such that all operators on the first (second) copy of $\mathcal{H}$ are translated forward (backward) in time. The SWAP operator in the middle of the slices guarantees that a closed loop is considered. With this definition we can introduce the extended QA.

\begin{definition}\textbf{Extended Quantum action (EQA)}. \label{def:extendedqa}
\begin{equation}
    e^{i\mathcal{S}_{\text{ext}}}= e^{i\epsilon \mathcal{P}_{\text{ext}}}\,e^{-i\epsilon (\mathcal{K}\otimes \mathbbm{1}-\mathbbm{1}\otimes \mathcal{K})}\,U^\dag_{N-1}U_{N}\,.
\end{equation}
\end{definition}
Let us add a small comment on the factor $U^\dag_{N-1}U_{N}$. Its only purpose is to cancel the last evolution operator  of the forward evolution and the first of the backward evolution. This will allow us to easily match the convention of other schemes in the literature but can also be omitted without changing the main properties of $\mathcal{S}_{\text{ext}}$. With this definition we can state the following result.

\begin{figure}[t!]
    \centering
    \includegraphics[width=\linewidth]{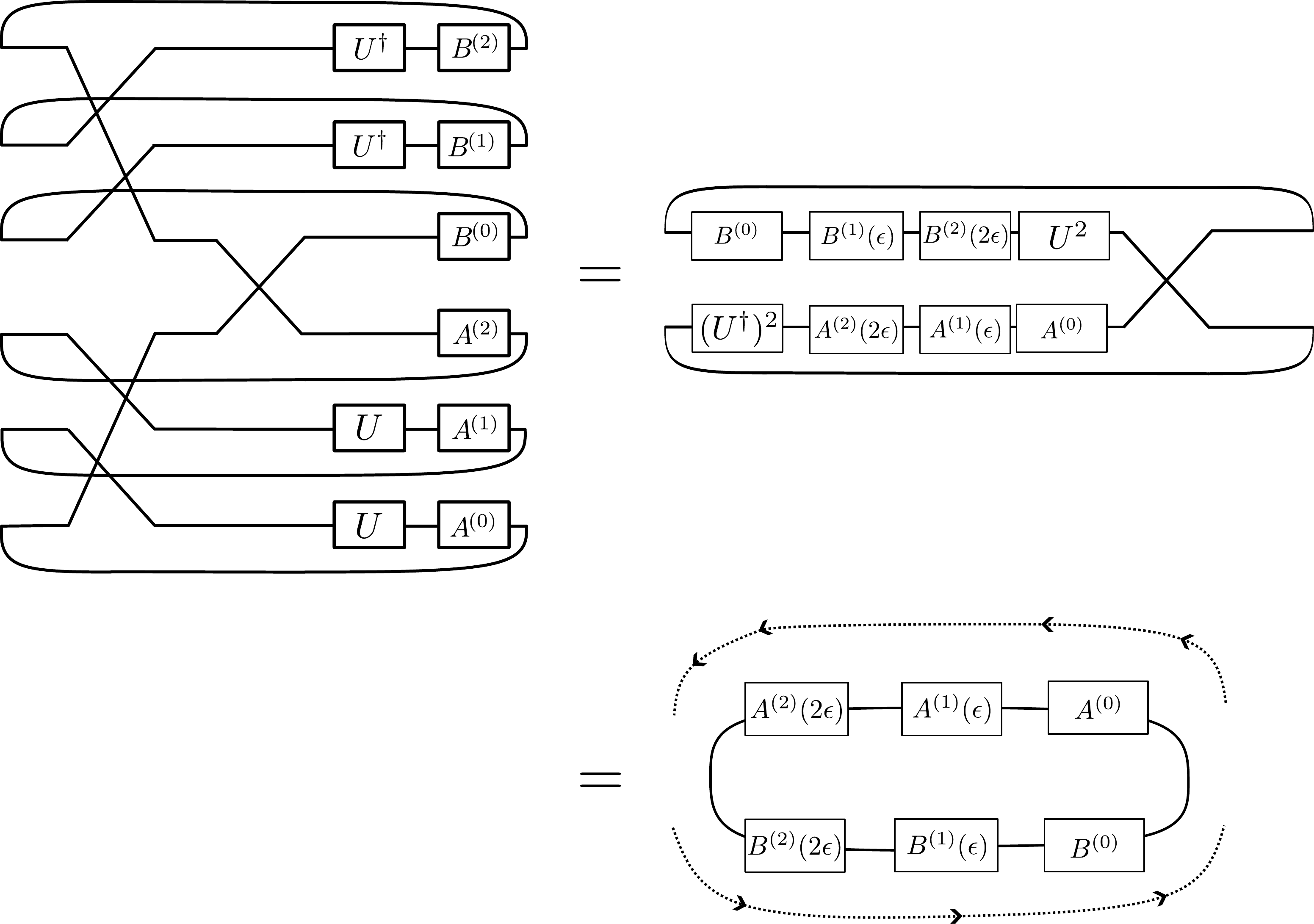}
    \caption{Graphical representation of Theorem \ref{th:extendedQA}. On the last part of the r.h.s. we have indicated  the direction of time (dotted lines) to emphasize that both a forward and backward evolution are involved. }
    \label{fig:theorem2}
\end{figure}
\begin{theorem}\label{th:extendedQA}
The correlators of the extended action yield        
        \begin{equation}
        \begin{split}
            &{\rm Tr}\left[e^{i\mathcal{S}_{\text{ext}}}\,\big(\otimes_{t=0}^{N-1}A^{(t)}_t\big)\otimes \big(\otimes_{t=0}^{N-1}B^{(t)}_t\big)\right]\\&={\rm tr}\left[\left(\bar{T}\smallprod\nolimits_{t=0}^{N-1} B^{(t)}(\epsilon t)\right)\,\left(\hat{T}\smallprod\nolimits_{t=0}^{N-1} A^{(t)}(\epsilon t)\right)\right]\,.
        \end{split}
        \end{equation}
    \end{theorem}
Interestingly,  we can consider this result as a version of Theorem \ref{th:theorem1} where times first goes forward and then returns to the initial slice. We depict an example in Figure \ref{fig:theorem2}. It is also instructive to notice that rearranging the times on the second copy of $\mathcal{H}$ in decreasing order starting from top $e^{i\epsilon \mathcal{P}_{\text{ext}}}$ becomes equal to the operator $e^{i\epsilon \mathcal{P}}$ for $2N$ times (so we may more explicitly understand Theorem \ref{th:extendedQA} in analogy with Theorem \ref{th:theorem1}). This is explained in Appendix \ref{app:spacetimest}.

Consider now that all the insertions in the backward direction are identities i.e., $B^{(t)}=\mathbbm{1}_t$. The corresponding trace is equal to the one we obtain with the SQA, thus leading directly to the following corollary.
    \begin{corollary}
    If we label $\mathcal{H}_{\text{ext}}\equiv \mathcal{H}_1 \otimes \mathcal{H}_2$, with $\mathcal{H}_2$ ($\mathcal{H}_1$) corresponding to backward (forward) in time evolution, we can write
        \begin{equation}
            e^{i\tilde{\mathcal{S}}}={\rm Tr}_{\mathcal{H}_2}[e^{i\mathcal{S}_{\text{ext}}}]\,,\quad e^{-i\tilde{\mathcal{S}}}={\rm Tr}_{\mathcal{H}_1}[e^{i\mathcal{S}_{\text{ext}}}]\,.
        \end{equation}
    \end{corollary}
Notice that while on the r.h.s. all the time evolution operators correspond to a single time slice, after the partial trace the factor $U^\dag_0(T)$ emerges. We can also use this corollary to recover the spacetime states as 
\begin{equation}
   \mathcal{R}= {\rm Tr}_{\mathcal{H}_2}[\rho_0 e^{i\mathcal{S}_{\text{ext}}}]\equiv  {\rm Tr}_{\mathcal{H}_2}[\mathcal{R}_{\text{ext}}]\,,
\end{equation}
where we have defined $\mathcal{R}_{\text{ext}}=\rho_0 e^{i\mathcal{S}_{\text{ext}}}$. This shows that we might study the properties of spacetime states directly from $\mathcal{R}_{\text{ext}}$.  In addition, this extended version of the spacetime state satisfies
\begin{equation}\label{eq:Rextexpect}
    \begin{split}
          &{\rm Tr}\left[\mathcal{R}_{\text{ext}}\big(\otimes_{t=0}^{N-1}A^{(t)}_t\big)\otimes \big(\otimes_{t=0}^{N-1}B^{(t)}_t\big)\right]\\&={\rm tr}\left[ \rho\left(\bar{T}\smallprod\nolimits_{t=0}^{N-1} B^{(t)}(\epsilon t)\right)\,\left(\hat{T}\smallprod\nolimits_{t=0}^{N-1} A^{(t)}(\epsilon t)\right)\right]\,.
    \end{split}
\end{equation}
Similarly, we have
\begin{equation}
    \mathcal{R}^\dag={\rm Tr}_{\mathcal{H}_1}[\mathcal{R}_{\text{ext}}]\,,
\end{equation}
showing that both spacetime states with time and anti-time orderings can be recovered from the same object.

We discuss additional properties of spacetime states with ``folded'' evolution in Section \ref{sec:otocs} in relation to OTOCs, showing, in particular, that many properties of $\mathcal{R}$ hold for $\mathcal{R}_{\text{ext}}$. Let us here just notice that 
 for a quantum channel and $N=2$ one finds 
\begin{equation}\label{eq:extquantumchannel}
    \mathcal{R}_{\text{ext}, S}=(\rho \otimes \mathbbm{1}\otimes \mathbbm{1}\otimes \mathbbm{1})\,e^{i\epsilon \mathcal{P}_{\text{ext}}}\sum_k E_k\otimes \mathbbm{1}\otimes \mathbbm{1}\otimes E_k^\dag\,,
\end{equation}
which follows in complete analogy with Theorem \ref{th:quantumchannel}. \\

\def\arraystretch{1.4}
 \begin{table*}[t!]
 \footnotesize
     \centering
 \begin{tabular}{|c|c|c|}
     \hline
         Approach & Main objects & Relation to spacetime QM \\\hline\hline
         Spacetime quantum mechanics& Spacetime Hilbert space $\mathcal{H}=h^{\otimes N}$; $\mathcal{H}^{ \otimes 2}$  & {} \\
        (SQM; this work and \cite{diaz2023spacetime, diaz2021path, diaz2025spacetime}) & Quantum actions $\mathcal{S}$, $\tilde{\mathcal{S}}$, $\mathcal{S}_{\text{ext}}$  & - \\
          {} & Spacetime states $\mathcal{R}=\rho_0 e^{i\tilde{\mathcal{S}}}$, $\mathcal{R}_{\text{ext}}=\rho_0 e^{i\mathcal{S}_{\text{ext}}}$  & {} \\\hline\hline
           Path integral formulation & Classical trajectories & Basis $|\textbf{q}\rangle=\otimes_t |q_t\rangle$ of $\mathcal{H}$\\ 
        (PI; \cite{feynman1948space, schwinger1961brownian, keldysh2024diagram}) & Classical actions & Matrix elements of quantum actions\\ 
       {} & Feynman PIs & Evaluation of $\rm{Tr}[\rho_0e^{i\mathcal{S}}\dots]$ in the $|\textbf{q}\rangle$ basis \\ 
         {} & Schwinger-Keldyish PIs & Evaluation of $\rm{Tr}[\rho_0e^{i\mathcal{S}_\text{ext}}\dots]$ in the $|\textbf{q}\rangle$ basis\\ \hline
           Quantum states over time & Spacetime Hilbert space $\mathcal{H}$  & Same Hilbert space\\
         (QSOT; \cite{horsman2017can}) & Jamiolkowski matrix $J(\mathcal{E})=\sum_{i,j}|i\rangle \langle j|\otimes \mathcal{E}(|j\rangle \langle i|)$  & $J(\mathcal{E})\equiv e^{i\tilde{\mathcal{S}}}$ for quantum channels ($N=2$)\\ {} & States over time ($\mathcal{H}=h_A\otimes h_B$) $\rho_{AB}=J(\mathcal{E})\star \rho_A$  & Example 0 of $\star$-product: $\rho_{AB}=\mathcal{R}$\\
         {} & Star product $\star$  & Example 1 of $\star$-product: $\rho_{AB}=\frac{\mathcal{R}+\mathcal{R}^\dag}{2}$\\
         {} & {} & Example 2 of $\star$-product: $\rho_{AB}=\rho_0^{-1/2}\,\mathcal{R}\,\rho_0^{1/2}$\\\hline
         Pseudo density matrices & Spacetime Hilbert space $\mathcal{H}$  &  Same Hilbert space\\
        (PDMs; \cite{fitzsimons2015quantum, fullwood2024operator}) & Pseudo density matrix $\mathcal{R}_{\text{pdm}}$  &  
         $\mathcal{R}_{\text{pdm}}=\Phi_{\text{UT}}(\mathcal{R})$ for $\Phi_{\text{UT}}$ CPTP Unimodal twirl 
        \\\hline
         Superdensity operators  & $\mathcal{L}(\mathcal{H})$  & Operator space of $\mathcal{H}$\\
        (SO; \cite{cotler2018superdensity}) & Superdensity operator $\varrho:\mathcal{L}(\mathcal{H})\to \mathcal{L}(\mathcal{H})$  &  $ \varrho=\frac{1}{\dim(\mathcal{H})}\mathcal{M}(\mathcal{R}_{\text{ext}})$  \\ 
       {} & {}  &   for $\mathcal{M}$  partial transpose and realignment map \\ \hline
          Page and Wootters mechanism  & Hilbert space $h_{\text{PW}}=h_T\otimes h$  & Second quantization: $\text{\upshape  Sym} \oplus_{n=0}^\infty h^{\otimes n}_{\text{\upshape PW}}\cong \mathcal{H}$\\
        (PW; {\cite{page1983evolution, giovannetti2015quantum, boette2016system}}) & Universe operator $J$  & $J$ first quantization version of $\tilde{\mathcal{S}}$\\
          & Physical states $|\Psi\rangle_{\text{PW}}$ & Reduced single particle operators\\{}& {}& for one particle states and free theories\\\hline
          Timelike entanglement approach  & Two times Hilbert space $\mathcal{H}=h\otimes h$  & Spacetime Hilbert space for $N=2$ \\ \cite{milekhin2025observable} & $T_{AB}$  & $T_{AB}={\rm Tr}_{\,\overline{A \cup B}}\,[\mathcal{R}^\dag]$ \\{} & Entropies of $T_{AB}$  & Entropies of reduced spacetime states\\\hline
     \end{tabular}
     \caption{Emergence of different spacetime approaches from spacetime QM.}
     \label{tab:unified}
 \end{table*}
\normalsize

Let us  add that while considering expectation values  of projectors  for $\mathcal{R}$ led to $Q_{KD}$ (Eq.\ \eqref{eq:KDquasiprobability})
if one considers only projectors in Eq.\ \eqref{eq:Rextexpect} one recovers general \emph{decoherence functionals}   $d(\boldsymbol{\alpha},\boldsymbol{\beta})={\rm tr}[C^\dag_\alpha\rho C_\beta]
$
for $C_\alpha=\alpha_{t_0}(t_0)\alpha_{t_1}(t_1)\dots$ where the different operators are evolving in the Heisenberg picture. Each operator $\alpha_{t_i}$ is a projector corresponding to ``propositions'' about the system \cite{griffiths1984consistent, gell2010quantum, isham1994quantum}. That is, the decoherence functional is the central object in the  coherent history approach which postulates that $d(\boldsymbol{\alpha},\boldsymbol{\beta})$  is a fundamental entity from which other physical quantities may be derived. Then, a series of criteria (such as Hermiticity and positivity for diagonal entries) for what constitutes a valid physical decoherence functional can be introduced. Interestingly, the decoherence functional is also a central object in the proposal \cite{isham1994quantum} which makes use of the spacetime Hilbert space $\mathcal{H}\otimes \mathcal{H}$ to characterize decoherence functionals through a generalization of Gleason theorem of standard QM \cite{isham1994classification}. We refer the reader to \cite{diaz2025spacetime} for a comparison between the spacetime approach and \cite{isham1994quantum}.

\section{Unifying spacetime approaches}\label{sec:unifying}

Having presented the main formalism and explored important properties of spacetime states, we are now in a position to introduce the main unifying results of the manuscript. We will establish how different proposals to generalize QM to a spacetime symmetric form can be rederived from the general framework of spacetime QM described in Section \ref{sec:formalism}. Our unifying scheme is summarized in Table \ref{tab:unified} and the proofs of our theorems are explicitly spelled out throughout the text for ease of readability unless indicated otherwise.

\subsection{Path integral formulation}\label{sec:PIs}

The first spacetime approach to QM, and by far the most used and known, is Feynman's formulation \cite{feynman1948space}. Feynman's approach replaces the canonical quantum formalism with a ``sum over histories'' weighted by $e^{iS_{\text{cl}}}$, with $S_{\text{cl}}$ the classical action. Since the Lagrangian formulation makes spacetime symmetries manifest a spacetime approach to QM follows. Moreover, in the case of QFTs both space and time appear on equal footing under the PI formulation as the history or trajectory of a field corresponds to its configurations in spacetime.

The price to pay by using Feynman's approach is  to abandon the familiar Hilbert space structure. Notably, the spacetime QM formulation of Section \ref{sec:formalism} can provide PIs with a Hilbert space embedding, whereby particular evaluations of  ${\rm Tr}[e^{i\mathcal{S}}\mathcal{O}]$ are time-sliced PIs.
Before making direct contact with standard PIs let us provide a direct intuition on why it is so. Consider a generic standard Hilbert space $h$ with a basis $|i\rangle$. Then a product-in-time basis of $\mathcal{H}$ is $|i_0i_1\dots i_{N-1}\rangle={\textstyle\bigotimes}_{t=0}^{N-1} |i_t\rangle$. We can think of this state as representing a \emph{trajectory} $\textbf{i}=(i_0,i_1,...,i_{N-1})$. Then one can write a trace in $\mathcal{H}$ as 
\begin{equation}
    {\rm Tr}[e^{i\mathcal{S}}\mathcal{O}]=\sum_{\textbf{i}}\;\langle i_0i_1\dots i_{N-1}|e^{i\mathcal{S}}\mathcal{O} | i_0i_1\dots i_{N-1}\rangle\,.
\end{equation}
 This is already a ``sum over histories'' parameterized by $\textbf{i}$. 
 Moreover, the matrix elements of the exponential of the quantum action provide precisely the classical phase associated with each trajectory, so that the Feynman path integral emerges directly from the spacetime Hilbert-space formalism. We now show this explicitly.

\subsubsection{Feynman PI}

Let us first showcase how PIs emerge from the spacetime approach in the familiar example of a single particle. We stress that  the  considerations on spacetime QM discussed in Section \ref{sec:formalism} hold for infinite dimensional Hilbert spaces as long as the traces of the operators involved are well-defined \footnote{Several results of this subsection are based on \cite{diaz2021path}, where further mathematical details have been discussed}.

For a single quantum particle defined by $[q,p]=i$ the spacetime formalism is characterized by $N$ independent modes satisfying $[q_t,p_{t'}]=i \delta_{tt'}$. The copies in time of the Hilbert space correspond to independent modes at each instant. This defines a basis $\{|\textbf{q}\rangle=\otimes_t |q_t\rangle\}$ so that $\hat{q}_t|\textbf{q}\rangle=q_t|\textbf{q}\rangle$. 
Similarly, we might use the momentum basis $\{|\textbf{p}\rangle=\otimes_t |p_t\rangle \}$. The corresponding completeness relations in $\mathcal{H}$ can then be written as
\begin{equation}
    \int \prod_{t=0}^{N-1}dq_t|\textbf{q}\rangle \langle \textbf{q}|=\int \prod_{t=0}^{N-1}dp_t|\textbf{p}\rangle \langle \textbf{p}|=\mathbbm{1}\,,
\end{equation}
while the overlap between eigenstates is 
\begin{equation}\label{eq:trajoverlap}
    \langle \textbf{p}|\textbf{q}\rangle=\frac{1}{\sqrt{2\pi }^{N}}e^{-i\sum_{t=0}^{N-1} p_t q_t}\,.
\end{equation}
This shows that the spacetime Hilbert space of a single particle is essentially the space of possible trajectories. This is the first contact of the spacetime approach with the ``sum over histories'' of Feynman: in the spacetime approach the sum arises from completeness relations or, equivalently, particular spacetime trace evaluations.

Let us now discuss how one can make use of standard properties of bosonic modes to find an explicit expression for $\mathcal{S}$. The idea is to introduce Fourier in time modes, namely consider the Fourier transform in time of the original operators. This corresponds to a simple canonical/Bogoliubov transformation where the Fourier frequencies are $\omega_n= \frac{2\pi n}{T}$ and which 
can be used to write 
\begin{equation}\label{eq:Plegendre}
 \mathcal{P}=\sum_t p_t \dot{q}_t\,.
\end{equation}
Here we have defined a discrete derivative $\dot{q}_t=\sum_{t',n} \frac{i\epsilon\omega_n}{N} e^{i\omega_n (t-t')\epsilon}q_{t'}$ induced by the Fourier transformation (see Appendix \ref{app:PI} for details). One can easily verify that
\begin{equation}
    e^{i\epsilon \mathcal{P}}q_t e^{-i\epsilon \mathcal{P}}=q_{t+1}\,,    e^{i\epsilon \mathcal{P}}p_t e^{-i\epsilon \mathcal{P}}=p_{t+1}\,,
\end{equation}
in accordance with Definition~\ref{def:eip}. Notably, the logarithm of the concatenation of SWAPs have the form of the \emph{classical Legendre transform}. Moreover, for the time translations within slices we find $\mathcal{K}=\sum_t H_t$ where the sum over $t$ follows from the tensor product of the exponentials in Definition~\ref{def:translationwithin}. We also recall that $H_t$ denotes an operator acting as $H$  on $h_t$ and trivially otherwise such that $\mathcal{K}=H\otimes \mathbbm{1}\otimes ...+ \mathbbm{1}\otimes H \otimes ...+...+\mathbbm{1}\otimes \mathbbm{1} \otimes ... \otimes H$. In addition, considering that for time-independent Hamiltonians $[\mathcal{P},\mathcal{K}]=0$, we can write $\mathcal{S}=\epsilon \mathcal{P}-\epsilon \mathcal{K}$  leading directly to 
\begin{equation}\label{eq:actionparticle}
 \mathcal{S}= \epsilon\sum_t \left[p_t \dot{q}_t-\frac{p_t^2}{2m}-V(q_t)\right]\,,
\end{equation}
where we considered a standard Hamiltonian $H=\frac{p^2}{2m}+V(q)$ for concreteness. 
We see that \emph{the quantum action has the form of the action of classical mechanics in phase-space variables}. 
This is rather surprising: 
while in classical mechanics the action is a functional of possible trajectories, the quantum counterpart  $\mathcal{S}$ is an operator acting on $\mathcal{H}$ with the time index indicating on which Hilbert space $h_t$ each operator acts and defined to yield Wightman functions (Theorem \ref{th:theorem1}). 
A priori there is no apparent reason for these two objects to have the same functional form. A nice conceptual interpretation is obtained by developing the spacetime counterpart of classical mechanics, where the role of the action in both quantum and classical mechanics is to define the mismatch between translations across and within time slices. This result is developed in Section \ref{sec:classical} independently from the PI formulation.

We can now make direct contact with Feynman's approach via spacetime expectation values. For concreteness let us consider the transition $|q_i\rangle_0 \langle q_f|$ and an operator $\mathcal{O}=q_{t_1} q_{t_2}$. By using the completeness relation in the position basis we obtain
\begin{equation}\label{eq:twopointpi}
\begin{split}
       {\rm Tr}\big[|q_i\rangle_0 \langle q_f|e^{i\mathcal{S}}q_{t_1} q_{t_2}\big]= \int \prod_{t=1}^{N-1}dq_t\, \langle \textbf{q}_f|e^{i\mathcal{S}} |\textbf{q}_i\rangle\, q_{t_1} q_{t_2}\,,
\end{split}
\end{equation}
where we have defined $\textbf{q}_i=(q_i,q_1,\dots, q_{N-1})$ and $\textbf{q}_f=(q_f,q_1,\dots, q_{N-1})$. Interestingly, this expression has already a PI-like form, with a sum over histories inherited from the completeness relation and the insertion of $q_{t_1} q_{t_2}$. We also remark that on the l.h.s. $q_{t_1} q_{t_2}$ are operators while on the r.h.s. they are c-numbers. Our previous results also guarantee that this is equal to the two-point Wightman function, namely, Corollary \ref{cor:wigthman} states that
\begin{equation}\label{eq:twopoincorr}
    {\rm Tr}\big[|q_i\rangle_0 \langle q_f|e^{i\mathcal{S}}q_{t_1} q_{t_2}\big]=\langle q_f,T|\hat{T}q(\epsilon t_1)q(\epsilon t_2)|q_i\rangle\,,
\end{equation}
for any number of time slices. 
To make direct contact with Feynman's PIs we simply need to give an expression for the matrix elements of the exponential of the QA. Considering that its operator expression \eqref{eq:actionparticle} already has the form of the classical action  one can easily relate the matrix elements of $e^{i\mathcal{S}}$ 
with $e^{iS_{\text{cl}}}$  for small $\epsilon$. We show this explicitly in Appendix \ref{app:PI}. The result is: 
\small
\begin{equation}
\begin{split}
    \langle \textbf{q}_f|e^{i\mathcal{S}} |\textbf{q}_i\rangle
    \!&=\!\int \!\prod_{t=0}^{N-1}\frac{dp_t}{2\pi}\, e^{i\sum_t \epsilon \big[p_t\left(\frac{q_{t+1}-q_t}{\epsilon}\right)-\frac{p_t^2}{2m}-V(q_t)\big]\big|_{q_0=q_i}^{q_N=q_f}}\\ 
    &=\frac{1}{(\sqrt{2\pi i \epsilon/m})^N}\, e^{i\sum_t \epsilon\, \left(\frac{1}{2}m \dot{q}^2_t-V(q_t)\right)\big|_{q_0=q_i}^{q_N=q_f}}\,,
\end{split}
\end{equation}
\normalsize
which holds up to second order in $\epsilon^2$ and where we introduced the notation (for classical variables) $\dot{q}_t\equiv \frac{q_{t+1}-q_t}{\epsilon}$. We have thus recovered the classical action evaluated along the trajectory $\textbf{q}$ with $q_0=q_i$, $q_N=q_f$ from the matrix elements of the QA. This also means, that if we write $\mathcal{R}_{q_i,q_f}\equiv |q_i\rangle_{0} \langle q_f|\,e^{i\mathcal{S}}$ we have
\begin{equation}\label{eq:shadowPI}
   \langle \textbf{q} |\mathcal{R}_{q_i,q_f} |\textbf{q}\rangle\propto  e^{iS_{\text{cl}}}\,,
\end{equation}
for $S_{\text{cl}}$ the classical action evaluated in the trajectory $\textbf{q}$ with border conditions $q_i,q_f$ (we are omitting a delta $\delta(q_0-q_i)$  in the proportionality constant). This has a clear interpretation: the ``shadows'' of $\mathcal{R}_{q_i,q_f}$ in the classical-like basis of trajectories give the standard phase of the PI formulation.

Now, if we put this result back in \eqref{eq:twopointpi} we recover the time-sliced PI expression for the two point Wightman function of Eq.\ \eqref{eq:twopoincorr}. The main insight is that expressions of the form ${\rm Tr}[|q_i\rangle_0\langle q_f|e^{i\mathcal{S}}\dots]$, which define spacetime QM in Hilbert space, when evaluated in particular basis yield standard PIs. We can state the following result.
\begin{theorem} (Informal) Consider the spacetime position basis $\{|\textbf{q}\rangle=|q_0\rangle\otimes\cdots\otimes|q_{N-1}\rangle\}$, where each element represents a discrete trajectory 
$\textbf{q}=(q_0,\dots,q_{N-1})$. Then the following matrix elements of $e^{iS}$ reproduce the classical phase
\[
\langle \textbf{q}\,|\,e^{iS}\,|\textbf{q}\rangle \propto  e^{iS_{\rm cl}[\textbf{q}]}\,.
\]
 Moreover, spacetime traces evaluated in this basis become sums/integrals over trajectories weighted by $e^{iS_{\rm cl}[\textbf{q}]}$, i.e., standard Feynman path integrals.\\
\end{theorem}

Let us now remark that the generalization to an arbitrary number of particles and to fields is completely straightforward. Indeed, for a standard algebra $[q_i,p_j]=i\delta_{ij}$ the spacetime formalism imposes instead $[q_{ti},p_{t'j}]=i\delta_{ij}\delta_{tt'}$ so that all modes are on an equal footing. 
 See also the discussion in Sections \ref{sec:QFT} and \ref{sec:fermions}. On the other hand, while we focused in the position basis, and on the phase-space PI, quite interestingly, some changes of bases in Hilbert space are related to different PI representations. For example, one might use coherent states to evaluate the trace and the coherent state PI follows. On the other hand, Matsubara-like expansions are related to the use of Fourier in time basis, which from the Hilbert space point of view correspond to a non-local in time Bogoliubov transformation. This is particularly useful when an euclidean action is considered (see Eq.\ \eqref{eq:Reuclid} and comments below).
Furthermore, the trace might be evaluated in basis for which a PI representation is not possible. We refer the reader to \cite{diaz2021path} for details. Moreover, fermionic PIs can also be obtained form the formalism by similar means \cite{diaz2025spacetime}. All the previous results follow directly from the spacetime formalism showing that many PI techniques can be easily accommodated in the extended Hilbert space.

Let us also make a small example on how similar ideas hold for a  continuum time formalism. If we consider a continuum time from the start, $\mathcal{H}$ is defined by $[q(t),p(t')]=i \delta(t-t')$. Then the generator of time translations takes the form
\begin{equation}
    \mathcal{P}=\int dt\, p(t)\dot{q}(t)\,,
\end{equation}
which makes the relation with the Legendre transform explicit. The spacetime algebra implies
\begin{equation}
    e^{i\tau \mathcal{P}}O(t) e^{-i\tau \mathcal{P}}=O(t+\tau)\,,
\end{equation}
in agreement with the discrete time case. As an example of QA, take $H=\frac{p^2}{2m}+\frac{1}{2}m\omega^2 q^2$ leading to $\mathcal{S}=\int dt\, \left( p(t)\dot{q}(t)-\frac{p^2(t)}{2m}-\frac{1}{2}m\omega^2 q^2(t)\right)$. For general (non quadratic) Hamiltonians there are some additional subtleties to be considered. These are related to the fact that for continuum time, single step translations become meaningless and Theorem \ref{th:theorem1} must be modified accordingly. This is not much of an obstacle as one can introduce a time scale $\tau$ and consider trace-expressions in the small $\tau$ limit. One can show that in this limit standard QM is recovered. This also provides a different perspective on the continuum time PIs \cite{diaz2021path}.

\subsubsection{Schwinger–Keldysh PI}\label{sec:SKPI}
In the previous we made exclusive use of the QA. Let us now discuss the role of the SQA.

Standard PIs describe ``in''-``out'' transitions, namely quantities of the form $\langle q',T|O(t)|q\rangle$, where we have taken a single operator for simplicity. However, in many physical applications one requires $\langle q'|O(t)|q\rangle$ which in particular allows one to write conventional expectation values of the form ${\rm tr}[O(t)\rho]$ \footnote{In many field theoretical applications at zero temperature a PI suffices as one is interested in vacuum correlation functions, which can be obtained by the quotient of two Feynman PIs. At finite temperature and out of equilibrium standard PI techniques are  no longer feasible \cite{kamenev2005course}. }. The Schwinger–Keldysh PI provides an expression for this expectation value. Let us provide an elementary derivation of this PI in the example of a single particle and taking $O=\int dq\, O(q)|q\rangle \langle q|$ for simplicity. We have,
\small
\begin{equation}\label{eq:SKPI}
\begin{split}
    \langle q'|O(t)|q\rangle&=\langle q'|U(t_i,t_f)U(t_f,t)OU(t,t_i)|q\rangle \\
    &=\int dq_f \langle q'|U(t_i,t_f)|q_f\rangle \langle q_f|U(t_f,t)OU(t,t_i)|q\rangle
    \\
    &=\int dq_f \int_{q^-(t_i)=q'}^{q^-(t_f)=q_f} \mathcal{D}q^-  \\&\qquad\int_{q^+(t_i)=q}^{q_+(t_f)=q_f}\mathcal{D}q^+ e^{iS[q^+]-iS[q^-]}\,O(q^+(t))\,,\end{split}
\end{equation}
\normalsize
where we used $U^\dag(t,t_i)=U(t_i,t_f)U(t_f,t)$. Notice that we can think of the evolution starting at $t_i=0$, moving to $t_f$ and then coming back $t_i$. For this reason this type of PI is said to correspond to a closed time loop, while the standard PI correspond to the $q_+$ variable going from $t_i$ to $t_f$. Notice also that the action corresponding to the $q_-$ trajectory has a minus sign that arises from $\langle q'|U(t_i,t_f)|q_f\rangle=\langle q_f|U(t_f,t_i)|q'\rangle^\ast$.

The closed time loop associated with the Schwinger–Keldysh PI has the form of the r.h.s. of Figure \ref{fig:theorem2}. 
We will now  show that indeed the PI emerges from evaluating  ${\rm Tr}[|q\rangle_0\langle q'|e^{i\mathcal{S}_{\text{ext}}}\dots]$ in the position basis. First let us make a few useful definitions to compare with \eqref{eq:SKPI}. We will write $$(\textbf{q}^+,\textbf{q}^-)=(q^+_0, q_1^+,\dots,q^+_{N-1},q^{-}_0,q^{-}_{1},\dots, q^{-}_{N-1})\,,$$
and 
$$(\textbf{p}^+,\textbf{p}^-)=(p^+_0, p_1^+,\dots,p^+_{N-1},p^{-}_{0},p^{-}_{1},\dots, p^{-}_{N-1})\,,$$ both of which are vectors with  $2N$ entries.  The completeness relations of position and momentum in $\mathcal{H}\otimes \mathcal{H}$ can be written as
\begin{align}
      &\int d^{N}q^+ d^{N}q^-\, |\textbf{q}^+,\textbf{q}^-\rangle \langle \textbf{q}^+,\textbf{q}^-|\nonumber\\&=\int d^{N}p^+ d^{N}p^-\, |\textbf{p}^+,\textbf{p}^-\rangle \langle \textbf{p}^+,\textbf{p}^-|=\mathbbm{1}\,,
\end{align}
where the integration measures are given by $d^{N}q^{\pm}=\prod_{t=0}^{N-1}dq^\pm_t$,  $d^{N}p^\pm=\prod_{t=0}^{N-1}dp^\pm_t$.

Just as for Feynman PIs we are interested in the matrix elements of the exponential of the action, in the current case given by $\langle \textbf{q}'^+\textbf{q}^-|e^{i\mathcal{S}_{\text{ext}}} |\textbf{q}^+\textbf{q}^-\rangle$
for ${\textbf{q}'^+=(q',q^+_1,\dots, q^+_{N-1})}$ and $\textbf{q}^+=(q,q^+_1,\dots, q^+_{N-1})$ since within the trace the first entries are fixed by the operator $|q\rangle_0 \langle q'|$. A direct calculation, explicitly provided in Appendix \ref{app:PI}, leads to 
\begin{equation}
    \langle \textbf{q}'^+\textbf{q}^-|e^{i\mathcal{S}_{\text{ext}}} |\textbf{q}^+\textbf{q}^-\rangle=\frac{1}{(\sqrt{2\pi i \epsilon/m})^{2N}}e^{i S_{\text{cl}}[\textbf{q}^+]-i S_{\text{cl}}[\textbf{q}^-]}\,,
\end{equation}
with $S_{\text{cl}}[\textbf{q}^{\pm}]$ the discrete time classical action of the corresponding variables and  border conditions $q^+_0=q, q^+_N=q_f$ and $q^-_{-1}=q', q^-_{N-1}=q_f$. This equation, holding for small $\epsilon$, reflects the structural form of $\mathcal{S}_{\text{ext}}$ in Definition~\ref{def:extendedqa} which might be associated with the difference between two standard QAs with the addition of border conditions among the backward and forward evolution.

On the other hand, by making use of the completeness relations and for $O(q_t)$ in the first half of $\mathcal{H}\otimes \mathcal{H}$ we can write
\begin{align}
&{\rm Tr}[|q\rangle_0 \langle q'|e^{i\mathcal{S}_{\text{ext}}}O(q_t)]\nonumber\\&= \int dq_fd^{N-1}q^+ d^{N-1}q^- \langle \textbf{q}'^+\textbf{q}^-|e^{i\mathcal{S}_{\text{ext}}} |\textbf{q}^+\textbf{q}^-\rangle\, O(q^+_t)\,,\label{eq:trSKPI}
\end{align}
which in combination with \eqref{eq:matrixelsext} lead to the time-sliced Schwinger–Keldysh PI of \eqref{eq:SKPI}. The insertion of more operators within the ``loop'' $t_i\to t_f\to t_i$ correspond to adding functionals within the PI just as in the standard Schwinger-Keldysh PI, and in agreement with the Wightman functions of Theorem \ref{th:extendedQA}.

\subsubsection{Path integrals for general systems}
Let us make a few comments regarding arbitrary quantum systems for which a classical analogue is not as straightforward. In this scenario it is natural to employ the so-called generalized coherent states (GCS). These states are defined by a Lie algebra $\mathfrak{g}$ whose Cartan-Weyl decomposition defines a Highest weight state $|hw\rangle\in h$. We assume that the group $G\equiv e^\mathfrak{g}$ is compact and has elements $T(\Omega)\in G$ parameterized by $\Omega$. Then the GCSs are given by $|\Omega\rangle=T(\Omega)|hw\rangle$. Notice that these are the states employed to map the Hilbert space of the system to a phase-space in the Stratonovich-Weyl formulation \cite{brif1999phase}.
They also provide an overcomplete basis of $h$ satisfying
\begin{equation}
    \int d\mu(\Omega)|\Omega\rangle \langle \Omega|=\mathbbm{1}\,.
\end{equation}
A common example of GCS are the spin coherent states defined by the algebra of $\mathfrak{su}(2)$. In particular, standard (bosonic) coherent states are recovered from the non-compact Heisenberg-Weyl algebra.

We can use GCS in $\mathcal{H}$, which define separable-in-time basis of the form $|\boldsymbol{\Omega}\rangle=|\Omega_0\rangle \otimes |\Omega_1\rangle \otimes \dots |\Omega_{N-1}\rangle$. The completeness relation reads
\begin{equation}
    \int d^N\mu(\boldsymbol{\Omega})|\boldsymbol{\Omega}\rangle\langle \boldsymbol{\Omega}|=\mathbbm{1}\,.
\end{equation}
Then, one can recover the coherent state PI of any quantum system by noting that
\small
\begin{equation}
\begin{split}
     \langle \boldsymbol{\Omega}_f|e^{i\mathcal{S}}|\boldsymbol{\Omega}_i\rangle&=\langle \Omega_f\Omega_1 ...\,  \Omega_{N-1}|e^{i\epsilon \mathcal{P}}e^{-i\epsilon\mathcal{K}}|\Omega_i\Omega_1...\,\Omega_{N-1}\rangle\\&=\langle \Omega_1...\,\Omega_{N-1}\Omega_f|\otimes_t e^{-i\epsilon H}|\Omega_i\Omega_1... \,\Omega_{N-1}\rangle\\&=
     \prod_t \langle \Omega_{t+1}|e^{-i\epsilon H}|\Omega_t\rangle\,.
\end{split}
\end{equation}
\normalsize
Now under the usual approximation one can write $\langle \Omega_{t+1}|e^{-i\epsilon H}|\Omega_t\rangle\approx e^{-i\epsilon H[\Omega_t,\Omega_{t+1}]} \langle \Omega_{t+1}|\Omega_t\rangle$ (where $\Omega_{t+1},\Omega_t$ replace the ladder operators that define a normally ordered $H$). Then one can consider the logarithm of the product in time to obtain the sum $\sum_t \log[\langle \Omega_{t+1}|\Omega_t\rangle]$ which is related to the Legendre transform of the action (e.g., for spin one obtains the Berry phase; for bosons the Legendre transform in coherent-state variables and so on).

In particular, we find it useful to discuss the special case of fermionic systems. The standard fermionic PI relies on the use of Grassmann variables to define fermionic coherent states that have the same form as the bosonic ones. It was recently shown in \cite{diaz2025spacetime} that one can introduce the corresponding  states in the spacetime approach leading to the fermionic PI just as we showed for GCS. However, Hilbert space properties and manipulations of the fermionic quantum action operator reproduce the conventional PI techniques without the need to introduce Grassmann variables. See Section \ref{sec:fermions} for a brief introduction to the spacetime formalism for fermionic systems.

Let us also mention that one might develop a phase space formulation of spacetime QM itself. That is,  given a map from Hilbert space $h\to X$, with $X$ a phase-space, one can map $\mathcal{H}\to \times_t X_t$, with $\times_t X_t$ the Cartesian product of a single time phase-space $X$. As a consequence, the inner products ${\rm Tr}[e^{i\mathcal{S}}\mathcal{O}]$ are naturally represented in $\times_t X_t$ as sum over histories. See also Section \ref{sec:classical} where we develop the classical counterpart of the spacetime formalism.

\subsection{Quantum states over time}\label{sec:QSOT}
In \cite{horsman2017can} the following question was posed: \emph{Can a quantum state over time
resemble a quantum state at a
single time?} Therein the authors pose the problem of defining a quantum state over time (QSOT) $\rho_{AB}$ such that 
\begin{equation}\label{eq:stateovertimecond}
    \rho_A={\rm Tr}_B[\rho_{AB}]\,,\quad \rho_B={\rm Tr}_A[\rho_{AB}]= \mathcal{E}(\rho_A)
\end{equation}
 for $\rho_A$ an initial quantum state, $\mathcal{E}$ a quantum channel and $\rho_{AB}$ acting on $\mathcal{H}=h_A\otimes h_B$. Here $h$ corresponds to the standard Hilbert space of the system while the tensor product indicates two time slices.

A central quantity in their discussion is the Jamiolkowski matrix $J(\mathcal{E})=\sum_{i,j}|i\rangle \langle j|\otimes \mathcal{E}(|j\rangle \langle i|)$. The authors propose to define a star product $\star$ such that $\rho_{AB}=J(\mathcal{E}) \star \rho_A$. 
They prove that no choice of the star product can satisfy all the criteria that would lead to $\rho_{AB}$ being a quantum state in the usual sense. Notice also that fixing the marginals \eqref{eq:stateovertimecond} does not fix  $\rho_{AB}$ which allows different possible definitions of $\star$.

Let us now discuss the relation between this scheme and spacetime QM. First of all we notice that the Hilbert space is the same in both constructions, as a tensor product structure across slices is imposed. To further the connection, let us focus for simplicity in the case of $N=2$, in which case we write $h_A\equiv h_0$, $h_B\equiv h_1$. Then, as we showed in Section \ref{sec:pureststates} we have $e^{i\tilde{\mathcal{S}}}\equiv J(\mathcal{E})$, where we are considering the SQA corresponding to a quantum channel (see Theorem \ref{th:quantumchannel} and the related discussion). We see that one of the main elements of the QSOT is also present in the spacetime formalism. 
On the other hand, the corresponding spacetime state is given by $\mathcal{R}=(\rho\otimes \mathbbm{1})J(\mathcal{E})$ and satisfies the conditions \eqref{eq:stateovertimecond} (see Eq.\ \eqref{eq:marginalschannel}). In this sense spacetime states $\mathcal{R}$ are a particular choice of QSOT where one relaxes the Hermiticity condition (not considered in \cite{horsman2017can})
in favor of allowing for timelike correlators. In fact, as we previously discussed the spacetime state $\mathcal{R}$ is unique, since it not only satisfies the conditions \eqref{eq:stateovertimecond} but also yields Wightman functions when the expectation value of operators at different times are considered.

To compare with other proposals in \cite{horsman2017can} we need to use Hermitian operators. One possibility is to consider the operator
\begin{equation}\label{eq:simistar}
    \rho_{AB}\equiv  (\rho^{-1/2}\otimes \mathbbm{1})\mathcal{R}  (\rho^{1/2}\otimes \mathbbm{1})
\end{equation}
with $\rho^{-1/2}$ indicating pseudo-inverse if $\rho$ is not invertible. This congruence transformation recovers one of the star products considered in \cite{horsman2017can} corresponding to  the state over time $\rho_{AB}=(\rho^{1/2}\otimes \mathbbm{1})J(\mathcal{E}) (\rho^{1/2}\otimes \mathbbm{1})$. Notice that since $\rho_{AB}$ is Hermitian, it follows that $\mathcal{R}$ has real eigenvalues for $\rho$ invertible.

Another natural choice is
\begin{equation}\label{eq:qsot}
    \rho_{AB}\equiv  \frac{\mathcal{R}+\mathcal{R}^\dag}{2}= \frac{(\rho\otimes \mathbbm{1})J(\mathcal{E})+J(\mathcal{E})(\rho\otimes \mathbbm{1})}{2}\,.
\end{equation}
This operator  also arises in our discussion about PDMs where we show that $\rho_{AB}$  is equal to a quantum channel acting on $\mathcal{R}_{\text{ext}}$ (see Section \ref{sec:PDM}), and in the context of Leggett-Garg inequalities (see Section \ref{sec:legget}). This choice of QSOT also corresponds to another star product considered in \cite{horsman2017can}. However, for multiple time slices this definition of QSOT differs from the proposal in \cite{lie2025multipartite} based on a repeated use of a star product (see also comments on section \ref{sec:PDM}).

Interestingly, the discussion of Section \ref{sec:pureststates} which shows that $\mathcal{R}$ can collapse to a standard quantum state under particular conditions shows that $\rho_{AB}$ can indeed become a quantum state as well, although this can only happen if the system is open.
Notice also that as a consequence of Corollary \ref{cor:schrodheis} these 3 definitions of state over time satisfy the generalized version of Eq.\ \eqref{eq:stateovertimecond} to an arbitrary number $N$ of time slices.

There is one final proposal in \cite{horsman2017can} inspired by Wigner quasiprobabilities in discrete settings that makes use of a point basis $\{K_i\}$ satisfying ${\rm tr}[K_i K_j]=d \delta_{ij}$ and $\sum_i K_i=d\mathbbm{1}$. If we write   $r_A(i)=\frac{1}{d}\,\mathrm{Tr}(\rho_A K_i)$ and 
    $r_{B|A}(j|i)=\frac{1}{d}\,\mathrm{Tr}\big(J(\mathcal{E})K_{i}^A\otimes K_j^B\big)$ the corresponding definition of state over time is \begin{equation}\label{eq:rhoW}
    \rho^{(W)}_{AB}=\sum_{i,j} r_A(i)\,r_{B|A}(j|i)\; K_i\otimes K_j.
\end{equation}
Consider instead the expansion of the spacetime state. One finds
$
    \mathcal{R}=\sum_{i,i',j} r_A(i)\,r_{B|A}(j|i')\;\big(K_i K_{i'}\big)\otimes K_j
$. To recover the previous QSOT one would need a map  $K_i K_{i'}\to \delta_{ii'}K_i$. One can see, e.g., by considering vanishing  linear combinations of products of $K_i$, that in general a linear map implementing this transformation cannot exist. While we leave open the possibility of providing an interpretation to $\rho^{(W)}_{AB}$ via the spacetime formalism, there is a simple (but rather  artificial) way to recover both $\mathcal{R}$ and  $\rho^{(W)}_{AB}$ from a single object: one can duplicate $h_A\to h_{A_1}\otimes h_{A_2}$ to write $$\mathcal{R}=\Phi_1( \rho_{A_2}\otimes J(\mathcal{E}))\,,\quad \rho^{(W)}_{AB}=\Phi_2( \rho_{A_2}\otimes J(\mathcal{E}))$$ where the Jamiolkowski operators act on $h_{A_1}\otimes h_B$.
Here the
maps $\Phi_{i}:\mathcal{L}(h_{A_1}\otimes h_{A_2})\to \mathcal{L}(h_A)$ are defined so that $\Phi_1(K_i \otimes K_{i'})=K_i K_{i'}$ and $\Phi_2(K_i \otimes K_{i'})= \delta_{ii'} K_i$. The first map can be trivially found as a SWAP between $A_1,A_2$ and a partial trace over $A_2$ while the second follows immediately from the orthogonality relation as $\Phi_2(X)=\frac{1}{d^2}\sum_i{\rm Tr}\big[(K_i\otimes K_i)\,X\big]\;K_i$. \\

The central message of \cite{horsman2017can} is that there is no universally satisfactory prescription $\rho_{AB}=J(\mathcal E)\star\rho_A$ that yields an ordinary  quantum state while simultaneously reproducing the correct marginals \eqref{eq:stateovertimecond} and retaining a sensible notion of temporal correlations. From the spacetime QM perspective this tension is natural: the canonical object is the spacetime state $\mathcal R=(\rho_A\otimes\mathbbm 1)J(\mathcal E)$, which satisfies \eqref{eq:stateovertimecond} but is generically non-Hermitian precisely because it encodes timelike correlators and causality. Hermitian QSOTs as the one considered in \cite{horsman2017can} are then obtained by applying additional  prescriptions, such as the similarity-transformed star product \eqref{eq:simistar} or the symmetrized choice \eqref{eq:qsot}, which  discard part of the temporal information. Moreover, in the regime where $\mathcal{R}$ collapses to a standard quantum state (e.g., under sufficiently strong decoherence; see Section.~\ref{sec:pureststates}), no Leggett-Garg violation is possible  (this is proven in Section \ref{sec:legget}) in line with the no-go intuition of \cite{horsman2017can}.

\subsection{Pseudo density matrix}\label{sec:PDM}

\subsubsection{Introduction to PDMs}
A proposal to generalize QM to spacetime tackles the problem by considering measurements at different points in spacetime, or equivalently, at different events. To be more explicit, in \cite{fitzsimons2015quantum} the authors identify that a standard density matrix of distinguishable subsystems (e.g., separated in space) is completely specified by the expectation value of the product of
measurement outcomes in space. If one wants to generalize this to time, one can define an object usually denoted as Pseudo-density matrix (PDM), specified by the expectation value of product of measurements in spacetime. In the following we will refer to this scheme as the PDM approach, a formalism explored by many authors, see e.g., \cite{marletto2020non, song2024causal,fullwood2025quantum}. Let us also remark that most of the PDM approach has been developed for qubits. The case of arbitrary dimensional systems is still under active research \cite{fullwood2024operator}. It has also been shown recently how to define PDMs  for gaussian states of continuous variable systems \cite{zhang2020different}. In the following, we will make a definition of PDM that holds for arbitrary systems, and that is compatible with the most recent developments.

Pseudo-density matrices deal with the expectation value of the product of outcomes from measurements separated in spacetime. Let us recall some basic facts. First of all, notice that the case of spatial separations, or subsystems, corresponds to conventional mean values and can be obtained from a conventional density matrix. In fact, if we consider a product operator $O=O^{(0)}\otimes O^{(1)} \otimes ... \otimes O^{(m-1)}$ we can write
$${\rm tr}[\rho O]=\sum_{\{\textbf{i}\}}{\rm tr}\Big[\rho \,P^{(0)}_{i_0}\otimes ... \otimes P^{(m-1)}_{m-1}\Big ]\,\lambda^{(0)}_{i_0}\dots \,\lambda^{(m-1)}_{m-1}\,,$$ 
for $O^{(k)}=\sum_i \lambda_i^{(k)}P_i^{(k)}$, namely, $\lambda^{(k)}_i$ are the possible outcomes of measuring $O^{(k)}$ and $P_i^{(k)}$ are the corresponding  orthogonal projectors. At the same time, according to the standard axioms of QM ${\rm tr}\big[\rho \, P^{(0)}_{i_0}\otimes \dots \otimes P^{(m-1)}_{m-1}\big]=p_{i_0\wedge  \dots \wedge i_{m-1}}$, the joint probability of obtaining the results labeled by $\{\textbf{i}\}=\{i_0,\dots ,i _{m-1}\}$. In summary,   for space-like separated operators
\begin{equation}
   {\rm tr}[\rho \, O^{(0)}\otimes ... O^{(m-1)}]= \mathbb{E}_{\text{space}}\big[O^{(0)},..., O^{(m-1)}\big]\,,
\end{equation}
where $\mathbb{E}_{\text{space}}\big[O^{(0)}, O^{(1)}, ..., O^{(m-1)}\big]$ denotes the expectation value of the product of the outcomes of measuring the different $O^{(i)}$ operators at a given time. Specifying these averages for a complete set of operators completely defines the quantum state of the system.
As is well known, the expectation value can exhibit non-classical behavior, codified in the state $\rho$, which can be e.g., entangled across space.

Pseudo density matrices (PDMs) \cite{fitzsimons2015quantum} generalize the previous considerations to scenarios where time-like separations are allowed.
Let us then describe the case of $N$ different operators separated in time. The PDM proposal is to define an operator $\mathcal{R}_{\text{pdm}}$ acting on $\mathcal{H}$, the same Hilbert space employed in the spacetime QM approach, such that
\begin{equation}\label{eq:PDMexp}
    {\rm Tr}[\mathcal{R}_{\text{pdm}} O^{(0)}\otimes ...  O^{(N-1)}]= \mathbb{E}_{\text{time}}\big[O^{(0)}, ..., O^{(N-1)}\big]\,,
\end{equation}
where the tensor product denotes time (each $h_t$ might still be a composed Hilbert space) and $\mathbb{E}_{\text{time}}\big[O^{(0)}, ..., O^{(N-1)}\big]$ is the expectation value of the product of the results of the measurements at different times.  These measurements are defined as follows. First we consider an initial state $\rho$ and immediately measure the observable $O^{(0)}$, then we let the system evolve unitarily and measure $O^{(1)}$ after the first time step, and so on \footnote{We consider that between each measurements a fixed amount of time $\epsilon$ lapsed. We also mostly focus on unitary evolution as our discussion can be extended to quantum channels by simply extending the system.}. 
Under this scheme, we can write $ \mathbb{E}_{\text{time}}\big[O^{(0)}, ..., O^{(N-1)}\big]=\sum_{\{\textbf{i}\}} p_{i_0\wedge i_1... \wedge i_{N-1}}\lambda_{i_0}^{(1)}\lambda_{i_1}^{(2)} \dots \lambda_{i_{N-1}}^{(N-1)}$ where the joint probabilities are now given by 
\begin{widetext}
   \begin{equation}\label{eq:jointprobs}
\begin{split}
      p_{i_0\wedge i_1\dots \wedge i_{N-1}} 
      &={\rm tr}\big[\rho \,P^{(0)}_{i_{0}}P^{(1)}_{i_1}(\epsilon)P^{(2)}_{i_2}(2\epsilon)\dots P^{(N-1)}_{i_{N-1}}((N-1)\epsilon) \dots P^{(2)}_{i_2}(2\epsilon)P^{(1)}_{i_1}(\epsilon)P^{(0)}_{i_0}\big]\,,
\end{split}
\end{equation} 
\end{widetext}
with 
$O^{(k)}=\sum_i \lambda_i^{(k)}P^{(k)}_i$ the eigendecomposition of the operators in orthogonal projectors (if the operator $O^{(k)}$ is degenerate then $P^{(k)}_i$ indicates the projector onto the corresponding subspace). To justify this expression let us consider the example of an initial pure state $|\psi\rangle$ undergoing a projective measurement, then unitary evolution and another projective measurement. One obtains 
\begin{equation*}
    \begin{split}
      p_{i_0\wedge i_1}=p_{i_1|i_0}p_{i_0}&=\frac{\langle \psi| P_{i_0}^{(0)}U^\dag P_{i_1}^{(1)}UP_{i_0}^{(0)}|\psi\rangle}{\langle \psi|P^{(0)}_{i_0}|\psi\rangle}\langle \psi|P^{(0)}_{i_0}|\psi\rangle  \\&=\langle \psi| P_{i_0}^{(0)}U^\dag P_{i_1}^{(1)}UP_{i_0}^{(0)}|\psi\rangle
    \end{split}
\end{equation*}
 where we used that the state of the system immediately after obtaining the result $\lambda_{i_0}$ from the first measurement is $|\psi\rangle'=P^{(0)}_{i_0}|\psi\rangle/\sqrt{\langle \psi|P^{(0)}_{i_0} |\psi\rangle}$ so that $p_{i_0|i_1}=\langle \psi'|U^\dag P^{(1)}_{i_1}U|\psi'\rangle$. 
 The general case follows analogously, yielding Eq.\ \eqref{eq:jointprobs}.

Notice that fixing the ``mean values'' across space and time over a complete set of operators completely defines the $\mathcal{R}_{\text{pdm}}$. In other words, if we consider a complete orthonormal set of operators in $h$, denoted as $\{O^{(i)}\}$, a complete set in $\mathcal{H}=\otimes_{t,x} h$ is given by $\{\otimes_{t,x}O^{(i_{tx})}\}$ and the corresponding PDM is 
${\mathcal{R}_{\text{pdm}}=\sum_{\textbf{i}}{\rm Tr}[\mathcal{R}_{\text{pdm}}\otimes_{t,x}O_{i_{tx}}]\otimes_{t,x}O^{(i_{tx})}}$ with 
\begin{equation}
    {\rm Tr}[\mathcal{R}_{\text{pdm}}\otimes_{t,x}O^{(i_{tx})}]=\mathbbm{E}_{\text{spacetime}}[O^{(i_{00})}, O^{(i_{01})}, \dots]\,.
\end{equation}
So the PDM is the unique object whose spacetime correlators are the expectation values of product of outcomes of measurements, defined by a given set of complete operators. To be more explicit , if we consider a grid of points in space and time, the corresponding joint probabilities defining $\mathbbm{E}_{\text{spacetime}}$ are given by 
\begin{equation}
\begin{split}
       p_{\wedge_{t,x}i_{tx}}&={\rm tr}\Big[\,\rho\, \otimes_x P^{(0x)}_{i_{0x}}\, U^\dag \otimes_x P^{(1x)}_{i_{1x}}U\dots\\&\dots (U^\dag)^{N-1} \otimes_x P^{(N-1,x)}_{i_{N-1x}}U^{N-1}\dots  \otimes_x P^{(0x)}_{i_{0x}}\Big]\,,
\end{split}
\end{equation}
 which is just Eq.\ \eqref{eq:jointprobs} with the subsystem structure explicit. Notice that the evolution operator $U$ is not necessarily separable in space. 
The previous definition leads to the following properties of the PDMs: they are Hermitian operators, have trace one, and their partial trace over all times except one is just the standard density matrix at that time. However, PDM are not positive semidefinite. Instead, the negative eigenvalues can be thought of a type of quantum correlation indicative of causality among events \cite{fitzsimons2015quantum}.

\subsubsection{PDMs from spacetime QM}

We now discuss how to obtain PDMs from the SQM formalism. Let us first emphasize that since the PDM approach focuses on measurements at different times and not on spacetime correlators the connection is not apparent. Yet, we can exploit the operators $\mathcal{S}_{\text{ext}}$ and Theorem \ref{th:extendedQA} to write the joint probabilities of measurements across time in a simple way. From these expressions, a closed form for PDMs follow.

The extended QA provides an explicit expression for the probabilities of sequential in time measurements as
\small
\begin{equation}
\begin{split}
     p_{i_0\wedge i_1 \dots \wedge i_{N-1}}
    =\\{\rm Tr}[\mathcal{R}_{\text{ext}}\underbrace{P^{(0)}_{i_0}\otimes P^{(1)}_{i_1}... P^{(N-1)}_{i_{N-1}}}_{\mathcal{H}} \otimes \underbrace{P^{(0)}_{i_0}\otimes P^{(1)}_{i_1}... P^{(N-1)}_{i_{N-1}}}_{\mathcal{H}}]\,,
\end{split}
\end{equation}
\normalsize
where we have emphasized the use of both copies of $\mathcal{H}$ to accommodate the presence of two projectors for each operator.
Using this result, we can write
\small
\begin{equation}
\begin{split}
&\mathbb{E}_{\text{time}}\big[O^{(0)}, ..., O^{(N-1)}\big]=\\&\sum_{\{\textbf{i}\}}{\rm Tr}\big[\mathcal{R}_{\text{ext}}\,P^{(0)}_{i_0}...\otimes P^{(N-1)}_{i_{N-1}}\otimes P^{(0)}_{i_0}\otimes ...\big]\,\lambda^{(0)}_{i_0}...\,\lambda^{(N-1)}_{i_{N-1}}\,.
\end{split}
\end{equation}
\normalsize
We see that the spacetime formalism provides a simple expression for the average of product of measurements at different times. If we now include also space we obtain
\small
\begin{equation}\label{eq:Espacetimeexpr}
\begin{split}
&\mathbb{E}_{\text{spacetime}}\big[O^{(00)},O^{(01)},...\big]=\sum_{\{\textbf{i}\}}\prod_{t,x}\lambda^{(tx)}_{i_{tx}}\times \\ &\times {\rm Tr}\big[\mathcal{R}_{\text{ext}}\,\mathop{{\textstyle\bigotimes}}\limits_{x} P^{(0x)}_{i_{0x}}\otimes    P^{(1x)}_{i_{1x}}\dots   P^{(N-1,x)}_{i_{N-1x}} \otimes P^{(0x)}_{i_{0x}}\otimes ...\big] \,.
\end{split}
\end{equation}
\normalsize
Notice that we simply needed to consider an ``overall'' tensor product in space (the tensor product between projectors correspond to time). As such,  the dynamics within slices, which can entangle across space, are codified in the QA.

To make further progress notice now 
that for each operator $O=\sum_i \lambda_i P_i$ on the l.h.s. of Eq. \eqref{eq:Espacetimeexpr}  the r.h.s. has a term $\phi(O\otimes \mathbbm{1})=\sum_i \lambda_i P_i\otimes P_i$, where we introduced the linear map $\phi$ for convenience. 
If we consider a complete set of operators and  find the map $\phi$ holding for the whole set, we can do the same for each copy $h_t\otimes h_t$ and define $\Phi(.)=\otimes_{t,x} \phi_{t,x}(.)$ such that 
\small
\begin{align*}
\mathbb{E}_{\text{spacetime}}\big[O^{(00)},O^{(01)},...\big]&={\rm Tr}[\mathcal{R}_{\text{ext}}\,\Phi(\otimes_{t,x} O_{t,x} \otimes \mathbbm{1}_\mathcal{H})]\,.
\end{align*}
\normalsize
This is almost the equation the PDM needs to satisfy, i.e., Eq.\ \eqref{eq:PDMexp}. To get the PDM we need to ``invert'' the map. Namely we need a $\Phi'$ such that
 ${{\rm Tr}[\mathcal{R}_{\text{ext}}\,\Phi(\otimes_{t,x} O_{tx} \otimes \mathbbm{1}_\mathcal{H})]={\rm Tr}[\Phi'(\mathcal{R}_{\text{ext}})\,\otimes _{t,x} O_{tx} \otimes \mathbbm{1}_\mathcal{H}]}$. Considering that these correlators completely specify the operator we conclude that
\begin{equation}
    \mathcal{R}_{\text{pdm}}={\rm Tr}_{\mathcal{H}_2}[\Phi'(\mathcal{R}_{\text{ext}})]\,,
\end{equation}
where the partial trace is over the ``backward evolution'' Hilbert space $\mathcal{H}_2$, namely we are labeling  $\mathcal{H}^{\otimes 2}\equiv \mathcal{H}_1 \otimes \mathcal{H}_2$ according with the conventions of Section \ref{sec:fbevolution}.

Now the crucial part becomes finding $\Phi'$. In general, the choice of $\Phi$ depends on the basis of operators chosen to define the PDM. To make direct contact with the literature (we recall that the definition of PDMs for qudits is under development), we will follow the proposal in \cite{fullwood2024operator} that generalizes the qubit case by employing ``light-touch observables''. In the Appendix \ref{app:PDM} we also discuss the use of Weyl observables, thus showing that our scheme can be used to generalize PDMs to other types of measurements.

In \cite{fullwood2024operator} the authors showed that one can always span $\mathcal{L}(h)$ with observables that have eigenvalues  $\pm 1$ (dichotomic measurements) with possible degeneracies, and the identity. These are the so-called light-touch observables, defined up to an arbitrary rescaling. Since for these observables we can write either $O=P_+-P_-$ or $O=P_+ +P_-$ (with $P_\pm$ projectors) we find
\begin{equation}
    \phi(O\otimes \mathbbm{1})=\sum_i \lambda_i P_i\otimes P_i=\frac{O\otimes \mathbbm{1}+\mathbbm{1}\otimes O}{2}\,,
\end{equation}
where the sum is over distinct eigenvalues. 
We can then define $\phi(.)$ acting on $h\otimes h$ as the quantum map 
\begin{equation}
    \phi(.)=\frac{\mathbbm{1}}{\sqrt{2}}(.) \frac{\mathbbm{1}}{\sqrt{2}}+\frac{\text{SWAP}}{\sqrt{2}}(.)\frac{\text{SWAP}}{\sqrt{2}}
\end{equation}
with Hermitian Kraus operators $\{\frac{\mathbbm{1}}{\sqrt{2}},\frac{\text{SWAP}}{\sqrt{2}}\}$. This defines $\Phi(.)=\otimes_{t,x}\phi(.)=\Phi'(.)$ as it follows from the cyclicity of the trace and Hermiticity of the Kraus operators. From the previous, we can state the following Theorem.
\begin{theorem}\label{th:PDMextendedQA}
    The PDM corresponding to light-touch observables is a quantum channel acting on the spacetime state $\mathcal{R}_{\text{ext}}$:
    \begin{equation}
        \mathcal{R}_{\text{pdm}}={\rm Tr}_{\mathcal{H}_2}[\Phi(\mathcal{R}_{\text{ext}})]
    \end{equation}
    with $\Phi=\otimes_{t,x}\phi$ for $\phi$ defined by Kraus operators $\{\frac{\mathbbm{1}}{\sqrt{2}},\frac{\text{SWAP}}{\sqrt{2}}\}$ and where SWAP acts between $h_{t,x}\otimes h_{t,x}$.  
\end{theorem}

\begin{figure}
    \centering
    \includegraphics[width=\linewidth]{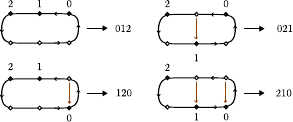}
    \caption{Graphical representation on how the channel $\Phi$ defined by SWAPs between $\mathcal{H}\otimes \mathcal{H}$ induces unimodal permutations in the ordering of operators for $N=3$. }
    \label{fig:PDMSWAPs}
\end{figure}

We have obtained a close expression for PDM, following from the spacetime QM formalism. This holds for arbitrary systems, as long as light-touch measurements are considered. Interestingly, the quantum channel $\Phi$ is independent on the evolution operators, a fact that can be exploited to generalize this result to PDM corresponding to a system undergoing a quantum channel evolution. In particular, as a direct application of the theorem one can use \eqref{eq:extquantumchannel} to  obtain 
\begin{equation}\label{eq:pdm2example}
    \mathcal{R}_{\text{pdm}}=\frac{\rho\otimes \mathbbm{1} J(\mathcal{E})+J(\mathcal{E})\rho\otimes \mathbbm{1}}{2}=\frac{\mathcal{R}+\mathcal{R}^\dag}{2}\,
\end{equation}
when  $N=2$.
 This is also equal to the state over time of Eq.\ \eqref{eq:qsot} that we discussed in Section \ref{sec:QSOT}. We explicitly derive this result in Appendix \ref{app:PDM} as an illustration of the theorem. One can also check that more generally, our closed form agrees with the recursive formula presented in \cite{liu2025quantum} for qubits. In addition, we also note that the repeated-star-product construction of \cite{lie2025multipartite}, when specialized to the canonical star product of Eq.\ \eqref{eq:qsot},  appears to reproduce the same unimodal symmetrization defining the multipartite PDM above \footnote{In this sense, that Markovian QSOT extension may be naturally identified with the multipartite PDM. To our knowledge, this connection was not made explicit in \cite{lie2025multipartite}.}.

So far, we have considered the extended QA operator acting on $\mathcal{H}\otimes \mathcal{H}$. 
Interestingly, we can develop things further. As shown in Figure \ref{fig:PDMSWAPs}, since only insertion of operators in the first copy of  $\mathcal{H}$ are to be considered (that is why we trace over the second copy of $\mathcal{H}$), the SWAP operators between $h_t\otimes h_{t}$ induce particular permutations on the usual time ordering. Notice that in the loop, reading from right to left (following the time arrows), even after the SWAPs have acted, one is necessarily increasing the time index till reaching $N-1$. Then, the order necessarily decreases.
We thus find that the only permutations among time slices appearing are the \emph{unimodal permutations}, rearrangements of the positions in time that first strictly increase to the peak and then strictly decreases.
This is a direct consequence of this ``loop'' structure depicted in the Figure \ref{fig:PDMSWAPs} and of Theorem \ref{th:PDMextendedQA}. To give a few concrete examples, for $N=3$ we have the orderings $012$, $021$, $120$ and $210$. For $N=4$ we obtain $3210$, $2310$, $1320$, $1230$, $0321$, $0231$, $0132$ and $0123$. In general one obtains $2^{N-1}$ out of the $N!$ permutations.

An equivalent, possibly more formal way to understand the appearance of unimodal permutations is the following. For light-touch observables, and considering Theorem \ref{th:PDMextendedQA} each local channel $\phi_t$ either leaves the insertion at time $t$ on the forward branch or moves it to the backward branch. Expanding $\Phi=\otimes_t \phi_t$ therefore amounts to summing over subsets $S\subseteq\{0,\dots,N-1\}$, where $t\in S$ indicates that the operator at time $t$ is assigned to the backward branch. By Theorem \ref{th:extendedQA}, the operators in $S$ then appear in an anti-time-ordered block, while those in the complement appear in a time-ordered block. As a result, each subset induces a permutation whose labels first increase up to the maximal time slice and then decrease, i.e., a unimodal permutation. Since the position of the maximal time is fixed, this yields precisely $2^{N-1}$ distinct permutations.

Notice also that under Theorem \ref{th:PDMextendedQA} the amount of evolution of the individual operators respect the initial labels: for example, the permutation $021$ induces ${\rm tr}[\rho O_1(\epsilon)O_2(2\epsilon)O_0]$, which 
following Theorem \ref{th:theorem1}, can be recovered form the spacetime state $\mathcal{R}$ after the adjoint action of  $\text{SWAP}_{21}\mathbbm{1}\otimes U\otimes U^\dag$. Notice that this adjoint action involves evolution operators (here assumed unitary again) and not only a permutation (in Theorem \ref{th:PDMextendedQA} the SWAP operators acted among time slices that correspond to the same time so no additional evolution operator was involved; instead, by working with $\mathcal{R}$ we are considering permutations among different times). 
 As a consequence, what we need to recover the PDM from a spacetime state is the \emph{unimodal twirl} $\Phi_{\text{UT}}:\mathcal{L}(\mathcal{H})\to \mathcal{L}(\mathcal{H})$ defined as
\begin{equation}\label{eq:unimodal-twirl}
\begin{split}
    \Phi_{\text{UT}}(X)= \frac{1}{|\mathcal U_N|} \sum_{\pi\in\mathcal U_N} U_\pi\,X\,U_\pi^\dagger\,,
\end{split}
\end{equation}
with $\pi$ a permutation, $\mathcal U_N$ the set of unimodal permutations ($|\mathcal U_N|=2^{N-1}$) and 
\begin{equation}
    U_\pi=W_\pi \otimes_t U_t^{t-\pi(t)}\,,
\end{equation}
for $W_\pi$ the unitary representation of the permutation acting among slices. Here the factor $\otimes_t U_t^{t-\pi(t)}$ ensures that the proper amount of evolution at each slice is recovered after the twirl, and a negative exponent is defined by $U_t^{t-\pi(t)}\equiv (U^\dag_t)^{\pi(t)-t}$.
Notice that the $2^{N-1}$ operators 
\begin{equation}
    K_\pi=\frac{1}{\sqrt{2^{N-1}}}U_{\pi}
\end{equation}
are the Kraus operators of the map with $K_\pi K_\pi^\dag=K_\pi^\dag K_\pi=\frac{1}{2^{N-1}}\mathbbm{1}$. 
With this definition, and following the previous discussion we can state the following Theorem.

\begin{theorem}\label{th:pdmfromr}
The PDM corresponding to light-touch observables is a quantum channel acting on spacetime states
  \begin{equation}
      \mathcal{R}_{pdm}=\Phi_{\text{UT}}(\mathcal{R})\,,
  \end{equation}
  with $\Phi_{\text{UT}}$ the CPTP channel implementing unimodal permutations among slices.
\end{theorem}
Let us  consider again the basic example of  $N=2$. In this case there is only one non trivial permutation leading to $\Phi_{\text{UT}}(\mathcal{R})=\frac{1}{2}\mathcal{R}+K_{\pi}\mathcal{R}K_\pi^\dag\,,$ for $\pi(0,1)=(1,0)$ corresponding to the Kraus operator $K_\pi=\frac{1}{\sqrt{2}}\text{SWAP}\, U\otimes U^\dag =\frac{1}{\sqrt{2}}e^{i\tilde{\mathcal{S}}}$. To compute the second term notice that $\mathcal{R}K_\pi^\dag=\frac{1}{\sqrt{2}}(\rho\otimes \mathbbm{1})$ while $K_\pi (\rho\otimes \mathbbm{1}) =\frac{1}{\sqrt{2}}\mathcal{R}^\dag$. In summary,
$
    \Phi_{\text{UT}}(\mathcal{R})=\frac{\mathcal{R}+\mathcal{R}^\dag}{2}\,,
$
in agreement with our previous discussion and with Eq.\ \eqref{eq:pdm2example}.

Considering that both $\mathcal{R}$ and $\mathcal{R}_{\text{ext}}$ have simple closed forms, we have obtained two \emph{closed formulas} for PDMs. This illustrates the advantage of working with spacetime states as the fundamental objects. Conversely, it also shows that the formalism of spacetime QM can naturally accommodate sequential-in-time measurements, in which case it collapses to the PDM formalism. While our derivation was purely mathematical, it would be interesting to understand possible physical interpretations of the quantum channels $\Phi,\Phi_{\text{UT}}$, perhaps in terms of a more complete description of the measurement process. Since these are genuine quantum channels, the resulting PDMs inherit a convex structure from the Kraus decomposition. However, these convex combinations are more general than those discussed in Eq.\ \eqref{eq:convextoJ}, since they also involve permutations across slices. This motivates the study of a broader class of convex mixtures of spacetime states, including those generated by channels that permute degrees of freedom across slices.

Let us also make a few comments about the extension of PDM proposed in \cite{zhang2020different} for gaussian continuum variable systems. First of all let us notice that the construction in \cite{zhang2020different} that defines modes at different times is equivalent to our spacetime Hilbert proposal when applied to bosons (see Sections \ref{sec:formalism} and \ref{sec:PIs} and the discussion in \cite{diaz2025spacetime}). Then, it is straightforward to verify that for two times and $U\equiv \mathbbm{1}$ the operator $\frac{\mathcal{R}+\mathcal{R}^\dag}{2}$ has as a covariance matrix the ``covariance matrix in time'' defined in \cite{zhang2020different} for the same scenario. We develop that example explicitly in  Appendix \ref{app:PDM}. 
Let us also remark that the operator $\mathcal{R}$ has an explicit closed form for bosons, directly related to the action in classical mechanics, as shown in Section \ref{sec:PIs}. Our approach is thus also elucidating the relation between PDMs and the PI formulation, a problem that has been posed in 
\cite{zhang2020different}.

Another open problem posed in \cite{fullwood2024operator} is the definition of PDMs for other sets of observables, with the simple relation between PDMs and quantum states over time only holding for two time slices and light touch observables. In principle, the strategy we presented throughout this section could be extended to any observable.
The problem is reduced to finding $\Phi$ which in general is not a linear function of observables but still a linear map. Then, since a Kraus-like representation always exists (with possible signs if the map is not CPTP), it can be ``inverted'' using the cyclicity of the trace. We provide an explicit example in the Appendix \ref{app:PDM}. Although the structure of the PDM and relation with spacetime states could vary notably from case to case at the very least
the formalism suggests a clear path to developing the required generalizations of PDMs, with the latter always retrievable as a linear map acting on spacetime states.

\subsection{Superdensity operators}

Superdensity operators (SO) are density-operator analogues on spacetime \emph{operator space}. Namely, they are objects acting as
\begin{equation}
    \varrho: \mathcal{L}(\mathcal{H})\to \mathcal{L}(\mathcal{H})\,,
\end{equation}
or equivalently as a bilinear form $\varrho: \mathcal{L}(\mathcal{H}^\ast)\otimes \mathcal{L}(\mathcal{H})\to \mathbbm{C}$, where $\mathcal{H}=h_0\otimes h_1\otimes \dots h_{N-1}$ is the spacetime Hilbert space of Section \ref{sec:formalism}. Since we are working in operator space, it is useful to introduce the notation $|O\rangle \rangle=O\otimes \mathbbm{1} |\Phi^+\rangle \rangle$ for $|\Phi^+\rangle \rangle$ the Choi state.

The essential idea of the SO formalism is that one can operationally construct  an operator $\varrho$ that captures spacetime correlations. In this sense, such  $\varrho$ leads to a definition for a state in spacetime. Following \cite{cotler2018superdensity} this SO can be obtained as the reduced state of  auxiliary registers after sequentially coupling them (in superposition over an operator basis) to the system at times $t_0,\dots,t_{N-1}$ and tracing out the system. Under this procedure (see \cite{cotler2018superdensity} for the details) one obtains 
\begin{equation}
\label{eq:sdo-expansion}
  \varrho=\frac{1}{\dim(\mathcal{H})} \sum_{\textbf{i},\textbf{j}}
  C_{\textbf{i},\textbf{j}}\;
  \mathop{{\textstyle\bigotimes}}\limits_{t=0}^{N-1} |O_{i_t}\rangle \rangle \langle \langle O_{j_t}|\,,
\end{equation}
where $\textbf{i}=(i_0,i_1,..., i_{N-1})$, the $O_i$ form a complete orthonormal basis of $\mathcal{L}(h)$ ($\langle \langle O_j|O_i\rangle\rangle={\rm tr}[O_j^\dag O_i]=\delta_{ij}$),
and 
\begin{equation}
\label{eq:Ctensor}
  C_{\textbf{i},\textbf{j}}
  = \mathrm{tr}\Big[
    O_{i_{N-1}}\,U\cdots U\,O_{i_0}\,\rho\,
    O_{j_0}^\dag\,U^\dag\cdots U^\dag\,O_{j_{N-1}}^\dag
  \Big].
\end{equation}
For simplicity and ease of notation we assume a time independent evolution such that $U=e^{-i\epsilon H}$. Notice that the tensor $C$ contains all the time-ordered correlation functions (of ascending and descending orders)
of a complete basis of operators. This also means that any time-ordered correlation function can be recovered from $C$ and, as consequence from $\varrho$: 
\begin{equation}\label{eq:correlfromSO}
    \begin{split}
        &C[A^{(0)},A^{(1)},\dots, B^{(0)}, B^{(1)},\dots]\\&= \mathrm{tr}\Big[
    A^{(i_{N-1})}\,U\cdots U\,A^{(0)}\,\rho\,
    B^{(0)\dag}\,U^\dag\cdots U^\dag\,B^{(N-1)\dag}
  \Big]\\&=\langle \langle A^{(0)},A^{(1)},\dots|\varrho|B^{(0)\dag},    B^{(1)\dag},\dots\rangle\rangle\,,
    \end{split}
\end{equation}
with e.g., $|B^{(0)}, B^{(1)},\dots\rangle\rangle\equiv |B^{(0)}\rangle \rangle \otimes  |B^{(1)}\rangle \rangle \otimes \dots\,$. This is a direct consequence of the $O_i$ being a complete orthonormal basis in which the $A^{(t)}$, $B^{(t)}$ can be expanded. Here, we find it important to remark that $\varrho$ is Hermitian, has unit trace and is positive definite. In other words, $\varrho$ is a genuine quantum state \cite{cotler2018superdensity} (technically one can call any quantum state acting in $\mathcal{L}(\mathcal{H})$ a superdensity operator; however we will focus on $\varrho$ as is the one naturally capturing correlations across time).

\begin{figure}[h!]
    \centering
\includegraphics[width=0.35\linewidth]{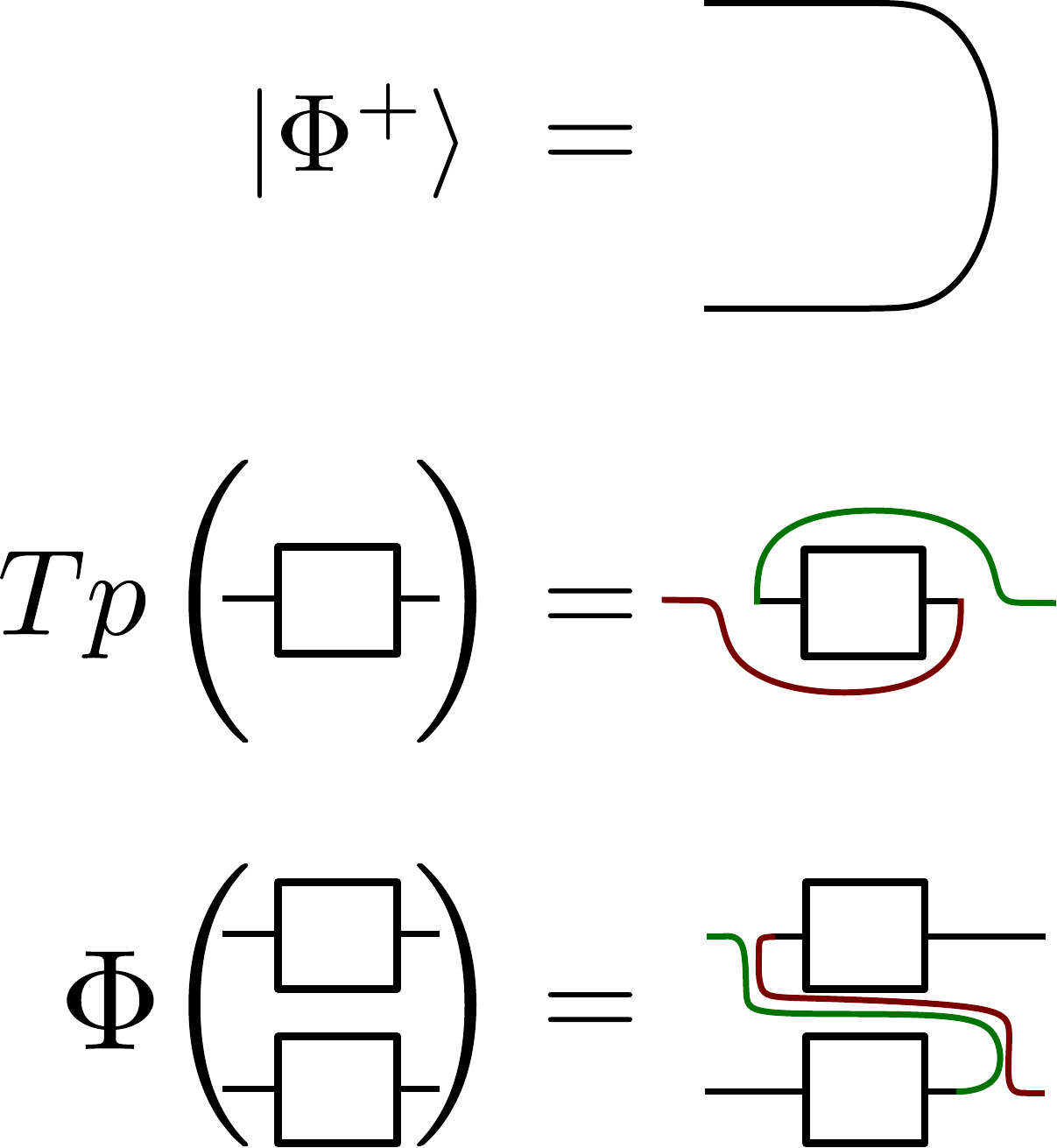}
    \caption{Pictorial representations of the different maps involved in Theorem \ref{th:SO} and of the unnormalized Choi state $|\Phi^+\rangle$. }
    \label{fig:TNmaps}
\end{figure}

\begin{figure*}[t!]
    \centering
\includegraphics[width=0.9\linewidth]{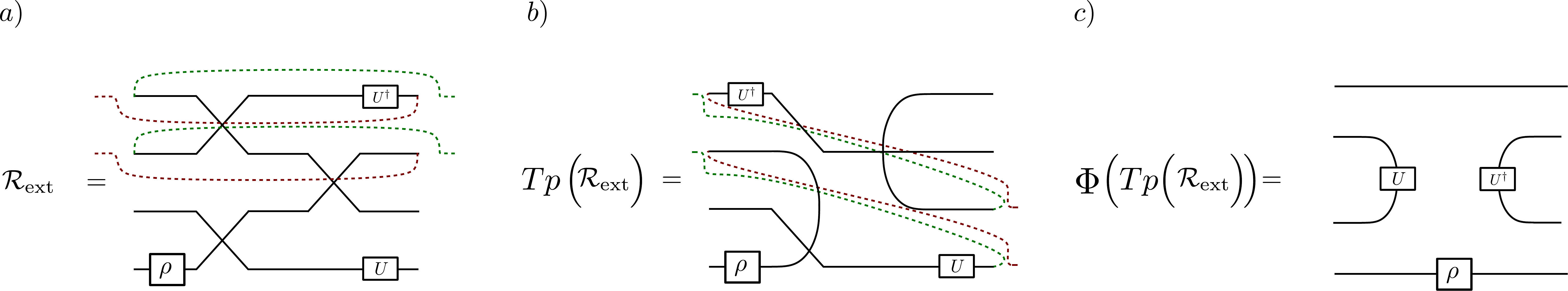}
    \caption{Diagrammatic derivation of Eq.\ \eqref{eq:SOfinal}. On panel a) we depict $\mathcal{R}_{\text{ext}}$ for two times (solid lines), together with the action of $Tp$ (dashed lines). On panel b) the result of $Tp(\mathcal{R}_{\text{ext}})$ (solid lines) with the action of $\Phi$ (dashed lines). On panel c) the result of having applied the full map, which results in $\varrho$. }
    \label{fig:soderivation}
\end{figure*}

Let us now discuss how to recover SOs from the spacetime formalism. Notably, as we shall prove, a simple rearrangement relates $\varrho$ to $\mathcal{R}_{\text{ext}}$. 
Let us first notice that the tensor $C$ can be rewritten as 
\small
\begin{equation}
    \begin{split}\label{eq:cijSO}
        C_{\textbf{i},\textbf{j}}
  &= \mathrm{tr}\big[\rho O^\dag_{j_0}O^\dag_{j_1}(\epsilon)...O^\dag_{j_{N-1}}\![(N-1)\epsilon]O_{i_{N-1}}\![(N-1)\epsilon]...O_{i_0}\!\big]\\
  &={\rm Tr}\big[ \mathcal{R}_{\text{ext}}(\otimes_{t=0}^{N-1} O_{i_t})\otimes  (\otimes_{t=0}^{N-1} O^\dag_{j_{t}})\, \big]\,,
    \end{split}
\end{equation}
\normalsize
where the last equality is a direct consequence of Theorem \ref{th:extendedQA}.
On the other hand, since the operators form a complete basis we can expand $\mathcal{R}_{\text{ext}}$ as follows
\begin{equation}
    \begin{split}
 \mathcal{R}_{\text{ext}}
  =\sum_{\textbf{i},\textbf{j}}  C_{\textbf{i},\,\textbf{j}}\,  \left[\,\mathop{{\textstyle\bigotimes}}\limits_{t=0}^{N-1} O_{i_t}\,\right]\otimes \left[\,\mathop{{\textstyle\bigotimes}}\limits_{t=0}^{N-1} O^\dag_{j_{t}}\,\right]\,,
    \end{split}
\end{equation}
where we use the dagger basis in the second copy of $\mathcal{H}$ (also a complete basis) in accordance with Eq.\ \eqref{eq:cijSO}. Let us now stress that this operator has almost the same form of $\varrho$ but without vectorizations. We also see that both operators have dimension $d^{2N}\times d^{2N}$. 
As a matter of fact, to explicitly relate them we need the following rearrangement 
$
        \mathcal{R}_{\text{ext}}\in \mathcal{L}(\mathcal{H})\otimes \mathcal{L}(\mathcal{H})=\mathcal{H}\otimes \mathcal{H}^\ast\otimes \mathcal{H}\otimes \mathcal{H}^\ast\to (\mathcal{H}\otimes \mathcal{H})\otimes (\mathcal{H}^\ast\otimes \mathcal{H}^\ast)\,.
$
    This corresponds to the $O_{i_t}$ becoming kets and the $O_{j_t}$ becoming bras. In a given basis we thus need  
    \begin{equation}
        \Phi(|i\rangle \langle j|\otimes |k\rangle \langle l|)= |ij\rangle \langle kl|=|i\rangle \langle k|\otimes |j\rangle \langle l|\,
    \end{equation}
    holding for all time slices. Interestingly, this is the realignment that also appears in the  Computable Cross-Norm or Realignment (CCNR) criterion to detect entanglement \cite{rudolph2003some, chen2002matrix}. 
 But in addition since the operators on $\mathcal{H}_2$ appear with a dagger we need to consider a partial transpose acting on them. Putting all together we find.

\begin{theorem}\label{th:SO}
    Consider the linear map \begin{equation}
        \mathcal{M}\equiv \Phi\circ(\mathbbm{1}_{\mathcal{H}_1}\otimes Tp)\,,
    \end{equation} with $Tp$ indicating  transpose on $\mathcal{H}_2$ and $ \Phi(|i\rangle_t \langle j|\otimes |k\rangle_t \langle l|)=|ij\rangle_t \langle kl|$ for all $t$. It holds that
    \begin{equation}
    \varrho=\frac{1}{\dim(\mathcal{H})}\mathcal{M}(\mathcal{R}_{\text{ext}})\,,\;\; \mathcal{R}_{\text{ext}}=\dim (\mathcal{H})\,\mathcal{M}(\varrho)\,.
\end{equation}
\end{theorem}
The theorem shows that $\varrho$ and $\mathcal R_{\rm ext}$ are related by the linear rearrangement
$\mathcal M=\Phi\circ(\mathbbm 1\otimes Tp)$, namely a partial transpose on the second copy  followed by a realignment (index reshuffling) $\Phi$. Importantly, $\mathcal M$ is not a physical map. For instance, the Peres-Horodecki test \cite{peres1996separability,horodecki2001separability} relies on the fact that the partial transpose preserves positivity for separable states but can produce negative eigenvalues for entangled states. Similarly, the CCNR/realignment criterion is based on a reshuffling of matrix indices which obeys a norm bound for separable states but can be violated by entanglement. 
This shows that \emph{timelike correlations in 
$\varrho$} (between forward and backward evolution) \emph{manifest  themselves precisely in how 
$\mathcal{R}_{\text{ext}}$
 departs from an ordinary quantum state}.

\begin{figure}[h!]
    \centering
\includegraphics[width=0.4\linewidth]{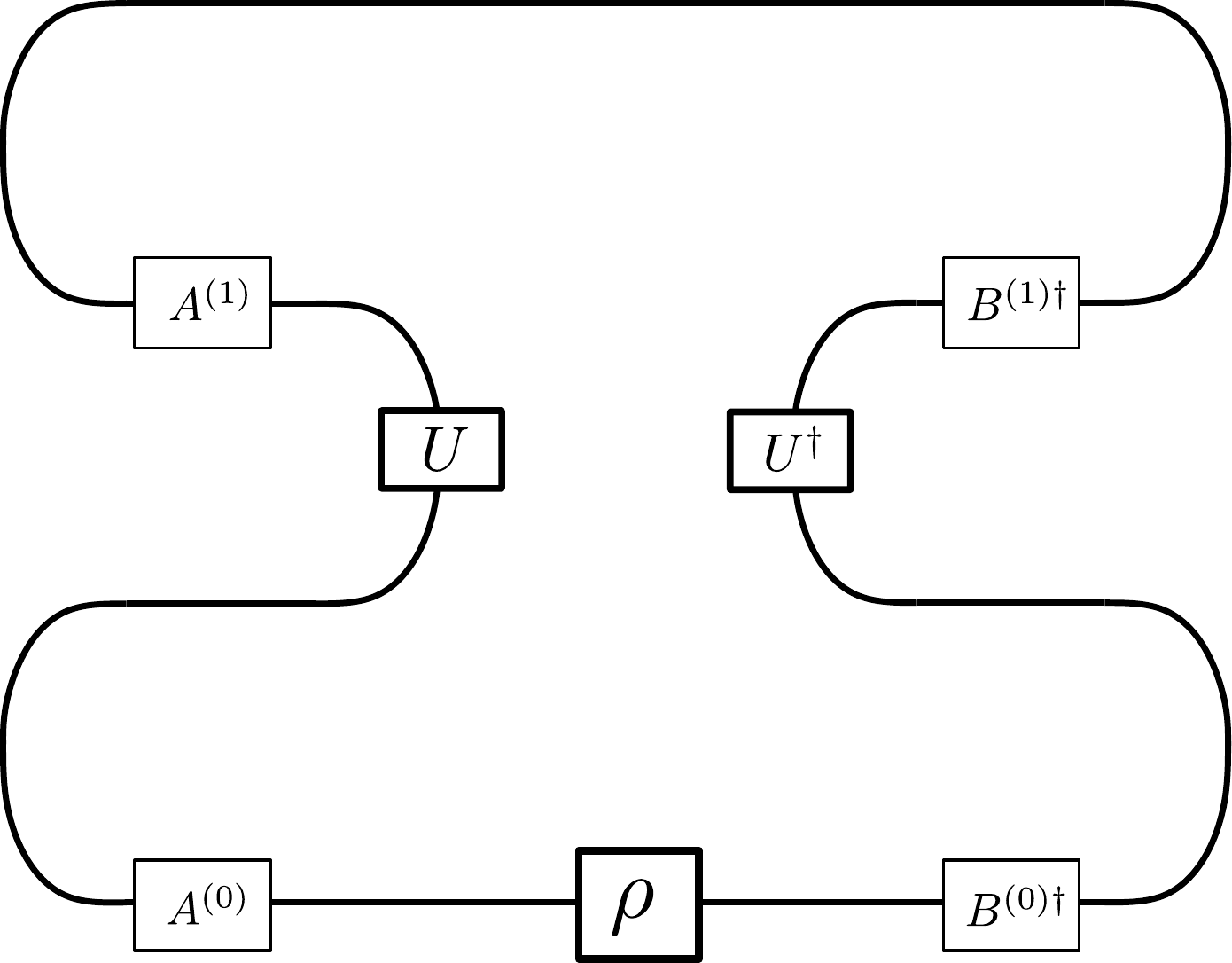}
    \caption{Forward and backward correlator of Eq.\ \eqref{eq:correlfromSO}  as derived from Eq.\ \eqref{eq:SOfinal}. }
    \label{fig:sotrace}
\end{figure}

As a direct application of the theorem let us obtain an explicit expression for $\varrho$. For concreteness we focus on the case $N=2$ but a similar procedure can be applied to any $N$. It is convenient to represent the maps $\Phi, Tp$ in TN notation as we depict in Figure \ref{fig:TNmaps}, where we also recall how to represent the unnormalized Choi state. 
Then, a direct application of the theorem to $\mathcal{R}_{\text{ext}}$, shown step by step in Figure \ref{fig:soderivation}, leads to
\begin{equation}\label{eq:SOfinal}
    \varrho=\rho \otimes \frac{1}{d}|U\rangle \rangle \langle \langle U| \otimes \frac{1}{d}\mathbbm{1}\,.
\end{equation}
Notably we have an explicit expression for $\varrho$. One can check that it satisfies Eq.\ \eqref{eq:correlfromSO} which is the defining property of the SO. An easy verification is obtained once again in TN notation, as depicted in Figure \ref{fig:sotrace}.

Having a closed expression for 
$\varrho$ allows us to make a few interesting comments. First, even when the initial state is pure and the evolution is unitary, 
$\varrho$ has nonzero entropy. This is in contrast with the picture arising from 
spacetime states, where closed dynamics correspond to vanishing pseudo-entropy. Second, SOs do not reproduce the time marginals by partial traces in the way spacetime states do (Corollary \ref{cor:schrodheis}). This is not a problem for SOs, since the equivalent description corresponds to identity-operator insertions (after all, SOs act on operator space). In this sense, while SOs contain the same information as $\mathcal{R}_{\text{ext}}$ rearranged in a convenient and useful way that makes $\varrho$ operational,  the resulting object does not play the same role as a spacetime state in the SQM framework. Rather, 
$\varrho$ may be viewed as a standard-state encoding of the same multitime information.

\subsection{The Page and Wootters mechanism}\label{sec:PW}

\subsubsection{Introduction to the PW mechanism for discrete time}
The Page and Wootters (PW) mechanism \cite{page1983evolution} aims to replace quantum evolution with correlations. Given a quantum system, traditionally evolving under the Schr\"{o}dinger equation, one considers instead an enlarged system with Hilbert space $h_{\text{PW}}=h_T\otimes h$ for $h$ the original Hilbert space of the system and $h_T$ the Hilbert space of the rest (sometimes also referred to as clock system). The essential idea of PW is that under a proper correlation between $S$ and $T$, which defines the so-called physical subspace, the system, conditioned to values of an observable of $T$ (which might be thought as ``clock readings''), evolves according to unitary evolution.  Initially proposed as a solution for a ``timeless universe'' in the context of the problem of time \cite{dewitt1967quantum, unruh1989time, kuchavr2011time, anderson2012problem}, the PW mechanism has recently attracted renewed interest in the setting of quantum information and computation \cite{giovannetti2015quantum, boette2016system, boette2018history, paiva2022non, diaz2023parallel,  cafasso2024quantum, coppo2026quantum}. Moreover, in the case of single particles, it has been employed to define models where  Lorentz covariance is explicit \cite{diaz2019history, diaz2019historystate, giovannetti2023geometric}.

One convenient and rigorous way to introduce the PW mechanism in a discrete time setting, developed in \cite{boette2016system}, is as follows. Assuming a basis of states $|t\rangle$ with $t=0,1,\dots ,N-1$ for $h_T$ we can define
\begin{equation}
    e^{i\epsilon P_T}=\sum_{t=0}^{N-1}|t+1\rangle \langle t|\otimes \mathbbm{1}_S
\end{equation}
with $|N\rangle\equiv |0\rangle$. This is clearly a time translation operator with $P_T$ the discrete version of the momentum conjugated to the ``time operator'' $T=\sum_t t\, |t\rangle \langle t|$ (notice that $e^{i \epsilon P_T}$ and $e^{i2\pi T/N}$ are just the usual Weyl operators $X, Z$ respectively).  Then, let us introduce the unitary operator 
\begin{equation}
    V=\sum_{t=0}^{N-1}|t\rangle \langle t|\otimes e^{i \epsilon  t H}\,,
\end{equation}
which is implementing unitaries on the system controlled by states on $T$ (let us notice that this approach is very natural in the quantum computing scenario, as shown in \cite{boette2016system} and further developed in \cite{diaz2023parallel}).
We can now define the exponential of the universe operator $J$ as
\begin{equation}\label{eq:eijdef}
    e^{i J}=V^\dag e^{i\epsilon P_T} V\,.
\end{equation}
With these definitions the Page and Wootters ``universe equation'' becomes
\begin{equation}\label{eq:universeconstr}
e^{iJ}|\Psi_{\text{PW}}\rangle=|\Psi_{\text{PW}}\rangle \,,
\end{equation}
or equivalently 
$
    J|\Psi_{\text{PW}}\rangle=0\,,
$
(modulo $2\pi$) which is the form used in continuum time formulations and which models the
timeless Wheeler-DeWitt equation.
Let us also notice that a direct calculation leads to 
\begin{equation}\label{eq:eij}
    e^{i J}=e^{i T (|0\rangle \langle 0|\otimes H)}\,e^{i \epsilon P_T}\otimes e^{-i\epsilon H}\,.
\end{equation}
In particular, for cyclic evolution $e^{-iTH}=\mathbbm{1}$, the universe operator becomes  $J=P_T\otimes \mathbbm{1}-\mathbbm{1}\otimes H$ thus taking the same form as in continuum (unbounded) time schemes.

One can easily prove that $J$ has always a zero eigenvalue with degeneracy equal to the dimension of $h$, leading to solutions of the form 
\begin{equation}
|\Psi_{\text{PW}}\rangle=\sum_{t=0}^{N-1} |t\rangle |\psi(t)\rangle    
\end{equation}
 for $|\psi(t)\rangle=e^{-i\epsilon t H}|\psi\rangle$ and arbitrary initial conditions $|\psi\rangle$. One might interpret \eqref{eq:universeconstr} as a frozen in time Schr\"{o}dinger equation with $J$ the global Hamiltonian. In this sense, the quantum state $|\Psi_{\text{PW}}\rangle$ is static but we can recover standard unitary evolution of the system by conditioning on the rest, namely $\langle t|\Psi_\text{PW}\rangle=|\psi(t)\rangle$. 
This is the essence of the PW mechanism. Equivalently, we can define the global operator $O_t= |t\rangle \langle t|\otimes O_S$ leading to 
\begin{equation}
    \langle \Psi_{\text{PW}}|O_t|\Psi_{\text{PW}}\rangle=\langle \psi(t)|O_S|\psi(t)\rangle=\langle \psi|O_H(\epsilon t)|\psi\rangle\,.
\end{equation}
Let us also notice that given a basis $|k\rangle$ of the system we can write 
\begin{equation}\label{eq:pwstate}
|\Psi_{\text{PW}}\rangle=\sum_{t,k}\psi_{k}(\epsilon t)|t,k\rangle\,,
\end{equation}
with $\psi_{k}(\epsilon t)$ the evolved wavefunction, and 
\begin{equation}\label{eq:operatorpaw}
    O_t= \sum_{k,k'} O_{kk'}|t,k\rangle \langle t,k'|\,.
\end{equation}

\subsubsection{Second quantization of the PW mechanism and  spacetime QM}

Given the description of a single system, typically a particle, second quantization provides the framework to describe many indistinguishable copies of it. In other words, starting from a fundamental single-particle Hilbert space, one constructs the corresponding Fock space that accommodates an arbitrary number of particles.
Let us apply this scheme to the PW mechanism assuming that the system is a single particle and taking as the single particle basis $|t,k\rangle$. The idea is to introduce creation and annihilation operators such that $a^\dag_{t,k}|\Omega\rangle=|t,k\rangle$ leading to (we assume a bosonic algebra for simplicity) 
\begin{equation}
[a_{t,k},a^\dag_{t',k'}]=\delta_{tt'}\delta_{kk'}\,,
\end{equation} 
and with $|\Omega\rangle$ the associated vacuum. Notably, if we fix $t=t'$ we have the algebra of the traditional Fock space $h_F$ defined by $[a_k,a_{k'}^\dag]=\delta_{kk'}$. As a consequence, $|\Omega\rangle=\otimes_t |0\rangle$ for $|0\rangle$ the vacuum in $h_F$. The additional index $t$ can be regarded as a Hilbert space index so we reach the following conclusion.
\begin{theorem}
Consider the PW Hilbert space $h_{\text{\upshape PW}}=h_T\otimes h$, and the conventional Fock space of the system $h_F\equiv \text{\upshape  Sym} \oplus_{n=0}^\infty h^{\otimes n}$.
The Fock space arising from second quantization $$\mathcal{H}^{\text{\upshape PW}}_F\equiv \text{\upshape  Sym} \oplus_{n=0}^\infty h^{\otimes n}_{\text{\upshape PW}}$$  
is isomorphic to the spacetime Fock space, namely:
    \begin{equation}
  \mathcal{H}^{\text{\upshape PW}}_F\simeq\mathcal{H}\,,
\end{equation}
with $\mathcal{H}=h_F^{\otimes N}$.
\end{theorem}
Hence, the Fock space whose elementary excitations are PW states is isomorphic to the corresponding spacetime Fock space. We will now show how, in many ways, the PW scheme can then be regarded as a the single-particle subspace of $\mathcal{H}$. We should clarify, however, that this isomorphism arises only when the system is described through a Fock space. This includes many important scenarios, such as any QFT. On the other hand, if the system is not described through a Fock space (e.g., a qubit), $\mathcal{H}$ is not a Fock space and no isomorphism follows.

The second quantization scheme also allows one to ``promote'' operators that were defined in the original Hilbert space $h_{\text{PW}}$ to $\mathcal{H}_F$. For a single particle operator the relation is given by \cite{schwabl2008advanced}
\begin{equation}
\begin{split}
      O=\sum_{t,i,t',j}\langle t,i|O|t',j\rangle |t,i\rangle\langle t',j|\\\to \mathcal{O}=\sum_{t,i,t',j}\langle t,i|O|t',j\rangle a^\dag_{ti}a_{t'j}\,.
\end{split}
\end{equation}
Notice that the matrix elements of $\mathcal{O}$ in the subspace of a single particle are the same as those of $O$, namely $\langle 0| a_{ti} \mathcal{O}a^\dag_{t'j}|0\rangle=\langle t,i|O|t',j\rangle$. In particular, for an operator of the form of \eqref{eq:operatorpaw} we obtain an operator $(\mathcal{O})_t$ acting on the slice $t$ of $\mathcal{H}$.
Let us now apply this scheme to the universe operator. 
A direct calculation yields
\begin{equation}\label{eq:matrixelJ}
    \langle t_2,k_2|e^{iJ}|t_1,k_1\rangle=e^{i\lambda_{k_1}(\epsilon-T\delta_{t_20})}\delta_{t_2,t_1+1}\delta_{k_1k_2}
\end{equation}
in the eigenbasis of the Hamiltonian, with $H|k\rangle=\lambda_k |k\rangle$. Similarly, consider the second quantized Hamiltonian $H_t=\sum_k\lambda_k a^\dag_{tk}a_{tk}$. The corresponding SQA leads to 
\begin{equation}
\begin{split}
     \langle \Omega|a_{t_2k_2}e^{i\tilde{\mathcal{S}}}a^\dag_{t_1k_1}|\Omega\rangle&=e^{i\lambda_{k_1}(\epsilon-T\delta_{t_20})} \langle \Omega|a_{t_2k_2}a^\dag_{t_1+1,k_1}|\Omega\rangle\\
     &= e^{i\lambda_{k_1}(\epsilon-T\delta_{t_20})} \delta_{t_2,t_1+1}\delta_{k_1k_2}\,,
\end{split} 
\end{equation}
which are precisely the matrix elements of $e^{i J}$ in Eq.\ \eqref{eq:matrixelJ}. In summary we can state the following result.
\begin{theorem}
The operator $e^{i\tilde{\mathcal{S}}}$ is the second quantized version of $e^{iJ}$.
\end{theorem}
This explains the similarities between Eqs.\ \eqref{eq:eijdef}, \eqref{eq:eij} and \eqref{eq:qainin}. Moreover, it is straightforward to check that the second quantization of $P_T\otimes \mathbbm{1}$ is $\mathcal{P}$ (one can use a Fourier in time basis to show this explicitly; see \cite{diaz2021spacetime,diaz2025spacetime}) while $\mathbbm{1}\otimes H=\sum_{t,k} \lambda_k |t\rangle \langle t|\otimes |k\rangle \langle k|\to \mathcal{K}=\sum_t \lambda_k a^\dag_{tk}a_{tk}$. Thus the form of the universe operator with cyclic evolution, given by $J=P_T\otimes \mathbbm{1}-\mathbbm{1}\otimes H$, leads precisely to $\mathcal{S}$ after second quantization. This reasoning has been recently applied to the relativistic extension of the PW formalism  developed in \cite{diaz2019history, diaz2019historystate} to show that the free Dirac and Klein-Gordon quantum actions arise from second quantization as well \cite{diaz2025spacetime}. 
In general, we see that promoting the PW mechanism corresponding to single particle interactions to second quantization, leads to \emph{free theories} \footnote{One could attempt to introduce e.g., two-particle interaction within the PW scheme. Nonetheless, since usually a single time is considered, in practice this is never done. Instead, introducing a ``clock'' for each particle give rise to a complete new set of challenges as one has to deal with a multi-time scheme. None of the difficulties arise in the second quantized approach. }, namely particle preserving gaussian QA operators.

\subsubsection{PW states as Reduced density matrices}

So far we have only made considerations regarding the Hilbert space and the universe operator. Now we discuss the role of spacetime states. In particular, there is a natural way to recover PW states from $\mathcal{R}$. The connection arises from the concept of reduced density matrices (RDM). Let us recall  first that given a density matrix $\rho$ in a Fock space $h_F$ with associated ladder operators $a_i$, we can always define the single particle reduced density matrix $\rho^{\text{sp}}=\sum_{i,j}{\rm tr}[\rho a^\dag_j a_i]|i\rangle \langle j|$. Then, given any single particle operator $O=\sum_{i,j} O^{\text{sp}}_{ij} a^\dag_i a_j$, we can compute expectations values as 
\begin{equation}
    {\rm tr}_F[\rho O]={\rm tr}[\rho^{\text{sp}}O^{\text{sp}}]\,.
\end{equation}
Notice that we have replaced a mean value in Fock space with a mean value in essentially a single particle Hilbert space of dimension $L$ for $L$ the number of modes in $h_F$ (e.g., for fermions leading to $\text{dim}(h_F)=2^L$). Since ${\rm tr}[\rho^{\text{sp}}]=\langle \hat{N}\rangle$, i.e., the mean value of particle number, for a single particle state $\rho$ the corresponding $\rho^{\text{sp}}$ is a genuine quantum state (the positive definiteness is also easily proven). Let us recall the importance of RDMs in applications such as in estimating energies in the context of quantum chemistry \cite{coleman1963structure, mazziotti2012structure} where the state of a many-particle system contains more information than necessary for the task.

Since $\mathcal{H}_F^{\text PW}$ is a Fock space we can apply the RDM formalism directly. Let us first notice that the single particle RDM of an operator $\mathcal{O}$ in 
$\mathcal{H}_F^{\text PW}$ has the following structure $$\rho^{(\text{sp})}[\mathcal{O}]=\sum_{t_1,k_1,t_2,k_2}{\rm Tr}[Oa^\dag_{t_1k_1}a_{t_2k_2}]|t_1,k_1\rangle \langle t_2,k_2|$$ which, as a matrix, acts on the Hilbert space $h_T\otimes h\equiv h_{\text{PW}}$, meaning that we recover the tensor product structure ``system-rest'' that characterized the PW mechanism. 
We now want to discuss the RDM of spacetime states. 
For reasons that will become apparent below, it will be convenient to consider $\mathcal{R}_{\text{ext}}=|\psi\rangle_0\langle \psi|e^{i\mathcal{S}_{\text{ext}}}$,   the ``extended'' version of the spacetime state, with $\tilde{\mathcal{S}}$ replaced by $\mathcal{S}_{\text{ext}}$. By using a notation inspired by Section \ref{sec:SKPI}, we will denote as $a^{\pm}_{ti}$ the operators acting on the first and second half of $\mathcal{H}\otimes \mathcal{H}$ respectively. The RDM of $\mathcal{R}_{\text{ext}}$ is then composed of four blocks that might be schematically denoted as $++,+-,-+,--$, each having dimension $\text{dim}(h_{\text{PW}})\times \text{dim}(h_{\text{PW}})$. These blocks correspond to different time orderings, e.g.,
\begin{equation}
\begin{split}
\rho^{\text{sp}}_{++}[\mathcal{R}_{\text{ext}}]=\sum_{t_1,i,t_2,j}{\rm Tr}[\mathcal{R}_{\text{ext}}\,a^{+\dag}_{t_1,i}a^+_{t_2,j}]|t_2,j\rangle \langle t_1,i|\\ 
=\sum_{t_1,i,t_2,j}\langle \psi|\hat{T}a^{\dag}_{i}(t_1)a_{j}(t_2)|\psi\rangle|t_2,j\rangle \langle t_1,i|\,,
\end{split}
\end{equation}
where we used Theorem \ref{th:extendedQA}. Notice that a time ordering is imposed in this case. Instead,
\begin{equation}
\begin{split}
\rho^{\text{sp}}_{+-}[\mathcal{R}_{\text{ext}}]=\sum_{t_1,i,t_2,j}{\rm Tr}[\mathcal{R}_{\text{ext}}\,a^{+\dag}_{t_1,i}a^-_{t_2,j}]|t_2,j\rangle \langle t_1,i|\\ 
=\sum_{t_1,i,t_2,j}\langle \psi|a^{\dag}_{i}(t_1)a_{j}(t_2)|\psi\rangle|t_2,j\rangle \langle t_1,i|\,,
\end{split}
\end{equation}
with no time-ordering: the creation operator is always at the left of the annihilation operator. Now, given a  sp operator of the form $\mathcal{O}=\sum_{t,i,t',j}\langle t,i|O^{\text{sp}}|t',j\rangle a^{+\dag}_{ti}a^-_{t'j}$ it is easy to show that 
\begin{equation}
    {\rm Tr}[\mathcal{R}_\text{ext}\mathcal{O}]={\rm tr}[\rho^{\text{sp}}_{+-} O^{\text{sp}}]\,.
\end{equation}
In particular, if we take $O^{\text{sp}}=|t\rangle \langle t|\otimes O^{\text{sp}}_S$ we obtain
\begin{equation}\label{eq:spresult}
    {\rm tr}[\rho^{\text{sp}}_{+-} O^{\text{sp}}]=\sum_{i,j}(O^{\text{sp}}_S)_{ij}\langle \psi(t)|a_i^\dag a_j|\psi(t)\rangle\,,
\end{equation}
i.e., the mean value at time $t$ of the operator $O=\sum_{i,j}(O^{\text{sp}}_S)_{ij}a_i^\dag a_j$ acting on $h_F$. This is an important result that requires some discussion. In the PW mechanism, the mean value of operators having a conditioning structure (see Eq.\ \eqref{eq:operatorpaw}) with respect to physical states are equal to the conventional mean value of states evolving unitarily.  Equation \eqref{eq:spresult} is thus a very nontrivial generalization of the mechanism, as no assumption was made on the initial state or on the evolution. Hence, through this scheme one can codify in a PW-like state $\rho^{\text{sp}}_{+-}$ any sp mean value. This includes interacting QFTs and states with arbitrary (or non-fixed) number of particles. The PW states, defined by \eqref{eq:universeconstr} are recovered under the assumption of $|\psi\rangle$ describing a single particle and the QA being the one of a free theory. This is the content of the following theorem.  
\begin{theorem}\label{th:RDMPW}
    Consider the  RDM of $\mathcal{R}_{\text{ext}}$ defined by 
\begin{equation}
\begin{split}
\rho^{\text{sp}}_{\pm}=\sum_{t_1,i,t_2,j}{\rm Tr}[\mathcal{R}_{\text{ext}}\,a^{+\dag}_{t_1,i}a^-_{t_2,j}]|t_2,j\rangle \langle t_1,i|\,.
\end{split}
\end{equation}
When $|\psi\rangle$ is a single particle state and $\mathcal{S}_{\text{ext}}$ a free quantum action one obtains 
\begin{equation}
\rho^{\text{sp}}_{\pm}=|\Psi_{\text{\upshape PW}}\rangle \langle \Psi_{\text{\upshape PW}}|\,,
\end{equation}
with $|\Psi_{\text{\upshape PW}}\rangle$ the physical PW state describing $|\psi\rangle$.
\end{theorem}
We see that under these additional assumptions $\rho^{\text{sp}}_{\pm}$ is exactly the PW state of Eq. \eqref{eq:pwstate}. 
In this precise sense, \emph{the PW mechanism applied to a single particle arises from the RDM formalism of spacetime QM applied to many particles.} Standard PW states correspond to single particle states and free theories. 
In general, $\rho^{\text{sp}}_{\pm}$ provides a direct generalization of the PW scheme to many-particles and interacting theories. Higher order RDMs might also be considered. The proof of the theorem follows by direct evaluation and is provided in the Appendix \ref{app:reducedPW}. One important detail, is the ordering of the ladder operators, with only the $+-$ ordering corresponding to the PW scheme. \\

Let us conclude this section with a few remarks. It was recently noted that the SQM approach applied to QFTs may be regarded as a natural many-body completion of the PW formalism of a particle \cite{diaz2026quantum}. The results of this section reinforce this idea and, in particular, clarify the role of standard PW states as RDMs of the spacetime approach. Although we did not focus on a particular QFT here, and we worked in a discrete-time setting, the results generalize naturally to QFT scenarios (see also section \ref{sec:QFT}).
Let us also comment on the  ``Geometric event-based quantum mechanics'' proposal of \cite{giovannetti2023geometric}. There, the authors introduce a Hilbert space framework for events and use a foliation/conditioning  to recover standard relativistic QM and QFTs at a given spacelike surface. The framework may also be regarded as a non-trivial generalization of the PW scheme to QFT that leverages second quantization concepts. In fact, one aspect in common with the SQM description of a Fock space concerns the single-particle space. While their interpretation is in terms of event states, from which a corresponding Fock space is built, these states are in one-to-one correspondence with the single-particle space described in this section, or more precisely with its relativistic version \cite{diaz2019historystate,diaz2019history}.
The main difference with the present approach is that \cite{giovannetti2023geometric} is formulated in terms of ordinary quantum states in the extended Hilbert space. By contrast, SQM introduces spacetime states as the central objects, focusing not only on recovering predictions at a given time, but also providing meaning to timelike correlators.
This is precisely the structure that allowed us to recover PW states as RDMs.

\subsection{Timelike entanglement approach}\label{sec:timelikeent}

In \cite{milekhin2025observable} the authors proposed a definition of spacetime density matrix that can take into account correlation between two regions separated in time. Their proposal provides a microscopic definition for the entropy of two disjoint regions $A$ and $B$ deformed from a spacelike separation (standard QM and QFT) to a timelike distance. Indeed such entropic quantity has attracted recent interest within the context of holographic theories \cite{harper2023timelike}.

If $A$ and $B$ constitute the entire system at two different points in time, the main quantity of the proposal is
\begin{equation}
    T_{AB}=J(U) (\rho \otimes \mathbbm{1}_B)\,,
\end{equation}
with $J(U)=\sum_{i,j}^{\dim( h)} U^\dag |i\rangle \langle j| U\otimes |j\rangle \langle i|$ the Jamiolkowski state
associated with the (unitary) quantum channel $\text{Ad}_U(.)=U^\dag \,.\, U$. 
This operator is defined in such a way to obtain Wightman functions as
\begin{equation}\label{eq:TABwightman}
\begin{split}
    {\rm Tr}[T_{AB}\,O_A\otimes O_B]&={\rm tr}[\rho \, O_A O_B(t)]\\
       {\rm Tr}[T^\dag_{AB}\,O_A\otimes O_B]&={\rm tr}[\rho \, O_B(t) O_A]
\end{split}    
\end{equation}
with $O_B(t)=U^\dag O_B U$. Notice that $T_{AB}$ here acts on $\mathcal{H}=h\otimes h$.
In general, one would be interested in arbitrary regions $A$ and $B$, such as two subregions of space at different times, with $T_{AB}$ defined so that it satisfies Eqs.\ \eqref{eq:TABwightman} for the subregions.  The entropies of 
$T_{AB}$ are the main focus of \cite{milekhin2025observable}.

To make direct contact with our scheme, consider first spacetime states for $N=2$. We can write $\mathcal{R}=(\rho \otimes \mathbbm{1})\, \text{SWAP} \,U(\epsilon)\otimes U^\dag(\epsilon)$. Here, $U(\epsilon)=e^{-i\epsilon H}$ constitutes a single step evolution. Interestingly, if we consider arbitrary $N$ and then a partial trace over all times except the initial slice and some final slice at position $t$, we obtain
\begin{equation}\label{eq:invariance}
    {\rm Tr}_{\,\overline{h_0 \otimes h_t}}[\mathcal{R}]=(\rho\otimes \mathbbm{1})\,\text{SWAP} \,U(t)\otimes U^\dag(t)\,,
\end{equation}
for $U(t)\equiv e^{-i \epsilon t H}$ and $\overline{h_0 \otimes h_t}\equiv \otimes_{t'\neq 0, t} h_{t'}\cong h\otimes h$. We see that under partial traces over all times but two, the spacetime state preserves its structure but $U\to U^t$. The proof of this relation is elementary and illustrated in Figure \ref{fig:invariance}. On the other hand, one can easily check that ${\rm Tr}_{\,\overline{h_0 \otimes h_t}}[\mathcal{R}^\dag\,]=\text{SWAP} \,U(t)\otimes U^\dag(t) (\rho\otimes \mathbbm{1})$.
This leads directly to the following result.

\begin{figure}
    \centering
\includegraphics[width=0.72\linewidth]{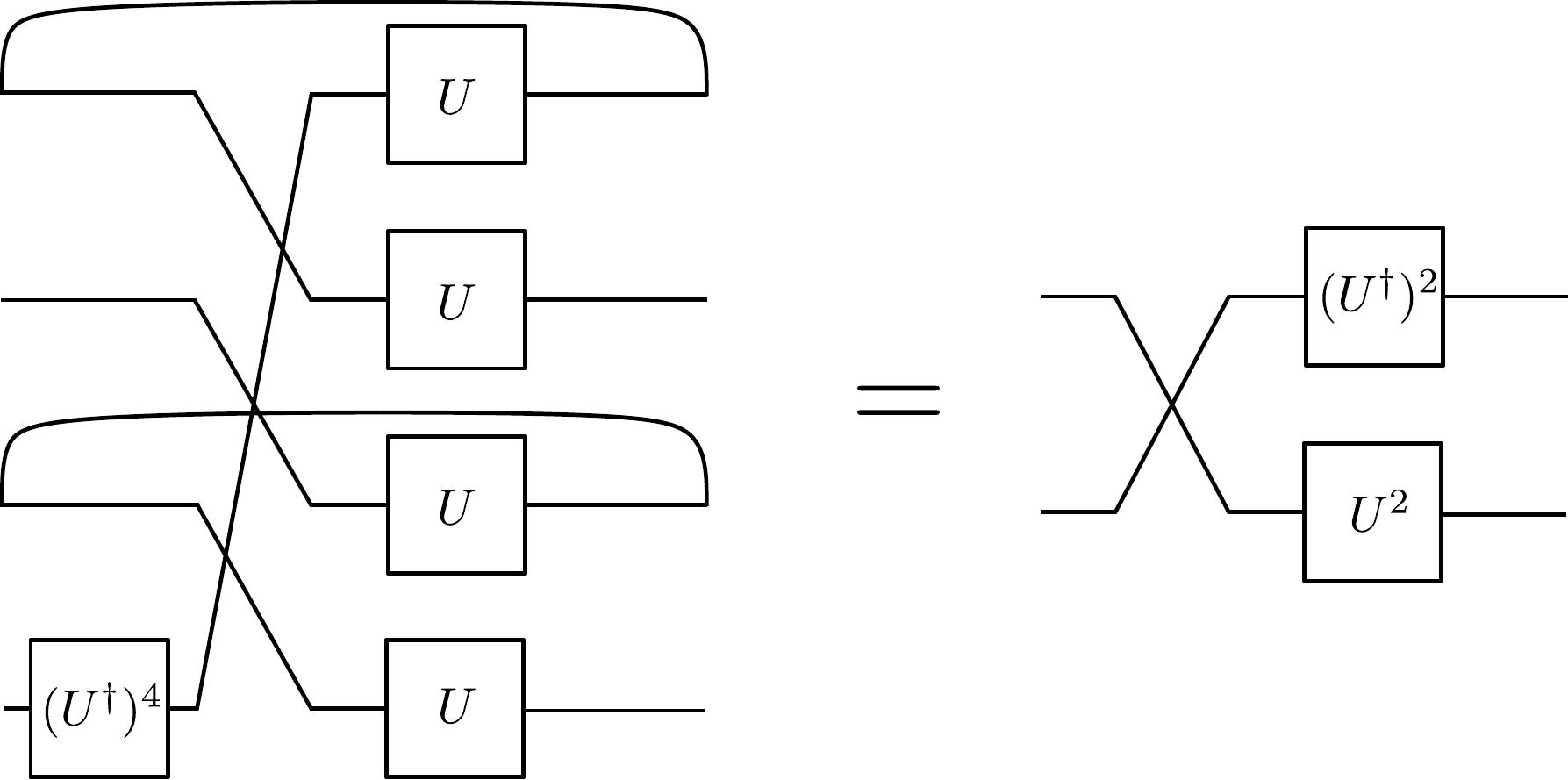}
    \caption{Illustrative example of Equation \eqref{eq:invariance}. We consider the diagrammatic representation of ${\rm Tr}_{\, \overline{h_0 \otimes h_2}}[e^{i\tilde{\mathcal{S}}}]$ for $N=4$. Here $U\equiv U(\epsilon)$. Notice how the structure of $e^{i\tilde{\mathcal{S}}}$ is preserved after the partial trace  but the step of evolution has been replaced by $\epsilon \to 2\epsilon$. In general, a partial trace over $\overline{h_0 \otimes h_t}$ corresponds to $\epsilon \to t \epsilon$ as it follows analogously  for any $N$. }
    \label{fig:invariance}
\end{figure}

\begin{lemma}\label{th:lemmaTAB}
\label{th:entanglement}
For $A,B$ representing the entire system at an initial and final time $t$, the operator $T_{AB}$ is given by
    \begin{equation}
        T_{AB}={\rm Tr}_{\,\overline{h_0 \otimes h_t}}[\mathcal{R}^\dag\,]\,.
    \end{equation}
\end{lemma}
Notice that $\overline{h_0 \otimes h_t}\equiv \, \overline{A \cup B}$, namely the complement of $A\cup B$.  One can prove this lemma by direct evaluation and by noting that  
\begin{equation}
\begin{split}
     J(U)&=\sum_{i,j}^{\dim( h)} U^\dag |i\rangle \langle j| U\otimes |j\rangle \langle i|\\&= (U^\dag \otimes \mathbbm{1})\, \text{SWAP} \,(U\otimes \mathbbm{1})={\rm Tr}_{\,\overline{h_0 \otimes h_t}}[e^{-i\tilde{\mathcal{S}}}]\,.
\end{split}
\end{equation}
It is also instructive to use a diagrammatic representation. In Figure \ref{fig:timelikeent1} we represent $T_{AB}$ for $A$ and $B$ representing the entire system, following the diagrams in \cite{milekhin2025observable}. Then, we deform the diagram to show the equivalence with spacetime states, described by the diagrams in Section \ref{sec:formalism}. 
\begin{figure}[h!]
    \centering
\includegraphics[width=0.95\linewidth]{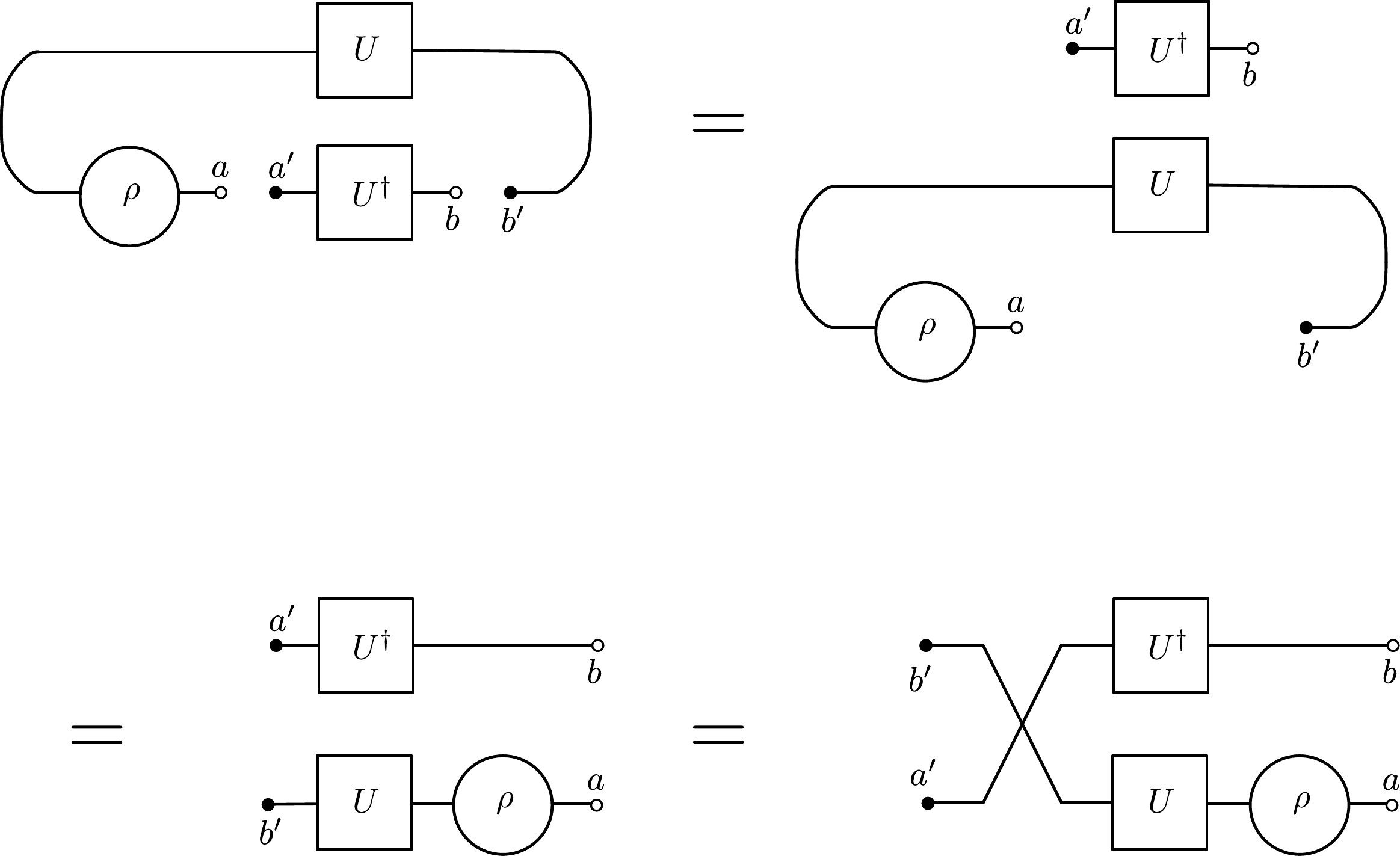}
    \caption{Diagrammatic proof of Lemma \ref{th:lemmaTAB}. We start with the representation of $T_{AB}$ provided in \cite{milekhin2025observable}. The diagram can be broken in two parts and rearranged to explicitly show its tensor structure in $h\otimes h$. Then, the indices can be adjusted so that the $A$ ($B$) indices $a,a'$ ($b,b'$) act on the corresponding time slice. This introduces a SWAP operator making the connection with spacetime states explicit. }
    \label{fig:timelikeent1}
\end{figure}

The natural generalization of Lemma \ref{th:lemmaTAB} to arbitrary regions is provided by the following theorem. 
\begin{theorem}\label{th:TABgeneral}
Consider regions of space (or subsystems) $A$ and $B$ at different times and their corresponding  complements $\bar{A}$ and $\bar{B}$ for the corresponding slices. Then
    \begin{equation}
        T_{AB}={\rm Tr}_{\,\overline{A \cup B}}[\mathcal{R}^\dag]\,.
    \end{equation}
\end{theorem}
By definition, $\mathcal{R}$ and $T_{AB}$ provide time ordered correlation functions so the theorem follows by uniqueness. We also provide an instructive diagrammatic proof in Figure \ref{fig:timelikeent2} based on deforming the diagram in \cite{milekhin2025observable} to the partial trace of $\mathcal{R}^\dag$.

\begin{figure}[t!]
    \centering
    \includegraphics[width=\linewidth]{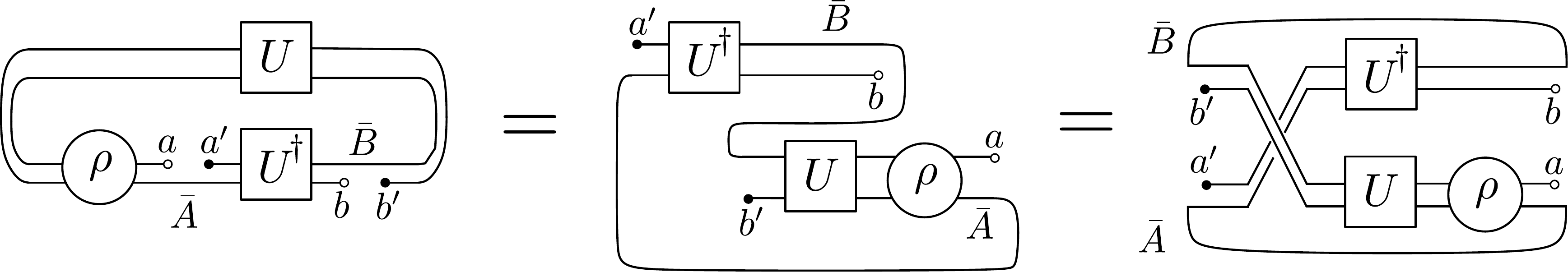}
    \caption{Diagrammatic proof of Theorem \ref{th:TABgeneral}. The main ideas of Figure \ref{fig:timelikeent1} stand but with the addition of the complements $\bar{A}$ and $\bar{B}$. The previously broken pieces of the diagram are now connected in these additional regions. The final diagram is ${\rm Tr}_{\, \overline{h_0\otimes h_t}}[\mathcal{R}^\dag]$ (for $U\equiv U(\epsilon)^t$) but with closed loops in $\bar{A}$ and $\bar{B}$, indicating the partial trace over these regions.}
    \label{fig:timelikeent2}
\end{figure}

The previous theorem establish the general connection between the spacetime scheme and $T_{AB}$. Then, we note that one of the main motivations of \cite{milekhin2025observable} concerns the definition of timelike entanglement in QFTs, and in particular providing a ``microscopic'' basis to the recently introduced \cite{harper2023timelike} timelike pseudoentropy in the context of the AdS-CFT and dS-CFT correspondence. Let us then discuss the consequences of the results of this section for QFTs (see also Section \ref{sec:QFT}).

\begin{figure*}[t!]
    \centering
\includegraphics[width=0.95\linewidth]{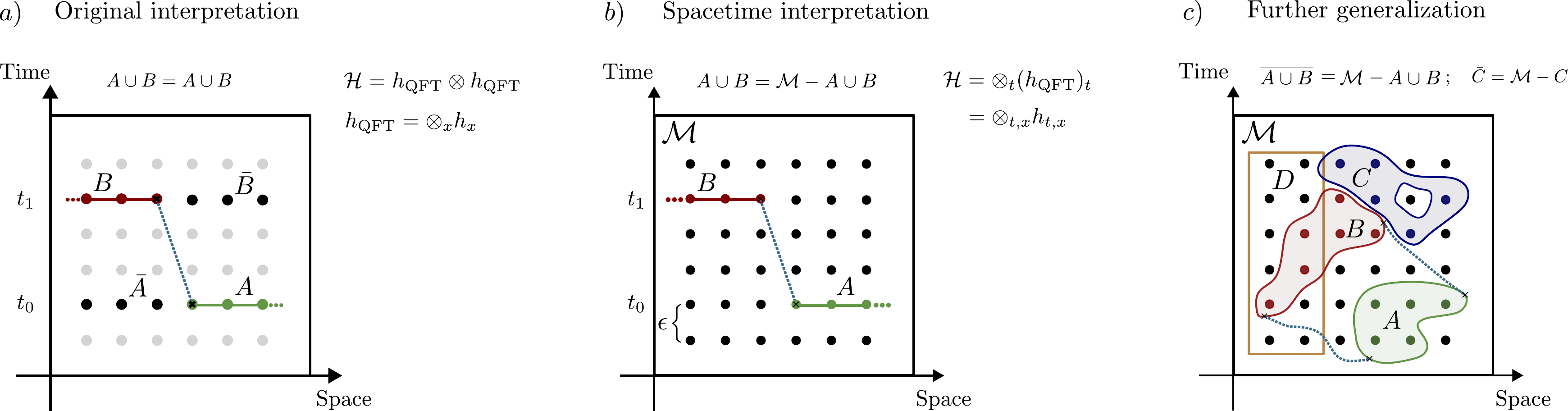}
    \caption{Different interpretations of timelike entanglement and generalizations. In Panel a) we depict the formulation presented in \cite{milekhin2025observable} where the entropy of $T_{AB}$ corresponds to a partial trace over the complements $\bar{A}$ and $\bar{B}$ at two different Cauchy surfaces. The partial trace is thus defined over the Hilbert space $\mathcal{H}=h_{\text{QFT}}\otimes h_{\text{QFT}}$ for $h_{\text{QFT}}$ the standard Hilbert space of a QFT (the one defining usual spacelike entanglement). The spacetime formalism shares this interpretation for $N=2$ but it also provides a novel point of view: the partial trace is over all Minkowski spacetime $\mathcal{M}$ other than $A$ and $B$. This is depicted in panel b) and follows from $\mathcal{H}=\otimes_t (h_{\text{QFT}})_t$, namely from the use of the spacetime Hilbert space. As depicted in Panel c), this perspective paves the way to considering partial traces over arbitrary subregions of $\mathcal{M}$, either connected or disconnected.
    Let us also recall \cite{milekhin2025observable} that for CFTs the entropy $-{\rm Tr}[T_{AB}\log{T_{AB}}]$ coincides with the pseudo-entanglement entropy \cite{harper2023timelike} for a single timelike interval (blue, dashed; panels a) and b)). The implications for CFT temporal pseudo-entanglement of scenarios depicted on panel c) are yet to be explored. }
    \label{fig:spacetimetraces}
\end{figure*}

One convenient consequence of the spacetime formalism is its direct relation to the PI formulation we discussed in Section \ref{sec:PIs}. Let us exploit it to connect Theorem \ref{th:TABgeneral} with the PI-based diagrams provided in \cite{milekhin2025observable} to justify the relation between entropies of $T_{AB}$ and analytical extensions of $\rho_{AB}$ (standard entanglement entropies) in QFT. Let us recall first that in the infinite volume limit, and considering a small imaginary part $i \epsilon$ for time, we can write the spacetime state of the vacuum of a QFT directly as $\mathcal{R}=e^{i\mathcal{S}}/{\rm Tr}[e^{i\mathcal{S}}]$ \cite{diaz2023spacetime} (see also Section \ref{sec:QFT}). Then, omitting overall constants \footnote{We refer the reader to \cite{diaz2023spacetime,diaz2021path} for subtleties in the continuum time limit. } the matrix elements of $T_{AB}^\dag$ in the field basis are given by  
\begin{equation}
\begin{split}
    T^\dag_{AB}[\phi^{-}_A, \phi^{-}_B,\phi^{+}_A, \phi^{+}_B]\equiv \langle \phi^{+}_A, \phi^{+}_B|\,{\rm Tr}_{\,\overline{A \cup B}}[e^{i\mathcal{S}}]\,|\phi^{-}_A, \phi^{-}_B\rangle\\=\int \mathcal{D}\phi_{\,\overline{A \cup B}}\,\langle \phi^{+}_A, \phi^{+}_B,\phi_{\,\overline{A \cup B}} |\,e^{i\mathcal{S}}\,|\phi^{-}_A, \phi^{-}_B, \phi_{\,\overline{A \cup B}}\rangle\,.
\end{split}
\end{equation}
Notice that we just used the definition of spacetime state and partial trace in $\mathcal{H}$ (the field $\phi_{\,\overline{A \cup B}}$ is defined in spacetime).
On the other hand, the results of Section \ref{sec:PIs} imply that $\langle \phi^{+}_A, \phi^{+}_B,\phi_{\,\overline{A \cup B}} |\,e^{i\mathcal{S}}\,|\phi^{-}_A, \phi^{-}_B, \phi_{\,\overline{A \cup B}}\rangle=e^{iS_{\text{cl}}}$, namely the exponential of the classical action evaluated along $\phi_{\,\overline{A \cup B}}$ with border conditions around $A$ and $B$. Then, the partial trace over $\overline{A \cup B}$ becomes a PI with proper boundary conditions defining the matrix elements of $T_{AB}^\dag$ (similar considerations hold for the adjoint). Thus, we have recovered the PI-based  diagram suggested in \cite{milekhin2025observable} to justify the definition of $T_{AB}$ (we are working with slightly imaginary Lorentzian time instead of Euclidean time but the diagrams are isomorphic). Besides working as a proof of consistency among approaches, this line of thought provides additional operational meaning to the entropy of $T_{AB}$. Namely, $(A\cup B) \cup \, (\overline{A\cup B})$ is the whole spacetime, meaning that the entropy of $T_{AB}$ is the entropy of $\mathcal{R}$ reduced over all spacetime points but $A\cup B$. This new interpretation is illustrated in Figure \ref{fig:spacetimetraces} where we also emphasize that the formalism paves the way to considering arbitrary regions of spacetime. 
We also recall \cite{milekhin2025observable} that $-{\rm Tr}[T_{AB}\log T_{AB}]$ in CFT is equal to the timelike pseudo entropy \cite{harper2023timelike} corresponding to the timelike separation between $A$  and $B$. This is the precise sense in which the approach described in this section provides a microscopic definition to the proposal in 
\cite{harper2023timelike} based on analytically extending spacelike regions to timelike ones. The current considerations pave the way to studying pseudo entropies of far more general regions, the meaning of which remains to be study. Some recent investigations among these lines have been considered in \cite{guo2025spacetime, guo2025entanglement}.

The previous comments on QFT also apply to general systems with $\mathcal{H}_{(A\cup B) \cup \, (\overline{A\cup B})}\equiv \mathcal{H}$. Importantly, let us remark that the results in Section \ref{sec:pureststates} apply to any bipartition and in particular to  $A'=A\cup B$, $B'=\overline{A\cup B}$ (whether for QFT or other quantum theories). This means that the interpretation of Figure \ref{fig:spacetimetraces} is rigorously justified and directly linked to the joint Schmidt decomposition of Eq.\ \eqref{eq:jointschmidt}. In summary, the SQM formalism (together with its direct relation with the scheme in \cite{milekhin2025observable}), and the results we presented in Section \ref{sec:pureststates}, provide a solid basis for a rigorous definition of timelike entanglement, properly grounding in a Hilbert-space formalism different discussions \cite{harper2023timelike, chu2023time, narayan2023notes, nunez2025timelike, heller2025temporal, jiang2025timelike} based on holographic and/or path integral considerations.

\section{Additional features and perspectives}\label{sec:additional}

Having introduced the SQM formalism and clarified its relation to other spacetime approaches, it is natural to take a step back and ask how it fits within a broader physical context. In this section we revisit several topics in the literature connected to temporal properties in QM in light of the SQM formalism. We also clarify its relation to classical physics.

\subsection{Leggett-Garg inequalities}\label{sec:legget}
The Leggett-Garg inequality \cite{leggett1985quantum}, sometimes also called Bell-inequality in time, is a statement on sequential measurements at different times: for a dichotomic measurement taking values $Q=\pm 1$ under the assumptions of  realism and
noninvasive measurability one obtains
\begin{equation}
    -3 \leq K=C_{21}+C_{32}-C_{31}\leq 1
\end{equation}
for $C_{ij}=\langle Q(t_i)Q(t_j)\rangle$ the expectation value of the measurements at times $t_i,t_j$.

The natural generalization of $C_{ij}$ to the quantum case \cite{emary2014leggett} is given by $C_{ij}=\frac{1}{2}{\rm tr}[\rho \{Q(t_i),Q(t_j)\}]$ where the anticommutator arises 
from considering sequential measurements (just as we discussed in Section \ref{sec:PDM}). Under this scheme, one can properly choose the evolution of a quantum system to find $K>1$ thus violating the classical inequality. At the heart of this violation lies the fact that in QM the order in which measurements described by non-commuting observables are performed affects the outcome.

As a direct consequence of the spacetime formalism we can write
\begin{equation}
    C_{ij}={\rm Tr}\left[\frac{\mathcal{R}+\mathcal{R}^\dag}{2}Q_{t_i}Q_{t_j}\right]\,.
\end{equation}
We can then use this simple observation, and linearity, to obtain
\small
\begin{equation}
    K={\rm Tr}\left[\frac{\mathcal{R}+\mathcal{R}^\dag}{2}\left(Q\otimes Q\otimes \mathbbm{1}+\mathbbm{1}\otimes Q \otimes Q-Q\otimes \mathbbm{1}\otimes Q\right)\right]
\end{equation}
\normalsize
where we considered the case $N=3$ for concreteness since only $3$ operators have been inserted \footnote{arbitrary evolutions can be considered by adding a simple adjoint action on $\mathcal{R}$; equivalently one can take arbitrary $N$ and partial traces}. Notice also that the quantity determining the values of $K$ is the quantum state over time generalized to $N\geq 3$ (Section \ref{sec:QSOT}). Thus while $\mathcal{R}-\mathcal{R}^\dag$ is linked to causality, $\mathcal{R}+\mathcal{R}^\dag$ is linked to the violation of Leggett-Garg inequalities.

Let us now notice that the operator $\mathcal{K}=Q\otimes Q\otimes \mathbbm{1}+\mathbbm{1}\otimes Q \otimes Q-Q\otimes \mathbbm{1}\otimes Q$ has eigenvalues $\lambda=\{-3,1\}$, as one can easily verify by considering the possible eigenvalues of the (commuting) $Q_t$. This leads to a simple but notable conclusion: if $\frac{\mathcal{R}+\mathcal{R}^\dag}{2}$ were a quantum state it would not be able to violate the Leggett-Garg inequality. In other words, for $|\Psi\rangle\in \mathcal{H}$
\begin{equation}
    -3\leq \langle \Psi|\mathcal{K} |\Psi\rangle\leq 1
\end{equation}
which is just the classical inequality: \emph{no standard quantum state defined in spacetime can violate the Leggett-Garg inequality}. This conclusion is independent of how much ordinary entanglement the state contains, and applies equally to mixed states.

The previous result is a direct consequence of the fact that the spacetime point of view associates each $C_{ij}$ to a pair of  \emph{commuting} observables $Q_{t_i}, Q_{t_j}$ (while generally $[Q(t_i),Q(t_j)]\neq 0$ in the Heisenberg picture). This means that the information of ``invasiveness'' of the measurement process needs to be contained exclusively at the level of the spacetime state by means that cannot be captured by standard quantum states. 
We can be even more precise. Write $\mathcal{K}=\mathbbm{1}-4\,\Pi$ for $\Pi=P_+ \otimes P_- \otimes P_++P_- \otimes P_+ \otimes P_-$ the projector onto the alternating sector, with $P_\pm = (\mathbbm{1}\pm Q)/2$. Notice in fact that these alternating cases $(+,-,+)$, $(-,+,-)$ are the only two choices of eigenstates of $Q_t$ giving an eigenvalue $\lambda=-3$ of $\mathcal{K}$. 
Then, 
\begin{equation}
    K=1-4\, {\rm Tr}\left[\Big(\frac{\mathcal{R}+\mathcal{R}^\dag}{2}\Big) \Pi\right]
\end{equation}
showing that the violation of the Leggett-inequality corresponding to $K>1$ is precisely the condition
\begin{equation}
    {\rm Tr}\left[\Big(\frac{\mathcal{R}+\mathcal{R}^\dag}{2}\Big) \Pi\right]< 0\,.
\end{equation}
 One thus needs that the Hermitian state over time object assigns negative weight to the alternating-history sector.
This is impossible if $(\mathcal{R}+\mathcal{R}^\dag)/2$
 is a positive semi-definite quantum state.

This conclusion can be generalized to higher order inequalities such as $-2\leq C_{21}+C_{32}+C_{43}-C_{41}\leq 2$ \cite{emary2014leggett} which can be recovered from the eigenvalues of $Q\otimes Q\otimes \mathbbm{1}\otimes \mathbbm{1}+\mathbbm{1}\otimes Q \otimes Q\otimes \mathbbm{1}+ \mathbbm{1}\otimes \mathbbm{1} \otimes Q \otimes Q-Q\otimes \mathbbm{1}\otimes  \mathbbm{1}\otimes Q$ which are equal to $\lambda=\pm 2$. 
This is a rather peculiar result because it allows one to derive \emph{classical} inequalities in time by considering the bounds that follow from standard \emph{quantum} states in $\mathcal{H}$. This means that  spacetime states $\mathcal{R}$ in general contain 
``correlations'' 
that go beyond what is captured by standard states.
Moreover, if one allows for quantum channels, one can find examples, such as the one considered in Eqs.\ \eqref{eq:collapsingdecoherence} and \eqref{eq:collapsingchannel}, of spacetime states collapsing to standard quantum states, which  correspond to quantum situations where the Leggett-Garg inequalities cannot be violated.

\subsection{Folded spacetime states and OTOCs}\label{sec:otocs}

The formalism we presented in Section \ref{sec:formalism} mostly involved (anti) time-ordered correlation functions. One can easily apply a similar scheme to out of order correlators (OTOCs), thus also recovering arbitrary Wightman functions. As a matter of fact, the forward and backward evolution obtained through $\mathcal{S}_{\text{ext}}$ is a particular case of what is needed in general. In this section we  explain how spacetime states naturally accommodate foldings of time, thus generalizing the theorems of Section \ref{sec:formalism} to OTOCs. As we now explain, \emph{the folded case can still be understood as a particular case of the formalism of  Section \ref{sec:formalism}} where the evolution within slices is no longer uniform.

Let us first consider the example of up to $4$ point correlators. We define
\small
\begin{equation}
\mathcal{V}_4=\mathop{{\textstyle\bigotimes}}\limits_{t=0}^{N-1}e^{i\epsilon tH}\mathop{{\textstyle\bigotimes}}\limits_{t=0}^{N-1}e^{-i\epsilon (N-t)H}\mathop{{\textstyle\bigotimes}}\limits_{t=0}^{N-1}e^{i\epsilon tH}\mathop{{\textstyle\bigotimes}}\limits_{t=0}^{N-1}e^{-i\epsilon (N-t)H}
\end{equation}
\normalsize
an operator acting on $\mathcal{H}^{\otimes 4}$ (equivalently we can divide  $\mathcal{H}$ in $4$ pieces). Then, for $e^{i\epsilon \mathcal{P}_4}$ the time translation operator across $4N$ slices we introduce
\begin{equation}\begin{split}
      e^{i\mathcal{S}_4}&=\mathcal{V}_4^\dag e^{i\epsilon \mathcal{P}_4} \mathcal{V}_4\\&=e^{i\epsilon \mathcal{P}_4}e^{-i\epsilon \mathcal{K}}\otimes e^{i\epsilon \mathcal{K}}\otimes e^{-i\epsilon \mathcal{K}}\otimes e^{i\epsilon \mathcal{K}}\,.
\end{split}
\end{equation}
While the time translation operator across slices is the same as before, now on each block $\mathcal{H}$ we find different time-translations within slices. 
In complete analogy with the results of Section \ref{sec:formalism} this leads to
\begin{equation}
\begin{split}
    {\rm Tr}[e^{i\mathcal{S}_4}A_{1,t_1}\otimes B_{2,N-t_2}\otimes C_{3,t_3}\otimes D_{4,N-t_4}]\\={\rm tr}[D(t_4)C(t_3)B(t_2)A(t_1)] 
\end{split}
\end{equation}
where $O_{i,t}$ indicates an operator acting on the $i$ copy of $\mathcal{H}$ and time slice $t$. While the index $t$ determines the amount of evolution, the index $i$ determines the ordering of the operators. If multiple operators are inserted in the same $\mathcal{H}$ one obtains either a time-ordered product or an anti-time ordered product within each ``slot''. The convention for the different operator insertions and how this relates to the corresponding Schwinger-Keldysh contour is depicted in Figure \ref{fig:S4}.

\begin{figure}[t!]
    \centering
    \includegraphics[width=0.75\linewidth]{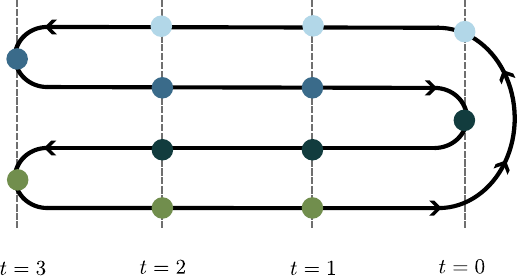}
    \caption{Example of the folded time contour corresponding to $\mathcal{S}_4$ for $N=3$. Each dot corresponds to a local Hilbert space $h$ with $\mathcal{H}^{\otimes 4}=(h^{\otimes 3})^{\otimes 4}=h^{\otimes 12}$. We use a different color for each $\mathcal{H}$. Following the contour we can assign to each operator insertion in a given $h$ the corresponding evolution time. The arrows point toward the direction of increasing index $t$ in the decomposition $\mathcal{H}=\otimes_{t=0}^{11} h_t$.}
    \label{fig:S4}
\end{figure}

Just as before, we can define a corresponding spacetime state $\mathcal{R}_4=\rho_0 e^{i\mathcal{S}_4}$ ($\rho_0$ acts on the first of the $4N$ copies of $h$) satisfying 
\begin{equation}\label{eq:basicotoc}
\begin{split}
    &{\rm Tr}[\mathcal{R}_4 A_{1,t_1}\otimes B_{2,N-t_2}\otimes C_{3,t_3}\otimes D_{4,N-t_4}]\\&={\rm tr}[\rho\, D(t_4)C(t_3)B(t_2)A(t_1)]\,. 
\end{split}
\end{equation}
From this relation we can recover common OTOCs. In particular, choosing $A=C^\dag=V$, with $t_1=t_2$ and $B=D^\dag=W$, with $t_2=t_4$ we obtain the OTOC
$F={\rm tr}[\rho\, W^\dag(t_2) V^\dag(t_1)W(t_2)V(t_1)]$, a quantity widely used in the literature, see e.g., \cite{roberts2017chaos, xu2020accessing}. We see that the norm of $\mathcal{R}_4$ can be used to bound $F$ directly. If in addition one properly inserts operators $y$ satisfying $y^4=e^{-\beta H}/Z$ on each copy (namely we split the initial state in $4$ parts distributed among copies) one can easily represent ${\rm tr}[y\, W(t_2) y V(t_1) y W(t_2)y V(t_1)]$,  a quantity appears in the seminal work \cite{maldacena2016bound}.

\begin{figure*}[ht!]
    \centering
    \includegraphics[width=0.95\linewidth]{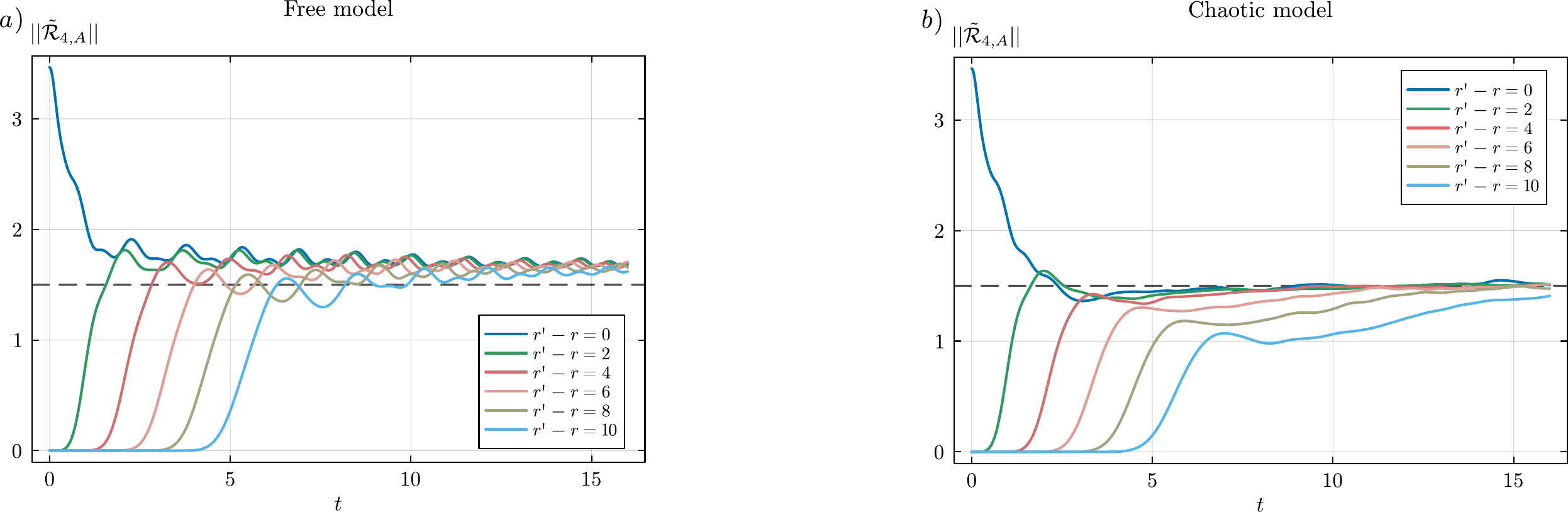}
    \caption{
Scrambling diagnostics for the tilted-field Ising chain of $51$ sites. 
Panels (a) and (b) show $||\widetilde{\mathcal R}_{4,A}||$ for the free model, $h_z=0$, and the chaotic model, $J=1$, $h_x=-1.05$, $h_z=0.5$, respectively, for several separations $r'-r$ ($r=26$). The dashed line marks the value $3/2$, corresponding to the isotropic late-time limit $G\simeq 2\mathbbm{1}_9$. 
}
    \label{fig:OTOCnorm}
\end{figure*}

In general we can define the $k$-folded spacetime state
\begin{equation}\label{eq:Rk}
    \mathcal{R}_k=\mathcal{V}^\dag_k \rho_0 e^{i\epsilon \mathcal{P}_k} \mathcal{V}_k
\end{equation}
acting on the spacetime Hilbert space with $kN$ time slots with corresponding time translation operator across slices $\mathcal{P}_k$ and with
\begin{equation}
\mathcal{V}_k=\mathop{{\textstyle\bigotimes}}\limits_{l=1}^{k} \mathop{{\textstyle\bigotimes}}\limits_{t=0}^{N-1}U^\dag_l (t\epsilon)\,,
\end{equation}
for $U_{l\; \text{odd}}(t\epsilon)=e^{-i\epsilon t H}$,  $U_{l\; \text{even}}(t\epsilon)=e^{i \epsilon(N-t)H}$. At this point we find it important to note that we can think of $\mathcal{R}_k$ as a particular case of the concept of spacetime state  $\mathcal{R}$ that we defined in Section \ref{sec:formalism}. The only difference is a time-dependent evolution induced by the folding. Considering that all the properties described in Section \ref{sec:pureststates} naturally generalize to the time-dependent case, we see that  $\mathcal{R}_k$ satisfy the same properties. In particular, when $\rho$ is pure $\mathcal{R}_k$ has vanishing entropies. The isospectrality across bipartitions also holds, and so on. Moreover, the spacetime state with no folding $\mathcal{R}$ acting on $h^{\otimes k N}$ is unitarily related with $\mathcal{R}_k$: as long as we chose the same number of time slices, we can unitarily relate arbitrarily folded spacetime states with same initial conditions. 
It is also clear that an inclusion between different $k$ arises, e.g., 
$\mathcal{R}={\rm Tr}_{\mathcal{H}_2}[\mathcal{R}_2]={\rm Tr}_{\mathcal{H}_2\otimes\mathcal{H}_3\otimes\mathcal{H}_4}[\mathcal{R}_4]$. Notice that here $\mathcal{R}_2$ is closely related to $\mathcal{R}_{\text{ext}}$ except for a convention regarding the arrangement of the backward evolution Hilbert space; see details on Appendix \ref{app:spacetimest}.

Let us now recall that the more direct connection between OTOCs and quantum chaos appears through unequal time commutators. Before exploring this it is interesting to connect the discussion with the imagitivity \cite{milekhin2025observable} we discussed in Sections \ref{sec:pureststates} and \ref{sec:timelikeent}. Before we had ${\rm Tr}\big[(\mathcal{R}-\mathcal{R}^\dag) V_{t_1}W_{t_2}\big]={\rm tr}\big\{\rho [W(t_2),V(t_1)]\big\}$ for $t_2>t_1$ leading to a bound on the expectation value of the unequal time commutator through $||\mathcal{R}_A-\mathcal{R}^\dag_A||$ with $A$ indicating the spatiotemporal support of the operators $V_{t_1},W_{t_2}$. 
At the same time, we can write the same quantity as 
\begin{equation}
\begin{split}
          & {\rm Tr}\big[\mathcal{R}_2 (W_{2,N-t_2}V_{1,t_1}-W_{1,t_2}V_{2,N-t_1})\big] \\&={\rm tr}\big\{\rho [W(t_2),V(t_1)]\big\}\,,
\end{split}
\end{equation}
where we used the two-folded spacetime state $\mathcal{R}_2$. 
Then, by  writing $\Phi_{12}(\mathcal{R}_2)=\text{SWAP}_{12} \tilde{U}\mathcal{R}_2 \tilde{U}^\dag\text{SWAP}_{12}$, where the SWAP operators act between the copies $\mathcal{H}_1\otimes \mathcal{H}_2$ and the $\tilde{U}$ operator readjust the evolution times, we can write
${\rm Tr}\big[\Phi_{12}(\mathcal{R}_2) W_{2,N-t_2}V_{1,t_1}\big]= {\rm Tr}\big[\mathcal{R}_2 W_{1,t_2}V_{2,N-t_1}\big]$. Now we can define $\widetilde{\mathcal{R}}_2=\mathcal{R}_2-\Phi_{12}(\mathcal{R}_2)$ so that 
\begin{equation}
        {\rm Tr}\big[\widetilde{\mathcal{R}}_2\, W_{2,N-t_2}V_{1,t_1}\big] ={\rm tr}\big\{\rho\, [W(t_2),V(t_1)]\big\}\,,
\end{equation}
with the commutator appearing directly. Notice that $\Phi(\mathcal{R}_2)$ corresponds to a different folding convention than the one defining $\mathcal{R}_2$ since the SWAP operator also affects the time translations across slices. We can then conclude that $\widetilde{\mathcal{R}}_{2,A'}=\mathcal{R}_A-\mathcal{R}^\dag_A$ for $A'$ the shared support of $W_{1,t_2},V_{2,N-t_1}$ so that we can replace the imagitivity with $||\widetilde{\mathcal{R}}_{2,A'}||$ in the bound of unequal time commutators of Eq.\ \eqref{eq:imagitivbound}.

 Let us now generalize this discussion to 4-point OTOCs.  By using Eq.\ \eqref{eq:basicotoc} we can write (we omit the symbols $\otimes$ for ease of notation)
\begin{equation}
\begin{split}\label{eq:otocC}
        &{\rm Tr}\big[\,\mathcal{R}_4\, \big(V^\dag_{4,N-t_1}W^\dag_{3,t_2}-V^\dag_{3,t_1}W^\dag_{4,N-t_2}\big)\\ &(W_{2,N-t_2}V_{1,t_1}-W_{1,t_2}V_{2,N-t_1}\big)\, \big]\\&={\rm tr}\{\rho\,  [W(t_2),V(t_1)]^\dag [W(t_2),V(t_1)]\}\,.
\end{split}
\end{equation}
In analogy to what we did with $\mathcal{R}_2$,  we can define 
\begin{equation}
\widetilde{\mathcal{R}}_4=\mathcal{R}_4+\mathcal{R}^\dag_4-\Phi_{12}(\mathcal{R}_4)-\Phi_{34}(\mathcal{R}_4)\,, 
\end{equation}
with $\Phi_{ij}$ properly exchanging $\mathcal{H}_i\otimes \mathcal{H}_j$ to generate the crossed terms, such that 
\begin{equation}
\begin{split}\label{eq:otocCcomm}
        {\rm Tr}\big[\widetilde{\mathcal{R}}_4\, V^\dag_{4,N-t_1}W^\dag_{3,t_2}W_{2,N-t_2}V_{1,t_1}\big]=C(t_1,t_2)\,,
\end{split}
\end{equation}
where we defined the commutator based OTOC
\begin{equation}
    C(t_1,t_2)={\rm tr}\{\rho\,  [W(t_2),V(t_1)]^\dag [W(t_2),V(t_1)]\}\,.
\end{equation}
From this relation
 we obtain the following bound on  $C$
\begin{equation}
    C(t_1,t_2)\leq  ||W||^2 ||V||^2 ||\widetilde{\mathcal{R}}_{4,A}||\,,
\end{equation}
with $A$ the shared support of the operators $V^\dag_{4,N-t_1},W^\dag_{3,t_2},W_{2,N-t_2},V_{1,t_1}$. This is the analogous of the imagitivity bound but generalized to $4$-point OTOCs involving commutators. Let us remark that the quantity $C(t_1,t_2)$ is ubiquitous to studies of quantum chaos and scrambling in the literature \cite{xu2020accessing}. 
\\

As an illustrative example of the previous ideas, in Figure \ref{fig:OTOCnorm} we depict the behavior of $||\widetilde{\mathcal{R}}_{4,A}||$ for the one-dimensional model  
\begin{equation}
        H = J\sum_{i=1}^{L} Z_i Z_{i+1} + h_x \sum_{i=1}^{L+1} X_i + h_z \sum_{i=1}^{L+1} Z_i \,,
    \end{equation}
    with $A=(t, r')\cup (0,r=L/2)$, namely initial time slice at the middle of the chain and arbitrary points $ (t, r')$. 
    We consider both the chaotic model with parameters $J=1$, $h_x=-1.05$, $h_z=0.5$, and the free model with $h_z\to 0$. We take $N=8000$ times with $\epsilon=0.002$ and $L=50$. The evolution is performed via a combination of matrix product states (MPS) evolution under Trotterization and matrix product operator (MPO) decompositions (see Appendix \ref{app:numerics} for details).  
Notice that by construction $||\widetilde{\mathcal{R}}_4||=\frac{1}{4}\sqrt{\sum_{i,j,k,l}|C_{ij,kl}(r,r',t)|^2}$  where we have defined the tensor 
\begin{equation}
     C_{ij;kl}(r,r';t)={\rm tr}\left\{ [P_{r,i}(t),P_{r',j}]^\dag[P_{r,k}(t),P_{r',l}]\right\}\,,
\end{equation}
with $P_{r,i}$ denoting the Pauli matrix $i$ at site $r$. 
It is interesting to notice that, for the chaotic model, we find 
$||\widetilde{\mathcal{R}}_{4,A}||\to 3/2$ at late times, while the free model exhibits rapid oscillations around, but systematically above, this value. This behavior can be understood directly from the structure of the tensor $C_{ij;kl}(r,r';t)$. Indeed, let us write
\begin{equation}
    C_{ij;kl}(r,r';t)
    =
    \langle\!\langle v_a|v_b\rangle\!\rangle
    \equiv G_{ab}\,,
\end{equation}
where $a=(i,j)$, $b=(k,l)$, and
$
    |v_a\rangle\!\rangle
    \equiv
    |[P_{r,i}(t),P_{r',j}]\rangle\rangle
$.
Thus, after vectorizing the commutators, $C_{ij;kl}$ can be regarded as a $9\times 9$ Gram matrix $G$ and $ ||\widetilde{\mathcal{R}}_{4,A}||
   =
    ||G||/4$. In the chaotic model, we numerically find that, at late times,
$G \simeq 2\mathbbm{1}_{9}$. 
This indicates that the commutator vectors become approximately orthogonal and have equal norm, i.e., that the late-time OTOC response becomes essentially independent of the particular local Pauli operators being probed. In this sense, the chaotic dynamics scrambles the local operator information isotropically within this operator sector. This  explains the asymptotic value of the norm, since
$
    ||\widetilde{\mathcal{R}}_{4,A}||
    \simeq 
    ||2\mathbbm{1}_9||/4
    \simeq
   \sqrt{9\times 2^2}/4
    =
    3/2
$.

\begin{figure*}[t!]
    \centering
    \includegraphics[width=0.94\linewidth]{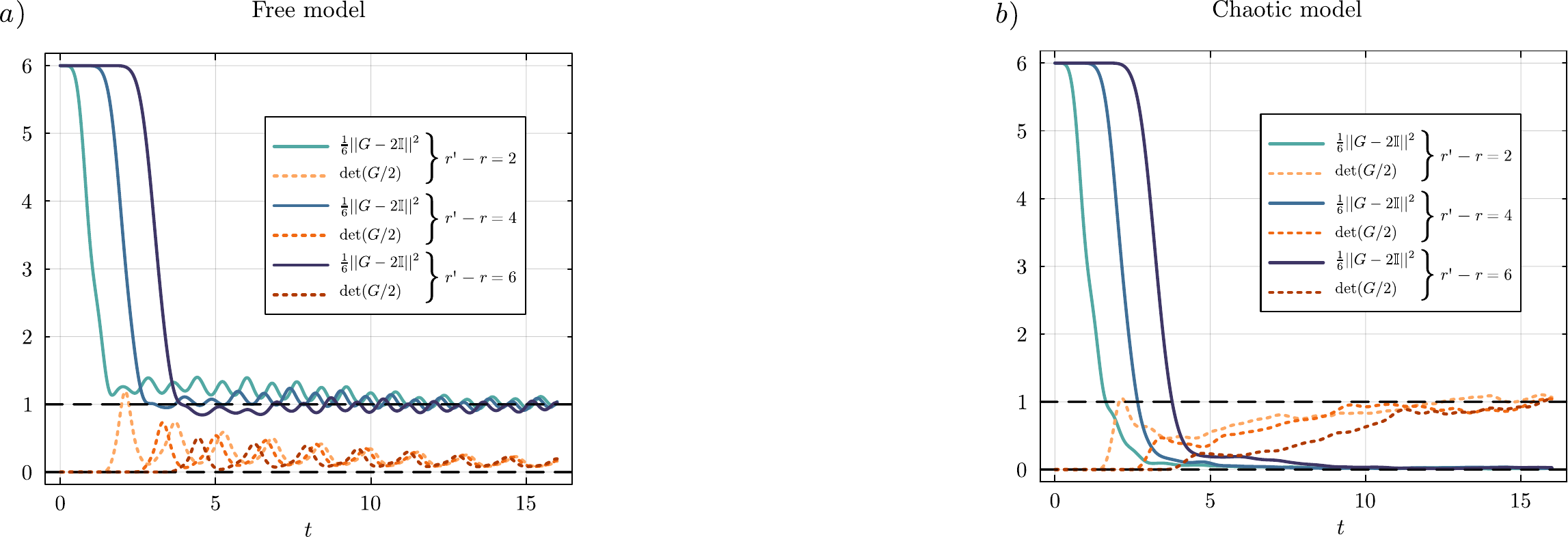}
    \caption{
Scrambling diagnostics for the tilted-field Ising chain. 
Panels (a) and (b) show the normalized distance $||G-2\mathbbm{1}_9||^2/6$ and $\det(G/2)$ for selected separations. The chaotic model approaches $G\simeq 2\mathbbm{1}_9$ (implying $\det(G/2)\simeq 1$), while the free model remains singular, with $\det(G/2)\simeq 0$, despite having a comparable value of $||\tilde{\mathcal R}_{4,A}||$.
}
    \label{fig:OTOCnormdiag}
\end{figure*}

The free model displays a different behavior. Although  $\mathrm{tr}[G]$ also approaches approximately the same late-time value $\mathrm{tr}[G]\simeq 18$, the spectrum of $G$ remains highly anisotropic (see Appendix \ref{app:numerics} for a depiction of  the eigenvalues of the Gram matrix). In particular, we find numerically that $\det [G]\simeq 0$ at late times, indicating that at least one direction in the space of Pauli-pair commutators remains linearly dependent, or has vanishing weight. This distinction, depicted in Figure \ref{fig:OTOCnormdiag} is important: individual OTOC witnesses $C(t)$, especially for operators that are nonlocal in the fermionic representation \cite{lin2018out}, can display behavior that superficially resembles chaotic scrambling even in a free theory.  By contrast, the Gram-matrix diagnostic clearly separates the two cases. That is, the chaotic model approaches a ``fully mixed'' Gram matrix, whereas the free model remains singular.
This also explains why $||\widetilde{\mathcal{R}}_{4,A}||$ stays above $3/2$ in the free model, as for a Gram matrix with fixed trace, the Frobenius norm is minimized when the nonzero (positive) eigenvalues are distributed uniformly. In the chaotic case, the weight is distributed over all nine directions while in the free case one eigenvalue is constrained to vanish.

The analysis above is largely numerical, but it clearly illustrates the potential of the spacetime-state formalism as a diagnostic of scrambling. In this analysis we have focused on  quantities derived from $\tilde{\mathcal R}_{4,A}$, such as its Frobenius norm and the associated Gram matrix of Pauli-pair commutators. One may also fully analyze $\tilde{\mathcal R}_{4,A}$ itself as a $16\times 16$ operator acting on four copies of the local Hilbert space, and the corresponding possible products of the form ${\rm Tr}[\tilde{\mathcal R}_{4,A} \mathcal{O}]$, where $\mathcal{O}$ may not be restricted to a separable operator in general. One may also combine information from different partitions.  
Let us also recall that the imagitivity of $\mathcal{R}$ can be estimated through quantum computing protocols \cite{milekhin2025observable}, so one may attempt to devise similar schemes to access properties of $\tilde{\mathcal R}_{4,A}$. An additional interesting perspective is the study of the entropic properties of folded spacetime states and their potential relation with timelike entanglement in CFTs (see Section \ref{sec:timelikeent}).
We leave these avenues for future work.\\

Returning to a conceptual discussion, let us finally point out that the results of this section also clarify the scope of the causal structure assumed in this manuscript. Most of our constructions use a fixed ordering of the time slices, encoded both by the SWAP chain and the time translations within slices entering the quantum action. Folded spacetime states do not remove this assumption, but they show that the formalism can naturally handle different time orderings, both at the level of time translations within slices and also across slices. This shows that the 
formalism already contains the ingredients needed to compare, mix, and coherently combine different orderings or foldings, and thus to encompass indefinite causal orderings \cite{castro2018dynamics}. 
For example one could coherently combine different foldings of spacetime such as the two-time pair $\mathcal R$ and $\mathcal R^\dagger$ . As we have seen both can be obtained from a common enlarged object, namely $\mathcal{R}= {\rm Tr}_{\mathcal{H}_2}[\mathcal{R}_{\text{ext}}]$, $\mathcal{R}^\dag= {\rm Tr}_{\mathcal{H}_2}[W\mathcal{R}_{\text{ext}}W^\dag]$ (see also Figure \ref{fig:PDMSWAPs} in Appendix~\ref{app:PDM}). This means that one may in principle also define ``non-diagonal'' reduced objects $\mathcal{R}_{01}\equiv {\rm Tr}_{\mathcal{H}_2}[\mathcal{R}_{\text{ext}}W^\dag]$, $\mathcal{R}_{10}\equiv {\rm Tr}_{\mathcal{H}_2}[W\mathcal{R}_{\text{ext}}]$ that would naturally arise if the folding is controlled by a quantum system.
See also the discussion in the Conclusions \ref{sec:conclusions}.

\subsection{Spacetime states as tensor networks and relation with TN-temporal entanglement}\label{sec:TN}

Over the past few years, the TN community has begun to study how to extract information from time-evolved quantum states without explicitly carrying out the evolution \cite{banuls2009matrix, lerose2021scaling, ye2021constructing}. By fully laying out the TN corresponding to, e.g., the calculation of an expectation value of a time-evolved quantum state represented as a MPS, one can devise contraction strategies that improve upon simply applying the evolution operator to the state. In particular, a notion of temporal entanglement naturally arises in the context of influence functionals (IFs) \cite{lerose2021scaling, park2025simulating}, a framework in which contractions are carried out ``orthogonally'' to the time direction.

In this section, we discuss how spacetime states admit a natural TN representation with physical legs corresponding to different spacetime points. The main result we want to emphasize here is that contracting this network in the appropriate way reproduces the IF used to define temporal entanglement, while leaving selected spacetime regions open, or equivalently partially tracing the complementary region, gives the reduced spacetime operators whose spectra define pseudo-entropic quantities. Thus, spacetime states naturally provide a common object for IF temporal entanglement and spacetime pseudo-entropies. 
While our focus will mostly be conceptual, this common TN representation also indicates where standard numerical tools may enter: the contraction of the network can be treated with existing approximate contraction methods, with IF-type contractions providing a natural baseline, even when operators are inserted at different times. We will also briefly suggest some possible alternative routes for exploiting the TN representation of spacetime states for actual numerical implementations. 

\begin{figure}[t!]
    \centering
    \includegraphics[width=0.74\linewidth]{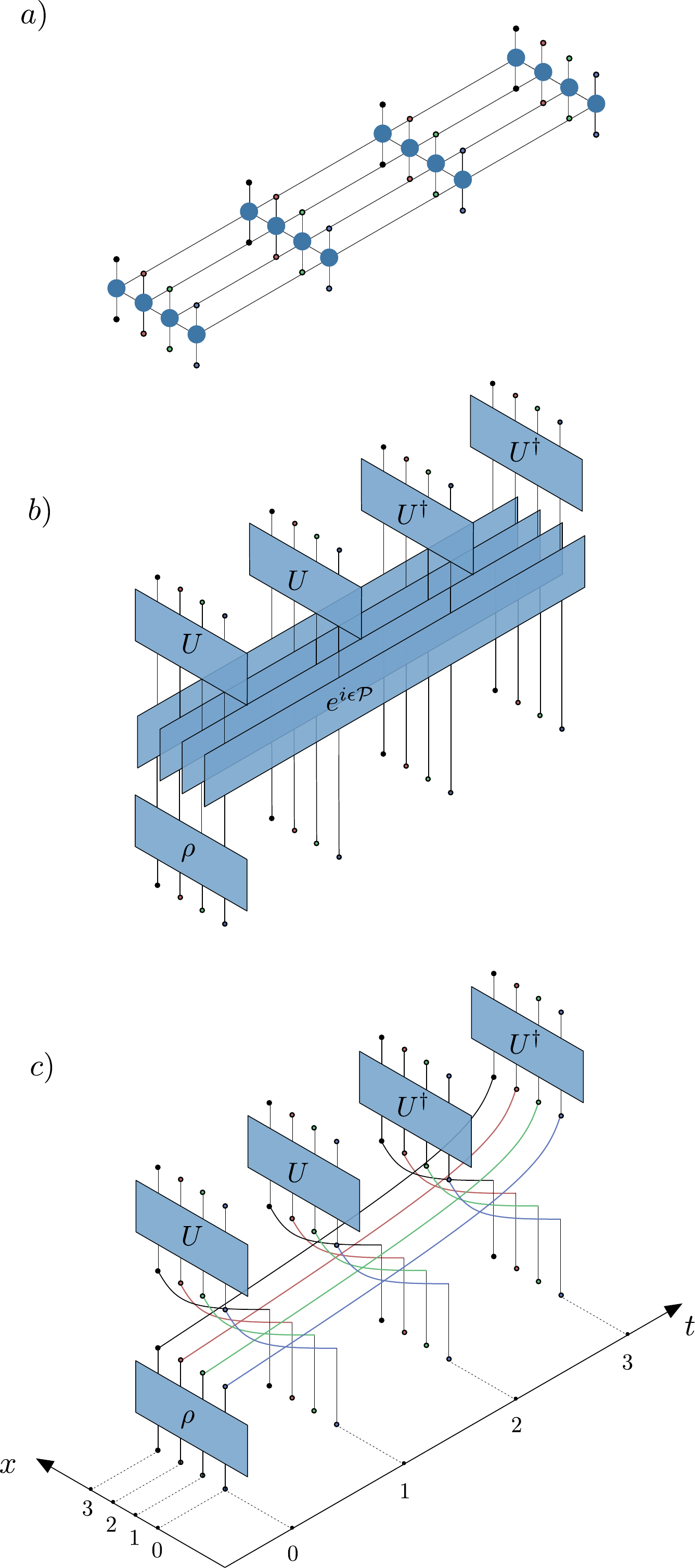}
    \caption{Spacetime state $\mathcal{R}_2$ as a TN with increasing level of detail. We are considering $4$ spatial sites and $4$ total time slices, with the first  half corresponding to forward evolution and the second half to backward evolution. On panel a) we simply emphasize that we can think of $\mathcal{R}_2$ as a PEPO whose structure comes from treating  time as a dimension, i.e., we have $d+1$ dimensional tensor with a physical leg for each spacetime point (in the representation we take $d=1$). On panel b) we represent the structure of $\mathcal{R}_2$:  the evolution within slices is separable in time, and the evolution across slices is separable in space. On panel c) we are explicit on the structure of each time translation across slices, thus represented as a cyclic permutation among the indices. For each space slice this tensor looks just like in Figure \ref{fig:tntheorem}. We have used different colors for each spatial site for clarity. }
    \label{fig:tnofR}
\end{figure}

Spacetime states can naturally be viewed as operators living on a $(d+1)$-dimensional lattice, generated by first placing an initial quantum state on the first temporal slice, and then applying the necessary operations along both the spatial ($d$-dimensional) and temporal directions. 
Since the IF framework best applies to the calculation of expectation values, where forward and backward time evolution appear symmetrically, and for the sake of keeping the discussion as simple as possible, we now focus on $\mathcal{R}_2$ defined as in Eq.~\eqref{eq:Rk} (for $k=2$). We also recall that $\mathcal{R}={\rm Tr}_{\mathcal{H}_2}[\mathcal{R}_2]$ can be represented by very similar means by just considering forward evolution only and adding a border term. Lastly, $\mathcal{R}_{\text{ext}}$ is also a slight variation of $\mathcal{R}_2$ as explained in Appendix \ref{app:spacetimest}.

In Figure~\ref{fig:tnofR} we fix $d=1$ and represent $\mathcal{R}_2$ as a $2$-dimensional TN, which could further be encoded as a projected entangled pair operator (PEPO). Notice that this would allow one to access approximate contraction methods such as belief propagation \cite{park2025simulating} and edge MPS \cite{lubasch2014unifying}. For concreteness, we consider just two (doubled) time slices and four spatial slices. The general case has the same structure. 
Each external physical leg corresponds to a spacetime point. Temporal translations within a slice are separable in time, while those across slices are separable in space, as explicitly depicted in the Figure. Note also that the operators $e^{i\epsilon \mathcal{P}}$ on each spatial slice are precisely of the form shown in Figure~\ref{fig:tntheorem}, panel a). In addition, each translation across slices may be decomposed as an MPO of bond dimension $\chi=\dim(h)^2$ (with $h$ the local Hilbert space), as follows directly from the SWAP-chain representation in Definition~\ref{def:eip} and from the fact that each local SWAP has full rank. 
Let us finally remark that the TN representing $\mathcal{R}_2$ is simple to prepare once the elementary Trotter layers have been specified. Indeed, each $U,U^\dagger$ operator can be treated as a single MPO layer, while the translations across time slices admit a simple MPO representation with fixed bond dimension $\chi=\dim(h)^2$. In this way, the main complexity of the time evolution is shifted to the contraction of the full spacetime network.

\begin{figure}[t!]
    \centering
    \includegraphics[width=\linewidth]{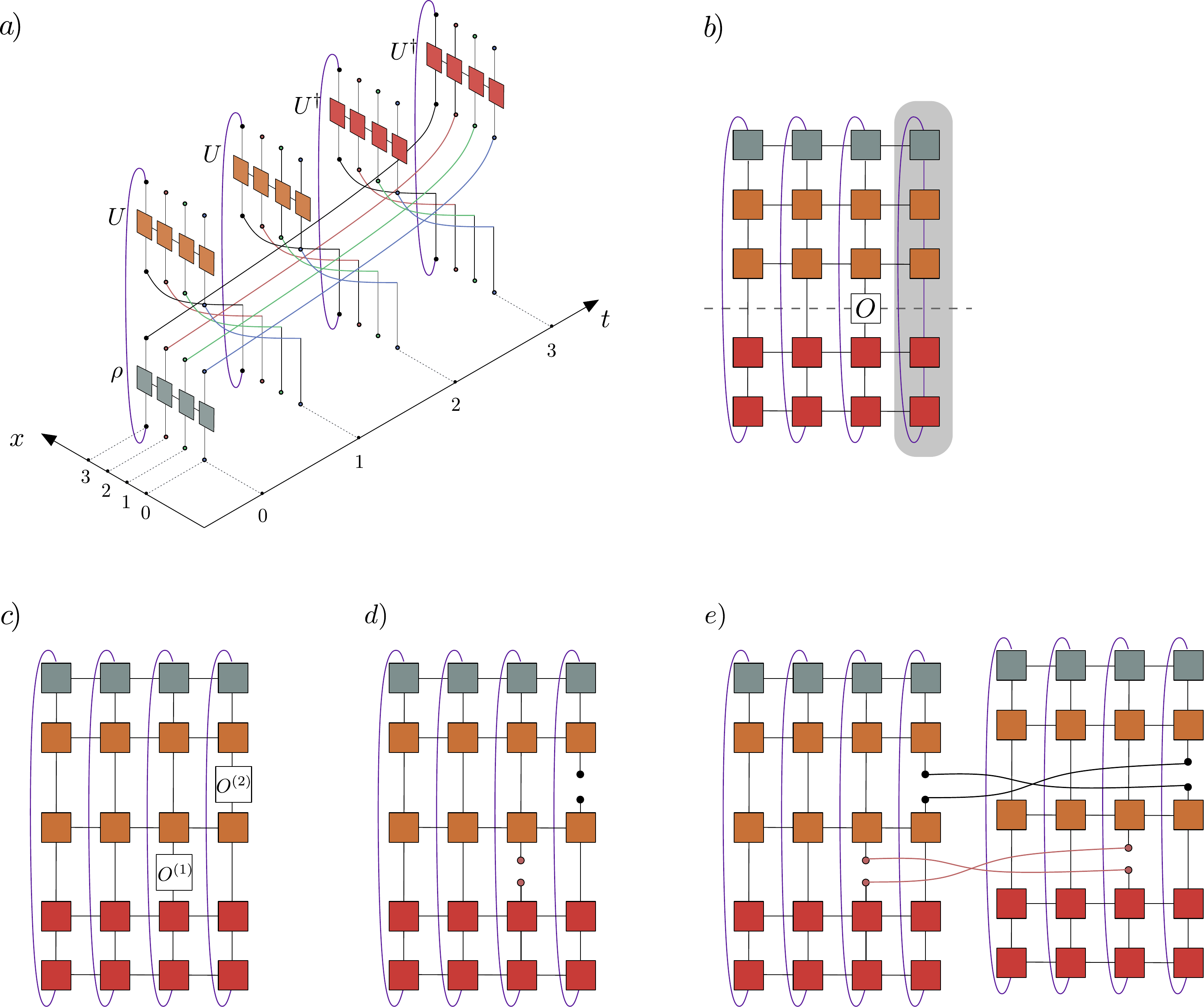}
    \caption{Temporal entanglement and pseudo-entropies. On panel a) we consider $\mathcal{R}_2$ as the TN in Figure \ref{fig:RasPEPO} with all single time gates replaced by their corresponding MPO. We also depicted a partial trace with respect to the spatial slice $x=3$ with violet lines. On panel b) we depict the result of this partial contraction (shadowed column) and the full contraction of the ensuing PEPO with the additional insertion of an operator $O$ at coordinates $(t=2,x=2)$. The shadowed column, as well as the ones on the l.h.s. of $O$ define influence functional (typically after a folding across the dashed line, which, as in our formalism, defines the Schwinger-Keldyish contour). One may treat these columns as independent MPS extended across the time direction thus defining the notion of temporal entanglement of \cite{lerose2021scaling, park2025simulating}. We see that in our current scheme this tensor corresponds to first taking a partial trace of $\mathcal{R}_2$ with respect to a spatial slice and isolating the corresponding tensor. On panel c) we generalize the trace to the case of two operator insertions at different times, namely an $O^{(1)}$ at same coordinates and $O^{(2)}$ at $(t=1,x=3)$. Notice that this contraction admits a similar representation as in the case of one operator. On panel d) we leave these coordinates empty, thus refining $\mathcal{R}_A={\rm Tr}_{\overline{(2,2)\cup (1,3)}}[\mathcal{R}_2]$. Since we are not considering the backward slice we can also write $\mathcal{R}_A={\rm Tr}_{\overline{(2,2)\cup (1,3)}}[\mathcal{R}]$ for $\mathcal{R}$ the standard spacetime state. On panel e) we depict the pseudo-purity ${\rm Tr}[\mathcal{R}_A^2]$. }
    \label{fig:RasPEPO}
\end{figure}

The IF framework for computing expectation values is naturally recovered from our construction. Indeed, in that framework one first lays out the full TN corresponding to: (i) evolving an initial MPS state by applying Trotter gates in MPO form, (ii) inserting the MPO observable whose expectation value we are interested in, (iii) applying the backward evolution, and (iv) finally contracting with the MPS encoding the bra of the initial state. One then regards all the tensors lying over the same spatial site as composing an MPS/MPO with bond dimension $\chi=\dim(h)$, and carries out the contraction of the network in the direction orthogonal to time, treating the original virtual indices as physical ones and vice versa. During these contractions, the growth of the bond dimension is used to define a notion of temporal entanglement in the system~\cite{banuls2009matrix, lerose2021scaling, ye2021constructing}.
Now, as depicted in Figure~\ref{fig:RasPEPO}, once the MPO structure of all the Trotter gates appearing in $\mathcal{R}_2$ is made explicit, carrying out the partial traces over each spatial slice of $\mathcal{R}_2$ reorganizes its TN exactly into the structure of the IF framework. More generally, one may consider arbitrary insertions of operators at different times, thereby allowing the study of temporal correlation functions in addition to expectation values.
If we leave these insertions open, which is equivalent to considering a partial trace of $\mathcal{R}_2$, we can graphically depict the corresponding pseudo-purity. Thus, while the latter can be represented in a somewhat similar scheme to the one employed in the IF setup, the notion of temporal entanglement arising in the IF framework and the one quantified by pseudo-entropies are quite different. Let us also note that the representation of the pseudo-purity shown in the Figure may be rewritten as in the supplemental Lemma \ref{applemm:purities} of Appendix~\ref{app:spacetimest}.

Next, we note that contrary to the notion of temporal entanglement in IF settings, the pseudo purity appears to quantify a very different quantity. Indeed, the spectral properties of $\mathcal{R}_{AB}$ (with $\mathcal{H}=\mathcal{H}_{A\cup B}\cup \mathcal{H}_{\overline{A\cup B}}$) encode how much the regions $A,B$ communicate with each other through the rest of the network (here we can interchange $\mathcal{R}_2$ with $\mathcal{R}$ if the regions only lie in the forward evolution sector). This is the interpretation also suggested in \cite{milekhin2025observable} for different regions of space separated in time (see Section \ref{sec:timelikeent}). A direct interpretation is thus obtained through the use of singular values. On the other hand, we recall that $\mathcal{R}$ is not a standard quantum state so that singular values and eigenvalues provide very different information. In particular, the isospectrality condition between bipartitions only holds at the level of eigenvalues (see Section \ref{sec:pureststates}). This suggests to study possible bounds among the two notions, or, on a more speculative note, one may also attempt to exploit the biorthogonal Schmidt decomposition of Eq.~\eqref{eq:jointschmidt} to define eigenvalue-based TN representations. \\

Let us add a few comments on $\mathcal{R}_E$, namely spacetime states for Euclidean time as defined in Eq.\ \eqref{eq:Reuclid}. The corresponding TN has the same form as $\mathcal{R}_2$ in Figure \ref{fig:tnofR} but without the insertion of the initial state and with all $U,U^\dag$ replaced by $e^{-\epsilon H}$. In this way the TN exhibits time translation invariance with identified borders. The total length of Euclidean time is $\beta$. If in addition translational symmetry in space is present, and again one assumes periodic boundary conditions, the geometry of the network is a torus. In this scenario it has also been recently proven \cite{diaz2025spacetime} that $\mathcal{R}_E$ lies in a stationary point of the functional $F[\Gamma]=\langle \mathcal{S}_E\rangle_\Gamma +{\rm Tr}[\Gamma \log(\Gamma)]$, with $\Gamma$ an operator acting on $\mathcal{H}$.

This setup paves the way for applying different approximate methods to estimate the trace of the non-Hermitian but normal operator $\mathcal{R}_E$ (with the possible insertion of local operators) from a few dominant spectral contributions. 
A natural possibility is to reinterpret the full contraction along one of the two directions as an effective boundary problem, so that the toroidal TN is replaced by an induced one-dimensional operator acting on a virtual boundary space. In such a formulation, one may attempt to approximate the trace from the leading eigenvalues of this effective operator, in a way reminiscent of transfer-matrix techniques. This suggests that TN methods based on boundary states, such as MPS approximations of the dominant left and right fixed points, could provide a useful starting point. Closely related alternatives are corner-transfer-matrix renormalization schemes, which construct effective environments directly from the two dimensional network, as well as tensor renormalization-group approaches based on iterative coarse graining of the torus. 
Let us also recall that the quantity being estimated can ultimately be written as a sum over Euclidean histories (see Section \ref{sec:PIs}), so that comparisons with standard PI quantum Monte Carlo methods are natural. The present formulation is nevertheless different in spirit: instead of sampling from classical configurations, one works with an ``operator-valued'' TN
 representation of the same quantity, to which one may then apply truncation, boundary-state, or coarse-graining approximations. We leave a detailed investigation of these approximation strategies, and of their possible use for computing Euclidean (or partially Wick rotated) spacetime correlators, for future work.

\subsection{Fermionic systems: formalism, tensor-product reinterpretation, and relation to other approaches}\label{sec:fermions}

Here we briefly describe the formalism of SQM for fermions. The main difference with bosons is that their canonical algebraic structure is based on anticommutators, so a direct tensor-product construction across time is no longer adequate for a space-time symmetric approach. Instead, following \cite{diaz2025spacetime} one can introduce independent fermionic operators $\tilde{a}_{t i},\tilde{a}_{t i}^\dagger$ for each time slice $t$, satisfying
\begin{equation}\label{eq:fermalg}
    \{\tilde{a}_{t i},\tilde{a}_{t' j}^\dagger\}=\delta_{t t'}\delta_{ij},\qquad
    \{\tilde{a}_{t i},\tilde{a}_{t' j}\}=0,
\end{equation}
with the corresponding spacetime Hilbert space built from the associated Fock representation. In this way, time is treated as an additional discrete label on the same footing as the other (standard) modes. Then, one can introduce a fermionic version of the time translations across slices satisfying 
\begin{equation}
    e^{i\epsilon \mathcal{P}_F} \tilde{a}_{t i} e^{-i\epsilon \mathcal{P}_F}=\tilde{a}_{t+1,i}\,,
\end{equation}
with anti-periodic boundary conditions. Using Fourier in time modes one can provide an explicit definition for $\mathcal{P}_F$ as a quadratic Legendre-like operator \cite{diaz2025spacetime}. The antiperiodic boundary conditions are linked to half-integer Matsubara frequencies $\omega_n=(2 n+1)\pi/T$.  
Fermionic spacetime states are then defined by $\mathcal{R}_F=\rho_0 e^{i \tilde{\mathcal{S}}_F}$ with $e^{i \tilde{\mathcal{S}}_F}=\Gamma e^{i H_0 T}e^{i\epsilon \mathcal{P}_F} \prod_t e^{-i \epsilon H_t}$ and where $\Gamma=e^{i\pi \sum_{ti}\tilde{a}^\dag_{ti}\tilde{a}_{ti}}$ is the total parity operator. These spacetime states satisfy properties that are the fermionic counterpart of what we developed in Section \ref{sec:formalism}. For example,
\begin{equation}\label{eq:fermionicmap}
    {\rm Tr}[\mathcal{R}_F \tilde{a}_{t_1 i}\tilde{a}^\dag_{t_2 j}]={\rm tr}[\rho \hat{T}a_i(t_1)a^\dag_j(t_2)]\,,
\end{equation}
where now the time-ordering follows the standard fermionic definition. Here it is also assumed that the evolution and the initial state are genuinely fermionic, namely they respect parity.

Let us now use Eq.\ \eqref{eq:fermionicmap} to introduce the fermionic version of the imagitivity discussed in Section \ref{sec:pureststates} for bosonic-like systems. For fermions, it is natural to discuss causality in terms of unequal anticommutation relations of elementary ladder operators since all observables can be built from an even number of them. Let us then focus on ${\rm tr}[\rho \{a_i(t_1),a^\dag_j(t_2)\}]={\rm tr}[\rho (\hat{T}-\hat{\bar{T}})a_i(t_1)a^\dag_j(t_2))]$ where we assumed $t_1>t_2$ for concreteness. Notice that we added a sign to the second term since the anti-time ordering operator yield an extra sign when interchanging the two ladder operators. This leads to 
\begin{equation}
 {\rm tr}[\rho \{a_i(t_1),a^\dag_j(t_2)\}]={\rm Tr}[(\mathcal{R}_F-\mathcal{R}^\dag_F)\tilde{a}_{t_1i}\tilde{a}^\dag_{t_2j}]
\end{equation}
with ensuing bound \footnote{we use the Cauchy-Schwarz inequality, but one may also use the more general H{\"{o}}lder inequality and Schatten norms. }
\begin{equation}
 |{\rm tr}[\rho \{a_i(t_1),a^\dag_j(t_2)\}]|\leq 4||\mathcal{R}_F-\mathcal{R}^\dag_F||\,. 
\end{equation}
Here $\mathcal{R}_F$ can be restricted to the support $(t_1,i)\cup (t_2,j)$ (see e.g., \cite{gigena2016one,vidal2021quantum,diaz2025spacetime} for discussions on fermionic partial traces and \cite{diaz2025spacetime} for a specific spacetime discussion). We see that even if we bound unequal-time anticommutators of ladder operators, instead of commutators, the \emph{fermionic imagitivity} is defined in complete analogy with the bosonic-like case. In addition, if we consider genuine observables $O^{(i)}$ commuting with parity, we can still bound their unequal time \emph{commutators} as 
$|{\rm tr}\big[\rho\, [O^{(2)}(t_2),O^{(1)}(t_1)]\big]|\leq ||O^{(1)}|| \,||O^{(2)}||\, ||\mathcal{R}_F-\mathcal{R}_F^\dag||$, where again we might consider a marginal of the spacetime state. This follows since observables are not affected by signs under time-ordering. In conclusion, the fermionic imagitivity bounds both the unequal-time commutators of fermionic observables and the unequal-time anticommutators of individual ladder operators, with the latter the basic building block defining causality (e.g., in QFT the microcausality condition is dictated by anticommutators of Dirac fields).

Let us now give a complementary point of view of the previous fermionic construction. While the natural setting is the one provided above and fully developed in \cite{diaz2025spacetime} one may also want to understand the fermionic case in direct analogy of what we developed throughout the manuscript for general systems. A direct way to do so is to consider the fermionic Hilbert space $h$ and build $\mathcal{H}=h^{\otimes N}$ as before so that the ladder operators $a_{ti}, a^\dag_{ti}$ satisfy anticommutation relations at the same time slice but they commute at different times. Then, we can ``dress'' these operators as follows $\tilde{a}_{ti}=\prod_{t'=0}^{t-1}\Gamma_{t'}a_{ti}$ with $\Gamma_{t}$ the local (in time) parity operator. The dressed operators are genuinely fermionic and satisfy Eq.\ \eqref{eq:fermalg}. Now, consider the action of the fermionic time translation operator on a Fock state 
\begin{equation}
\begin{split}
     &e^{i\epsilon \mathcal{P}_F}|\textbf{n}_0 \textbf{n}_1\dots \textbf{n}_{N-1}\rangle\\=&(-1)^{M_{N-1}(\sum_{t=0}^{N-2} M_t+1)}|\textbf{n}_{N-1} \textbf{n}_0 \textbf{n}_{1}\dots\rangle
\end{split}
\end{equation}
where we used the notation $\textbf{n}_t$ to indicate the fermionic modes at a given time and $M_t=\sum_i \tilde{a}^\dag_{ti}\tilde{a}_{ti}$ is the total number of fermions at time $t$. The phase is a consequence of the antiperiodic boundary conditions and the need to commute the creation operators that were previously on the last slice past all operators belonging to the previous slices (the Fock basis assumes an ordering). We can then write $e^{i\epsilon \mathcal{P}_F}=e^{i\epsilon \mathcal{P}}D$ for $D=e^{i\pi(M_{N-1}[\sum_{t=0}^{N-2} M_t+1])}$ the operator responsible of the phases. This also implies $ \Gamma e^{i\epsilon \mathcal{P}_F}=Ge^{i\epsilon \mathcal{P}}$
 for $G=e^{i\pi(1+M_0)\sum_{t=1}^{N-1}M_t}$, where we used that both $G$ and $\Gamma$ are diagonal in the number basis. Moreover, the previous implies 
 \begin{equation}
     e^{i \tilde{\mathcal{S}}_F}=G\,e^{i\tilde{\mathcal{S}}}\,.
 \end{equation}
 This result shows that we can recover the genuinely fermionic QA from the QA defined in \eqref{eq:qainin} by simply multiplying on the left by $G$. Notice also that the evolution within is automatically dressed for parity preserving Hamiltonians (e.g., $\tilde{a}^\dag_{1i}\tilde{a}_{1j}=\Gamma_0 \Gamma_0 \otimes a^\dag_{i}a_{j}=a^\dag_{1i}a_{1j}$). A similar result holds for the spacetime state 
\begin{equation}
    \mathcal{R}_F=G_\psi\,   \mathcal{R}\,,
\end{equation}
 for $G_\psi=\big(\prod_{t=1}^{N-1}\Gamma_t\big)^{\frac{1+P_\psi}{2}}$ where $\Gamma|\psi\rangle=P_\psi |\psi\rangle$, i.e., $P_\psi$ is the parity of the state. As a corollary of this relation one can immediately conclude that the singular values of $\mathcal{R}_F$ are equal to those of $\mathcal{R}$ and, for the odd parity sector, that their eigenvalues, and hence pseudo-entropies are also equal. One can easily verify that the pseudo-purities of all orders for the even sector are also unaffected by $G_\psi$. We can readily conclude that all fermionic pure spacetime states have vanishing pseudo-entropies, namely for closed systems, they satisfy  Eq.\ \eqref{eq:vanishent} just as their bosonic counterpart.

 With the previous  results at hand one can also 
 reinterpret the fermionic maps, such as Eq.\ \eqref{eq:fermionicmap} in analogy with the results of Section \ref{sec:formalism}
\begin{equation}
\begin{split}
    {\rm Tr}\big[\mathcal{R}_F \tilde{a}_{t_1 i}\tilde{a}^\dag_{t_2 j}\big]&={\rm Tr}\Big[G_\psi\,\mathcal{R}\, \smallprod_{t'=0}^{t_1-1}\Gamma_{t'}a_{t_1 i}\smallprod_{t''=0}^{t_2-1}\Gamma_{t''}a^\dag_{t_2 j}\Big]\\
    &={\rm Tr}\big[\mathcal{R}\, a_{t_1 i}a^\dag_{t_2 j}\,G_{\psi,t_1,t_2}\big]\,,
\end{split}
\end{equation}
where we have absorbed all dressing factors into the operator
$G_{\psi,t_1,t_2}=-\mathrm{sgn}(t_1-t_2)\prod_{t=\min(t_1,t_2)}^{\max(t_1,t_2)-1}\Gamma_t\, G_\psi$, which follows by keeping track of the signs produced when parity operators are moved through ladder operators on the same time slice (for $t_1\neq t_2$; $G_{\psi,t,t}=G_\psi$).
Since the latter is a product-in-time operator and all the operators on the r.h.s. are bosonic-like, one can now apply Corollary \ref{cor:wigthman} to recover Eq.\ \eqref{eq:fermionicmap}. When doing so, the extra factor $G_{\psi,t_1,t_2}$  takes into account the fermionic properties and in particular provides the sign in the fermionic time-ordering operator. Similar considerations hold for arbitrary correlators. 
We have thus seen that the genuinely fermionic spacetime construction can be re-expressed in direct analogy with the tensor-product-in-time formalism developed for general systems, at the cost of introducing a temporal parity string.

Let us add a final comment on how the discussion of this section fits within the unifying spacetime picture summarized in Table \ref{tab:unified}. The first point to notice is that QSOT and PDMs have not, to our knowledge, been explicitly formulated for fermions in the literature. Since these approaches rely on tensor products across time, and fermions do not naturally admit such an ordinary tensor-product structure, it is natural instead to seek their fermionic generalization by combining the results of the present section with the derivations of Section \ref{sec:unifying}. In this regard, the reinterpretation of the fermionic spacetime formalism in terms of tensor products provides a particularly natural tool.

Regarding PIs, we recall that the standard fermionic formulation makes use of Grassmann variables. While one may indeed introduce Grassmann variables to define coherent-state bases in spacetime, the use of traces in $\mathcal{H}$ already captures the fermionic sum-over-histories intuition by itself. In fact, one can bypass the explicit use of Grassmann variables altogether and still reproduce PI-like techniques directly within the Hilbert-space formalism. This was recently discussed in \cite{diaz2025spacetime}.

Regarding SOs, we note that in \cite{cotler2018superdensity} the authors also discuss the fermionic case through the use of spacetime anticommutation relations. It is therefore natural to develop the fermionic analogue of Theorem \ref{th:SO}. The main caveat is that this requires fermionic versions of the partial transposition and realignment maps. We leave these subtleties for future work.

Finally, consider the timelike entanglement approach. Since $\mathcal{R}_F$, and its marginals, play precisely the role of $T_{AB}$, they may be used to define both fermionic pseudo-entropies and fermionic imagitivity. In this sense, the discussion of Section \ref{sec:timelikeent} extends naturally to the fermionic setting. It is also interesting to note that in \cite{milekhin2025observable} the authors already study timelike entanglement in a free fermionic model and relate it explicitly to the timelike entanglement of \cite{harper2023timelike} in the corresponding CFT. In that particular case, however, the analysis relies directly on two-point correlations, exploiting the Gaussian nature of the model, rather than on a microscopic definition. The results of the present section provide precisely such a microscopic definition.

\subsection{Relativistic QFTs}\label{sec:QFT}
Up to this point we have presented many proposals that aim to develop a spacetime symmetric form of QM and showed how to unify them by using the spacetime QM framework. 
On the other hand, it is widely known that QFTs provide the framework underlying the Standard Model of particle physics and the
standard language for relativistic quantum theories. It is natural to wonder where QFTs fit in our picture. In this section we discuss the precise relation existing between SQM and QFT.

Let us start by recalling that there are two main approaches to QFT. Namely, the canonical one, and the PI formulation. Since the PI formulation has been directly related to the spacetime approach in Section \ref{sec:PIs} let us focus on the first (we add some QFT specific  comments related to PIs afterwards). 
For simplicity we will also focus on the case of a single scalar field. In the canonical formulation the Hilbert space is defined by the equal-time commutation relations $[\phi(t,\textbf{x}),\pi(t,\textbf{y})]=i\delta^{(d)}(\textbf{x}-\textbf{y})$ where we are considering $d$ spatial dimensions. This very definition requires a classically defined foliation of spacetime. From this one can also introduce a Hamiltonian which in turn defines what we mean by a field in spacetime: $\phi_H(t,\textbf{x})= U^\dag(t) \phi(\textbf
{x})U(t)$, namely a quantum field in spacetime is a dynamical-dependent notion, with the Heisenberg picture relating the initial field with the field at different times. 
We see that in the standard canonical approach fields and momenta at different spacetime points are in general causally dependent. We recall, however, that an important result of special relativistic QFTs,  is the microcausality condition $[\phi_H(x),\phi_H(y)]=0$ for $x,y$ two spacelike separated points. This allows one to recover Lorentz invariance a posteriori, even if the very definition of the Hilbert space and the related quantization process breaks Lorentz invariance explicitly.

Let us now make a few remarks regarding the roles of space and time. To avoid subtleties linked to the continuum let us first discretize the canonical algebra as  
$[\phi_\textbf{x},\pi_\textbf{y}]=i\delta_{\textbf{x},\textbf{y}}$ (with other commutators vanishing), where we assumed the fields are at a common reference time. We see that there is a different independent field for each value of $\textbf{x}$. This is in agreement with QM describing the complete system as the tensor product of the Hilbert spaces of their parts. In other words, we have independent modes with $\textbf{x}$ being a simple label of the fields of interest. On the other hand, when we write $\phi_\textbf{x}(t)$, where we omit the subindex $H$ for ease of notation, the parameter $t$ has a completely different meaning: it is indicating how much we evolved the field through a parameterized adjoint action. So even if both $t$ and $\textbf{x}$ are classical parameters, their role is completely different.

Let us now contrast the canonical QFT approach with the spacetime QM formalism. First of all we recall that the SQM approach can be applied to any quantum system, so let us apply to the scalar fields in discrete space. The definition of $\mathcal{H}$ corresponds to 
$[\phi_{t\textbf{x}},\pi_{t'\textbf{y}}]=i\delta_{tt'}\delta_{\textbf{x},\textbf{y}}$. Now both space and time are classical labels of independent fields and their role is on an equal footing. For each spacetime point we have independent modes. In the continuum limit we write $\phi(x)=\frac{\phi_{t\textbf{x}}}{\sqrt{\epsilon^D}}$, $\pi(x)=\frac{\pi_{t\textbf{x}}}{\sqrt{\epsilon^D}}$, where we assumed a common uniform spacing $\epsilon$ across all spacetime dimensions, leading to 
\begin{equation}\label{eq:spacetimefieldalg}
    [\phi(x),\pi(y)]=i\delta^{(D)}(x-y)\,,
\end{equation}
with $D=d+1$ \footnote{Let us recall that while for discrete spacetime the Stone-Von Neummann theorem guarantees that the Hilbert space is well-defined by the canonical algebra (up to unitary equivalence), the continuum limit carries the usual mathematical subtleties. Nonetheless, these subtleties are the standard ones, as the algebra of Eq.\ \eqref{eq:spacetimefieldalg} is isomorphic to the equal-time algebra of  standard QFT with an additional dimension. As such, they may be tackled by proper regularization. }. Thus the field algebra becomes manifestly covariant. This also allows one to treat Lorentz transformations in analogy with rotations with the operators 
\begin{equation}\label{eq:poincare}
     L_{\mu\nu}=\int d^{4}x\, \pi(x_\mu \partial_\nu-x_\nu \partial_\mu) \phi\,,\quad\mathcal{P}_\mu=\int d^4x\, \pi\partial_\mu\phi
\end{equation}
closing the Poincaré algebra. Notice that both the generator of boosts $L_{0i}$ and of time translations $\mathcal{P}_0\equiv \mathcal{P}$ involve \emph{translations across Hilbert spaces} (see Definition \ref{def:eip}) and not Hamiltonian evolution.

One apparent obstacle when defining the algebra of Eq.\ \eqref{eq:spacetimefieldalg} is the fact that it seems to conflict with standard unitary evolution. Fortunately, the framework of spacetime QM described in Section \ref{sec:formalism} clearly shows us how to proceed: we simply need to exponentiate the quantum action corresponding to a given theory
and recover observables from its correlators. Time-ordered correlation functions are a fundamental quantity of QFT and the spacetime formalism leads us precisely to them. Schematically, from the results of Section \ref{sec:basicformalism} we  obtain e.g., 
\begin{equation}\label{eq:feynmanprop}
    {\rm Tr}[\mathcal{R} \phi(x)\phi(y)]\equiv \langle 0| \hat{T}\phi_H(x^0,\textbf{x})\phi_H(y^0,\textbf{y})|0\rangle
\end{equation}
 which for $\mathcal{R}=\frac{1}{{\rm Tr}[e^{i\mathcal{S}}]}e^{i\mathcal{S}}$ with quantum action
 \small
 \begin{equation}\label{eq:kleingordonaction}
   \mathcal{S}=\int d^Dx \, \left[\pi(x)\dot{\phi}(x)-\frac{\pi^2(x)}{2}-\frac{(\nabla \phi(x))^2}{2}-\frac{m^2\phi^2(x)}{2} \right] 
 \end{equation}
 \normalsize
 is just Feynman propagator (similar considerations hold for interacting theories). Here a small imaginary part has to be added to the Hamiltonian part of the action, with $\mathcal{R}$ the spacetime state corresponding to the vacuum state. Indeed, the small imaginary part is precisely what projects onto the vacuum state inserted at the initial time slice (see comment below Eq.\ \eqref{eq:Reuclid}).
The continuum time case involves some additional subtleties we are ignoring here. While the trace may be understood as the limit $\epsilon \to 0$ at fixed $T$ and then large $T$, there are more direct treatments one may employ in the continuum case. We refer the reader to \cite{diaz2023spacetime,diaz2026quantum} for explicit discussions on the continuum limit and  interacting QFTs.

Let us add that in this context the imagitivity  is directly tied to microcausality
\begin{equation}
{\rm Tr}[(\mathcal{R}-\mathcal{R}^\dag)\phi(x)\phi(y)]\equiv [\phi_H(x),\phi_H(y)]\,,
\end{equation}
where we assume that the commutator is a 
$c$-number, as in free theories. Since local operators can in principle be constructed from the field and its derivatives, microcausality formally implies that the reduced spacetime state associated with spacelike regions is Hermitian.

Notice also that the l.h.s. of Eq.\ \eqref{eq:feynmanprop}
 has precisely the form of a PI but it is not yet one: only after we evaluate in the field spacetime basis we recover Feynman's formalism. This is what we explained for a single particle in Section \ref{sec:PIs}, which for fields corresponds to using a field eigenbasis $\hat{\phi}(x)|\phi(x)\rangle=\phi(x) |\phi(x)\rangle$, i.e., a basis of field configurations in spacetime. Another interesting insight concerns the notion of particle: if one diagonalizes the free quantum action normal ladder operators satisfying $[a(p),a^\dag(k)]=(2\pi)^D\delta^{(D)}(p-k)$ are obtained. This shows that the theory is intrinsically off-shell. However, only on-shell modes appear when considering scattering processes, a result that can be recovered from considerations on spacetime states \cite{diaz2023spacetime,diaz2026quantum}. Interestingly, these comments also link the spacetime approach to QFT with the PW mechanism, as single particle states of the QFT in $\mathcal{H}$ have the structure of PW states, as we described in Section \ref{sec:PW}.

In conclusion, when quantizing a QFT under the spacetime QM framework Lorentz covariance becomes explicit and Hamiltonian independent at the level of canonical algebras. Particles are naturally off-shell in this formulation. Correlation functions are recovered from expectation values computed with respect to the exponential of the action, thus reflecting Feynman's formulation. In principle, these insights apply also to curved spacetime with covariance under diffeomorphisms also manifest. Similar considerations hold for fermionic fields, with commutators replaced by anticommutators. For example, for a Dirac field one imposes $\{\psi_a(x),\psi^\dag_b(y)\}=\delta_{ab}\delta^{(D)}(x-y)$ with $a,b$ spinor indices and where we are considering the continuum spacetime limit. The corresponding free quantum action takes the form ${\mathcal{S}=\int d^4x\, \bar{\psi}(x)(\gamma^\mu i\partial_\mu-m){\psi}(x)}$ (see \cite{diaz2025spacetime} for details). 
In principle, and in analogy with the PI formulation, the SQM formalism is naturally expected to accommodate gauge theories as well. However, a detailed study of how gauge redundancies and constraints are encoded in spacetime states is still under development.

We also remark that algebraic QFT \cite{haag1964algebraic} also differs substantially from SQM. There, local algebras are tied to dynamical evolution and are typically generated from the fields $\phi$, with $\pi$ formally recovered from time derivatives. By contrast, an algebraic counterpart of SQM would treat $\phi$ and $\pi$ symmetrically and assign commuting algebras to any disjoint spacetime regions. We refer the reader to \cite{diaz2026quantum} for further discussion on this specific point.

\subsection{Quantum reference frames}\label{sec:QRF}

Another line of work, related in spirit to a ``spacetime approach'', introduces the concept of quantum reference frame (QRF), namely, the idea that a more relational and, in some cases, more space-time symmetric approach to QM may emerge when the reference frames themselves are treated quantum mechanically \cite{aharonov1984quantum,giacomini2019quantum,giacomini2019relativistic,apadula2022quantum}\footnote{We note that the literature also includes resource-theoretic approaches to QRFs; see, e.g., \cite{bartlett2007reference,marvian2014modes}.}.

While the usual QRFs are proposed within the framework of non-relativistic and/or single particle QM, the application of spacetime QM to QFT we described in Section \ref{sec:QFT} allow one to introduce a QRF in QFT settings. As a matter of fact, they almost appear as a mathematical necessity, in the sense that while the algebra of Eq.\ \eqref{eq:spacetimefieldalg} is symmetric in spacetime, the action in Eq.\ \eqref{eq:kleingordonaction} breaks Lorentz covariance. The problem is that the Legendre transform part  requires one to choose a foliation. Thus, while at the algebraic level we can make Lorentz invariance manifest, the need to employ phase-space variables forces us to inherit an asymmetry caused by the Hamiltonian formulation. Notably this obstacle can be immediately circumvented if we let the choice of the foliation to stay arbitrary thus replacing $\int d^Dx\, \pi(x)\dot{\phi}(x)\to \int d^Dx\, \pi(x) n^\mu \partial_\mu \phi(x)$ with $n^\mu$ a timelike vector orthogonal to the Cauchy hypersurfaces of the corresponding foliation. 
However, to truly recover Lorentz symmetry we need the vector $n^\mu$ to transform under the action of some unitary operator acting on Hilbert space. To do so $n^\mu$ has to be treated as a quantum degree of freedom. Let us remark that this is a natural extension of the SQM construction, not a standard QRF result. 
The simplest choice corresponds to imposing $[n^\mu, \kappa_\nu]=i\delta^{\mu}_{\; \nu}$. Then the quantum action becomes a control-like operator $\mathcal{S}=\int dn\, \mathcal{S}_n \otimes |n\rangle \langle n|$ for $\mathcal{S}_n$ the quantum action of fields, defined by the Legendre transform at fixed $n$, and with $\hat{n}^\mu |n\rangle= n^\mu |n\rangle$ (the separability is a consequence of the fields and foliation operators commuting with e.g., $[\phi(x),n^\mu]=0$). With these definitions it is straightforward to define a total angular momentum operator $J_{\mu\nu}=L_{\mu\nu}+l_{\mu\nu}$ such that $[J_{\mu\nu}, \mathcal{S}]=0$ for relativistic actions. Here $l_{\mu\nu}=n_\mu \kappa_\nu-n_\nu \kappa_\mu$ is the angular momentum of the foliation while $L_{\mu\nu}$ is defined in Eq.\ \eqref{eq:poincare}. The formalism we described here has been introduced in \cite{diaz2023spacetime} where it is also shown how to apply it to classical mechanics.

In light of the above, we can claim that the concept of QRF arises naturally within SQM when applied to QFT: if rather than an abstract tool to make Lorentz invariance explicit we provide the quantization of $n^\mu$ a genuine dynamical  origin, we are indeed defining a QRF. To give a simple example, we may associate $n^\mu$ with the worldline of a particle so that $n^\mu \propto p^\mu$ with $p^\mu$ its tetramomentum. More generally, this interpretation, extended in particular to non-inertial observers,  would require to consider a foliation field, i.e.,  $n^\mu\equiv n^\mu(x)$.

Let us remark that, at least formally, this approach to QRFs is not directly equivalent to previous proposals in the literature, particularly because our construction is intrinsically formulated within QFT and in the SQM framework. Nonetheless, the considerations above provide a natural and tightly constrained framework in which to study the consequences of QRFs in QFT, together with their possible low-energy limits. In those limits, the effective description may be compared with the QRF constructions developed in standard QM (or relativistic single particle settings), thus opening a route toward an alternative, symmetry-motivated foundation for QRFs, as well as for their  QFT completion, a problem recently posed in \cite{apadula2022quantum}.

\subsection{Spacetime classical mechanics}\label{sec:classical}

Having presented a general framework for formulating quantum mechanics in spacetime, it is natural to ask how classical mechanics fits into this picture. In particular, since the Hamiltonian formulation provides the canonical route from classical theories to standard quantum mechanics, one may ask whether there exists an analogous spacetime formulation of classical mechanics whose quantization leads to SQM, and conversely, which features of SQM admit a natural classical counterpart.

We begin by recalling a few basic facts about the standard Hamiltonian formulation of classical mechanics. The fundamental structure is provided by the Poisson brackets (PBs), which for two phase-space functions $f(q,p)$ and $g(q,p)$ are defined as
\begin{equation}\label{eq:PBsstandard}
    \{f,g\}
    =
    \sum_i
    \left(
    \frac{\partial f}{\partial q_i}\frac{\partial g}{\partial p_i}
    -
    \frac{\partial f}{\partial p_i}\frac{\partial g}{\partial q_i}
    \right).
\end{equation}
In particular, they imply the canonical relations
$
    \{q_i,p_j\}=\delta_{ij}$. 
PBs play a central role both  in the passage to quantization and in the description of the dynamics. Indeed, time evolution is generated by the Hamiltonian through Hamilton's equations, which through PBs take the compact  form $
    \dot q_i=\{q_i,H\}=\frac{\partial H}{\partial p_i}$ $,
    \dot p_i=\{p_i,H\}=-\frac{\partial H}{\partial q_i}$.
More generally, any phase-space function $f$ evolves according to
\begin{equation}\label{eq:Hamiltoneqs}
    \dot f=\{f,H\}\,,
\end{equation}
assuming $f$ has no explicit time dependence. Thus a function $f(q,p)$ gets parameterized by an external time according to $f(q,p)\to f(q(t),p(t))$. A convenient way to derive Hamilton equations is to impose the principle of stationary action $\delta S=0$ which may be rewritten as $\delta S=\sum_i \frac{\delta S}{\delta q_i(t)}\delta q_i(t)+\frac{\delta S}{\delta p_i(t)}\delta p_i(t)=0$ when we treat the action as a functional of phase-space variables. Then, one obtains Hamilton equations from  $\frac{\delta S}{\delta q_i(t)}=0$, $ \frac{\delta S}{\delta p_i(t)}=0$ imposed for all $i$ and times, 
where one assumes the standard form of the action in phase-space variables $S=\int dt\, \left[\sum_i p_i(t) \dot{q}_i(t)-H(q(t),p(t))\right]$ and the functional derivative are defined so that $\delta q_j(t')/\delta q_i(t)=\delta_{ij}\delta(t-t')=\delta p_j(t')/\delta p_i(t)$, $\delta q_j(t')/\delta p_i(t)=0$  (which can be interpreted as the continuum limit of discrete derivatives such as $\frac{\delta S}{\delta q_{i}(t)}\equiv \frac{1}{\epsilon}\frac{ \partial S}{\partial q_{ti}}$ for $\epsilon$ the time scaling going to zero).

Let us also add that in principle, the pairs $q_i,p_i$ are arbitrary canonical variables, namely they do not need to represent position and momentum of a particle. They could represent instead canonical pairs obtained from a previously solved constrained system (e.g., $q=\theta$, $p=\frac{\partial L}{\partial \dot{\theta}}$ for a pendulum with Lagrangian $L$). 
More generally, constraints can be imposed directly in phase-space as weak equations
\begin{equation}\label{eq:stconstraint}
    \phi_a(q,p)\approx 0,
\end{equation}
where each $\phi_a$ is a function on phase-space. These equations define a constraint surface, namely the submanifold of allowed phase-space points. The symbol $\approx$ indicates that the equality is required only on this submanifold, so that Poisson brackets are first computed in the full phase space and only afterwards restricted to the constraint surface. Notice also that consistency under evolution requires imposing  $\{ H,\phi_a\}\approx 0$ as well, which may produce additional constraints.

Let us now develop a classical counterpart of SQM. For simplicity we will employ a discrete time formulation. Our starting point will be to make use of the classical analogue of spacetime commutators such as \eqref{eq:spacetimealg}. We then define spacetime PBs
\begin{equation}\label{eq:SPBs}
    \{f,g\}_{\text{ST}}=\sum_{t,i} \left(\frac{\partial f}{\partial q_{ti}}\frac{\partial g}{\partial p_{ti}}-\frac{\partial f}{\partial p_{ti}}\frac{\partial g}{\partial q_{ti}}\right),
\end{equation}
which in particular leads to the canonical relations $\{q_{ti},p_{t'j}\}_{\text{ST}}=\delta_{tt'}\delta_{ij}$ and to $\{q_{ti},\cdot\}_{\text{ST}}=\frac{\partial \,(\cdot)}{\partial p_{ti}}$, $\{p_{ti},\cdot\}_{\text{ST}}=-\frac{\partial \,(\cdot)}{\partial q_{ti}}$. According to these relations
``space'' and time are indistinguishable so that an independent degree of freedom is assigned for each value $(t,i)$. Equivalently, the phase-space factorizes as $X_{\text{spacetime}}=\times_t (X_{\text{space}})_t=\times_{t,i}(X_{\text{single-mode}})_{ti}$, for $X_{\text{space}}$ the standard phase-space defined by Eq.\ \eqref{eq:PBsstandard} (here ``$\times$'' denotes standard Cartesian product). The apparent problem of this formulation is that evolution needs to somehow be recovered from a ``timeless'' picture, or, in other words, from a phase-space where time is internal. Notably, the solution is straightforward by considering the action in phase-space variables 
\begin{equation}\label{apeq:discreteaction}
 S= \epsilon\sum_t \Big\{\sum_i p_{ti} \dot{q}_{ti}-H(p_t,q_t)\Big\}\,,
\end{equation}
where for discrete time one needs to define a discrete derivative such as $\dot{q}_{ti}=\sum_{t',n} \frac{i\epsilon\omega_n}{N} e^{i\omega_n (t-t')\epsilon}q_{t'i}$ for $\omega_n= \frac{2\pi n}{T}$ (other possibilities, such as $\dot{q}_t=\frac{q_{t+1}-q_t}{\epsilon}$ are possible). Then, it follows from Eq.\ \eqref{eq:SPBs} that for any function $f_t\in X_{\text{spacetime}}$
\begin{equation}\label{eq:PBwithaction}
    \{f_t,S\}_{\text{ST}}=\epsilon\, \big(\dot{f}_t-\{f_t,H\}\big)\,,
\end{equation}
where we used that a direct evaluation leads to $\{f_t,\sum_{t,i} p_{ti} \dot{q}_{ti}\}_{\text{ST}}=\sum_i \frac{\partial f_t}{\partial q_{ti}}\dot{q}_{ti}+\frac{\partial f_t}{\partial p_{ti}}\dot{p}_{ti}= \dot{f}_t$. The derivative with respect to time is thus provided by the Legendre transform part, which generates time translations in the index $t$, the classical version of \emph{time translations across slices}.  The second term is just the bracket with the Hamiltonian part with the sum over $t$ fixed to the time of $f_t$. If we now consider Eqs.\ \eqref{eq:PBwithaction} and \eqref{eq:Hamiltoneqs} we see that evolution is recovered by imposing
\begin{equation}\label{eq:constraintev}
     \{f_t,S\}_{\text{ST}}\approx 0\,.
\end{equation}
Namely, the submanifold of points satisfying \eqref{eq:constraintev} is in one-to-one correspondence with the standard phase-space (assuming a well-posed initial conditions problem). 
A simple way to understand why this condition leads precisely to Hamilton's equations is to connect this discussion with the principle of stationary action.
\begin{theorem}
    The constraints $\{q_{ti},S\}_{\text{ST}}\approx 0$, $\{p_{ti},S\}_{\text{ST}}\approx 0$ defined by the spacetime brackets are equivalent to the stationary condition
    \begin{equation}
        \delta S= 0\,,
    \end{equation}
    for the action regarded as a functional of the phase-space variables. 
\end{theorem}
The proof of the Theorem is straightforward: 
\begin{equation}
\begin{split}
     \delta S&=\sum_{t,i}  \frac{\partial S}{\partial p_{ti}}\delta p_{ti}+\frac{\partial S}{\partial q_{ti}}\delta q_{ti}\\&=\sum_{t,i}  \{q_{ti},S\}_{\text{ST}}\,\delta p_{ti}-\{p_{ti},S\}_{\text{ST}}\,\delta q_{ti}\,,
\end{split}
\end{equation} 
In the first line we have written the infinitesimal variation of the action under independent variations of the $2MN$ spacetime variables. This is precisely the standard variational step underlying the usual action principle: before the equations of motion are imposed, the variables at different times are varied independently, subject only to the chosen boundary conditions. In the second line we used the definition of the spacetime Poisson brackets, which express derivatives with respect to independent canonical variables at different times.  We see
that requiring the variation of the action to vanish for arbitrary infinitesimal variations is the same as imposing the spacetime brackets to vanish. This completes the proof.

Let us add a few remarks. The reason the action principle can be directly compared with the spacetime scheme is that, in the variational formulation, the variables at different times are treated as independent before the equations of motion are imposed. At this level, therefore, the two descriptions are straightforwardly related, since the spacetime Poisson brackets simply provide a convenient way of rewriting the derivatives of the action with respect to the independent spacetime variables.
However, the spacetime formulation adds an interpretation that is not explicit in the standard variational statement. In the spacetime scheme, the action principle can be understood as imposing the equivalence between the geometrical notion of translations across time slices and the Hamiltonian evolution generated within each slice.
We also remark that the condition of the Theorem is analogous to standard constraints of the form \eqref{eq:stconstraint} where $\phi^{(q)}_{ti}=\{q_{ti}, S\}_{\text{ST}}$, $\phi^{(p)}_{ti}=\{p_{ti}, S\}_{\text{ST}}$ are first constructed from $S$. Then $\phi^{(q,p)}_{ti}\approx 0$ are the Hamilton equations.

In the language of classical mechanics we can also say that only the functions invariant under the flow of $S$ are ``physical''. This flow is precisely the difference between translations across slices (generated by the Legendre transform) and the time translations ``within'' (generated by the Hamiltonian). In the classical case, the initial discussion of Section \ref{sec:basicformalism} about these two types of translations is thus even more direct.  
Notice also that initial (and/or final) conditions can be imposed either in the definition of $S$ or as additional constraints. 
Let us finally notice that a symmetry in the action translates immediately to a symmetry in evolution. That is, if $\{G,S\}_{\text{ST}}=0$ then Jacobi identity implies $0=\{f,\{G,S\}\}_{\text{ST}}+\{G,\{S,f\}\}_{\text{ST}}+\{S,\{f,G\}\}_{\text{ST}}=\{S,\{f,G\}\}_{\text{ST}}$, namely if $f$ is physical the transformed function $\{f,G\}$ (and higher order nested PBs) is as well. Clearly a particular set of symmetries of $S$ follow from $G=\sum_t g(q_t,p_t)$ for $\{g,H\}=0$ and $g$ independent of time.

Let us add a conceptual discussion. In the standard Hamiltonian picture, the variables $q_i,p_i$ provide canonical coordinates on phase space. A point in phase space, specified by $2M$ real numbers, represents the instantaneous state of the system at a given time. Dynamics is then described by a trajectory in phase space, namely by the time-parameterized family $(q_i(t),p_i(t))$ generated by the Hamiltonian flow.
In the spacetime picture, the canonical coordinates are instead $q_{ti},p_{ti}$. A point in this enlarged phase space, specified by $2MN$ real numbers, corresponds to a full history of the system rather than to an instantaneous state. A priori, these histories are arbitrary.
The dynamical information is then encoded not as motion on phase space, but rather as a restriction on physical histories through the constraint \eqref{eq:constraintev}. 
In particular, the variables at one time cannot be chosen independently of those at another time as they must satisfy the equations of motion together with the appropriate initial (and/or final) conditions, which in the spacetime picture are viewed not as an evolution law, but as a coupled system of equations for the full set of variables.
Thus, in the spacetime picture, evolution is replaced by a constraint on histories. Different physical theories are then characterized by different constraints on the same kinematical spacetime phase space.

We have seen that it is perfectly plausible to formulate classical mechanics directly in spacetime. At the kinematical level this construction is in one-to-one correspondence with SQM, including the fact that the bosonic SQM Hilbert space follows from promoting spacetime Poisson brackets to commutators. However, a crucial obstruction appears if one attempts to quantize the theory by strictly following Dirac's prescription \cite{dirac1950generalized}. The reason is that the constraints $\phi^{(q,p)}_{ti}$ implementing evolution in spacetime classical mechanics are second class, namely, their mutual Poisson brackets do not vanish weakly \cite{diaz2026quantum}. To understand why this is a problem consider the example of a two-dimensional particle constrained by $\phi_1=q_x\approx 0$ and $\phi_2=p_x\approx 0$ we have $\{\phi_1,\phi_2\}=1$, so the constraints are second class. This means that, although it is classically consistent to restrict the particle to a line, these constraints cannot be imposed a posteriori as operator conditions after quantization, i.e., $q_x|\psi\rangle=p_x|\psi\rangle=0$ leads to $|\psi\rangle=0$. Instead, under these condition Dirac proposed replacing Poisson brackets by Dirac brackets, which effectively reduces the number of dynamical degrees of freedom prior to quantization; in the present example, one is left with the quantum theory of a one-dimensional particle as $q_x,p_x$ are never quantized.
A similar phenomenon occurs in spacetime classical mechanics with the constraints $\phi^{(q,p)}_{ti}$. If one follows Dirac's procedure directly, replacing Poisson brackets by Dirac brackets before quantization, the ``off-shell'' modes become non-dynamical and the formalism collapses back to standard QM. In this sense, Dirac quantization removes the spacetime structure rather than preserving it. This no-go theorem has been recently proven in \cite{diaz2026quantum}.

As argued in \cite{diaz2026quantum}, this obstruction is closely tied to the need to generalize the notion of quantum state itself: preserving the spacetime picture requires implementing the constraints at the level of correlators through a quantum action, rather than through standard state/subspace quantization. In this sense, we can think of SQM as the quantum version of spacetime classical mechanics. In particular, the use of the exponential of the action to define spacetime states defines a natural classical limit: for small $\epsilon$ we can write 
 $\langle [\mathcal{S},O_t]\rangle_\mathcal{S}=0$, where $\langle ...\rangle_\mathcal{S}={\rm Tr}[|q\rangle_0\langle q'| e^{i\mathcal{S}}...]$ (see Theorem \ref{th:theorem1}). Then  we may write
$0=\langle [\mathcal{S},O_t]\rangle_\mathcal{S}=\lim_{\hbar\to 0 }\langle [\mathcal{S},O_t]\rangle_\mathcal{S}\simeq i\hbar(\{S,O_t\}_{\text{ST}})|_{\text{classical-trajectory}}$ where in the last step we used the standard PI argument based on a stationary phase approximation, so  that for small $\hbar$ the dominant contributions come from classical trajectories (here we also reintroduce $\hbar$ as $\mathcal{S}\to \mathcal{S}/\hbar$ and make use of the results of Section \ref{sec:PIs}). We also replaced the spacetime commutator with the corresponding spacetime PB. So at least formally, it is clear that SQM implies the spacetime condition \eqref{eq:constraintev} in the classical limit.

\section{Conclusions}\label{sec:conclusions}

In this work we developed the concept of spacetime state as the central object of spacetime quantum mechanics. The basic idea is to assign the quantum description of a system not to a spacelike slice, but to spacetime itself. While the present formulation still assumes an underlying classical spacetime background, this does not amount to a deterministic classical block universe. Rather, it suggests a block universe of possible histories: the spacetime arena is fixed and ``four-dimensional'', but the physical content is encoded quantum mechanically. In this sense, spacetime states provide a
new synthesis of the tension emphasized in the Introduction between the geometrical point of view of Minkowski and the dynamical point of view of Dirac.

A central point of the manuscript is that spacetime states $\mathcal{R}$ are not merely formal mathematical objects. They reduce to standard quantum states in the regimes where they should, and even when they do not, they retain many state-like properties: pure spacetime states have vanishing entropies, the corresponding reduced spacetime states are isospectral under arbitrary bipartitions leading to biorthogonal Schmidt decompositions, and convex mixtures and coherences acquire natural spacetime-generalized interpretations. At the same time, the features that make spacetime states different from ordinary density matrices are not pathologies. Their non-Hermiticity encodes time ordering and causality; their pseudo-entropies quantify the loss of a closed-history, either through standard convex mixtures or through inaccessible globally entangled evolution; and the negativity of $\mathcal{R}+\mathcal{R}^\dag$  marks the presence of genuinely temporal quantum behavior, beyond ordinary spacelike quantum correlations, and required for violations of Leggett-Garg inequalities \cite{leggett1985quantum}. Conversely, we identified regimes, such as sufficiently strong decoherence, in which spacetime states collapse to standard density matrices and violations of temporal inequalities are no longer possible. Thus, while space and time are on an equal footing at the kinematical level, spacetime states make manifest how QM distinguishes the two. 
Our preliminary results suggest that this distinction is lost in the classical limit.

With this structure in place, we unified all the main spacetime approaches to QM in Table \ref{tab:unified} under a common framework. We showed that these are not unrelated proposals, but different manifestations of the same underlying formalism. 
This includes the path-integral formulation, which is essential for any spacetime-oriented formulation of quantum theory but is often outside Hilbert-space-based approaches. One general lesson of the unification is that the role of the action  is not lost in the other spacetime approaches, but hidden by the reductions, rearrangements, or non-invertible maps through which they arise from spacetime states. The formalism also applies naturally to relativistic QFTs, showing that ideas originally developed in non-relativistic or quantum-information settings can be connected to the standard language of high-energy physics. Finally, the classical counterpart of the formalism shows that the same logic can be applied to classical mechanics itself, whereby the action can be reinterpreted as the object relating geometrical translations across time with standard dynamics, so that the standard action principle corresponds to constraining these two notions to being equal.

At the same time, the previous spacetime approaches remain valuable in their own domains, and immediately provide insights into the relevance and properties of spacetime states, or of their corresponding reductions and maps, in regimes that may not have been apparent from SQM alone. This is particularly clear in the recent holographic setting. While SQM was heavily inspired by PIs and QFTs, only recent developments \cite{milekhin2025observable,guo2025spacetime} have suggested an explicit relevance of spacetime states for understanding holographic timelike entanglement \cite{harper2023timelike,chu2023time,narayan2023notes,nunez2025timelike,heller2025temporal,jiang2025timelike}. Having established this connection, the present results provide a rigorous Hilbert-space basis for further developments, in the same way that standard quantum-information tools support the usual theory of spatial entanglement and the corresponding holographic implications. In addition, we have provided a ``microscopic'' definition for timelike entanglement that holds for fermionic systems. The numerical examples involving imagitivity and folded spacetime states further indicate that reduced spacetime states and their folded generalizations can serve as concrete probes of causal propagation and scrambling, while their TN representation points to possible computational uses of the formalism.
More generally, by revealing interconnections among previous approaches and clarifying the role played by each of them, the unification presented in this work should provide both a useful guide and a natural benchmark for future developments.

One question, raised naturally in discussions of the formalism, is what SQM  (or other spacetime approaches) gives us beyond a convenient
rewriting of known quantities. The answer is that, besides its conceptual strengths, 
the formalism provides a quantum-information language in which the time-related assumptions of ordinary QM become visible. In the usual formulation, time evolution, the ordering of
events, the reference frame, and the distinction between space and time are already built into
the formalism in a way that presupposes a classical external structure that cannot be removed. In SQM these assumptions enter the theory in a transparent way so that they can in principle be relaxed. 
This makes it possible to ask, and possibly quantify, which aspects of time should be understood as effective classical approximations that may require revision in appropriate physical or experimental regimes. 
In this sense, what may appear as an equivalent scheme within standard, well-tested regimes could become physically meaningful, or even necessary, in regimes beyond the reach of ordinary QM, where the assumptions built into its standard formulation can no longer be taken for granted.

On this basis, one particular direction to be explored concerns the status of causal order. Throughout most of the manuscript we have assumed a fixed ordering of the time translations across slices.
At the same time, we have seen that QSOTs, PDMs,  and OTOCs can be understood in terms of different foldings, or as non-separable-in-time linear maps acting on enlarged spacetime states. In the present work these foldings were either fixed or combined only as classical mixtures. The only example we presented of genuine spacetime coherences was the quantum switch \cite{chiribella2013quantum},  which, however, does not affect the  time-slice ordering itself.
A genuine spacetime generalization would coherently combine different foldings of spacetime, as would naturally occur if the choice of folding were controlled by a quantum system (see the explicit construction  at the end of Section \ref{sec:otocs}).
This suggests a route toward a spacetime-state formulation of indefinite causal order, possibly connected with quantum reference frames (Section \ref{sec:QRF}) or, more speculatively, with quantum spacetime backgrounds entangled with matter, both of which may control the quantum action and thereby the order of events. A direct comparison with the process-matrix formalism \cite{castro2018dynamics}, a framework proposed precisely to relax the assumption of a definite causal structure among quantum parties, may therefore become feasible upon further extensions of the formalism presented here.

Regarding the possibility of revisiting old problems such as canonical quantization under the SQM intuition, let us remark that the appearance of a Hamiltonian should no longer be interpreted as the introduction of a preferred time direction. In standard canonical quantization the algebra is defined at equal times, so the split between space and time is already built into the kinematics, and treating this split as dynamical is indeed problematic. In particular, the momentum is fixed by this choice.
In SQM, instead, spacetime algebras are naturally independent of the foliation. This means that the time direction $n^\mu$, used to define a Hamiltonian decomposition, can itself be treated as dynamical without modifying the canonical algebra of the standard degrees of freedom (i.e., the matter and foliation algebras commute). As discussed in Section \ref{sec:QRF}, one natural reason to assign an algebra to the foliation itself is to give the choice of time a concrete physical realization, namely to regard it as determined by an actual physical reference system. The important conceptual point is that what breaks manifest covariance of a relativistic theory is not the Hamiltonian split itself, but treating this split as fixed external information. 
For these reasons, SQM may also provide a nontrivial way to revisit canonical approaches to gravity.

Let us finally comment on a broader foundational question suggested by the formalism. In ordinary QM, once the Hilbert space of a system is given, any positive normalized operator on that Hilbert space is a valid quantum state. The dynamics is then introduced through an additional, separate axiom. In the spacetime formalism developed here, the situation is different. We have not postulated that an arbitrary normalized operator on the spacetime Hilbert space is a valid spacetime state. Instead, spacetime states were defined constructively, through an initial state and a quantum action. This leaves open the possibility that this simple structure is itself an effective one, thus raising  a natural question: \textit{Can spacetime states be characterized from a more axiomatic point of view?} For example, if every suitably normalized spacetime operator satisfying appropriate consistency conditions could be obtained from an enlarged construction of the type described here (i.e., from marginals of pure spacetime states), then the formalism would admit an axiomatic formulation closer in spirit to ordinary QM. The key difference would be that this single axiom would unify what in ordinary QM are two separate postulates.
If not, the obstruction would be equally interesting: it could identify the precise conditions distinguishing spacetime quantum mechanics from more general, possibly post-quantum, theories, or reveal that the constructive spacetime states studied here arise only as effective descriptions selected by limited access to a more general spacetime object. In either case, the results of this work provide the necessary starting point for addressing this issue. \\



In summary, SQM does not simply add another spacetime approach to quantum mechanics. Instead,  it identifies the spacetime state as the object from which the existing spacetime approaches, the standard dynamical picture, and several new questions about time and causality can be understood within a common rigorous and intuitive framework.

\acknowledgments 
We thank Raúl Rossignoli and Lorenzo Maccone  for deep conceptual discussions on quantum mechanics and its extensions. In particular, we thank Raúl Rossignoli for suggesting the use of a biorthogonal Schmidt decomposition. We are also grateful to Akram Touil and \mbox{Wojciech~Zurek} for insightful questions regarding the scope and perspectives opened by the formalism. We acknowledge valuable discussions with Alexey Milekhin, Wu-zhong Guo, and James Fullwood on timelike entanglement and pseudo-density matrices. We would also like to thank Lukasz Cincio for his feedback on tensor-network aspects of the work. N. L. D., M. C. and P.B. were supported by the Laboratory Directed Research and Development (LDRD) program of Los Alamos National Laboratory (LANL) under project number 20260043DR. N. L. D. also acknowledges support by the Center for Nonlinear Studies at LANL.

\clearpage

\appendix

\section{Spacetime state properties: proofs and additional technical details}\label{app:spacetimest}

Here we provide proofs and additional comments for the results of Sections \ref{sec:pureststates} and \ref{sec:fbevolution}.\\

\emph{Pure spacetime states}. Consider first pure spacetime states of the form $\mathcal{R}=|\psi\rangle_0\langle \psi| e^{i\tilde{\mathcal{S}}}$.
As a useful tool, let us first notice that since $e^{i\tilde{\mathcal{S}}}=\mathcal{V}^\dag e^{i\epsilon \mathcal{P}}\mathcal{V}$ we can write 
\begin{equation}\label{appeq:adjV}
    \mathcal{R}=\mathcal{V}^\dag\, |\psi\rangle_0\langle \psi| e^{i\epsilon \mathcal{P}}\,\mathcal{V}\,,
\end{equation}
 as it follows from $\mathcal{V}$ acting trivially on the first slice. We will use this result to study spectral properties of $\mathcal{R}$ ignoring the Hamiltonian part.

\begin{figure}[t!]
    \centering
    \includegraphics[width=\linewidth]{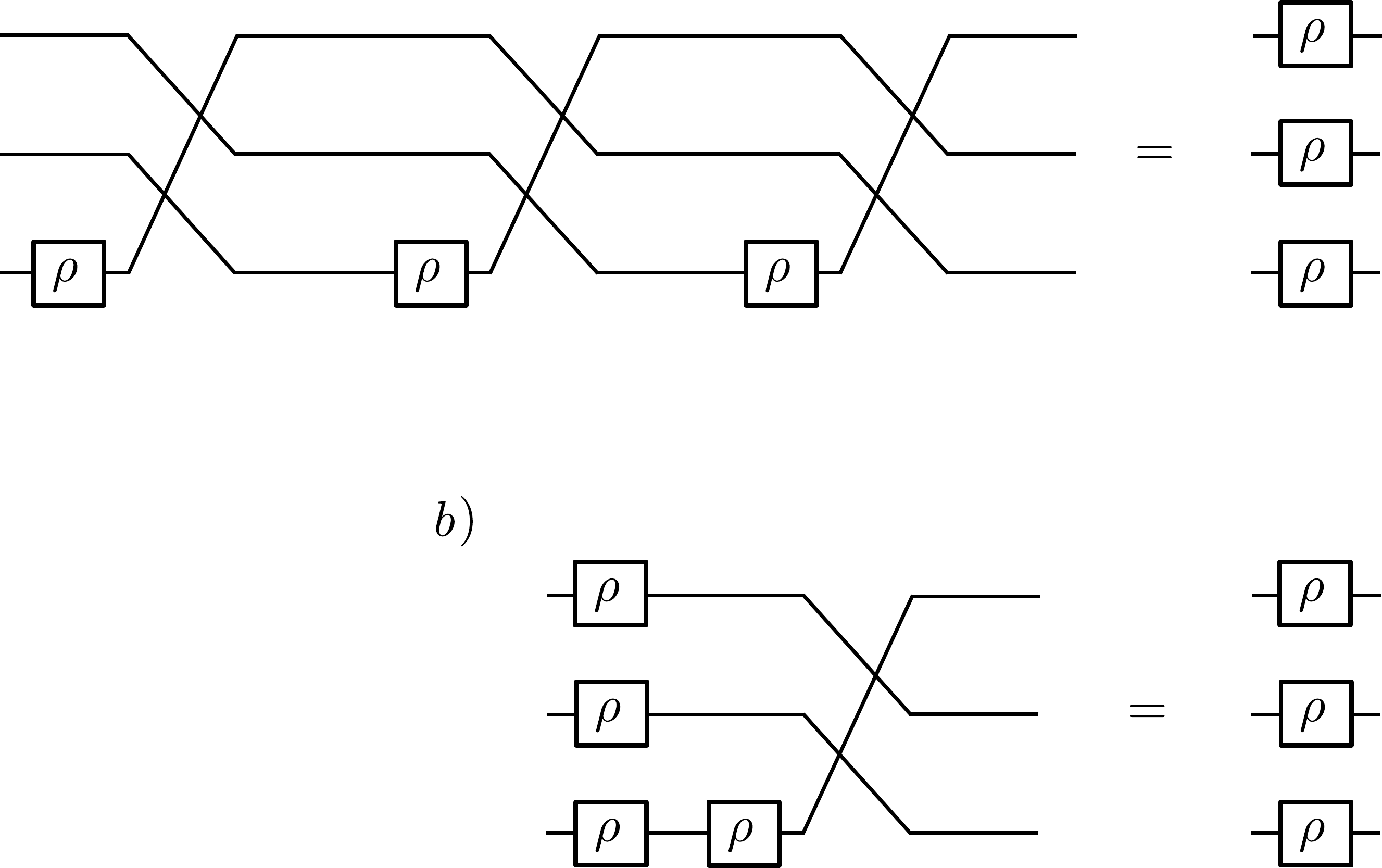}
    \caption{Diagrammatic proof of the decomposition $\mathcal{R}=\Pi+X$ without including evolution. On panel a) we depict $\mathcal{R}^3$ for $N=3$ leading to $\mathcal{R}^3=\rho^{\otimes 3}$. On panel b) we show that $\Pi \mathcal{R}=\Pi$ for pure states ($e^{-i\epsilon \mathcal{P}}|\psi\rangle^{\otimes N}=|\psi\rangle^{\otimes N}$). Similar considerations hold for arbitrary $N$. }
    \label{fig:decompproof}
\end{figure}

Let us now discuss the decomposition $  \mathcal{R}=\Pi+X$ described in Eq.\ \eqref{eq:decompositionpix} for $\mathcal{R}=\rho_0 e^{i\epsilon \mathcal{P}}$. A simple diagrammatic proof (see Figure \ref{fig:decompproof}) allows one to conclude $\mathcal{R}^N=\rho^{\otimes N}$ as it follows from the cyclical nature of $e^{i\epsilon \mathcal{P}}$. This shows that for an initial pure state $\mathcal{R}^N=\Pi$. We can then prove that $\Pi \mathcal{R}=\Pi$ and $\mathcal{R}\Pi=\Pi$ which is also easy to see graphically and depicted in  Figure \ref{fig:decompproof}. This means that $X=\mathcal{R}-\Pi$ satisfies $X\Pi=\Pi X=0$. Moreover, since $\mathcal{R}^N=\Pi+X^N$ we have $X^N=0$. Now, since the adjoint action of $\mathcal{V}$ in Eq.\ \eqref{appeq:adjV} cannot change the previous properties of $\Pi, X$ we conclude that the same results hold if evolution is included.

Notice also that $\mathcal{R}^\dag \mathcal{R}=\rho_0^2\otimes \mathbbm{1}$. Hence, for pure states $\mathcal{R}^\dag \mathcal{R}$ has an eigenvalue $1$ with degeneracy $d^{N-1}$, and one eigenvalue $0$ with degeneracy $d-1$. This defines the SVD of $\mathcal{R}$. \\

\emph{Pseudo-purities and Isospectrality}.
Here we discussed the properties of a pure spacetime state under bipartitions.
We will discuss now show how to fully characterize the pseudo-entropies of marginals. Let us first recall that given any operator $O$ we can write ${\rm tr}[O^k]={\rm tr}[O^{\otimes k} C]$ for $C$ a shift operator between copies. One can expand $C$ in terms of SWAPs just as we did with $e^{i\epsilon \mathcal{P}}$. Since the latter shifts among different time slices while $C$ among different copies of a single object we use different notation. Now, if instead we want ${\rm tr}[O^k_A]$ for $O_A={\rm tr}_B[O]$ what we need is the shift operator $C_A$, shifting only between copies of $A$. In this way the $B$ part gets traced away on each copy, i.e., ${\rm tr}[O^k_A]={\rm tr}[O^{\otimes k} C_A]$ where in the r.h.s. we use the full $O$.

Going back to the spacetime formalism, let us introduce a convenient notation to fully characterize the bipartition $A-B$ of $\mathcal{H}=\otimes_t h$: for each time slice we define a partition $A_t$ so that $A=\cup_t A_t$. Then it's clear that $B=\bar{A}=\cup_t B_t$ for $B_t=\bar{A}_t$ such that $h=h_{A_t}\otimes h_{B_t}$ for each time slice. Equivalently, $\mathcal{H}_{A(B)}=\otimes_t h_{A(B)}$. Notice also that $A_t$ can correspond to the full system or be empty. With this in mind we can introduce the cyclic operators $C_{A}=\otimes_t C_{A_t}$, $C_{B}=\otimes_t C_{B_t}$ which cyclically 
shift between copies of $\mathcal{H}_A$, $\mathcal{H}_B$ and where $C_{A_t(B_t)}$ do the same between copies of $h_{A_t}(B_t)$. We can then write
\begin{equation}\label{appeq:cyclictrick}
{\rm Tr}[\mathcal{R}_A^k]={\rm Tr}[\mathcal{R}^{\otimes k}C_A]\,,\quad  {\rm Tr}[\mathcal{R}_B^k]={\rm Tr}[\mathcal{R}^{\otimes k}C_B]\,.
\end{equation}
Notice also that since $\mathcal{R}^{\otimes k}$ is symmetric between copies we can always replace $C_B\to C_B^{-1}$ without changing the trace, i.e., ${\rm Tr}[\mathcal{R}_B^k]={\rm Tr}[\mathcal{R}^{\otimes k}C_B^{-1}]$ with $C_B^{-1}=C_B^{N-1}$.

\begin{figure*}[t!]
    \centering
    \includegraphics[width=0.9\linewidth]{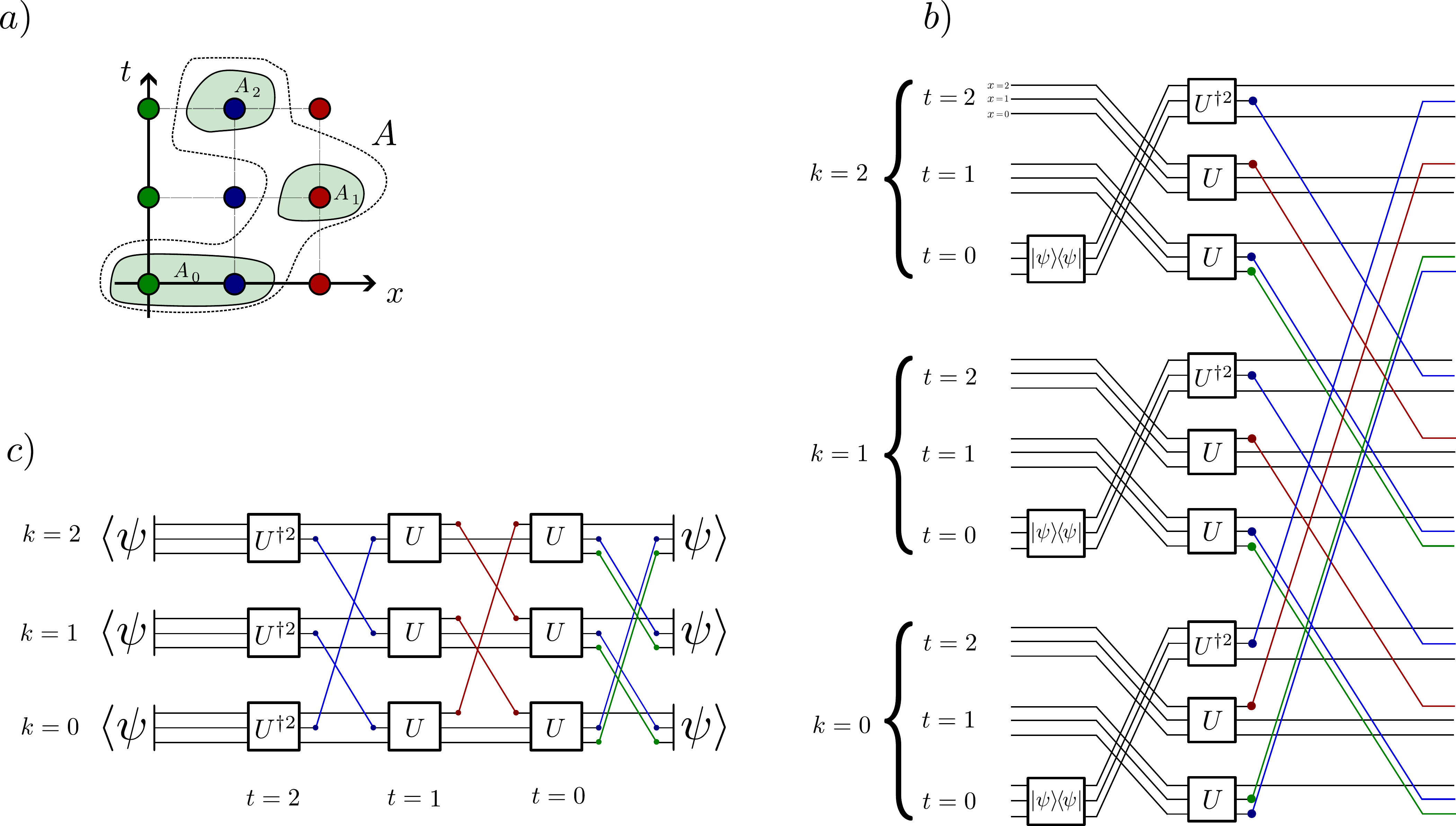}
    \caption{Example of the purity ${\rm Tr}[\mathcal{R}_A^3]$ for $3$ systems separated in space and $N=3$ times. We label the different copies with an index $k$ while for the systems at a given time, with local Hilbert spaces $h_0$, $h_1$, $h_2$, we use colors green, blue and red respectively. Here $A$ is defined by $\mathcal{H}_A=(h_0\otimes h_1)_{t=0}\otimes (h_2)_{t=1}\otimes (h_1)_{t=2}$ as depicted on panel a). On panel b) we depict the corresponding $\mathcal{R}^{\otimes 3} C_{A_0}\otimes C_{A_1}\otimes C_{A_2}$ whose trace yields ${\rm Tr}[\mathcal{R}_A^3]$. Notice how
    each $C_{A_t}$ acts only on the lines corresponding with the subystem $A_t$ with e.g., $C_{A_0}$ affecting the first two lines in space at the initial slice of each copy. 
    On panel c) we show how the trace of panel b) can be greatly simplified yielding  ${\rm Tr}[\mathcal{R}^{\otimes 3} C_{A_0}\otimes C_{A_1}\otimes C_{A_2}]=\langle \psi^{\otimes 3}| (U^{\dag 2})^{\otimes 2}C_{A_2}U^{\otimes 3}C_{A_1} U^{\otimes 3} C_{A_0} |\psi\rangle^{\otimes 3}$ which is a particular case of Eq.\  \eqref{appeq:puritiesform}. }
    \label{fig:bipartitions}
\end{figure*}
Having established the basic strategy and notation in Eq.\ \eqref{appeq:cyclictrick} we now proceed to evaluate the r.h.s. 
In Figure \ref{fig:bipartitions} we depict a particular example which is helpful to understand the following general argument.
The first thing to notice is that the type of contractions between objects are particularly constrained. The bra $\langle \psi|$ are entirely contracted with the $(U^\dag)^{N-1}$ of the same copy. The $(U^\dag)^{N-1}$ can only be contracted on the right with a $U$: if a line is not affected by $C_{A_{N-1}}$ (i.e., the line is not coming from $A_{N-1}$) then it contracts with the $U$ of the same copy and the previous time, as it follows from the structure of $\mathcal{R}$ in a given copy. Otherwise it gets contracted with the $U$ of the previous copy (and previous time), as it follows from the fact that the cyclic operator connects any given line in $A_t$ with the same line in the previous copy of $\mathcal{R}$. We can then repeat this analysis for all the right contractions on the operators $U$ showing that they get either connected with the previous (in time) $U$ on the same copy for lines not in $A_t$ or in the previous copy for lines in $A_t$. The same holds for the last contractions between $U$s and $|\psi\rangle$. In summary, we obtain a closed expression for the pseudo-entropies.
\begin{suplemma}\label{applemm:purities}
The pseudo-purities of a reduced spacetime state acting on $A=\cup_t A_t$ is given by the expectation value
 \begin{equation}\label{appeq:puritiesform}
\begin{split}
      &{\rm Tr}[\mathcal{R}_A^k]\\&=\langle \psi|^{\otimes k}(U^{\dag N-1})^{\otimes k}C_{A_{N-1}}U^{\otimes k}\dots   C_{A_1} U^{\otimes k}C_{A_0}|\psi\rangle^{\otimes k}\,,
\end{split}
\end{equation}
where $C_{A_t}$ is the cyclic shift operator acting on $A_t$. 
\end{suplemma}

Notably, the pseudo purities of degree $k$ are equal to an expectation value with respect to the standard quantum state $|\psi\rangle^{\otimes k}$. The corresponding operator is not Hermitian in general and involves single step evolutions intertwined with cyclic operators defining the partition at each time slice. When all $C_{A_t}$ are trivial we recover the pure condition ${\rm Tr}[\mathcal{R}_A^k]=1$. 
Moreover since the operators $\mathcal{W}=(U^{\dag N-1})^{\otimes k}C_{A_{N-1}}U^{\otimes k}\dots   C_{A_1} U^{\otimes k}C_{A_0}$ are unitary (for any choice of partitions $\mathcal{W}$ is a product of unitary operators) we must have 
\begin{equation}
    |{\rm Tr}[\mathcal{R}_A^k]|=|\langle \psi|^{\otimes k}\mathcal{W}|\psi\rangle^{\otimes k}|\leq 1\,,
\end{equation}
a property we discussed in Section \ref{sec:pureststates}.

With the supplemental Lemma \ref{applemm:purities} at hand we are also in a position to prove the isospectrality property. 
Notice first that a reasoning analogous to the previous allows us to write
\small
\begin{equation}
\begin{split}
      {\rm Tr}[\mathcal{R}_B^k]&={\rm Tr}[\mathcal{R}^{\otimes k}C_B^{-1}]\\&=\langle \psi|^{\otimes k}(U^{\dag N-1})^{\otimes k}C_{B_{N-1}}^{-1}\dots   C_{B_1}^{-1} U^{\otimes k}C^{-1}_{B_0}|\psi\rangle^{\otimes k}\,,
\end{split}
\end{equation}
\normalsize
where we employed the inverse of $C_B$ for convenience.  Now, since $C_{A_t}C_{B_t}=C_{h}$, namely the full cyclic operator, we can replace in Eq.\ \eqref{appeq:puritiesform} each $C_{A_t}= C_{A_t}C_{B_t}C_{B_t}^{-1}\to C_{B_t}^{-1}$, since $C_{h}$ acts trivially on all the operators involved. This proves the equality 
\begin{equation}
    {\rm Tr}[\mathcal{R}_A^k]={\rm Tr}[\mathcal{R}_B^k]\,.
\end{equation}
holding for all $k$. \\

\emph{Biorthogonal Schmidt}.
We now prove the biorthogonal Schmidt decomposition of Eq.\ \eqref{eq:jointschmidt}.
Let us first assume for simplicity that $d_A=d_B$ and both $\mathcal{R}_A$ and $\mathcal{R}_B$ are diagonalizable. Then we can write
\begin{equation}
    \mathcal{R}_A=\sum_\nu \lambda_\nu |\nu\rangle_A\langle \tilde{\nu}|\,,\quad \mathcal{R}_B=\sum_\nu \lambda_\nu |\nu\rangle_B\langle \tilde{\nu}|
\end{equation}
where we have defined the basis
\begin{equation}\label{appeq:bibasis}
\begin{split}
     |\nu\rangle_A&=\sum_i U_{i\nu}|i\rangle_A\,,\quad _A\langle \tilde\nu|=\sum_j U^{-1}_{\nu j} {}_A\langle j|\\
|\nu\rangle_B&=\sum_i V^{-1}_{\nu i}|i\rangle_B\,,\quad _B\langle \tilde\nu|=\sum_j V_{j\nu} {}_B\langle j|
\end{split}
\end{equation}
for 
$U,V$  the matrices diagonalizing $\mathcal{R}_A$ and $\mathcal{R}_B$ respectively. With these definition we find the biorthogonal conditions
 \begin{equation} _A\langle \tilde{\nu}'|\nu\rangle_A=_B\langle\tilde\nu'|\nu\rangle_B=\delta_{\nu\nu'}\,.
 \end{equation} 
Notice that we used the isospectrality condition to write the same eigenvalues $\lambda_\nu$ in both marginals. If we now define matrices
\begin{align}\label{appeq:defCD}
    C=U\Sigma V^{-1}\,,\quad
    D^\dag=V\Sigma^\dag U^{-1}\,,
\end{align}
with $\Sigma=\text{diag}(\sqrt{\lambda_\nu})$ we can introduce the pair of states
\begin{equation}
    |\Psi\rangle=\sum_{i,j}C_{ij}|i\rangle_A |j\rangle_B=\sum_\nu \sqrt{\lambda_\nu}|\nu\rangle_A |\nu\rangle_B\,,
\end{equation}
\begin{equation}
    |\Phi\rangle=\sum_{i,j}D_{ij}|i\rangle_A |j\rangle_B=\sum_\nu \sqrt{\lambda^\ast_\nu}|\tilde{\nu}\rangle_A |\tilde{\nu}\rangle_B\,,
\end{equation}
where we used the definition of the biorthogonal basis of Eq.\ \eqref{appeq:bibasis}. 
These states satisfy
\begin{equation}
    \langle \Phi| \Psi\rangle=\sum_\nu \lambda_\nu=1
\end{equation}
and
\begin{equation}\label{appeq:marginalspsiphi}
    \mathcal{R}_A={\rm Tr}_B\big[\,|\Psi\rangle \langle \Phi|\,\big]\,,\quad \mathcal{R}_B={\rm Tr}_A\big[\,|\Psi\rangle \langle \Phi|\,\big]
\end{equation}
as it follows from the biorthogonality condition. In particular, this shows that we have decomposed the spacetime states of the subsystems as $\mathcal{R}_A=\sum_{i,j}(CD^\dag)_{ij}|i\rangle \langle j|$ and $\mathcal{R}_B=\sum_{i,j}(D^\dag C)_{ji}|i\rangle \langle j|$. 
On the other hand, from Eq.\ \eqref{appeq:marginalspsiphi} if we write
\begin{equation}
    \mathcal{R}=|\Psi\rangle \langle \Phi|+X
\end{equation}
it must hold ${\rm Tr}_A [X]={\rm Tr}_B [X]=0$. This proves the result of Eq.\ \eqref{eq:jointschmidt}.

It is interesting to compare these decomposition with that of standard quantum states: we can replace $\mathcal{R}$ with a standard pure quantum state by taking $|\Phi\rangle\to |\Psi\rangle$ which corresponds to $D\to C$. Then the diagonalization of the marginals reduces to the diagonalization of $C^\dag C$, $CC^\dag$ which is what defines a standard SVD and the corresponding Schmidt decomposition. Conversely, if $\mathcal{R}_A(B)$ collapses to a standard positive semidefinite quantum state, it is natural to take $D=C\equiv \sqrt{\mathcal{R}_A}$ so that the joint Schmidt decomposition collapses to the standard one, even if $\mathcal{R}=|\Psi\rangle \langle \Psi|+X$ is not a standard quantum state.

Notice now that for $d_A\neq d_B$ but $\mathcal{R}_{A(B)}$ diagonalizable, the previous discussion holds by simply replacing $\Sigma$ in Eq.\ \eqref{appeq:defCD} with a rectangular operator of dimensions $d_A\times d_B$ containing only ``diagonal'' entries $\sqrt{\lambda_\nu}$ (and zeros otherwise).

Let us now briefly discuss the generalization to the case in which $\mathcal{R}_{A(B)}$ are not diagonalizable and only admit a nontrivial Jordan form. We assume that the two marginals admit compatible Jordan decompositions, and denote by $J_\beta(\lambda_\beta)$ the Jordan blocks (with size $m_\beta$). For each nonzero block, $ J_\beta(\lambda_\beta)=\lambda_\beta \mathbbm{1}+N_\beta$ with  $N_\beta^{m_\beta}=0$ and for $ \lambda_\beta\neq 0$
we choose a square root $J^{1/2}_\beta(\lambda_\beta)$.
Proceeding in analogy with the diagonalizable case, we construct the explicit part of the decomposition from the nonzero Jordan blocks only. Namely, if $|\beta,k\rangle_{A(B)}$ and $|\widetilde{\beta,k}\rangle_{A(B)}$ denote the corresponding right/left generalized eigenvector chains (biorthogonal Jordan bases), we define
\begin{equation}
\begin{split}
|\Psi\rangle&=\sum_{\beta:\,\lambda_\beta\neq 0}\sum_{k,\ell=1}^{m_\beta}
\big(J_\beta^{1/2}\big)_{k\ell}\,|\beta,k\rangle_A|\beta,\ell\rangle_B\\
|\Phi\rangle&=\sum_{\beta:\,\lambda_\beta\neq 0}\sum_{k,\ell=1}^{m_\beta}
\big(J_\beta^{1/2}\big)_{k\ell}^{\!*}\,|\widetilde{\beta,k}\rangle_A|\widetilde{\beta,\ell}\rangle_B\,.
\end{split}
\end{equation}
Equivalently, in matrix form this corresponds to choosing
\begin{equation}
    C=U\Sigma_JV^{-1}\,,\qquad D^\dag=V\Sigma_J^\dag U^{-1}
\end{equation}
for $\Sigma_J$ the (generally rectangular) matrix containing the blocks $J_\beta(\lambda_\beta)^{1/2}$ on the common nonzero Jordan sector and zeros otherwise.
The contribution $|\Psi\rangle\langle\Phi|$ reproduces the nonzero Jordan-sector part of the marginals. We then write the full spacetime state as
\begin{equation}
    \mathcal{R}=|\Psi\rangle\langle\Phi|+X\,,
\end{equation}
where all remaining terms (including the zero-eigenvalue sector) are absorbed into $X$. We then recover the conditions
$
    {\rm Tr}_A[X]={\rm Tr}_B[X]=0\,.
$

Notice that in contrast with the diagonalizable case, $J^{1/2}_\beta(\lambda_\beta)$ is not diagonal within each Jordan block, but still admits a finite polynomial expansion $J^{1/2}_\beta(\lambda_\beta)= \sqrt{\lambda_\beta} \mathbbm{1}_{m_\beta}+\text{polynomial in $N_\beta$}$. 
 Importantly, these nilpotent corrections modify the operator structure but not the eigenvalues of the marginals. Hence quantities depending only on the marginal spectra (such as pseudo-entropies defined from those eigenvalues) are unaffected by the nilpotent part.\\

\emph{Quantum channels and convex mixtures}.
Here we provide proofs related to  spacetime states corresponding to quantum channels and use them to provide an interpretation convex mixtures of closed spacetime states.\\

We start by proving  Theorem \ref{th:quantumchannel}.
\begin{proof}
    Consider a composite system $S+E$ in an initial state $\rho\otimes |0\rangle_E\langle 0|$. 
    The corresponding SQA is given by $$e^{i\tilde{\mathcal{S}}_{SE}}=(\mathbbm{1}\otimes U)\text{SWAP}_{SE}(\mathbbm{1}\otimes U^\dag)\,.$$ Here $\text{SWAP}_{SE}=\text{SWAP}\otimes \text{SWAP}_E$, namely a product operator in the partition system-environment. 
    The spacetime state is then
    \small
    \begin{equation}
        \begin{split}
            &{\rm Tr}_E[(\rho\otimes |0\rangle_E\langle 0|)_0e^{i\tilde{\mathcal{S}}_{SE}}]\\&=\rho_0\text{SWAP}\sum_{k_0,k_1}{}_{E}\langle k_0k_1|(|0\rangle_E \langle 0|\otimes U)\text{SWAP}_{E}(\mathbbm{1}\otimes U^\dag)|k_0k_1\rangle_{E}\\&=\rho_0\text{SWAP}\sum_{k_1}\langle k_1|U|0\rangle \otimes \langle 0|U^\dag|k_1\rangle\\&=\rho_0\text{SWAP}\sum_{k}E_k \otimes E_k^\dag\,.
        \end{split}
    \end{equation}
    \normalsize
    Notice that in the second line the tensor product symbol indicates time, and $E_k=\langle k|U|0\rangle$.
    Moreover,  \begin{equation}
        \begin{split}
            \text{SWAP}\sum_{k}E_k \otimes E_k^\dag\\=\sum_k\mathbbm{1}\otimes E_k \text{SWAP}\sum_{k}\mathbbm{1} \otimes E_k^\dag\\=\sum_{i,j}|i\rangle \langle j|\otimes \sum_k E_k |j\rangle \langle i|E_k^\dag= J(\mathcal{E})
        \end{split}\,.
    \end{equation}

\end{proof}

We now prove Lemma \ref{lemma:specquantumchannels} about the spectrum of a spacetime state $\mathcal{R}_S=|\psi\rangle \langle \psi|J(\mathcal{E})$. 
\begin{proof}
    Consider the spacetime state with $J(\mathcal{E})$ expanded in some basis:
    $\mathcal{R}_S=|\psi\rangle \langle \psi| \otimes \mathbbm{1}\sum_{i,j}|i\rangle \langle j|\otimes \mathcal{E}(|j\rangle \langle i|)=\sum_{i,j}\langle \psi|i\rangle |\psi\rangle \langle j|\otimes \mathcal{E}(|j\rangle \langle i|)$. Now by writing $\langle \psi|i\rangle=\psi_i^\ast$ we find $\mathcal{R}_S=\sum_{i,j}\psi_i^\ast |\psi\rangle \langle j|\otimes \mathcal{E}(|j\rangle \langle i|)$. Using linearity of the map we conclude that
    \begin{equation}\label{appeq:channelpure}
        \mathcal{R}_S=\sum_j |\psi\rangle \langle j|\otimes \mathcal{E}(|j\rangle \langle \psi|)\,.
    \end{equation}
Using this expression for the spacetime state, a direct evaluation leads to
$\mathcal{R}_S \mathcal{R}_S=\mathbbm{1}\otimes \mathcal{E}(|\psi\rangle \langle \psi|)\mathcal{R}_S$. Iterating  we find the property
\begin{equation}
    \mathcal{R}_S^k=\mathbbm{1}\otimes \mathcal{E}(|\psi\rangle \langle \psi|)^k\mathcal{R}_S\,.
\end{equation}
Combining this result with Eq.\ \eqref{appeq:channelpure} we find
\begin{equation}
    {\rm Tr}[\mathcal{R}_S^k]={\rm tr}[\mathcal{E}(|\psi\rangle \langle \psi|)^k]\,.
\end{equation}
Since this holds for all values of $k$ we have proven that $\mathcal{R}_S$ and $\mathcal{E}(|\psi\rangle \langle \psi|)$ share the same non-null spectrum.
\end{proof}

Let us also add that the vectors $|\psi\rangle\otimes |\phi_\alpha\rangle$, for $\mathcal{E}(|\psi\rangle \langle \psi|)=\sum_\alpha \lambda_\alpha |\phi_\alpha\rangle \langle \phi_\alpha|$ the eigendecomposition of the state after the action of the channel, are the eigenvectors of $\mathcal{R}_S$ with eigenvalues $\lambda_\alpha$. This is easily verified since \eqref{appeq:channelpure} implies $\mathcal{R}_S|x\rangle \otimes |v\rangle=|\psi\rangle \otimes \mathcal{E}(|x\rangle \langle \psi|)|v\rangle$. \\

Let us now discuss convex mixtures of spacetime states. Let us first notice that for the case considered in Eq.\ \eqref{eq:convexmix} of the main text we can provide an explicit ``purification'' as follows: define an environment reference state $|e_0\rangle=\sum_k \sqrt{p_k\!}|k\rangle$ and take as initial state $\rho_{SE}=\rho \otimes |e_0\rangle_E\langle e_0|$. Then, consider the  controlled unitary $W_{SE}=\sum_k U_k \otimes |k\rangle_E\langle k|$. It is straightforward to verify that $W_{SE}$ provides a  Stinespring dilation of the channel $\mathcal{E}$, namely $E_k=\sqrt{p_k\!}\,U_k={}_E\langle k|W_{SE}|e_0\rangle_E$. Then, following Theorem \ref{th:quantumchannel} we recover  the mixed $\mathcal{R}$ of Eq.\ \eqref{eq:convextoJ} as the marginal of $\mathcal{R}_{SE}
    =
    (\rho_{SE})_0\,
    \text{SWAP}_{SE}\,
    (W_{SE}\otimes W_{SE}^\dag)$. 
Let us now consider a more general convex mixture of closed two-time spacetime states of the form
$
    \mathcal{R}=\sum_k p_k\, \mathcal{R}_k,
$
with
\begin{equation}
    \mathcal{R}_k=
    |\psi_{i_k}\rangle_0\langle\psi_{i_k}|\,
    \text{SWAP} \,
    (U_{j_k}\otimes U_{j_k}^\dag )\,.
\end{equation}
Here the maps $k\to i_k$, $k\to j_k$ specify which initial pure state $|\psi_{i_k}\rangle$, and evolution $U_{j_k}$ are associated with the spacetime state  $k$. This includes as a special case the previous scenario where all initial states coincide, namely $|\psi_{i_k}\rangle=|\psi\rangle$ for all $k$.

A natural question is whether such a convex mixture can still be understood as the reduction of a larger closed history. The answer is yes, if we allow for initial correlations between the system and the environment:
Introduce an environment with orthonormal basis $\{|k\rangle_E\}$ and define the pure state
$
    |\Psi_{SE}\rangle
    =
    \sum_k \sqrt{p_k}\,
    |\psi_{i_k}\rangle_S\otimes |k\rangle_E\,.
$
The corresponding density matrix is
\small
\begin{equation}
\begin{split}
    \rho_{SE}
    =
    |\Psi_{SE}\rangle\langle\Psi_{SE}|
    =
    \sum_{k,k'}
    \sqrt{p_k p_{k'}\!}\,
    |\psi_{i_k}\rangle\langle\psi_{i_{k'}}|_S
    \otimes
    |k\rangle\langle k'|_E\,.
    \end{split}
\end{equation}
\normalsize
This is the natural  extension of the previous Stinespring-like construction, recovered for $|\psi_{i_k}\rangle=|\psi\rangle$. 
We now define the controlled unitary 
\begin{equation}
    W_{SE}
    =
    \sum_k U_{j_k}\otimes |k\rangle\langle k|_E\,.
\end{equation}
Notice that for all $U_{j_k}=U$ we obtain  $W_{SE}=U\otimes \mathbbm{1}_E$ so that any correlations between the system and the environment are merely contained in the initial state. 
At the level of ordinary quantum states one then finds
$
    W_{SE}\rho_{SE}W_{SE}^\dagger
    \!=\!
    \sum_{k,k'}
    \sqrt{p_kp_{k'}\!}\,
    U_{j_k}|\psi_{i_k}\rangle\langle\psi_{i_{k'}}|U_{j_{k'}}^\dagger
    \otimes
    |k\rangle\langle k'|_E\,.
$
Tracing out the environment gives
\begin{equation}
    \rho_S'
    =
    {\rm Tr_E}\!\big[W_{SE}\rho_{SE}W_{SE}^\dagger\big]
    =
    \sum_k p_k\,
    U_{j_k}|\psi_{i_k}\rangle\langle\psi_{i_k}|U_{j_k}^\dagger.
\end{equation}
So, already at the level of standard states, the reduced system is the desired convex mixture.
The corresponding global two-time spacetime state is instead
\begin{equation}
    \mathcal{R}_{SE}
    =
    (\rho_{SE})_0\,
    \text{SWAP}_{SE}\,
    (W_{SE}\otimes W_{SE}^\dag)\,.
\end{equation}
Tracing out the environment, one obtains
\begin{equation}
    {\rm Tr}_E[\mathcal{R}_{SE}]
    =
    \sum_k p_k\,
    |\psi_{i_k}\rangle_0\langle\psi_{i_k}|
    \,\text{SWAP}\,
    (U_{j_k}\otimes U_{j_k}^\dag)\,,
\end{equation}
which is just the convex mixture defining $\mathcal{R}$, i.e., there is always a closed spacetime states $\mathcal{R}_{SE}$ such that $\mathcal{R}={\rm Tr}_E[\mathcal{R}_{SE}]$. Let us prove this last relation. 
\begin{proof}
Consider
\small
\begin{equation}
\begin{split}
    &{\rm Tr}_E[\mathcal{R}_{SE}]
    = {\rm Tr}_E\!\left[(\rho_{SE})_0\,\text{SWAP}_{SE}\,(W_{SE}\otimes W_{SE}^\dagger)\right]\\
    &= \sum_{k,k'} \sqrt{p_kp_{k'}\!}\,|\psi_{i_k}\rangle_0\langle\psi_{i_{k'}}|\,\text{SWAP}\, \times \\
    &\sum_{k_0,k_1}{}_E\langle k_0k_1|(|k\rangle_E \langle k'|\otimes W_{SE})\text{SWAP}_E
    (\mathbbm{1}\otimes W_{SE}^\dagger)|k_0k_1\rangle_E\,,
   \end{split}
\end{equation} 
    \normalsize
    where we separated all the terms that depend only on the system. 
We can now focus on the terms under the sum over $k_0,k_1$. By writing $W_{SE}=\sum_l U_{j_l}\otimes |l\rangle_E \langle l|$ a straightforward calculation leads to 
\small
\begin{equation}
    \begin{split}
        &\sum_{k_0,k_1}{}_E\langle k_0k_1|(|k\rangle \langle k'|\otimes W_{SE})\text{SWAP}_E
    (\mathbbm{1}\otimes W_{SE}^\dagger)|k_0k_1\rangle_E\\&=\sum_{l,l',k_0,k_1} U_{j_l}\otimes U^\dag_{j_{l'}}\delta_{kk_0}\delta_{l k_1}\delta_{lk}\delta_{l'k'}\delta_{ll'}=\delta_{kk'} U_{j_k}\otimes U^\dag_{j_{k}}\,.
    \end{split}
\end{equation}
\normalsize
Notice that the factor $U_{j_l}\otimes U^\dag_{j_{l'}}$ can be pulled out from the start as it only acts on the system. The final deltas are thus obtained directly by computing overlaps among kets and bras of the environment. 
Replacing this result in ${\rm Tr}_E[\mathcal{R}_{SE}]$ gives the desired relation.
\end{proof}

We see that $\mathcal{R}$
is indeed the reduced spacetime state of a larger closed system which contains both initial state correlations and entangling evolution. \\

\emph{Forward and backward evolution}. Here we explain how Theorem \ref{th:extendedQA} can be proven in analogy with Theorem \ref{th:theorem1}. The first thing to notice is that  the operator $e^{i\epsilon \mathcal{P}_{\text{ext}}}$ is just $e^{i\epsilon \mathcal{P}}$ for $2N$ times after a rearrangement of the time slices in the second copy of $\mathcal{H}$. This is explained in Figure \ref{fig:eiprear}. The within evolution on the second copy $e^{i\mathcal{S}_{\text{ext}}}$ is symmetric under permutations, while the factor $U_N$ gets mapped to $U_{2N-1}$. 
So after the rearrangement we get $e^{i\mathcal{S}_{\text{ext}}}\to e^{i\epsilon \mathcal{P}} e^{-i\epsilon (\mathcal{K}\otimes \mathbbm{1}-\mathbbm{1}\otimes \mathcal{K})}U^\dag_{N-1}U_{2N-1}$ which is essentially the form of $e^{i\mathcal{S}}$ but with backward evolution on the second half. If one considers correlators Wightman functions follow in complete analogy with Theorem \ref{th:theorem1} but with $U^\dag$ operators accompanying the $B$ operators of Theorem \ref{th:extendedQA}.
Notice also that under this rearrangement we can write
\begin{equation}
    \mathcal{R}_{\text{ext}}\to \mathcal{R}_2U^\dag_{N-1}U_{2N-1}\,.
\end{equation}
Thus the only difference between $\mathcal{R}_{\text{ext}} $ and $\mathcal{R}_2$ is a convention in the Hilbert space ordering and a border term changing the amount of evolution of operators on the backward sector by a single time step. \\

\begin{figure}[t!]
    \centering
    \includegraphics[width=\linewidth]{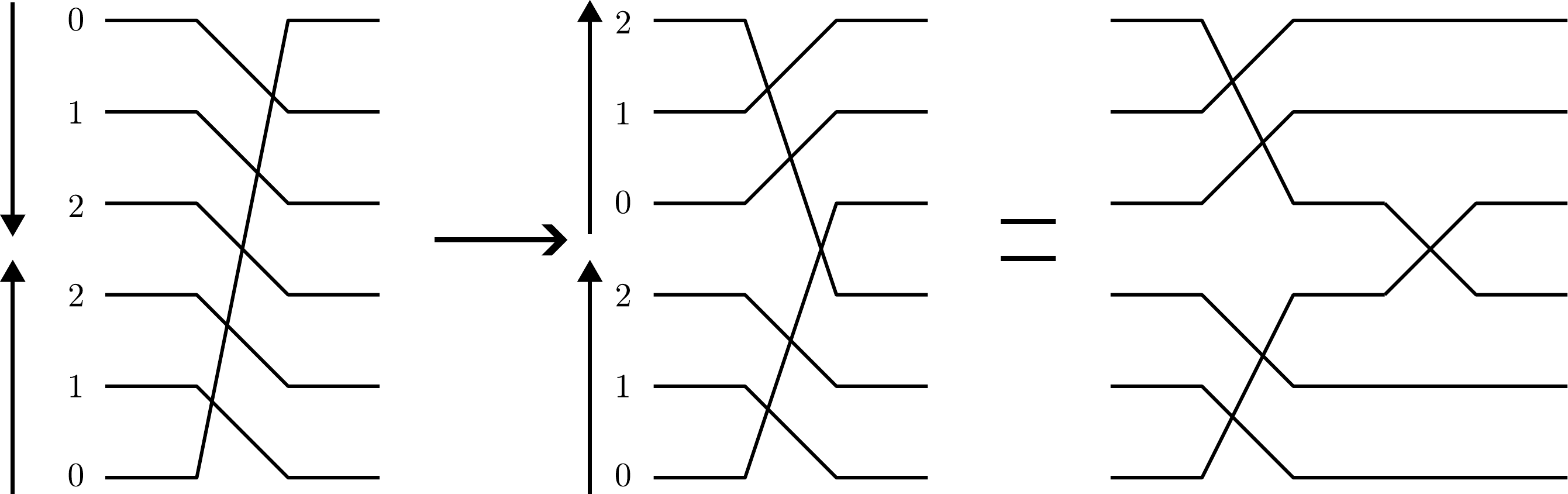}
    \caption{We show a rearrangement of the Hilbert spaces on the second copy of $\mathcal{H}$ acting on $e^{i\epsilon \mathcal{P}_{2}}$ and leading to $e^{i\epsilon \mathcal{P}_{\text{ext}}}$. At first, times are labeled backward. A rearrangement can always be understood as $e^{i\epsilon \mathcal{P}}\otimes e^{-i\epsilon \mathcal{P}} \text{SWAP}_{N-1,N}$ with the SWAP operator at the end being the only entangling operator between the two copies of $\mathcal{H}$.}
    \label{fig:eiprear}
\end{figure}

\emph{Time-dependent case and folded spacetime states}. 
For completeness let us recall here how the spacetime states of Section \ref{sec:formalism} generalize straightforwardly to a time-dependent evolution. Then, we will discuss how the  properties of Section \ref{sec:pureststates} proved in this manuscript generalize as well.

The time-dependent quantum actions are introduced as follows: we define $\mathcal{V}=\otimes_{t=0}^{N-1}U^\dag(\epsilon t)$ with $U(\epsilon t)=\hat{T}\exp(i\int_0^{\epsilon t} dt' H(t'))$ the general time evolution operator. Then we have
\begin{equation}
  e^{i\tilde{\mathcal{S}}}=\mathcal{V}^\dag e^{i\epsilon \mathcal{P}}\mathcal{V}\,,  
\end{equation}
and
\begin{equation}
     e^{i\mathcal{S}}= U_0(T)e^{i\tilde{\mathcal{S}}}=
     e^{i\epsilon \mathcal{P}}\otimes_{t=0}^{N-1} U[\epsilon(t+1),\epsilon t]\,.
\end{equation}
Clearly for time-independent Hamiltonians we recover the main text definition. The corresponding spacetime states are defined as before
\begin{equation}
    \mathcal{R}=\rho_0e^{i\tilde{\mathcal{S}}}=\mathcal{V}^\dag \rho_0e^{i\epsilon \mathcal{P}}\mathcal{V}
\end{equation}
and satisfy all the time-dependent version of the properties discussed in Section \ref{sec:basicformalism}.

We now discuss the spectral properties of reduced spacetime states when the Hamiltonian is time-dependent. Let us first discuss the supplemental Lemma \ref{applemm:purities}. Notably it is trivial to generalize to the time-dependent case: consider the example of the Figure \ref{fig:bipartitions}.  The operators $U\otimes U \otimes (U^\dag)^2$ acting on the left of the time translations across slices on panel b) get replaced by $U_{10}\otimes U_{21}\otimes U^\dag_{20}$ for $U_{tt'}\equiv U(\epsilon t', \epsilon t)$. This happens in the $3$-copies of $\mathcal{R}$. Then, on panel c) we just need to replace again each $(U^\dag)^2\to U^\dag_{20}$, each at time $t=1$ with $U_{21}$ and each $U$ at time $t=0$ with $U_{10}$. The same reasoning holds for any other case, allowing us to write 
\small
 \begin{equation}
\begin{split}
      &{\rm Tr}[\mathcal{R}_A^k]\\&=\langle \psi|^{\otimes k}(U^{\dag \otimes k}_{N-1,0})^{\otimes k}C_{A_{N-1}}U_{N-1,N-2}^{\otimes k}\dots   C_{A_1} U_{10}^{\otimes k}C_{A_0}|\psi\rangle^{\otimes k}\,.
\end{split}
\end{equation}
\normalsize
Since the same reasoning applies to the other partition, the proof of isospectrality holds for time-dependent Hamiltonians. As a corollary, considering the discussion on folded spacetime states of Section \ref{sec:otocs}, the isospectrality property also holds for $\mathcal{R}_k$. Moreover, since the joint Schmidt decomposition only employs the isospectrality assumption, Eq.\ \eqref{eq:jointschmidt} holds as well for diagonalizable marginals.

\section{Numerical methods}\label{app:numerics}

All of the following numerical investigations have been conducted via customized code developed on top of the \texttt{ITensors.jl} tensor network package~\cite{fishman2022itensor} for the \texttt{Julia} programming language~\cite{bezanson2017julia}.

\subsection{Imagitivity}
To describe our method let us first notice that since we are considering only two sites $\mathcal{R}_A=\frac{1}{4}\sum_{i,j}{\rm Tr}[\mathcal{R}P_i\otimes P_j] P_i\otimes P_j$ where $P_i$ are the Pauli matrices which form a complete orthogonal basis of $h$. By writing the corresponding expression for $\mathcal{R}^\dag_A$ and using Corollary \ref{cor:wigthman}
we obtain
\begin{equation}
    \mathcal{R}_A-\mathcal{R}_A^\dag=\frac{1}{4}\sum_{i,j}\langle \psi|[P_i(t),P_j] |\psi\rangle P_i\otimes P_j\,.
\end{equation}
where we assumed a pure initial state $|\psi\rangle$. 
From this expressions we can write the imagitivity as $|| \mathcal{R}_A-\mathcal{R}_A^\dag||^2=\frac{1}{4}\sum_{i,j}|\langle \psi|[P_i(t),P_j] |\psi\rangle|^2$.

Now we define $|\psi_j\rangle=U(t)P_j|\psi\rangle$, $|\varphi_i\rangle=P_iU(t)|\psi\rangle$ which allows to write the imagitivity as the overlap 
\begin{equation}
    ||\mathcal{R}_A-\mathcal{R}_A^\dag||^2=\sum_{i,j} \text{Im}[\langle \varphi_i|\psi_j\rangle]^2\,.
\end{equation}
This is the basis of our numerical calculation.

In our computation of the imagitivity depicted in Figure \ref{fig:lightcone} we obtain $|\psi\rangle$ from a DMRG approximation of the ground state with bound dimension $\chi=64$ and $25$ sweeps. Then $|\psi_j\rangle$ is obtained by simply acting on $|\psi\rangle$ with the local Pauli $P_j$ first and then evolving through a Trotter approximation updating our tensor on each step, namely we use the TEBD method with second order Trotter. To compute $|\varphi_i\rangle$ we first evolve with the same method and in the end we apply $P_i$. We used a time step of $\epsilon=0.05$ and the evolved state bound dimension was fixed to a maximum of $\chi=256$.

\subsection{OTOCs}

Consider the operator $\tilde {\mathcal{R}}_4$ restricted to the support of sites $r,r'$ and times $t_1=0, t_2\equiv t$. From Eq.\ \eqref{eq:otocCcomm} we get 
\begin{equation}
    \widetilde{\mathcal{R}}_{4,A}=\frac{1}{16}\sum_{i,j,k,l} C_{ij;kl}(r,r';t) P_i\otimes P_j \otimes P_k \otimes P_l
\end{equation}
where we recall that the tensor, for $r,r'$ sites and $P_i$ the Pauli matrices ($i=x,y,z$), is given by 
\begin{equation}
     C_{ij;kl}(r,r';t)={\rm tr}[ [P_{r,i}(t),P_{r',j}]^\dag[P_{r,k}(t),P_{r',l}]]\,.
\end{equation}
This implies $||\widetilde{\mathcal{R}}_{4,A}||^2=\frac{1}{16}\sum_{i,j,k,l}| C_{ij;kl}(r,r';t)|^2$. For our numerical computation we have rewritten the tensor in 
vectorized notation 
\begin{widetext}
\begin{equation}
    C_{ij;kl}(r,r';t)
=
\langle P_{r,i}(t)|
\Big(
P_{r',j} P_{r',l}\otimes \mathbbm 1
+\mathbbm 1\otimes P^\ast_{r',j}P^\ast_{r',l}
-P_{r',l}\otimes P_{r',j}^\ast
-P_{r',j}\otimes P_{r',l}^\ast
\Big)
|P_{r,k}(t)\rangle 
\end{equation}
\end{widetext}
Here the states represent the vectorization $|O\rangle=O\otimes \mathbbm{1}|\Phi^+\rangle$ and are evolved through the super Hamiltonian $H_s=H\otimes \mathbbm{1}-\mathbbm{1}\otimes H^\ast$. In the actual numerical computation we have represented the internal operators as local MPOs and the external ``states'' through MPS, evolved via Trotter steps, i.e., TEBD method applied to the vectorized ``external'' Paulis.

Notice that for the diagonal entries of the tensor we get a simpler expression $C_{ij;ij}(r,r';t)
=
2-2\Re\left(\langle P_{r,i}(t)|\,P_{r',j}\otimes P_{r',j}^\ast\,|P_{r,i}(t)\rangle\right)$. If moreover \(P_{r',j}\) is real in the chosen basis,
$$
C_{ij;ij}(r,r';t)
=
2-2\,\langle P_{r,i}(t)|\,P_{r',j}\otimes P_{r',j}\,|P_{r,i}(t)\rangle \,,
$$
which is the result in \cite{xu2020accessing}, employed to access scrambling witnesses through MPOs. 
In this sense, our current scheme is the direct generalization of 
\cite{xu2020accessing} to the full tensor $C_{ij;ij}(r,r';t)$.

For our computations of $C_{ij;kl}(r,r';t)$ presented in Section \ref{sec:otocs} we considered a time step of $\epsilon=0.002$ and a maximal bound dimension of $\chi=24$ for both the chaotic and free models. Contrary to what is reported in \cite{xu2020accessing} for the diagonal part of $C$ in the free case, the full tensor exhibits bound dimensions grow with time, rapidly exceeding $\chi=4$.

\begin{figure}[h!]
    \centering
    \includegraphics[width=0.85\linewidth]{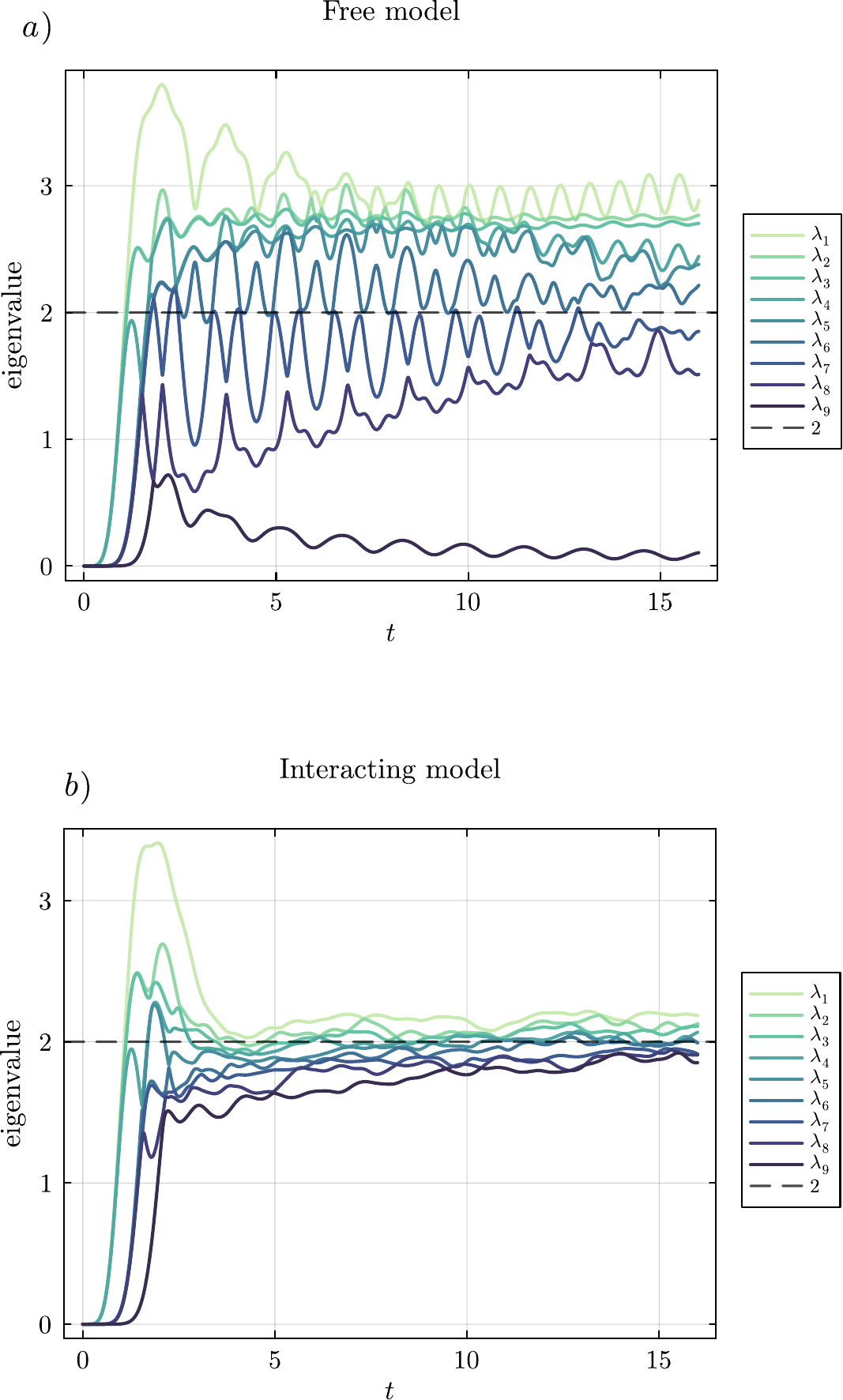}
    \caption{Eigenvalues of the Gram matrix $G$ for a fixed separation $r'-r=2$ as a function of time. For each instant of time, the eigenvalues are ordered in magnitude. }
    \label{fig:eigenvaluesotoc}
\end{figure}

Having computed the full tensor for each value of $r',t$ we reshape it onto the matrix $G_{ab}$. We then computed the eigenvalues of the Gram matrix to obtain the results reported in Figure \ref{fig:OTOCnorm}. In particular $||G||^2=\sum_{i=1}^9 \lambda_i^2$ from which $||\widetilde{\mathcal{R}}_{4,A}||$ is obtained. As an illustration, we plot the $9$ eigenvalues as a function of time in Figure \ref{fig:eigenvaluesotoc} for $r'=28$.

\section{Path Integrals: matrix elements of the exponential of QAs}\label{app:PI}

Here we present additional technical details on PIs mentioned
in Section \ref{sec:PIs}. In particular, we provide a close expression for the generator of time translations and obtain the matrix elements of the exponentials of $\mathcal{S}$ and $\mathcal{S}_{\text{ext}}$.

\subsection{Feynman PI}

Let us briefly explain how to obtain the close expression for $\mathcal{P}$ of Eq.\ \eqref{eq:Plegendre}. We first notice than rather than employing operators $q_t,p_t$ we can work with ladder operators $a_t,a_t^\dag$ satisfying 
\begin{equation}\label{eqapp:ladders}
    [a_t,a^\dag_{t'}]=\delta_{tt'}\,,
\end{equation}
 with other commutators vanishing. In fact, we can always define $a_t=\sqrt{\frac{\omega}{2}}(q_t+ip_t/\omega)$, $a^\dag_t=\sqrt{\frac{\omega}{2}}(q_t-ip_t/\omega)$ for arbitrary $\omega$.

Since Eq.\ \eqref{eqapp:ladders} shows that we have a mode for each $t$ we can define a Fourier in time transformation
\begin{equation}
    a_n=\frac{1}{\sqrt{N}}\sum_t e^{i \omega_n \epsilon t}a_t
\end{equation}
for $\omega_n=2\pi n/T$ such that $[a_n,a^\dag_{n'}]=\delta_{nn'}$ as one can easily verify. The new operators are genuine ladder operators in the $n$-modes showing that this is a well-defined canonical transformation. 

Now it becomes straightforward to show that
\begin{equation}
    \mathcal{P}=\sum_n \omega_n a_n^\dag a_n\,,
\end{equation}
generates time translations: $e^{i\epsilon \mathcal{P}}a_n e^{-i\epsilon\mathcal{P}}=e^{-i\epsilon\omega_n}a_n$
while the inverse Fourier transform gives $a_t=\frac{1}{\sqrt{N}}\sum_t e^{-i \omega_n \epsilon t}a_n$ leading to $e^{i\epsilon \mathcal{P}}a_t e^{-i\epsilon\mathcal{P}}=a_{t+1}$. 
We can also write 
\begin{equation}
    \mathcal{P}=\sum_{t,t'} a^\dag_t iD_{tt'}a_{t'}\equiv \sum_t a^\dag_t i\dot{a}_t\,,
\end{equation}
with $D_{tt'}=-\frac{1}{N}\sum_n i \omega_n e^{i\omega_n \epsilon (t-t')}$ a discrete derivative. This last expression already has the form of the classical Legendre transform in terms of conjugate variables $a^\dag, ia$. It is a straightforward calculation to obtain $ \mathcal{P}=\sum_t p_t \dot{q}_t$ by replacing the definition of ladder operators in terms of $q_t,p_t$. The final result is independent of $\omega$. As a matter of fact, the time translation operators exhibits invariance under a quite general set of transformation. See the discussion in \cite{diaz2021spacetime}.\\

Let us now show how the classical action arises from matrix elements of the QA. Consider 
\begin{equation}\label{eq:matrixeleis}
\begin{split}
    \langle \textbf{q}_f|e^{i\mathcal{S}} |\textbf{q}_i\rangle&=\langle \textbf{q}_f|e^{i\epsilon \mathcal{P}}e^{-i\epsilon \mathcal{K}} |\textbf{q}_i\rangle\\
    &=\int \prod_{t=0}^{N-1}dp_t\langle \textbf{q}_f|e^{i\epsilon \mathcal{P}}|\textbf{p}\rangle \langle \textbf{p}|e^{-i\epsilon \mathcal{K}} |\textbf{q}_i\rangle\,,
\end{split}
\end{equation}
where we used the momentum completeness relation. 
Now, by using Eq.\ \eqref{eq:trajoverlap} it is straightforward to show that 
\begin{equation}\label{eq:matrixellegendre}
    \langle \textbf{q}_f|e^{i\epsilon \mathcal{P}}|\textbf{p}\rangle \langle \textbf{p}|\textbf{q}_i\rangle=\frac{1}{(2\pi)^N}\,e^{\frac{i}{\hbar}\sum_t \epsilon\, p_t\left(\frac{q_{t+1}-q_t}{\epsilon}\right)\big|_{q_0=q_i}^{q_N=q_f}}\,,
\end{equation}
thus confirming that the translations across time-slices are linked to the classical Legendre transform.

On the other hand, consider a standard Hamiltonian $H=\frac{p^2}{2m}+V(q)$ such that  ${\epsilon\mathcal{K}=\sum_t \epsilon \big[\frac{p_t^2}{2m}+V(q_t)\big]}$. We have
\begin{equation}\label{eq:matrixelham}
    \langle \textbf{p}|e^{-i\epsilon \mathcal{K}}|\textbf{q}_i\rangle=\langle \textbf{p}|\textbf{q}_i\rangle\, e^{-i\sum_t \epsilon\big(\frac{p_t^2}{2m}+V(q_t)\big)|_{q_0=q_i}}+\mathcal{O}(\epsilon^2)\,,
\end{equation}
where we used the usual Trotter approximation to separate the kinetical and potential contributions. 
Combining \eqref{eq:matrixeleis}, \eqref{eq:matrixellegendre} and \eqref{eq:matrixelham} leads directly to 
\small
\begin{equation}
\begin{split}
    \langle \textbf{q}_f|e^{i\mathcal{S}} |\textbf{q}_i\rangle
    \!&=\!\int \!\prod_{t=0}^{N-1}\frac{dp_t}{2\pi}\, e^{i\sum_t \epsilon \big[p_t\left(\frac{q_{t+1}-q_t}{\epsilon}\right)-\frac{p_t^2}{2m}-V(q_t)\big]\big|_{q_0=q_i}^{q_N=q_f}}\\ 
    &=\frac{1}{(\sqrt{2\pi i \epsilon/m})^N}\, e^{i\sum_t \epsilon\, \left(\frac{1}{2}m \dot{q}^2_t-V(q_t)\right)\big|_{q_0=q_i}^{q_N=q_f}}\,,
\end{split}
\end{equation}
\normalsize
which holds up to second order in $\epsilon^2$ and where we are using the notation (for classical variables) $\dot{q}_t\equiv \frac{q_{t+1}-q_t}{\epsilon}$.

Notice that the integral over $p_t$ is standard for Hamiltonians which are quadratic in the momentum. This integral allows one to work directly in configuration space, and it is a basic assumption of Feynman's formulation. Nonetheless in the SQM approach one can work directly with the QA and there is no need for this extra assumption.

\subsection{Schwinger–Keldysh PI}
Here we provide an explicit derivation of the matrix elements of $e^{i\mathcal{S}_{\text{ext}}}$ needed to connect it to the Schwinger–Keldysh PI.

Let us focus first in the time translation operator defined in Eq.\ \eqref{eq:extendedP}. Notice that acting on the right it essentially translates one step in the future direction when acting on $\textbf{p}^+$ and one steps on the past direction when acting on $\textbf{p}^-$. In addition the operator $\text{SWAP}_{N-1,N}$  relates the two halves. By using these properties and   Eq.\ \eqref{eq:trajoverlap} for both copies of position/momentum operators
one finds 
\begin{align}
&\langle \textbf{q}^+\textbf{q}^-|e^{i\epsilon \mathcal{P}_{\text{ext}}}|\textbf{p}^+\textbf{p}^-\rangle \langle \textbf{p}^+\textbf{p}^-|\textbf{q}^+\textbf{q}^-\rangle\nonumber\\
&=\frac{1}{(2\pi)^{N}}e^{i\epsilon\left(\sum_{t=0}^{N-1}  p^+_t\dot{q}^+_t-\sum_{t=0}^{N-1}  p^-_t\dot{q}^-_t\right)\big|_{q^-_{-1}=q_0^+}^{q^+_N=q^-_{N-1}}}\,,
\end{align}
which is the difference between two Legendre transforms and where $\dot{q}_t^+\equiv (q^+_{t+1}-q_t)/\epsilon$, $\dot{q}_t^-\equiv (q^-_t-q^-_{t-1})/\epsilon$. A convenient way to understand the border terms is to notice that 
\small
\begin{equation}
    e^{-i\epsilon \mathcal{P}_{\text{ext}}}|\textbf{q}^+\textbf{q}^-\rangle=|q_1^+q_2^+\dots q_{N-1}^+ q^-_{N-1}, q_0^+ q_0^-q_1^-\dots q_{N-2}^-\rangle\,.
\end{equation}
\normalsize
This equation is obtained by  first translating $\textbf{q}^+$ to the left one step and simultaneously $\textbf{q}^-$ to the right, and then swapping $q_{N-1}^+ \leftrightarrow q_{0}^-$.  The border conditions can then be read directly from the resulting string, with  $q_{0}^+$ always to the left of $q_0^-$ and $q^{-}_{N-1}$ to the right of $q_{N-1}^+$, in agreement with the Schwinger–Keldysh loop.

If we now impose initial and final conditions we need to consider ${\textbf{q}'^+=(q',q^+_1,\dots, q^+_{N-1})}$ and $\textbf{q}^+=(q,q^+_1,\dots, q^+_{N-1})$ as it follows from the term $|q\rangle_0\langle q'|$ and as we explain in the main text. We thus have
\begin{align}
&\langle \textbf{q}'^+\textbf{q}^-|e^{i\epsilon \mathcal{P}_{\text{ext}}}|\textbf{p}^+\textbf{p}^-\rangle \langle \textbf{p}^+\textbf{p}^-|\textbf{q}^+\textbf{q}^-\rangle\nonumber\\
&=\frac{1}{(2\pi)^{N}}e^{i\epsilon\left(\sum_{t=0}^{N-1}  p^+_t\dot{q}^+_t-\sum_{t=0}^{N-1}  p^-_t\dot{q}^-_t\right)\big|_{q^+_0=q, q^-_{-1}=q'}^{q^+_N=q^-_{N-1}=q_f}}\,,\label{eq:dsffmatrixellegendre}
\end{align}
where we are writing $q_f=q_{N-1}^-$ to compare directly with the PI. For $q=q'=q^+_0$ we recover the previous expression.

Similarly, under Trotter approximation
\begin{align}
&\langle \textbf{p}^+ \textbf{p}^-|e^{i\epsilon(\mathcal{K}\otimes \mathbbm{1}-\mathbbm{1}\otimes \mathcal{K})} |\textbf{q}^+\textbf{q}^-\rangle\nonumber\\&=\langle \textbf{p}^+ \textbf{p}^-|\textbf{q}^+\textbf{q}^-\rangle\, e^{-i\epsilon \sum_t (H_t[\textbf{q}^+]-H_t[\textbf{q}^-])\big|_{q^+_0=q}^{q^+_N=q^-_{N-1}=q_f}}\,.\label{eq:matrixelsext}
\end{align}
We thus obtain a PI over momenta with the classical action corresponding to the difference between the action in the $+$ and $-$ variables (including or not the factor $U^\dag_{N-1}U_N$, which we added in Definition~\ref{def:extendedqa}, is an order $\epsilon$ correction). After integrating over the momenta we obtain
\begin{equation}
    \langle \textbf{q}'^+\textbf{q}^-|e^{i\mathcal{S}_{\text{ext}}} |\textbf{q}^+\textbf{q}^-\rangle=\frac{1}{(\sqrt{2\pi i \epsilon/m})^{2N}}e^{i S_{\text{cl}}[\textbf{q}^+]-i S_{\text{cl}}[\textbf{q}^+]}\,,
\end{equation}
with $S_{\text{cl}}[\textbf{q}^{\pm}]$ the discrete time classical action of the corresponding variables and  border conditions $q^+_0=q, q^+_N=q_f$ and $q^-_{-1}=q', q^-_{N-1}=q_f$.

\section{Pseudo Density matrices: additional details}\label{app:PDM}

Here we discuss additional details on PDMs mentioned in Section \ref{sec:PDM}.\\

\emph{Illustration of Theorem \ref{th:PDMextendedQA}}.
Let us first illustrate a basic application of Theorem \ref{th:PDMextendedQA} corresponding to two time slices. The first thing to notice is that for two times the only Kraus operators to be considered are either the identity or a SWAP between slices $0\leftrightarrow 2$, where for convenience we are labeling $\mathcal{H}^{\otimes 2}\equiv h_0 \otimes h_1 \otimes h_2 \otimes h_3$. Notice that $h_2$, $h_3$ correspond to the time slices $0$ and $1$ respectively but on the second copy of $\mathcal{H}$. Then we can write
\begin{equation}
    \Phi(\mathcal{R}_{\text{ext}})=\frac{1}{2}(\mathcal{R}_{\text{ext}}+\text{SWAP}_{02}\mathcal{R}_{\text{ext}}\text{SWAP}_{02})\,.
\end{equation}
The theorem requires one to consider the partial trace over $h_2\otimes h_3$.

\begin{figure}[t!]
    \centering
\includegraphics[width=\linewidth]{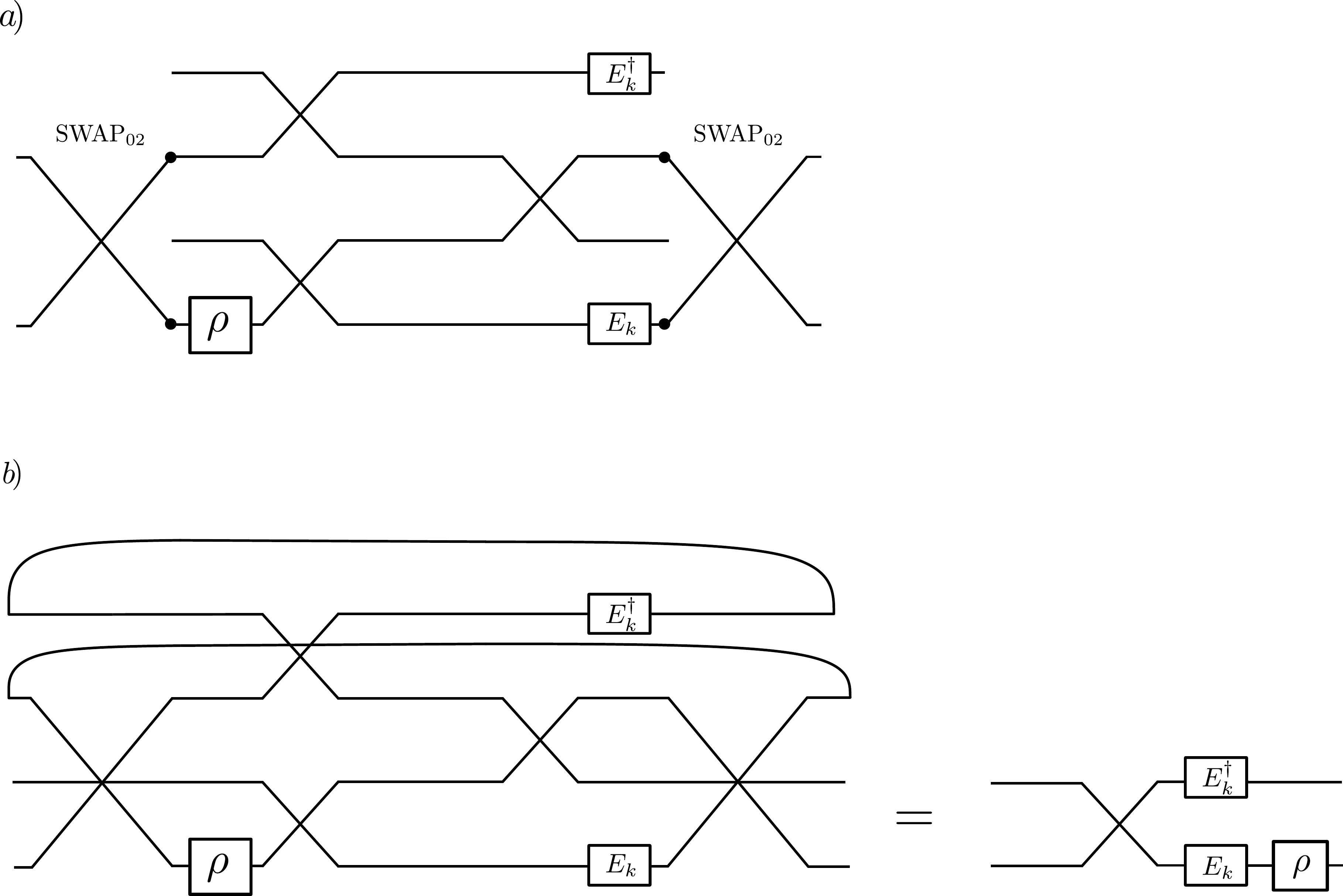}
    \caption{On panel a) we represent $\mathcal{R}_{\text{ext}}$ under the adjoint action of $\text{SWAP}_{02}$. On panel  b) we consider the partial trace over $h_2\otimes h_3$ leading to $J(\mathcal{E})\rho_0$.}
    \label{fig:PDMapp}
\end{figure}

Since the partial trace is linear, and we already know from Section \ref{sec:fbevolution} that ${\rm Tr}_{h_2\otimes h_3}[\mathcal{R}_{\text{ext}}]=\rho_0 J(\mathcal{E})$ we just need to compute the partial trace of the second term. 
The latter is easily obtained via tensor network notation and by using Eq.\ \eqref{eq:extquantumchannel}. As we depict in Figure \ref{fig:PDMapp} we find 
\begin{equation}
    {\rm Tr}_{h_2\otimes h_3}[\text{SWAP}_{02}\mathcal{R}_{\text{ext}}\text{SWAP}_{02}]= J(\mathcal{E})\rho_0\,.
\end{equation}
In summary, Theorem \ref{th:PDMextendedQA} yields
\begin{equation}\label{apeq:pdm2example}
    \mathcal{R}_{\text{pdm}}=\frac{\rho\otimes \mathbbm{1} J(\mathcal{E})+J(\mathcal{E})\rho\otimes \mathbbm{1}}{2}=\frac{\mathcal{R}+\mathcal{R}^\dag}{2}\,.
\end{equation}

One can apply the theorem for any number of slices by similar means. For unitary evolution this is straightforward as one can use the explicit form of spacetime states. For open quantum systems one needs to find the corresponding $\mathcal{R}_{\text{ext}}$ first, just as we did for $N=2$. \\

\emph{The bosonic gaussian case}.
In \cite{zhang2020different} the authors defined PDMs for gaussian states. Let us first recall that Gaussian states are a special case in continuous variables  where the first two statistical moments of the quantum
states, the quadrature mean values and the covariance matrix, fully determine the state. Following the notation in \cite{zhang2020different} the mean values of $L$ modes $\{q_k, p_k\}$ arranged as 
$x_k=(q_0, p_0,..., q_L, p_L)^t$ 
define $d_k={\rm tr}[\rho x_k]$ and the covariance matrix 
\begin{equation}
    \sigma_{ij}={\rm tr}[\rho (x_j x_k+ x_k x_j)]-2 {\rm tr}[\rho x_i]{\rm tr}[\rho x_j]\,.
\end{equation}
In general, any gaussian operator is  determined by the first two moments. The proposal in \cite{zhang2020different} consists in defining a gaussian PDM such that the correlation matrix is defined as
\begin{equation}
\sigma^{\text{spacetime}}_{ij}=2\langle\{x_i,x_j\}\rangle-2 {\rm tr}[\rho x_i]{\rm tr}[\rho x_j]\,.
\end{equation}
where the correlator $\langle\{x_i,x_j\}\rangle$ is defined to be the expectation value for the
product of measurement results on these quadratures. Instead, the mean values are given by $d_j$ at different times.

A basic example consists of measurements at two times, with $U=\mathbbm{1}$ and an initial vacuum state $|0\rangle$. Then one can easily obtain  $d_j=0$ and $$\sigma^{\text{spacetime}}=\begin{pmatrix}
    1&0&1&0\\
    0&1&0&1\\
     1&0&1&0\\
    0&1&0&1
\end{pmatrix}\,,$$
notice that since we have a single mode and two times so that $\sigma^{\text{spacetime}}\equiv \sigma^{\text{time}}$ (only timelike correlations are considered in the example). 
On the other hand, if we define $\mathcal{R}'=\frac{\mathcal{R}+\mathcal{R}^\dag}{2}$, for $\mathcal{R}=|0\rangle \langle 0|\otimes \mathbbm{1} e^{i\epsilon \sum_t p_t \dot{q}_t}$, and use the basic theorems of Section \ref{sec:formalism} we obtain ${{\rm Tr}[\mathcal{R}' (x_k x_{k'}+x_{k'} x_{k})]=\frac{1}{2}\langle 0|(x_k x_{k'}+x_{k'}x_k)|0\rangle}$. Then, using $\langle 0|q^2|0\rangle=\langle 0|p^2|0\rangle=1/2$ and $\langle 0|(qp+pq)|0\rangle=0$ we find the following result:
\begin{equation}
 {\rm Tr}[\mathcal{R}' (x_k x_{k'}+x_{k'}x_k)]- 2{\rm Tr}[\mathcal{R}' x_k ]{\rm Tr}[\mathcal{R}'  x_{k'}]=   \sigma^{\text{spacetime}}_{kk'}\,,
\end{equation}
namely the \emph{standard covariance matrix} of $\mathcal{R}'$ is equal to the spacetime covariance matrix introduced in \cite{zhang2020different}. Notice, however, that while $\mathcal{R}$ is gaussian, $\mathcal{R}'$ needs not to be  (one can consider the mean value of e.g., $q_0^2\otimes p_1q_1$ to show that Wick's theorem doesn't hold) and so the PDM defined in \cite{zhang2020different} is the gaussian operator sharing the first two momenta with $\mathcal{R}'$ but not  $\mathcal{R}'$ itself.

Let us also recall that this form for $\mathcal{R}'$ arises from considering light touch observables. Since for more general observables this form does not hold, there is no reason to regard $\mathcal{R}'$ as the ``true'' PDM. Interestingly, the first two momenta of a properly defined PDM are matched nonetheless. Let us finally remark that the strategy of using a map $\Phi$ that we presented for finite dimensional systems could in principle be extended to bosons leading to a potential definition of PDMs beyond the gaussian case. \\

\emph{Extending the PDM to other observables}. 
Let us now show how one can develop an extension of the theorems in Section \ref{sec:PDM} to observables other than light-touch observables.

Consider a complete Hermitian operator basis $\{K_{i}\}$ 
satisfying
$
{\rm tr}[K_{i} K_{j}]=d\,\delta_{ij}
$ and
\begin{equation}
    K_{i}=\alpha (P_{i})_\alpha+\beta (P_{i})_\beta\,,
\end{equation}
where we denote the spectral projectors
$(P_{i})_\alpha$ (rank $1$) and $(P_{i})_\beta$ (rank $d-1$). 
As an example of these operators  one can consider operators defining a discrete-phase space \cite{wootters1987wigner} for $d=2$ in which case $\alpha=\frac{1+\sqrt{3}}{2}$, $\beta=\frac{1-\sqrt{3}}{2}$ so that the $K_i$ are not light-touch observables. 
Another example is given by a symmetric informationally complete POVM (SIC-POVM) \cite{renes2004symmetric}: 
recall that a SIC is a set of $d^2$ rank-one projectors $\{\Pi_i\}_{i=1}^{d^2}$ on $\mathcal H_d$ such that
$
{\rm Tr}(\Pi_i\Pi_j)=\frac{d\,\delta_{ij}+1}{d+1}.
$
From any SIC one can define the associated  operators
\begin{equation}\label{eq:KfromSIC}
K_i=\sqrt{d+1}\,\Pi_i+\frac{1-\sqrt{d+1}}{d}\,\mathbbm 1,
\end{equation}
which satisfy the previous conditions.
For example, for $d=3$ one has $\sqrt{d+1}=2$ and therefore
$
K_i = 2\Pi_i-\frac{1}{3}\,\mathbbm 1,
$
so each $K_i$ has spectrum $\{5/3,-1/3,-1/3\}$.

Returning to the general case, one can show that
\small
\begin{equation}
\begin{split}
        \sum_{\kappa=\alpha,\beta} \lambda_\kappa (P_{i})_\kappa\otimes (P_{i})_\kappa&=-\frac{2\alpha\beta}{(\alpha-\beta)^2}\,(K_{i}\otimes \mathbbm{1}+\mathbbm{1}\otimes K_{i})\\
&+\frac{\alpha+\beta}{(\alpha-\beta)^2}\,K_{i}\otimes K_{i}
\\&+\frac{\alpha\beta(\alpha+\beta)}{(\alpha-\beta)^2}\,\mathbbm{1}\otimes \mathbbm{1}\,,
\end{split}
\end{equation}
\normalsize
with $\lambda_\alpha \equiv \alpha$, $\lambda_\beta \equiv \beta$.
Notice that for $\beta=-\alpha$ we recover the linear relation of light-touch observables. The first term can then be recovered by a linear map acting on $K_i$ using SWAPs as before. To recover the other two terms from a linear map notice that
\begin{equation}
\Phi_1(X)=\frac{1}{d^{2}}\sum_{i}
\operatorname{Tr}\!\big[(K_{i}\otimes \mathbbm{1})\,X\big]\,K_{i}\otimes K_{i}
\end{equation}
maps $K_{i}\otimes \mathbbm{1}\to K_{i}\otimes K_{i}$ for all $i$ and its linear. Similarly, 
\begin{equation}
\Phi_2(X)=\frac{1}{d}{\rm Tr}[X]\, \mathbbm{1}\otimes \mathbbm{1}
\end{equation}
maps $K_{i}\otimes \mathbbm{1}\to \mathbbm{1}\otimes \mathbbm{1}$ up to a proportionality constant. Since a linear combination of linear maps is another linear map it is clear that we can construct  $\Phi$ such that
\begin{equation}
      \sum_i \lambda_i (P_{i})_i\otimes (P_{i})_i=\Phi(K_{i}\otimes \mathbbm{1})
\end{equation}
for all $i$ as required.

The idea of this example can be generalized since one can always find a polynomial expansion of $ \sum_i \lambda_i P_i\otimes P_i$ and then exploit orthonormality of the basis to represent this sum as a single linear map acting on the corresponding operator. In this sense, our scheme provides an explicit route to generalize PDMs to arbitrary observables (see also comments in Section \ref{sec:PDM}).

\section{Page and Wootters mechanism: proof of Theorem \ref{th:RDMPW}}\label{app:reducedPW}

Here we prove Theorem \ref{th:RDMPW} stating that the RDM 
of $\mathcal{R}_{\text{ext}}$ defined by 
$$
\rho^{\text{sp}}_{\pm}=\sum_{t_1,i,t_2,j}{\rm Tr}[\mathcal{R}_{\text{ext}}\,a^{+\dag}_{t_1,i}a^-_{t_2,j}]|t_2,j\rangle \langle t_1,i|\,,
$$
is the PW state when the initial state represents a single particle and the action is free. \\

\begin{proof}
We recall that the matrix elements of $\rho^{\text{sp}}_{+-}$ are given by  $ \langle \psi|a_{k_1}^\dag (t_1) a_{k_2}(t_2)|\psi\rangle$, as it follows from Theorem \ref{th:extendedQA}. Here ladder operators are evolved in the Heisenberg picture $a_k(t)=e^{iHt}a_k e^{-iHt}$ and $|\psi\rangle$ is a general state in $h_F$. Under the hypothesis of the theorem we consider a state of the form $|\psi\rangle=\sum_k \psi_k a^\dag_k|0\rangle$ and a quadratic (particle preserving) Hamiltonian $H=\sum_{k,k'}M_{kk'}a^\dag_k a_{k'}$. Then, ladder operator evolve as follows,
\begin{equation}
    \begin{split}
        a_k(t)&=\sum_{k'}\,\big(\,e^{-itM}\,\big)_{kk'}\,a_{k'} \\a^\dag_k(t)&=\sum_{k'}\,\big(\,e^{itM}\,\big)_{kk'}\,a^\dag_{k'}\,.
    \end{split}
\end{equation}
We can use these relations to write
\small
\begin{equation}
  \langle \psi|a_{k_1}^\dag (t_1) a_{k_2}(t_2)|\psi\rangle=\sum_{l_1,l_2}  \langle 0|\wick{ \c1 a_{l_2} \c1 a^\dag_{k_1}(t_1) \c2 a_{k_2}(t_2) \c2 a^\dag_{l_1}}|0\rangle \psi_{l_1}\psi^\ast_{l_2}
\end{equation}
\normalsize
with the contractions defined by
\begin{equation}
    \wick{\c2 a_{k}(t) \c2 a^\dag_{k'}}=\langle 0|a_{k}(t)  a^\dag_{k'}|0\rangle=[e^{-it M}]_{kk'}\,,
\end{equation}
with other contractions vanishing. 
 This yields
\begin{equation}
\begin{split}
    \langle \psi|a_{k_1}^\dag (t_1) a_{k_2}(t_2)|\psi\rangle&=\sum_{l_1,l_2}  \psi_{l_1}\psi^\ast_{l_2}(e^{-it_2 M})_{k_2l_1}(e^{it_1 M})_{l_2k_1}\\
    &=  \psi_{k_2}(t_2)\psi^\ast_{k_1}(t_1)\,,
\end{split}
\end{equation}
where we have defined the evolved wavefunction $\psi_{k}(t)=\sum_l [e^{-it M}]_{kl} \psi_{l}$. Using this relation in the definition of the RDM gives
\begin{equation}
\rho^{\text{sp}}_{+-}=\sum_{t_1,k_1,t_2,k_2}\psi_{k_2}(t_2)\psi^\ast_{k_1}(t_1)|t_2,k_2\rangle \langle t_1,k_1|\,.
\end{equation}
This is precisely the projector $\rho^{\text{sp}}_{+-}=|\Psi_{\text{PW}}\rangle\langle\Psi_{\text{PW}}|$ with
\begin{equation}
|\Psi_{\text{PW}}\rangle=\sum_{t,k}\psi_k(t)|t,k\rangle=\sum_t |t\rangle |\psi(t)\rangle\,,
\end{equation}
the PW state. In fact, the associated universe equation corresponds to the Hamiltonian $M$ acting on a single particle, and the state $|\psi(t)\rangle$ is the associated state in the Schr\"{o}dinger picture: 
\begin{equation}
\begin{split}
     |\psi(t)\rangle&=e^{-iMt}\sum_l \psi_l |l\rangle=\sum_k \sum_l \,\big(e^{-it M}\big)_{kl}\, \psi_{l}|k\rangle\\&=\sum_k \psi_k(t)|k\rangle\,.
\end{split}
\end{equation}
\end{proof}

\newpage
\clearpage


\begin{thebibliography}{115}%
\makeatletter
\providecommand \@ifxundefined [1]{%
 \@ifx{#1\undefined}
}%
\providecommand \@ifnum [1]{%
 \ifnum #1\expandafter \@firstoftwo
 \else \expandafter \@secondoftwo
 \fi
}%
\providecommand \@ifx [1]{%
 \ifx #1\expandafter \@firstoftwo
 \else \expandafter \@secondoftwo
 \fi
}%
\providecommand \natexlab [1]{#1}%
\providecommand \enquote  [1]{``#1''}%
\providecommand \bibnamefont  [1]{#1}%
\providecommand \bibfnamefont [1]{#1}%
\providecommand \citenamefont [1]{#1}%
\providecommand \href@noop [0]{\@secondoftwo}%
\providecommand \href [0]{\begingroup \@sanitize@url \@href}%
\providecommand \@href[1]{\@@startlink{#1}\@@href}%
\providecommand \@@href[1]{\endgroup#1\@@endlink}%
\providecommand \@sanitize@url [0]{\catcode `\\12\catcode `\$12\catcode `\&12\catcode `\#12\catcode `\^12\catcode `\_12\catcode `\%12\relax}%
\providecommand \@@startlink[1]{}%
\providecommand \@@endlink[0]{}%
\providecommand \url  [0]{\begingroup\@sanitize@url \@url }%
\providecommand \@url [1]{\endgroup\@href {#1}{\urlprefix }}%
\providecommand \urlprefix  [0]{URL }%
\providecommand \Eprint [0]{\href }%
\providecommand \doibase [0]{https://doi.org/}%
\providecommand \selectlanguage [0]{\@gobble}%
\providecommand \bibinfo  [0]{\@secondoftwo}%
\providecommand \bibfield  [0]{\@secondoftwo}%
\providecommand \translation [1]{[#1]}%
\providecommand \BibitemOpen [0]{}%
\providecommand \bibitemStop [0]{}%
\providecommand \bibitemNoStop [0]{.\EOS\space}%
\providecommand \EOS [0]{\spacefactor3000\relax}%
\providecommand \BibitemShut  [1]{\csname bibitem#1\endcsname}%
\let\auto@bib@innerbib\@empty
\bibitem [{\citenamefont {Minkowski}(1909)}]{minkowski1909raum}%
  \BibitemOpen
  \bibfield  {author} {\bibinfo {author} {\bibfnamefont {H.}~\bibnamefont {Minkowski}},\ }\bibfield  {title} {\bibinfo {title} {{Raum und Zeit}},\ }\href@noop {} {\bibfield  {journal} {\bibinfo  {journal} {Physikalische Zeitschrift}\ }\textbf {\bibinfo {volume} {10}},\ \bibinfo {pages} {104} (\bibinfo {year} {1909})}\BibitemShut {NoStop}%
\bibitem [{\citenamefont {Dirac}(1963)}]{dirac1963evolution}%
  \BibitemOpen
  \bibfield  {author} {\bibinfo {author} {\bibfnamefont {P.~A.~M.}\ \bibnamefont {Dirac}},\ }\bibfield  {title} {\bibinfo {title} {The evolution of the physicist's picture of nature},\ }\href {https://doi.org/10.1038/scientificamerican0563-45} {\bibfield  {journal} {\bibinfo  {journal} {Scientific American}\ }\textbf {\bibinfo {volume} {208}},\ \bibinfo {pages} {45} (\bibinfo {year} {1963})}\BibitemShut {NoStop}%
\bibitem [{\citenamefont {Eddington}(1920)}]{eddington1920space}%
  \BibitemOpen
  \bibfield  {author} {\bibinfo {author} {\bibfnamefont {A.~S.}\ \bibnamefont {Eddington}},\ }\href@noop {} {\emph {\bibinfo {title} {Space, Time and Gravitation: An Outline of the General Relativity Theory}}}\ (\bibinfo  {publisher} {Cambridge University Press},\ \bibinfo {address} {Cambridge},\ \bibinfo {year} {1920})\BibitemShut {NoStop}%
\bibitem [{\citenamefont {Dirac}(1958)}]{dirac1958theory}%
  \BibitemOpen
  \bibfield  {author} {\bibinfo {author} {\bibfnamefont {P.~A.~M.}\ \bibnamefont {Dirac}},\ }\bibfield  {title} {\bibinfo {title} {The theory of gravitation in hamiltonian form},\ }\href {https://doi.org/10.1098/rspa.1958.0142} {\bibfield  {journal} {\bibinfo  {journal} {Proceedings of the Royal Society of London. Series A. Mathematical and Physical Sciences}\ }\textbf {\bibinfo {volume} {246}},\ \bibinfo {pages} {333} (\bibinfo {year} {1958})}\BibitemShut {NoStop}%
\bibitem [{\citenamefont {Isham}(1993)}]{isham1993canonical}%
  \BibitemOpen
  \bibfield  {author} {\bibinfo {author} {\bibfnamefont {C.~J.}\ \bibnamefont {Isham}},\ }\bibfield  {title} {\bibinfo {title} {Canonical quantum gravity and the problem of time},\ }\href {https://doi.org/10.1007/978-94-011-1980-1_6} {\bibfield  {journal} {\bibinfo  {journal} {Integrable systems, quantum groups, and quantum field theories}\ ,\ \bibinfo {pages} {157}} (\bibinfo {year} {1993})}\BibitemShut {NoStop}%
\bibitem [{\citenamefont {Anderson}(2012)}]{anderson2012problem}%
  \BibitemOpen
  \bibfield  {author} {\bibinfo {author} {\bibfnamefont {E.}~\bibnamefont {Anderson}},\ }\bibfield  {title} {\bibinfo {title} {Problem of time in quantum gravity},\ }\href {https://doi.org/https://doi.org/10.1002/andp.201200147} {\bibfield  {journal} {\bibinfo  {journal} {Annalen der Physik}\ }\textbf {\bibinfo {volume} {524}},\ \bibinfo {pages} {757} (\bibinfo {year} {2012})}\BibitemShut {NoStop}%
\bibitem [{\citenamefont {Zurek}(2003)}]{zurek2003decoherence}%
  \BibitemOpen
  \bibfield  {author} {\bibinfo {author} {\bibfnamefont {W.~H.}\ \bibnamefont {Zurek}},\ }\bibfield  {title} {\bibinfo {title} {Decoherence, einselection, and the quantum origins of the classical},\ }\href {https://doi.org/10.1103/RevModPhys.75.715} {\bibfield  {journal} {\bibinfo  {journal} {Reviews of modern physics}\ }\textbf {\bibinfo {volume} {75}},\ \bibinfo {pages} {715} (\bibinfo {year} {2003})}\BibitemShut {NoStop}%
\bibitem [{\citenamefont {Abanin}\ \emph {et~al.}(2019)\citenamefont {Abanin}, \citenamefont {Altman}, \citenamefont {Bloch},\ and\ \citenamefont {Serbyn}}]{abanin2019colloquium}%
  \BibitemOpen
  \bibfield  {author} {\bibinfo {author} {\bibfnamefont {D.~A.}\ \bibnamefont {Abanin}}, \bibinfo {author} {\bibfnamefont {E.}~\bibnamefont {Altman}}, \bibinfo {author} {\bibfnamefont {I.}~\bibnamefont {Bloch}},\ and\ \bibinfo {author} {\bibfnamefont {M.}~\bibnamefont {Serbyn}},\ }\bibfield  {title} {\bibinfo {title} {Colloquium: Many-body localization, thermalization, and entanglement},\ }\href {https://journals.aps.org/rmp/abstract/10.1103/RevModPhys.91.021001} {\bibfield  {journal} {\bibinfo  {journal} {Reviews of Modern Physics}\ }\textbf {\bibinfo {volume} {91}},\ \bibinfo {pages} {021001} (\bibinfo {year} {2019})}\BibitemShut {NoStop}%
\bibitem [{\citenamefont {Zurek}(2009)}]{zurek2009quantum}%
  \BibitemOpen
  \bibfield  {author} {\bibinfo {author} {\bibfnamefont {W.~H.}\ \bibnamefont {Zurek}},\ }\bibfield  {title} {\bibinfo {title} {Quantum darwinism},\ }\href {https://doi.org/10.1038/nphys1202} {\bibfield  {journal} {\bibinfo  {journal} {Nature physics}\ }\textbf {\bibinfo {volume} {5}},\ \bibinfo {pages} {181} (\bibinfo {year} {2009})}\BibitemShut {NoStop}%
\bibitem [{\citenamefont {Touil}\ \emph {et~al.}(2024)\citenamefont {Touil}, \citenamefont {Anza}, \citenamefont {Deffner},\ and\ \citenamefont {Crutchfield}}]{touil2024branching}%
  \BibitemOpen
  \bibfield  {author} {\bibinfo {author} {\bibfnamefont {A.}~\bibnamefont {Touil}}, \bibinfo {author} {\bibfnamefont {F.}~\bibnamefont {Anza}}, \bibinfo {author} {\bibfnamefont {S.}~\bibnamefont {Deffner}},\ and\ \bibinfo {author} {\bibfnamefont {J.~P.}\ \bibnamefont {Crutchfield}},\ }\bibfield  {title} {\bibinfo {title} {Branching states as the emergent structure of a quantum universe},\ }\href {https://doi.org/10.22331/q-2024-10-10-1494} {\bibfield  {journal} {\bibinfo  {journal} {Quantum}\ }\textbf {\bibinfo {volume} {8}},\ \bibinfo {pages} {1494} (\bibinfo {year} {2024})}\BibitemShut {NoStop}%
\bibitem [{\citenamefont {Ryu}\ and\ \citenamefont {Takayanagi}(2006)}]{ryu2006holographic}%
  \BibitemOpen
  \bibfield  {author} {\bibinfo {author} {\bibfnamefont {S.}~\bibnamefont {Ryu}}\ and\ \bibinfo {author} {\bibfnamefont {T.}~\bibnamefont {Takayanagi}},\ }\bibfield  {title} {\bibinfo {title} {Holographic derivation of entanglement entropy from the anti–de sitter space/conformal field theory correspondence},\ }\href {https://doi.org/10.1103/PhysRevLett.96.181602} {\bibfield  {journal} {\bibinfo  {journal} {Physical review letters}\ }\textbf {\bibinfo {volume} {96}},\ \bibinfo {pages} {181602} (\bibinfo {year} {2006})}\BibitemShut {NoStop}%
\bibitem [{\citenamefont {Van~Raamsdonk}(2010)}]{van2010building}%
  \BibitemOpen
  \bibfield  {author} {\bibinfo {author} {\bibfnamefont {M.}~\bibnamefont {Van~Raamsdonk}},\ }\bibfield  {title} {\bibinfo {title} {Building up space--time with quantum entanglement},\ }\href {https://doi.org/10.1007/s10714-010-1034-0} {\bibfield  {journal} {\bibinfo  {journal} {International Journal of Modern Physics D}\ }\textbf {\bibinfo {volume} {19}},\ \bibinfo {pages} {2429} (\bibinfo {year} {2010})}\BibitemShut {NoStop}%
\bibitem [{\citenamefont {Cao}\ \emph {et~al.}(2017)\citenamefont {Cao}, \citenamefont {Carroll},\ and\ \citenamefont {Michalakis}}]{cao2017space}%
  \BibitemOpen
  \bibfield  {author} {\bibinfo {author} {\bibfnamefont {C.}~\bibnamefont {Cao}}, \bibinfo {author} {\bibfnamefont {S.~M.}\ \bibnamefont {Carroll}},\ and\ \bibinfo {author} {\bibfnamefont {S.}~\bibnamefont {Michalakis}},\ }\bibfield  {title} {\bibinfo {title} {Space from hilbert space: Recovering geometry from bulk entanglement},\ }\href {https://doi.org/10.1103/PhysRevD.95.024031} {\bibfield  {journal} {\bibinfo  {journal} {Physical Review D}\ }\textbf {\bibinfo {volume} {95}},\ \bibinfo {pages} {024031} (\bibinfo {year} {2017})}\BibitemShut {NoStop}%
\bibitem [{\citenamefont {Harper}\ \emph {et~al.}(2023)\citenamefont {Harper}, \citenamefont {Mollabashi}, \citenamefont {Takayanagi}, \citenamefont {Taki} \emph {et~al.}}]{harper2023timelike}%
  \BibitemOpen
  \bibfield  {author} {\bibinfo {author} {\bibfnamefont {J.}~\bibnamefont {Harper}}, \bibinfo {author} {\bibfnamefont {A.}~\bibnamefont {Mollabashi}}, \bibinfo {author} {\bibfnamefont {T.}~\bibnamefont {Takayanagi}}, \bibinfo {author} {\bibfnamefont {Y.}~\bibnamefont {Taki}}, \emph {et~al.},\ }\bibfield  {title} {\bibinfo {title} {Timelike entanglement entropy},\ }\href {https://doi.org/10.1007/JHEP05(2023)052} {\bibfield  {journal} {\bibinfo  {journal} {Journal of High Energy Physics}\ }\textbf {\bibinfo {volume} {2023}},\ \bibinfo {pages} {1} (\bibinfo {year} {2023})}\BibitemShut {NoStop}%
\bibitem [{\citenamefont {Heller}\ \emph {et~al.}(2025)\citenamefont {Heller}, \citenamefont {Ori},\ and\ \citenamefont {Serantes}}]{heller2025temporal}%
  \BibitemOpen
  \bibfield  {author} {\bibinfo {author} {\bibfnamefont {M.~P.}\ \bibnamefont {Heller}}, \bibinfo {author} {\bibfnamefont {F.}~\bibnamefont {Ori}},\ and\ \bibinfo {author} {\bibfnamefont {A.}~\bibnamefont {Serantes}},\ }\bibfield  {title} {\bibinfo {title} {Temporal entanglement from holographic entanglement entropy},\ }\href {https://doi.org/10.1103/qlsv-gp22} {\bibfield  {journal} {\bibinfo  {journal} {Physical Review X}\ }\textbf {\bibinfo {volume} {15}},\ \bibinfo {pages} {041022} (\bibinfo {year} {2025})}\BibitemShut {NoStop}%
\bibitem [{\citenamefont {Milekhin}\ \emph {et~al.}(2025)\citenamefont {Milekhin}, \citenamefont {Adamska},\ and\ \citenamefont {Preskill}}]{milekhin2025observable}%
  \BibitemOpen
  \bibfield  {author} {\bibinfo {author} {\bibfnamefont {A.}~\bibnamefont {Milekhin}}, \bibinfo {author} {\bibfnamefont {Z.}~\bibnamefont {Adamska}},\ and\ \bibinfo {author} {\bibfnamefont {J.}~\bibnamefont {Preskill}},\ }\bibfield  {title} {\bibinfo {title} {Observable and computable entanglement in time},\ }\href {https://arxiv.org/abs/2502.12240} {\bibfield  {journal} {\bibinfo  {journal} {arXiv preprint arXiv:2502.12240}\ } (\bibinfo {year} {2025})}\BibitemShut {NoStop}%
\bibitem [{\citenamefont {Horsman}\ \emph {et~al.}(2017)\citenamefont {Horsman}, \citenamefont {Heunen}, \citenamefont {Pusey}, \citenamefont {Barrett},\ and\ \citenamefont {Spekkens}}]{horsman2017can}%
  \BibitemOpen
  \bibfield  {author} {\bibinfo {author} {\bibfnamefont {D.}~\bibnamefont {Horsman}}, \bibinfo {author} {\bibfnamefont {C.}~\bibnamefont {Heunen}}, \bibinfo {author} {\bibfnamefont {M.~F.}\ \bibnamefont {Pusey}}, \bibinfo {author} {\bibfnamefont {J.}~\bibnamefont {Barrett}},\ and\ \bibinfo {author} {\bibfnamefont {R.~W.}\ \bibnamefont {Spekkens}},\ }\bibfield  {title} {\bibinfo {title} {Can a quantum state over time resemble a quantum state at a single time?},\ }\href {https://doi.org/10.1098/rspa.2017.0395} {\bibfield  {journal} {\bibinfo  {journal} {Proceedings of the Royal Society A: Mathematical, Physical and Engineering Sciences}\ }\textbf {\bibinfo {volume} {473}},\ \bibinfo {pages} {20170395} (\bibinfo {year} {2017})}\BibitemShut {NoStop}%
\bibitem [{\citenamefont {Fitzsimons}\ \emph {et~al.}(2015)\citenamefont {Fitzsimons}, \citenamefont {Jones},\ and\ \citenamefont {Vedral}}]{fitzsimons2015quantum}%
  \BibitemOpen
  \bibfield  {author} {\bibinfo {author} {\bibfnamefont {J.~F.}\ \bibnamefont {Fitzsimons}}, \bibinfo {author} {\bibfnamefont {J.~A.}\ \bibnamefont {Jones}},\ and\ \bibinfo {author} {\bibfnamefont {V.}~\bibnamefont {Vedral}},\ }\bibfield  {title} {\bibinfo {title} {Quantum correlations which imply causation},\ }\href {https://doi.org/10.1038/srep18281} {\bibfield  {journal} {\bibinfo  {journal} {Scientific reports}\ }\textbf {\bibinfo {volume} {5}},\ \bibinfo {pages} {18281} (\bibinfo {year} {2015})}\BibitemShut {NoStop}%
\bibitem [{\citenamefont {Fullwood}\ and\ \citenamefont {Parzygnat}(2024)}]{fullwood2024operator}%
  \BibitemOpen
  \bibfield  {author} {\bibinfo {author} {\bibfnamefont {J.}~\bibnamefont {Fullwood}}\ and\ \bibinfo {author} {\bibfnamefont {A.~J.}\ \bibnamefont {Parzygnat}},\ }\bibfield  {title} {\bibinfo {title} {Operator representation of spatiotemporal quantum correlations},\ }\href {https://arxiv.org/abs/2405.17555} {\bibfield  {journal} {\bibinfo  {journal} {arXiv preprint arXiv:2405.17555}\ } (\bibinfo {year} {2024})}\BibitemShut {NoStop}%
\bibitem [{\citenamefont {Cotler}\ \emph {et~al.}(2018)\citenamefont {Cotler}, \citenamefont {Jian}, \citenamefont {Qi},\ and\ \citenamefont {Wilczek}}]{cotler2018superdensity}%
  \BibitemOpen
  \bibfield  {author} {\bibinfo {author} {\bibfnamefont {J.}~\bibnamefont {Cotler}}, \bibinfo {author} {\bibfnamefont {C.-M.}\ \bibnamefont {Jian}}, \bibinfo {author} {\bibfnamefont {X.-L.}\ \bibnamefont {Qi}},\ and\ \bibinfo {author} {\bibfnamefont {F.}~\bibnamefont {Wilczek}},\ }\bibfield  {title} {\bibinfo {title} {Superdensity operators for spacetime quantum mechanics},\ }\href {https://doi.org/10.1007/JHEP09(2018)093} {\bibfield  {journal} {\bibinfo  {journal} {Journal of High Energy Physics}\ }\textbf {\bibinfo {volume} {2018}},\ \bibinfo {pages} {1} (\bibinfo {year} {2018})}\BibitemShut {NoStop}%
\bibitem [{\citenamefont {Page}\ and\ \citenamefont {Wootters}(1983)}]{page1983evolution}%
  \BibitemOpen
  \bibfield  {author} {\bibinfo {author} {\bibfnamefont {D.~N.}\ \bibnamefont {Page}}\ and\ \bibinfo {author} {\bibfnamefont {W.~K.}\ \bibnamefont {Wootters}},\ }\bibfield  {title} {\bibinfo {title} {Evolution without evolution: Dynamics described by stationary observables},\ }\href {https://doi.org/10.1103/PhysRevD.27.2885} {\bibfield  {journal} {\bibinfo  {journal} {Physical Review D}\ }\textbf {\bibinfo {volume} {27}},\ \bibinfo {pages} {2885} (\bibinfo {year} {1983})}\BibitemShut {NoStop}%
\bibitem [{\citenamefont {Giovannetti}\ \emph {et~al.}(2015)\citenamefont {Giovannetti}, \citenamefont {Lloyd},\ and\ \citenamefont {Maccone}}]{giovannetti2015quantum}%
  \BibitemOpen
  \bibfield  {author} {\bibinfo {author} {\bibfnamefont {V.}~\bibnamefont {Giovannetti}}, \bibinfo {author} {\bibfnamefont {S.}~\bibnamefont {Lloyd}},\ and\ \bibinfo {author} {\bibfnamefont {L.}~\bibnamefont {Maccone}},\ }\bibfield  {title} {\bibinfo {title} {Quantum time},\ }\href {https://doi.org/10.1103/PhysRevD.92.045033} {\bibfield  {journal} {\bibinfo  {journal} {Physical Review D}\ }\textbf {\bibinfo {volume} {92}},\ \bibinfo {pages} {045033} (\bibinfo {year} {2015})}\BibitemShut {NoStop}%
\bibitem [{\citenamefont {Boette}\ \emph {et~al.}(2016)\citenamefont {Boette}, \citenamefont {Rossignoli}, \citenamefont {Gigena},\ and\ \citenamefont {Cerezo}}]{boette2016system}%
  \BibitemOpen
  \bibfield  {author} {\bibinfo {author} {\bibfnamefont {A.}~\bibnamefont {Boette}}, \bibinfo {author} {\bibfnamefont {R.}~\bibnamefont {Rossignoli}}, \bibinfo {author} {\bibfnamefont {N.}~\bibnamefont {Gigena}},\ and\ \bibinfo {author} {\bibfnamefont {M.}~\bibnamefont {Cerezo}},\ }\bibfield  {title} {\bibinfo {title} {System-time entanglement in a discrete-time model},\ }\href {https://doi.org/10.1103/PhysRevA.93.062127} {\bibfield  {journal} {\bibinfo  {journal} {Physical Review A}\ }\textbf {\bibinfo {volume} {93}},\ \bibinfo {pages} {062127} (\bibinfo {year} {2016})}\BibitemShut {NoStop}%
\bibitem [{\citenamefont {Giovannetti}\ \emph {et~al.}(2023)\citenamefont {Giovannetti}, \citenamefont {Lloyd},\ and\ \citenamefont {Maccone}}]{giovannetti2023geometric}%
  \BibitemOpen
  \bibfield  {author} {\bibinfo {author} {\bibfnamefont {V.}~\bibnamefont {Giovannetti}}, \bibinfo {author} {\bibfnamefont {S.}~\bibnamefont {Lloyd}},\ and\ \bibinfo {author} {\bibfnamefont {L.}~\bibnamefont {Maccone}},\ }\bibfield  {title} {\bibinfo {title} {Geometric event-based quantum mechanics},\ }\href {https://doi.org/10.1088/1367-2630/acb793} {\bibfield  {journal} {\bibinfo  {journal} {New Journal of Physics}\ }\textbf {\bibinfo {volume} {25}},\ \bibinfo {pages} {023027} (\bibinfo {year} {2023})}\BibitemShut {NoStop}%
\bibitem [{\citenamefont {Guo}(2026)}]{guo2025spacetime}%
  \BibitemOpen
  \bibfield  {author} {\bibinfo {author} {\bibfnamefont {W.-z.}\ \bibnamefont {Guo}},\ }\bibfield  {title} {\bibinfo {title} {Spacetime density matrix: Formalism and properties},\ }\href {https://doi.org/https://doi.org/10.1007/JHEP01(2026)128} {\bibfield  {journal} {\bibinfo  {journal} {J. High Energ. Phys.}\ }\textbf {\bibinfo {volume} {2026}}\bibinfo  {number} { (128)}}\BibitemShut {NoStop}%
\bibitem [{\citenamefont {Giacomini}\ \emph {et~al.}(2019{\natexlab{a}})\citenamefont {Giacomini}, \citenamefont {Castro-Ruiz},\ and\ \citenamefont {Brukner}}]{giacomini2019quantum}%
  \BibitemOpen
\bibfield  {number} {  }\bibfield  {author} {\bibinfo {author} {\bibfnamefont {F.}~\bibnamefont {Giacomini}}, \bibinfo {author} {\bibfnamefont {E.}~\bibnamefont {Castro-Ruiz}},\ and\ \bibinfo {author} {\bibfnamefont {{\v{C}}.}~\bibnamefont {Brukner}},\ }\bibfield  {title} {\bibinfo {title} {Quantum mechanics and the covariance of physical laws in quantum reference frames},\ }\href {https://doi.org/https://doi.org/10.1038/s41467-018-08155-0} {\bibfield  {journal} {\bibinfo  {journal} {Nature communications}\ }\textbf {\bibinfo {volume} {10}},\ \bibinfo {pages} {494} (\bibinfo {year} {2019}{\natexlab{a}})}\BibitemShut {NoStop}%
\bibitem [{\citenamefont {Castro-Ruiz}\ \emph {et~al.}(2018)\citenamefont {Castro-Ruiz}, \citenamefont {Giacomini},\ and\ \citenamefont {Brukner}}]{castro2018dynamics}%
  \BibitemOpen
  \bibfield  {author} {\bibinfo {author} {\bibfnamefont {E.}~\bibnamefont {Castro-Ruiz}}, \bibinfo {author} {\bibfnamefont {F.}~\bibnamefont {Giacomini}},\ and\ \bibinfo {author} {\bibfnamefont {{\v{C}}.}~\bibnamefont {Brukner}},\ }\bibfield  {title} {\bibinfo {title} {Dynamics of quantum causal structures},\ }\href {https://doi.org/10.1103/PhysRevX.8.011047} {\bibfield  {journal} {\bibinfo  {journal} {Physical Review X}\ }\textbf {\bibinfo {volume} {8}},\ \bibinfo {pages} {011047} (\bibinfo {year} {2018})}\BibitemShut {NoStop}%
\bibitem [{\citenamefont {Diaz}\ \emph {et~al.}(2024)\citenamefont {Diaz}, \citenamefont {Matera},\ and\ \citenamefont {Rossignoli}}]{diaz2023spacetime}%
  \BibitemOpen
  \bibfield  {author} {\bibinfo {author} {\bibfnamefont {N.~L.}\ \bibnamefont {Diaz}}, \bibinfo {author} {\bibfnamefont {J.~M.}\ \bibnamefont {Matera}},\ and\ \bibinfo {author} {\bibfnamefont {R.}~\bibnamefont {Rossignoli}},\ }\bibfield  {title} {\bibinfo {title} {Spacetime quantum and classical mechanics with dynamical foliation},\ }\href {https://doi.org/10.1103/PhysRevD.109.105008} {\bibfield  {journal} {\bibinfo  {journal} {Phys. Rev. D}\ }\textbf {\bibinfo {volume} {109}},\ \bibinfo {pages} {105008} (\bibinfo {year} {2024})}\BibitemShut {NoStop}%
\bibitem [{\citenamefont {Diaz}\ and\ \citenamefont {Rossignoli}(2026)}]{diaz2025spacetime}%
  \BibitemOpen
  \bibfield  {author} {\bibinfo {author} {\bibfnamefont {N.~L.}\ \bibnamefont {Diaz}}\ and\ \bibinfo {author} {\bibfnamefont {R.}~\bibnamefont {Rossignoli}},\ }\bibfield  {title} {\bibinfo {title} {Spacetime quantum mechanics for bosonic and fermionic systems},\ }\bibfield  {journal} {\bibinfo  {journal} {Physical Review Research}\ }\textbf {\bibinfo {volume} {8}},\ \href {https://doi.org/10.1103/zmky-mmwb} {10.1103/zmky-mmwb} (\bibinfo {year} {2026})\BibitemShut {NoStop}%
\bibitem [{\citenamefont {Diaz}\ \emph {et~al.}(2021)\citenamefont {Diaz}, \citenamefont {Matera},\ and\ \citenamefont {Rossignoli}}]{diaz2021spacetime}%
  \BibitemOpen
  \bibfield  {author} {\bibinfo {author} {\bibfnamefont {N.~L.}\ \bibnamefont {Diaz}}, \bibinfo {author} {\bibfnamefont {J.~M.}\ \bibnamefont {Matera}},\ and\ \bibinfo {author} {\bibfnamefont {R.}~\bibnamefont {Rossignoli}},\ }\bibfield  {title} {\bibinfo {title} {Spacetime quantum actions},\ }\href {https://doi.org/10.1103/PhysRevD.103.065011} {\bibfield  {journal} {\bibinfo  {journal} {Physical Review D}\ }\textbf {\bibinfo {volume} {103}},\ \bibinfo {pages} {065011} (\bibinfo {year} {2021})}\BibitemShut {NoStop}%
\bibitem [{\citenamefont {Diaz}(2026)}]{diaz2026quantum}%
  \BibitemOpen
  \bibfield  {author} {\bibinfo {author} {\bibfnamefont {N.~L.}\ \bibnamefont {Diaz}},\ }\bibfield  {title} {\bibinfo {title} {From quantum time to manifestly covariant QFT: On the need for a quantum-action-based quantization},\ }\href {https://doi.org/10.3390/e28040425} {\bibfield  {journal} {\bibinfo  {journal} {Entropy}\ }\textbf {\bibinfo {volume} {28}},\ \bibinfo {pages} {425} (\bibinfo {year} {2026})}\BibitemShut {NoStop}%
\bibitem [{\citenamefont {Feynman}(1948)}]{feynman1948space}%
  \BibitemOpen
  \bibfield  {author} {\bibinfo {author} {\bibfnamefont {R.~P.}\ \bibnamefont {Feynman}},\ }\bibfield  {title} {\bibinfo {title} {Space-time approach to non-relativistic quantum mechanics},\ }\href {https://doi.org/10.1103/RevModPhys.20.367} {\bibfield  {journal} {\bibinfo  {journal} {Reviews of modern physics}\ }\textbf {\bibinfo {volume} {20}},\ \bibinfo {pages} {367} (\bibinfo {year} {1948})}\BibitemShut {NoStop}%
\bibitem [{\citenamefont {Diaz}\ \emph {et~al.}(2025{\natexlab{a}})\citenamefont {Diaz}, \citenamefont {Matera},\ and\ \citenamefont {Rossignoli}}]{diaz2021path}%
  \BibitemOpen
  \bibfield  {author} {\bibinfo {author} {\bibfnamefont {N.~L.}\ \bibnamefont {Diaz}}, \bibinfo {author} {\bibfnamefont {J.~M.}\ \bibnamefont {Matera}},\ and\ \bibinfo {author} {\bibfnamefont {R.}~\bibnamefont {Rossignoli}},\ }\bibfield  {title} {\bibinfo {title} {Path integrals from spacetime quantum actions},\ }\href {https://doi.org/10.1016/j.aop.2025.170052} {\bibfield  {journal} {\bibinfo  {journal} {Annals of Physics}\ ,\ \bibinfo {pages} {170052}} (\bibinfo {year} {2025}{\natexlab{a}})}\BibitemShut {NoStop}%
\bibitem [{Note1()}]{Note1}%
  \BibitemOpen
  \bibinfo {note} {For a finite number of time labels, $[q_t,p_{t'}]=i\delta _{tt'}$ defines a direct sum of independent canonical algebras. The corresponding Weyl algebra factorizes in time, and Stone-von Neumann gives $\protect \mathcal {H}\simeq \otimes _t h$ up to unitary equivalence.}\BibitemShut {Stop}%
\bibitem [{Note2()}]{Note2}%
  \BibitemOpen
  \bibinfo {note} {The case of fermions is however not equivalent to a tensor-product structure; for this reason we study it in the separate Section \ref {sec:fermions}.}\BibitemShut {Stop}%
\bibitem [{\citenamefont {Kirkwood}(1933)}]{kirkwood1933quantum}%
  \BibitemOpen
  \bibfield  {author} {\bibinfo {author} {\bibfnamefont {J.~G.}\ \bibnamefont {Kirkwood}},\ }\bibfield  {title} {\bibinfo {title} {Quantum statistics of almost classical assemblies},\ }\href {https://doi.org/https://doi.org/10.1103/PhysRev.44.31} {\bibfield  {journal} {\bibinfo  {journal} {Physical Review}\ }\textbf {\bibinfo {volume} {44}},\ \bibinfo {pages} {31} (\bibinfo {year} {1933})}\BibitemShut {NoStop}%
\bibitem [{\citenamefont {Dirac}(1945)}]{dirac1945analogy}%
  \BibitemOpen
  \bibfield  {author} {\bibinfo {author} {\bibfnamefont {P.~A.~M.}\ \bibnamefont {Dirac}},\ }\bibfield  {title} {\bibinfo {title} {On the analogy between classical and quantum mechanics},\ }\href {https://doi.org/https://doi.org/10.1103/RevModPhys.17.195} {\bibfield  {journal} {\bibinfo  {journal} {Reviews of Modern Physics}\ }\textbf {\bibinfo {volume} {17}},\ \bibinfo {pages} {195} (\bibinfo {year} {1945})}\BibitemShut {NoStop}%
\bibitem [{\citenamefont {Harper}\ \emph {et~al.}(2025)\citenamefont {Harper}, \citenamefont {Kawamoto}, \citenamefont {Maeda}, \citenamefont {Nakamura},\ and\ \citenamefont {Takayanagi}}]{harper2025non}%
  \BibitemOpen
  \bibfield  {author} {\bibinfo {author} {\bibfnamefont {J.}~\bibnamefont {Harper}}, \bibinfo {author} {\bibfnamefont {T.}~\bibnamefont {Kawamoto}}, \bibinfo {author} {\bibfnamefont {R.}~\bibnamefont {Maeda}}, \bibinfo {author} {\bibfnamefont {N.}~\bibnamefont {Nakamura}},\ and\ \bibinfo {author} {\bibfnamefont {T.}~\bibnamefont {Takayanagi}},\ }\bibfield  {title} {\bibinfo {title} {Non-hermitian density matrices from time-like entanglement and wormholes},\ }\href {https://arxiv.org/abs/2512.13800} {\bibfield  {journal} {\bibinfo  {journal} {arXiv preprint arXiv:2512.13800}\ } (\bibinfo {year} {2025})}\BibitemShut {NoStop}%
\bibitem [{\citenamefont {Aharonov}\ \emph {et~al.}(1988)\citenamefont {Aharonov}, \citenamefont {Albert},\ and\ \citenamefont {Vaidman}}]{aharonov1988result}%
  \BibitemOpen
  \bibfield  {author} {\bibinfo {author} {\bibfnamefont {Y.}~\bibnamefont {Aharonov}}, \bibinfo {author} {\bibfnamefont {D.~Z.}\ \bibnamefont {Albert}},\ and\ \bibinfo {author} {\bibfnamefont {L.}~\bibnamefont {Vaidman}},\ }\bibfield  {title} {\bibinfo {title} {How the result of a measurement of a component of the spin of a spin-1/2 particle can turn out to be 100},\ }\href {https://doi.org/10.1103/PhysRevLett.60.1351} {\bibfield  {journal} {\bibinfo  {journal} {Physical review letters}\ }\textbf {\bibinfo {volume} {60}},\ \bibinfo {pages} {1351} (\bibinfo {year} {1988})}\BibitemShut {NoStop}%
\bibitem [{\citenamefont {Leggett}\ and\ \citenamefont {Garg}(1985)}]{leggett1985quantum}%
  \BibitemOpen
  \bibfield  {author} {\bibinfo {author} {\bibfnamefont {A.~J.}\ \bibnamefont {Leggett}}\ and\ \bibinfo {author} {\bibfnamefont {A.}~\bibnamefont {Garg}},\ }\bibfield  {title} {\bibinfo {title} {Quantum mechanics versus macroscopic realism: Is the flux there when nobody looks?},\ }\href {https://doi.org/https://doi.org/10.1103/PhysRevLett.54.857} {\bibfield  {journal} {\bibinfo  {journal} {Physical Review Letters}\ }\textbf {\bibinfo {volume} {54}},\ \bibinfo {pages} {857} (\bibinfo {year} {1985})}\BibitemShut {NoStop}%
\bibitem [{\citenamefont {Singh}\ and\ \citenamefont {Datta}(2022)}]{singh2022detecting}%
  \BibitemOpen
  \bibfield  {author} {\bibinfo {author} {\bibfnamefont {S.}~\bibnamefont {Singh}}\ and\ \bibinfo {author} {\bibfnamefont {N.}~\bibnamefont {Datta}},\ }\bibfield  {title} {\bibinfo {title} {Detecting positive quantum capacities of quantum channels},\ }\href {https://doi.org/https://doi.org/10.1038/s41534-022-00550-2} {\bibfield  {journal} {\bibinfo  {journal} {npj Quantum Information}\ }\textbf {\bibinfo {volume} {8}},\ \bibinfo {pages} {50} (\bibinfo {year} {2022})}\BibitemShut {NoStop}%
\bibitem [{\citenamefont {Chiribella}\ \emph {et~al.}(2013)\citenamefont {Chiribella}, \citenamefont {D’Ariano}, \citenamefont {Perinotti},\ and\ \citenamefont {Valiron}}]{chiribella2013quantum}%
  \BibitemOpen
  \bibfield  {author} {\bibinfo {author} {\bibfnamefont {G.}~\bibnamefont {Chiribella}}, \bibinfo {author} {\bibfnamefont {G.~M.}\ \bibnamefont {D’Ariano}}, \bibinfo {author} {\bibfnamefont {P.}~\bibnamefont {Perinotti}},\ and\ \bibinfo {author} {\bibfnamefont {B.}~\bibnamefont {Valiron}},\ }\bibfield  {title} {\bibinfo {title} {Quantum computations without definite causal structure},\ }\href {https://doi.org/10.1103/PhysRevA.88.022318} {\bibfield  {journal} {\bibinfo  {journal} {Physical Review A—Atomic, Molecular, and Optical Physics}\ }\textbf {\bibinfo {volume} {88}},\ \bibinfo {pages} {022318} (\bibinfo {year} {2013})}\BibitemShut {NoStop}%
\bibitem [{\citenamefont {Ara{\'u}jo}\ \emph {et~al.}(2015)\citenamefont {Ara{\'u}jo}, \citenamefont {Branciard}, \citenamefont {Costa}, \citenamefont {Feix}, \citenamefont {Giarmatzi},\ and\ \citenamefont {Brukner}}]{araujo2015witnessing}%
  \BibitemOpen
  \bibfield  {author} {\bibinfo {author} {\bibfnamefont {M.}~\bibnamefont {Ara{\'u}jo}}, \bibinfo {author} {\bibfnamefont {C.}~\bibnamefont {Branciard}}, \bibinfo {author} {\bibfnamefont {F.}~\bibnamefont {Costa}}, \bibinfo {author} {\bibfnamefont {A.}~\bibnamefont {Feix}}, \bibinfo {author} {\bibfnamefont {C.}~\bibnamefont {Giarmatzi}},\ and\ \bibinfo {author} {\bibfnamefont {{\v{C}}.}~\bibnamefont {Brukner}},\ }\bibfield  {title} {\bibinfo {title} {Witnessing causal nonseparability},\ }\href {https://doi.org/10.1088/1367-2630/17/10/102001} {\bibfield  {journal} {\bibinfo  {journal} {New Journal of Physics}\ }\textbf {\bibinfo {volume} {17}},\ \bibinfo {pages} {102001} (\bibinfo {year} {2015})}\BibitemShut {NoStop}%
\bibitem [{\citenamefont {Cerezo}\ \emph {et~al.}(2017)\citenamefont {Cerezo}, \citenamefont {Rossignoli}, \citenamefont {Canosa},\ and\ \citenamefont {R{\'\i}os}}]{cerezo2017factorization}%
  \BibitemOpen
  \bibfield  {author} {\bibinfo {author} {\bibfnamefont {M.}~\bibnamefont {Cerezo}}, \bibinfo {author} {\bibfnamefont {R.}~\bibnamefont {Rossignoli}}, \bibinfo {author} {\bibfnamefont {N.}~\bibnamefont {Canosa}},\ and\ \bibinfo {author} {\bibfnamefont {E.}~\bibnamefont {R{\'\i}os}},\ }\bibfield  {title} {\bibinfo {title} {Factorization and criticality in finite $xxz$ systems of arbitrary spin},\ }\href {https://doi.org/10.1103/PhysRevLett.119.220605} {\bibfield  {journal} {\bibinfo  {journal} {Physical Review Letters}\ }\textbf {\bibinfo {volume} {119}},\ \bibinfo {pages} {220605} (\bibinfo {year} {2017})}\BibitemShut {NoStop}%
\bibitem [{\citenamefont {Des~Cloizeaux}\ and\ \citenamefont {Pearson}(1962)}]{des1962spin}%
  \BibitemOpen
  \bibfield  {author} {\bibinfo {author} {\bibfnamefont {J.}~\bibnamefont {Des~Cloizeaux}}\ and\ \bibinfo {author} {\bibfnamefont {J.}~\bibnamefont {Pearson}},\ }\bibfield  {title} {\bibinfo {title} {Spin-wave spectrum of the antiferromagnetic linear chain},\ }\href {https://doi.org/10.1103/PhysRev.128.2131} {\bibfield  {journal} {\bibinfo  {journal} {Physical Review}\ }\textbf {\bibinfo {volume} {128}},\ \bibinfo {pages} {2131} (\bibinfo {year} {1962})}\BibitemShut {NoStop}%
\bibitem [{\citenamefont {Lee}\ \emph {et~al.}(2026)\citenamefont {Lee}, \citenamefont {Kumaran}, \citenamefont {Pokharel}, \citenamefont {Scheie}, \citenamefont {Sarkis}, \citenamefont {Tennant}, \citenamefont {Humble}, \citenamefont {Schleife}, \citenamefont {Kandala},\ and\ \citenamefont {Banerjee}}]{lee2026benchmarking}%
  \BibitemOpen
  \bibfield  {author} {\bibinfo {author} {\bibfnamefont {Y.-T.}\ \bibnamefont {Lee}}, \bibinfo {author} {\bibfnamefont {K.}~\bibnamefont {Kumaran}}, \bibinfo {author} {\bibfnamefont {B.}~\bibnamefont {Pokharel}}, \bibinfo {author} {\bibfnamefont {A.}~\bibnamefont {Scheie}}, \bibinfo {author} {\bibfnamefont {C.~L.}\ \bibnamefont {Sarkis}}, \bibinfo {author} {\bibfnamefont {D.~A.}\ \bibnamefont {Tennant}}, \bibinfo {author} {\bibfnamefont {T.}~\bibnamefont {Humble}}, \bibinfo {author} {\bibfnamefont {A.}~\bibnamefont {Schleife}}, \bibinfo {author} {\bibfnamefont {A.}~\bibnamefont {Kandala}},\ and\ \bibinfo {author} {\bibfnamefont {A.}~\bibnamefont {Banerjee}},\ }\bibfield  {title} {\bibinfo {title} {Benchmarking quantum simulation with neutron-scattering experiments},\ }\bibfield  {journal} {\bibinfo  {journal} {arXiv preprint arXiv:2603.15608}\ }\href {https://doi.org/10.48550/arXiv.2603.15608} {10.48550/arXiv.2603.15608} (\bibinfo {year} {2026})\BibitemShut {NoStop}%
\bibitem [{\citenamefont {Schwinger}(1961)}]{schwinger1961brownian}%
  \BibitemOpen
  \bibfield  {author} {\bibinfo {author} {\bibfnamefont {J.}~\bibnamefont {Schwinger}},\ }\bibfield  {title} {\bibinfo {title} {Brownian motion of a quantum oscillator},\ }\href {https://doi.org/https://doi.org/10.1063/1.1703727} {\bibfield  {journal} {\bibinfo  {journal} {Journal of Mathematical Physics}\ }\textbf {\bibinfo {volume} {2}},\ \bibinfo {pages} {407} (\bibinfo {year} {1961})}\BibitemShut {NoStop}%
\bibitem [{\citenamefont {Keldysh}(1965)}]{keldysh2024diagram}%
  \BibitemOpen
  \bibfield  {author} {\bibinfo {author} {\bibfnamefont {L.~V.}\ \bibnamefont {Keldysh}},\ }\bibfield  {title} {\bibinfo {title} {Diagram technique for nonequilibrium processes},\ }\href@noop {} {\bibfield  {journal} {\bibinfo  {journal} {Sov. Phys. JETP}\ }\textbf {\bibinfo {volume} {20}} (\bibinfo {year} {1965})}\BibitemShut {NoStop}%
\bibitem [{\citenamefont {Griffiths}(1984)}]{griffiths1984consistent}%
  \BibitemOpen
  \bibfield  {author} {\bibinfo {author} {\bibfnamefont {R.~B.}\ \bibnamefont {Griffiths}},\ }\bibfield  {title} {\bibinfo {title} {Consistent histories and the interpretation of quantum mechanics},\ }\href {https://doi.org/https://doi.org/10.1007/BF01015734} {\bibfield  {journal} {\bibinfo  {journal} {Journal of Statistical Physics}\ }\textbf {\bibinfo {volume} {36}},\ \bibinfo {pages} {219} (\bibinfo {year} {1984})}\BibitemShut {NoStop}%
\bibitem [{\citenamefont {Gell-Mann}\ and\ \citenamefont {Hartle}(2010)}]{gell2010quantum}%
  \BibitemOpen
  \bibfield  {author} {\bibinfo {author} {\bibfnamefont {M.}~\bibnamefont {Gell-Mann}}\ and\ \bibinfo {author} {\bibfnamefont {J.~B.}\ \bibnamefont {Hartle}},\ }\bibfield  {title} {\bibinfo {title} {Quantum mechanics in the light of quantum cosmology},\ }in\ \href {https://doi.org/https://doi.org/10.1142/9789812836854_0021} {\emph {\bibinfo {booktitle} {Murray Gell-Mann: Selected Papers}}}\ (\bibinfo  {publisher} {World Scientific},\ \bibinfo {year} {2010})\ pp.\ \bibinfo {pages} {303--325}\BibitemShut {NoStop}%
\bibitem [{\citenamefont {Isham}(1994)}]{isham1994quantum}%
  \BibitemOpen
  \bibfield  {author} {\bibinfo {author} {\bibfnamefont {C.~J.}\ \bibnamefont {Isham}},\ }\bibfield  {title} {\bibinfo {title} {Quantum logic and the histories approach to quantum theory},\ }\href {https://doi.org/10.1063/1.530544} {\bibfield  {journal} {\bibinfo  {journal} {Journal of Mathematical Physics}\ }\textbf {\bibinfo {volume} {35}},\ \bibinfo {pages} {2157} (\bibinfo {year} {1994})}\BibitemShut {NoStop}%
\bibitem [{\citenamefont {Isham}\ \emph {et~al.}(1994)\citenamefont {Isham}, \citenamefont {Linden},\ and\ \citenamefont {Schreckenberg}}]{isham1994classification}%
  \BibitemOpen
  \bibfield  {author} {\bibinfo {author} {\bibfnamefont {C.~J.}\ \bibnamefont {Isham}}, \bibinfo {author} {\bibfnamefont {N.}~\bibnamefont {Linden}},\ and\ \bibinfo {author} {\bibfnamefont {S.}~\bibnamefont {Schreckenberg}},\ }\bibfield  {title} {\bibinfo {title} {The classification of decoherence functionals: An analog of gleason’s theorem},\ }\href {https://doi.org/https://doi.org/10.1063/1.530679} {\bibfield  {journal} {\bibinfo  {journal} {Journal of Mathematical Physics}\ }\textbf {\bibinfo {volume} {35}},\ \bibinfo {pages} {6360} (\bibinfo {year} {1994})}\BibitemShut {NoStop}%
\bibitem [{Note3()}]{Note3}%
  \BibitemOpen
  \bibinfo {note} {Several results of this subsection are based on \cite {diaz2021path}, where further mathematical details have been discussed}\BibitemShut {NoStop}%
\bibitem [{Note4()}]{Note4}%
  \BibitemOpen
  \bibinfo {note} {In many field theoretical applications at zero temperature a PI suffices as one is interested in vacuum correlation functions, which can be obtained by the quotient of two Feynman PIs. At finite temperature and out of equilibrium standard PI techniques are no longer feasible \cite {kamenev2005course}.}\BibitemShut {Stop}%
\bibitem [{\citenamefont {Brif}\ and\ \citenamefont {Mann}(1999)}]{brif1999phase}%
  \BibitemOpen
  \bibfield  {author} {\bibinfo {author} {\bibfnamefont {C.}~\bibnamefont {Brif}}\ and\ \bibinfo {author} {\bibfnamefont {A.}~\bibnamefont {Mann}},\ }\bibfield  {title} {\bibinfo {title} {Phase-space formulation of quantum mechanics and quantum-state reconstruction for physical systems with lie-group symmetries},\ }\href {https://doi.org/10.1103/PhysRevA.59.971} {\bibfield  {journal} {\bibinfo  {journal} {Physical Review A}\ }\textbf {\bibinfo {volume} {59}},\ \bibinfo {pages} {971} (\bibinfo {year} {1999})}\BibitemShut {NoStop}%
\bibitem [{\citenamefont {Lie}\ and\ \citenamefont {Fullwood}(2025)}]{lie2025multipartite}%
  \BibitemOpen
  \bibfield  {author} {\bibinfo {author} {\bibfnamefont {S.~H.}\ \bibnamefont {Lie}}\ and\ \bibinfo {author} {\bibfnamefont {J.}~\bibnamefont {Fullwood}},\ }\bibfield  {title} {\bibinfo {title} {Multipartite quantum states over time from two fundamental assumptions},\ }\href {https://doi.org/10.1103/lbf3-snp8} {\bibfield  {journal} {\bibinfo  {journal} {Physical Review Letters}\ }\textbf {\bibinfo {volume} {135}},\ \bibinfo {pages} {230204} (\bibinfo {year} {2025})}\BibitemShut {NoStop}%
\bibitem [{\citenamefont {Marletto}\ \emph {et~al.}(2020)\citenamefont {Marletto}, \citenamefont {Vedral}, \citenamefont {Virz{\`\i}}, \citenamefont {Rebufello}, \citenamefont {Avella}, \citenamefont {Piacentini}, \citenamefont {Gramegna}, \citenamefont {Degiovanni},\ and\ \citenamefont {Genovese}}]{marletto2020non}%
  \BibitemOpen
  \bibfield  {author} {\bibinfo {author} {\bibfnamefont {C.}~\bibnamefont {Marletto}}, \bibinfo {author} {\bibfnamefont {V.}~\bibnamefont {Vedral}}, \bibinfo {author} {\bibfnamefont {S.}~\bibnamefont {Virz{\`\i}}}, \bibinfo {author} {\bibfnamefont {E.}~\bibnamefont {Rebufello}}, \bibinfo {author} {\bibfnamefont {A.}~\bibnamefont {Avella}}, \bibinfo {author} {\bibfnamefont {F.}~\bibnamefont {Piacentini}}, \bibinfo {author} {\bibfnamefont {M.}~\bibnamefont {Gramegna}}, \bibinfo {author} {\bibfnamefont {I.~P.}\ \bibnamefont {Degiovanni}},\ and\ \bibinfo {author} {\bibfnamefont {M.}~\bibnamefont {Genovese}},\ }\bibfield  {title} {\bibinfo {title} {Non-monogamy of spatio-temporal correlations and the black hole information loss paradox},\ }\href {https://doi.org/https://doi.org/10.3390/e22020228} {\bibfield  {journal} {\bibinfo  {journal} {Entropy}\ }\textbf {\bibinfo {volume} {22}},\ \bibinfo {pages} {228} (\bibinfo {year} {2020})}\BibitemShut {NoStop}%
\bibitem [{\citenamefont {Song}\ \emph {et~al.}(2024)\citenamefont {Song}, \citenamefont {Narasimhachar}, \citenamefont {Regula}, \citenamefont {Elliott},\ and\ \citenamefont {Gu}}]{song2024causal}%
  \BibitemOpen
  \bibfield  {author} {\bibinfo {author} {\bibfnamefont {M.}~\bibnamefont {Song}}, \bibinfo {author} {\bibfnamefont {V.}~\bibnamefont {Narasimhachar}}, \bibinfo {author} {\bibfnamefont {B.}~\bibnamefont {Regula}}, \bibinfo {author} {\bibfnamefont {T.~J.}\ \bibnamefont {Elliott}},\ and\ \bibinfo {author} {\bibfnamefont {M.}~\bibnamefont {Gu}},\ }\bibfield  {title} {\bibinfo {title} {Causal classification of spatiotemporal quantum correlations},\ }\href {https://doi.org/https://doi.org/10.1103/PhysRevLett.133.110202} {\bibfield  {journal} {\bibinfo  {journal} {Physical Review Letters}\ }\textbf {\bibinfo {volume} {133}},\ \bibinfo {pages} {110202} (\bibinfo {year} {2024})}\BibitemShut {NoStop}%
\bibitem [{\citenamefont {Fullwood}(2025)}]{fullwood2025quantum}%
  \BibitemOpen
  \bibfield  {author} {\bibinfo {author} {\bibfnamefont {J.}~\bibnamefont {Fullwood}},\ }\bibfield  {title} {\bibinfo {title} {Quantum dynamics as a pseudo-density matrix},\ }\href {https://doi.org/https://doi.org/10.22331/q-2025-04-24-1719} {\bibfield  {journal} {\bibinfo  {journal} {Quantum}\ }\textbf {\bibinfo {volume} {9}},\ \bibinfo {pages} {1719} (\bibinfo {year} {2025})}\BibitemShut {NoStop}%
\bibitem [{\citenamefont {Zhang}\ \emph {et~al.}(2020)\citenamefont {Zhang}, \citenamefont {Dahlsten},\ and\ \citenamefont {Vedral}}]{zhang2020different}%
  \BibitemOpen
  \bibfield  {author} {\bibinfo {author} {\bibfnamefont {T.}~\bibnamefont {Zhang}}, \bibinfo {author} {\bibfnamefont {O.}~\bibnamefont {Dahlsten}},\ and\ \bibinfo {author} {\bibfnamefont {V.}~\bibnamefont {Vedral}},\ }\bibfield  {title} {\bibinfo {title} {Different instances of time as different quantum modes: quantum states across space-time for continuous variables},\ }\href {https://doi.org/10.1088/1367-2630/ab6b9f} {\bibfield  {journal} {\bibinfo  {journal} {New Journal of Physics}\ }\textbf {\bibinfo {volume} {22}},\ \bibinfo {pages} {023029} (\bibinfo {year} {2020})}\BibitemShut {NoStop}%
\bibitem [{Note5()}]{Note5}%
  \BibitemOpen
  \bibinfo {note} {We consider that between each measurements a fixed amount of time $\epsilon $ lapsed. We also mostly focus on unitary evolution as our discussion can be extended to quantum channels by simply extending the system.}\BibitemShut {Stop}%
\bibitem [{\citenamefont {Liu}\ \emph {et~al.}(2025)\citenamefont {Liu}, \citenamefont {Qiu}, \citenamefont {Dahlsten},\ and\ \citenamefont {Vedral}}]{liu2025quantum}%
  \BibitemOpen
  \bibfield  {author} {\bibinfo {author} {\bibfnamefont {X.}~\bibnamefont {Liu}}, \bibinfo {author} {\bibfnamefont {Y.}~\bibnamefont {Qiu}}, \bibinfo {author} {\bibfnamefont {O.}~\bibnamefont {Dahlsten}},\ and\ \bibinfo {author} {\bibfnamefont {V.}~\bibnamefont {Vedral}},\ }\bibfield  {title} {\bibinfo {title} {Quantum causal inference with extremely light touch},\ }\href {https://doi.org/https://doi.org/10.1038/s41534-024-00956-0} {\bibfield  {journal} {\bibinfo  {journal} {npj Quantum Information}\ }\textbf {\bibinfo {volume} {11}},\ \bibinfo {pages} {1} (\bibinfo {year} {2025})}\BibitemShut {NoStop}%
\bibitem [{Note6()}]{Note6}%
  \BibitemOpen
  \bibinfo {note} {In this sense, that Markovian QSOT extension may be naturally identified with the multipartite PDM. To our knowledge, this connection was not made explicit in \cite {lie2025multipartite}.}\BibitemShut {Stop}%
\bibitem [{\citenamefont {Rudolph}(2003)}]{rudolph2003some}%
  \BibitemOpen
  \bibfield  {author} {\bibinfo {author} {\bibfnamefont {O.}~\bibnamefont {Rudolph}},\ }\bibfield  {title} {\bibinfo {title} {Some properties of the computable cross-norm criterion for separability},\ }\href {https://doi.org/https://doi.org/10.1103/PhysRevA.67.032312} {\bibfield  {journal} {\bibinfo  {journal} {Physical Review A}\ }\textbf {\bibinfo {volume} {67}},\ \bibinfo {pages} {032312} (\bibinfo {year} {2003})}\BibitemShut {NoStop}%
\bibitem [{\citenamefont {Chen}\ and\ \citenamefont {Wu}(2003)}]{chen2002matrix}%
  \BibitemOpen
  \bibfield  {author} {\bibinfo {author} {\bibfnamefont {K.}~\bibnamefont {Chen}}\ and\ \bibinfo {author} {\bibfnamefont {L.-A.}\ \bibnamefont {Wu}},\ }\bibfield  {title} {\bibinfo {title} {A matrix realignment method for recognizing entanglement},\ }\href@noop {} {\bibfield  {journal} {\bibinfo  {journal} {Quantum Information and Computation}\ }\textbf {\bibinfo {volume} {3}},\ \bibinfo {pages} {193} (\bibinfo {year} {2003})}\BibitemShut {NoStop}%
\bibitem [{\citenamefont {Peres}(1996)}]{peres1996separability}%
  \BibitemOpen
  \bibfield  {author} {\bibinfo {author} {\bibfnamefont {A.}~\bibnamefont {Peres}},\ }\bibfield  {title} {\bibinfo {title} {Separability criterion for density matrices},\ }\href {https://doi.org/10.1103/PhysRevLett.77.1413} {\bibfield  {journal} {\bibinfo  {journal} {Physical Review Letters}\ }\textbf {\bibinfo {volume} {77}},\ \bibinfo {pages} {1413} (\bibinfo {year} {1996})}\BibitemShut {NoStop}%
\bibitem [{\citenamefont {Horodecki}\ \emph {et~al.}(2001)\citenamefont {Horodecki}, \citenamefont {Horodecki},\ and\ \citenamefont {Horodecki}}]{horodecki2001separability}%
  \BibitemOpen
  \bibfield  {author} {\bibinfo {author} {\bibfnamefont {M.}~\bibnamefont {Horodecki}}, \bibinfo {author} {\bibfnamefont {P.}~\bibnamefont {Horodecki}},\ and\ \bibinfo {author} {\bibfnamefont {R.}~\bibnamefont {Horodecki}},\ }\bibfield  {title} {\bibinfo {title} {Separability of n-particle mixed states: necessary and sufficient conditions in terms of linear maps},\ }\href {https://doi.org/10.1016/S0375-9601(01)00142-6} {\bibfield  {journal} {\bibinfo  {journal} {Physics Letters A}\ }\textbf {\bibinfo {volume} {283}},\ \bibinfo {pages} {1} (\bibinfo {year} {2001})}\BibitemShut {NoStop}%
\bibitem [{\citenamefont {DeWitt}(1967)}]{dewitt1967quantum}%
  \BibitemOpen
  \bibfield  {author} {\bibinfo {author} {\bibfnamefont {B.~S.}\ \bibnamefont {DeWitt}},\ }\bibfield  {title} {\bibinfo {title} {Quantum theory of gravity. i. the canonical theory},\ }\href {https://doi.org/10.1103/PhysRev.160.1113} {\bibfield  {journal} {\bibinfo  {journal} {Physical Review}\ }\textbf {\bibinfo {volume} {160}},\ \bibinfo {pages} {1113} (\bibinfo {year} {1967})}\BibitemShut {NoStop}%
\bibitem [{\citenamefont {Unruh}\ and\ \citenamefont {Wald}(1989)}]{unruh1989time}%
  \BibitemOpen
  \bibfield  {author} {\bibinfo {author} {\bibfnamefont {W.~G.}\ \bibnamefont {Unruh}}\ and\ \bibinfo {author} {\bibfnamefont {R.~M.}\ \bibnamefont {Wald}},\ }\bibfield  {title} {\bibinfo {title} {Time and the interpretation of canonical quantum gravity},\ }\href {https://doi.org/https://doi.org/10.1103/PhysRevD.40.2598} {\bibfield  {journal} {\bibinfo  {journal} {Physical Review D}\ }\textbf {\bibinfo {volume} {40}},\ \bibinfo {pages} {2598} (\bibinfo {year} {1989})}\BibitemShut {NoStop}%
\bibitem [{\citenamefont {Kucha{\v{r}}}(2011)}]{kuchavr2011time}%
  \BibitemOpen
  \bibfield  {author} {\bibinfo {author} {\bibfnamefont {K.~V.}\ \bibnamefont {Kucha{\v{r}}}},\ }\bibfield  {title} {\bibinfo {title} {Time and interpretations of quantum gravity},\ }\href {https://doi.org/10.1142/S0218271811019347} {\bibfield  {journal} {\bibinfo  {journal} {International Journal of Modern Physics D}\ }\textbf {\bibinfo {volume} {20}},\ \bibinfo {pages} {3} (\bibinfo {year} {2011})}\BibitemShut {NoStop}%
\bibitem [{\citenamefont {Boette}\ and\ \citenamefont {Rossignoli}(2018)}]{boette2018history}%
  \BibitemOpen
  \bibfield  {author} {\bibinfo {author} {\bibfnamefont {A.}~\bibnamefont {Boette}}\ and\ \bibinfo {author} {\bibfnamefont {R.}~\bibnamefont {Rossignoli}},\ }\bibfield  {title} {\bibinfo {title} {History states of systems and operators},\ }\href {https://doi.org/10.1103/PhysRevA.98.032108} {\bibfield  {journal} {\bibinfo  {journal} {Physical Review A}\ }\textbf {\bibinfo {volume} {98}},\ \bibinfo {pages} {032108} (\bibinfo {year} {2018})}\BibitemShut {NoStop}%
\bibitem [{\citenamefont {Paiva}\ \emph {et~al.}(2022)\citenamefont {Paiva}, \citenamefont {Te’eni}, \citenamefont {Peled}, \citenamefont {Cohen},\ and\ \citenamefont {Aharonov}}]{paiva2022non}%
  \BibitemOpen
  \bibfield  {author} {\bibinfo {author} {\bibfnamefont {I.~L.}\ \bibnamefont {Paiva}}, \bibinfo {author} {\bibfnamefont {A.}~\bibnamefont {Te’eni}}, \bibinfo {author} {\bibfnamefont {B.~Y.}\ \bibnamefont {Peled}}, \bibinfo {author} {\bibfnamefont {E.}~\bibnamefont {Cohen}},\ and\ \bibinfo {author} {\bibfnamefont {Y.}~\bibnamefont {Aharonov}},\ }\bibfield  {title} {\bibinfo {title} {Non-inertial quantum clock frames lead to non-hermitian dynamics},\ }\href {https://doi.org/10.1038/s42005-022-01081-0} {\bibfield  {journal} {\bibinfo  {journal} {Communications Physics}\ }\textbf {\bibinfo {volume} {5}},\ \bibinfo {pages} {298} (\bibinfo {year} {2022})}\BibitemShut {NoStop}%
\bibitem [{\citenamefont {Diaz}\ \emph {et~al.}(2025{\natexlab{b}})\citenamefont {Diaz}, \citenamefont {Braccia}, \citenamefont {Larocca}, \citenamefont {Matera}, \citenamefont {Rossignoli},\ and\ \citenamefont {Cerezo}}]{diaz2023parallel}%
  \BibitemOpen
  \bibfield  {author} {\bibinfo {author} {\bibfnamefont {N.~L.}\ \bibnamefont {Diaz}}, \bibinfo {author} {\bibfnamefont {P.}~\bibnamefont {Braccia}}, \bibinfo {author} {\bibfnamefont {M.}~\bibnamefont {Larocca}}, \bibinfo {author} {\bibfnamefont {J.~M.}\ \bibnamefont {Matera}}, \bibinfo {author} {\bibfnamefont {R.}~\bibnamefont {Rossignoli}},\ and\ \bibinfo {author} {\bibfnamefont {M.}~\bibnamefont {Cerezo}},\ }\bibfield  {title} {\bibinfo {title} {Parallel-in-time quantum simulation via Page and Wootters quantum time},\ }\href {https://doi.org/10.1103/wpnf-4nnn} {\bibfield  {journal} {\bibinfo  {journal} {Phys. Rev. Res.}\ }\textbf {\bibinfo {volume} {7}},\ \bibinfo {pages} {033294} (\bibinfo {year} {2025}{\natexlab{b}})}\BibitemShut {NoStop}%
\bibitem [{\citenamefont {Cafasso}\ \emph {et~al.}(2024)\citenamefont {Cafasso}, \citenamefont {Pranzini}, \citenamefont {Malo}, \citenamefont {Giovannetti},\ and\ \citenamefont {Chiofalo}}]{cafasso2024quantum}%
  \BibitemOpen
  \bibfield  {author} {\bibinfo {author} {\bibfnamefont {D.}~\bibnamefont {Cafasso}}, \bibinfo {author} {\bibfnamefont {N.}~\bibnamefont {Pranzini}}, \bibinfo {author} {\bibfnamefont {J.~Y.}\ \bibnamefont {Malo}}, \bibinfo {author} {\bibfnamefont {V.}~\bibnamefont {Giovannetti}},\ and\ \bibinfo {author} {\bibfnamefont {M.}~\bibnamefont {Chiofalo}},\ }\bibfield  {title} {\bibinfo {title} {Quantum time and the time-dilation induced interaction transfer mechanism},\ }\href {https://doi.org/https://doi.org/10.1103/PhysRevD.110.106014} {\bibfield  {journal} {\bibinfo  {journal} {Physical Review D}\ }\textbf {\bibinfo {volume} {110}},\ \bibinfo {pages} {106014} (\bibinfo {year} {2024})}\BibitemShut {NoStop}%
\bibitem [{\citenamefont {Coppo}\ \emph {et~al.}(2026)\citenamefont {Coppo}, \citenamefont {Pranzini},\ and\ \citenamefont {Verrucchi}}]{coppo2026quantum}%
  \BibitemOpen
  \bibfield  {author} {\bibinfo {author} {\bibfnamefont {A.}~\bibnamefont {Coppo}}, \bibinfo {author} {\bibfnamefont {N.}~\bibnamefont {Pranzini}},\ and\ \bibinfo {author} {\bibfnamefont {P.}~\bibnamefont {Verrucchi}},\ }\bibfield  {title} {\bibinfo {title} {Quantum model for black holes and clocks},\ }\bibfield  {journal} {\bibinfo  {journal} {arXiv preprint arXiv:2601.07437}\ }\href {https://doi.org/10.48550/arXiv.2601.07437} {10.48550/arXiv.2601.07437} (\bibinfo {year} {2026})\BibitemShut {NoStop}%
\bibitem [{\citenamefont {Diaz}\ \emph {et~al.}(2019)\citenamefont {Diaz}, \citenamefont {Matera},\ and\ \citenamefont {Rossignoli}}]{diaz2019history}%
  \BibitemOpen
  \bibfield  {author} {\bibinfo {author} {\bibfnamefont {N.~L.}\ \bibnamefont {Diaz}}, \bibinfo {author} {\bibfnamefont {J.~M.}\ \bibnamefont {Matera}},\ and\ \bibinfo {author} {\bibfnamefont {R.}~\bibnamefont {Rossignoli}},\ }\bibfield  {title} {\bibinfo {title} {History state formalism for scalar particles},\ }\href {https://doi.org/10.1103/PhysRevD.100.125020} {\bibfield  {journal} {\bibinfo  {journal} {Physical Review D}\ }\textbf {\bibinfo {volume} {100}},\ \bibinfo {pages} {125020} (\bibinfo {year} {2019})}\BibitemShut {NoStop}%
\bibitem [{\citenamefont {Diaz}\ and\ \citenamefont {Rossignoli}(2019)}]{diaz2019historystate}%
  \BibitemOpen
  \bibfield  {author} {\bibinfo {author} {\bibfnamefont {N.~L.}\ \bibnamefont {Diaz}}\ and\ \bibinfo {author} {\bibfnamefont {R.}~\bibnamefont {Rossignoli}},\ }\bibfield  {title} {\bibinfo {title} {History state formalism for Dirac’s theory},\ }\href {https://doi.org/10.1103/PhysRevD.99.045008} {\bibfield  {journal} {\bibinfo  {journal} {Physical Review D}\ }\textbf {\bibinfo {volume} {99}},\ \bibinfo {pages} {045008} (\bibinfo {year} {2019})}\BibitemShut {NoStop}%
\bibitem [{\citenamefont {Schwabl}(2008)}]{schwabl2008advanced}%
  \BibitemOpen
  \bibfield  {author} {\bibinfo {author} {\bibfnamefont {F.}~\bibnamefont {Schwabl}},\ }\href@noop {} {\emph {\bibinfo {title} {Advanced quantum mechanics}}}\ (\bibinfo  {publisher} {Springer},\ \bibinfo {year} {2008})\BibitemShut {NoStop}%
\bibitem [{Note7()}]{Note7}%
  \BibitemOpen
  \bibinfo {note} {One could attempt to introduce e.g., two-particle interaction within the PW scheme. Nonetheless, since usually a single time is considered, in practice this is never done. Instead, introducing a ``clock'' for each particle give rise to a complete new set of challenges as one has to deal with a multi-time scheme. None of the difficulties arise in the second quantized approach.}\BibitemShut {Stop}%
\bibitem [{\citenamefont {Coleman}(1963)}]{coleman1963structure}%
  \BibitemOpen
  \bibfield  {author} {\bibinfo {author} {\bibfnamefont {A.~J.}\ \bibnamefont {Coleman}},\ }\bibfield  {title} {\bibinfo {title} {Structure of fermion density matrices},\ }\href {https://doi.org/10.1103/RevModPhys.35.668} {\bibfield  {journal} {\bibinfo  {journal} {Reviews of modern Physics}\ }\textbf {\bibinfo {volume} {35}},\ \bibinfo {pages} {668} (\bibinfo {year} {1963})}\BibitemShut {NoStop}%
\bibitem [{\citenamefont {Mazziotti}(2012)}]{mazziotti2012structure}%
  \BibitemOpen
  \bibfield  {author} {\bibinfo {author} {\bibfnamefont {D.~A.}\ \bibnamefont {Mazziotti}},\ }\bibfield  {title} {\bibinfo {title} {Structure of fermionic density matrices: Complete n-representability conditions},\ }\href {https://doi.org/10.1103/PhysRevB.78.201103} {\bibfield  {journal} {\bibinfo  {journal} {Physical Review Letters}\ }\textbf {\bibinfo {volume} {108}},\ \bibinfo {pages} {263002} (\bibinfo {year} {2012})}\BibitemShut {NoStop}%
\bibitem [{Note8()}]{Note8}%
  \BibitemOpen
  \bibinfo {note} {We refer the reader to \cite {diaz2023spacetime,diaz2021path} for subtleties in the continuum time limit.}\BibitemShut {Stop}%
\bibitem [{\citenamefont {Guo}\ \emph {et~al.}(2025)\citenamefont {Guo}, \citenamefont {He},\ and\ \citenamefont {Liu}}]{guo2025entanglement}%
  \BibitemOpen
  \bibfield  {author} {\bibinfo {author} {\bibfnamefont {W.-z.}\ \bibnamefont {Guo}}, \bibinfo {author} {\bibfnamefont {S.}~\bibnamefont {He}},\ and\ \bibinfo {author} {\bibfnamefont {T.}~\bibnamefont {Liu}},\ }\bibfield  {title} {\bibinfo {title} {Entanglement of general subregions in time-dependent states},\ }\bibfield  {journal} {\bibinfo  {journal} {arXiv preprint arXiv:2512.19955}\ }\href {https://doi.org/10.48550/arXiv.2512.19955} {10.48550/arXiv.2512.19955} (\bibinfo {year} {2025})\BibitemShut {NoStop}%
\bibitem [{\citenamefont {Chu}\ and\ \citenamefont {Parihar}(2023)}]{chu2023time}%
  \BibitemOpen
  \bibfield  {author} {\bibinfo {author} {\bibfnamefont {C.-S.}\ \bibnamefont {Chu}}\ and\ \bibinfo {author} {\bibfnamefont {H.}~\bibnamefont {Parihar}},\ }\bibfield  {title} {\bibinfo {title} {Time-like entanglement entropy in ads/bcft},\ }\href {https://arxiv.org/abs/2304.10907} {\bibfield  {journal} {\bibinfo  {journal} {arXiv preprint arXiv:2304.10907}\ } (\bibinfo {year} {2023})}\BibitemShut {NoStop}%
\bibitem [{\citenamefont {Narayan}\ and\ \citenamefont {Saini}(2023)}]{narayan2023notes}%
  \BibitemOpen
  \bibfield  {author} {\bibinfo {author} {\bibfnamefont {K.}~\bibnamefont {Narayan}}\ and\ \bibinfo {author} {\bibfnamefont {H.~K.}\ \bibnamefont {Saini}},\ }\bibfield  {title} {\bibinfo {title} {Notes on time entanglement and pseudo-entropy},\ }\href {https://arxiv.org/abs/2303.01307} {\bibfield  {journal} {\bibinfo  {journal} {arXiv preprint arXiv:2303.01307}\ } (\bibinfo {year} {2023})}\BibitemShut {NoStop}%
\bibitem [{\citenamefont {Nunez}\ and\ \citenamefont {Roychowdhury}(2025)}]{nunez2025timelike}%
  \BibitemOpen
  \bibfield  {author} {\bibinfo {author} {\bibfnamefont {C.}~\bibnamefont {Nunez}}\ and\ \bibinfo {author} {\bibfnamefont {D.}~\bibnamefont {Roychowdhury}},\ }\bibfield  {title} {\bibinfo {title} {Timelike entanglement entropy: A top-down approach},\ }\href {https://doi.org/10.1103/vjyt-xc15} {\bibfield  {journal} {\bibinfo  {journal} {Physical Review D}\ }\textbf {\bibinfo {volume} {112}},\ \bibinfo {pages} {026030} (\bibinfo {year} {2025})}\BibitemShut {NoStop}%
\bibitem [{\citenamefont {Jiang}\ \emph {et~al.}(2026)\citenamefont {Jiang}, \citenamefont {Wu},\ and\ \citenamefont {Yang}}]{jiang2025timelike}%
  \BibitemOpen
  \bibfield  {author} {\bibinfo {author} {\bibfnamefont {X.}~\bibnamefont {Jiang}}, \bibinfo {author} {\bibfnamefont {H.}~\bibnamefont {Wu}},\ and\ \bibinfo {author} {\bibfnamefont {H.}~\bibnamefont {Yang}},\ }\bibfield  {title} {\bibinfo {title} {Timelike entanglement entropy revisited},\ }\href {https://doi.org/10.1103/rl9b-373v} {\bibfield  {journal} {\bibinfo  {journal} {Phys. Rev. D}\ }\textbf {\bibinfo {volume} {113}},\ \bibinfo {pages} {106021} (\bibinfo {year} {2026})}\BibitemShut {NoStop}%
\bibitem [{\citenamefont {Emary}\ \emph {et~al.}(2014)\citenamefont {Emary}, \citenamefont {Lambert},\ and\ \citenamefont {Nori}}]{emary2014leggett}%
  \BibitemOpen
  \bibfield  {author} {\bibinfo {author} {\bibfnamefont {C.}~\bibnamefont {Emary}}, \bibinfo {author} {\bibfnamefont {N.}~\bibnamefont {Lambert}},\ and\ \bibinfo {author} {\bibfnamefont {F.}~\bibnamefont {Nori}},\ }\bibfield  {title} {\bibinfo {title} {Leggett--garg inequalities},\ }\href {https://doi.org/10.1088/0034-4885/77/3/039501} {\bibfield  {journal} {\bibinfo  {journal} {Reports on Progress in Physics}\ }\textbf {\bibinfo {volume} {77}},\ \bibinfo {pages} {016001} (\bibinfo {year} {2014})}\BibitemShut {NoStop}%
\bibitem [{Note9()}]{Note9}%
  \BibitemOpen
  \bibinfo {note} {Arbitrary evolutions can be considered by adding a simple adjoint action on $\protect \mathcal {R}$; equivalently one can take arbitrary $N$ and partial traces}\BibitemShut {NoStop}%
\bibitem [{\citenamefont {Roberts}\ and\ \citenamefont {Yoshida}(2017)}]{roberts2017chaos}%
  \BibitemOpen
  \bibfield  {author} {\bibinfo {author} {\bibfnamefont {D.~A.}\ \bibnamefont {Roberts}}\ and\ \bibinfo {author} {\bibfnamefont {B.}~\bibnamefont {Yoshida}},\ }\bibfield  {title} {\bibinfo {title} {Chaos and complexity by design},\ }\href {https://doi.org/10.1007/JHEP04(2017)121} {\bibfield  {journal} {\bibinfo  {journal} {Journal of High Energy Physics}\ }\textbf {\bibinfo {volume} {2017}},\ \bibinfo {pages} {121} (\bibinfo {year} {2017})}\BibitemShut {NoStop}%
\bibitem [{\citenamefont {Xu}\ and\ \citenamefont {Swingle}(2020)}]{xu2020accessing}%
  \BibitemOpen
  \bibfield  {author} {\bibinfo {author} {\bibfnamefont {S.}~\bibnamefont {Xu}}\ and\ \bibinfo {author} {\bibfnamefont {B.}~\bibnamefont {Swingle}},\ }\bibfield  {title} {\bibinfo {title} {Accessing scrambling using matrix product operators},\ }\href {https://doi.org/10.1038/s41567-019-0712-4} {\bibfield  {journal} {\bibinfo  {journal} {Nature Physics}\ }\textbf {\bibinfo {volume} {16}},\ \bibinfo {pages} {199} (\bibinfo {year} {2020})}\BibitemShut {NoStop}%
\bibitem [{\citenamefont {Maldacena}\ \emph {et~al.}(2016)\citenamefont {Maldacena}, \citenamefont {Shenker},\ and\ \citenamefont {Stanford}}]{maldacena2016bound}%
  \BibitemOpen
  \bibfield  {author} {\bibinfo {author} {\bibfnamefont {J.}~\bibnamefont {Maldacena}}, \bibinfo {author} {\bibfnamefont {S.~H.}\ \bibnamefont {Shenker}},\ and\ \bibinfo {author} {\bibfnamefont {D.}~\bibnamefont {Stanford}},\ }\bibfield  {title} {\bibinfo {title} {A bound on chaos},\ }\href@noop {} {\bibfield  {journal} {\bibinfo  {journal} {Journal of High Energy Physics}\ }\textbf {\bibinfo {volume} {2016}},\ \bibinfo {pages} {1} (\bibinfo {year} {2016})}\BibitemShut {NoStop}%
\bibitem [{\citenamefont {Lin}\ and\ \citenamefont {Motrunich}(2018)}]{lin2018out}%
  \BibitemOpen
  \bibfield  {author} {\bibinfo {author} {\bibfnamefont {C.-J.}\ \bibnamefont {Lin}}\ and\ \bibinfo {author} {\bibfnamefont {O.~I.}\ \bibnamefont {Motrunich}},\ }\bibfield  {title} {\bibinfo {title} {Out-of-time-ordered correlators in a quantum {I}sing chain},\ }\href {https://doi.org/10.1103/PhysRevB.97.144304} {\bibfield  {journal} {\bibinfo  {journal} {Physical Review B}\ }\textbf {\bibinfo {volume} {97}},\ \bibinfo {pages} {144304} (\bibinfo {year} {2018})}\BibitemShut {NoStop}%
\bibitem [{\citenamefont {Banuls}\ \emph {et~al.}(2009)\citenamefont {Banuls}, \citenamefont {Hastings}, \citenamefont {Verstraete},\ and\ \citenamefont {Cirac}}]{banuls2009matrix}%
  \BibitemOpen
  \bibfield  {author} {\bibinfo {author} {\bibfnamefont {M.-C.}\ \bibnamefont {Banuls}}, \bibinfo {author} {\bibfnamefont {M.~B.}\ \bibnamefont {Hastings}}, \bibinfo {author} {\bibfnamefont {F.}~\bibnamefont {Verstraete}},\ and\ \bibinfo {author} {\bibfnamefont {J.~I.}\ \bibnamefont {Cirac}},\ }\bibfield  {title} {\bibinfo {title} {Matrix product states for dynamical simulation of infinite chains},\ }\href {https://doi.org/doi.org/10.1103/PhysRevLett.102.240603} {\bibfield  {journal} {\bibinfo  {journal} {Physical review letters}\ }\textbf {\bibinfo {volume} {102}},\ \bibinfo {pages} {240603} (\bibinfo {year} {2009})}\BibitemShut {NoStop}%
\bibitem [{\citenamefont {Lerose}\ \emph {et~al.}(2021)\citenamefont {Lerose}, \citenamefont {Sonner},\ and\ \citenamefont {Abanin}}]{lerose2021scaling}%
  \BibitemOpen
  \bibfield  {author} {\bibinfo {author} {\bibfnamefont {A.}~\bibnamefont {Lerose}}, \bibinfo {author} {\bibfnamefont {M.}~\bibnamefont {Sonner}},\ and\ \bibinfo {author} {\bibfnamefont {D.~A.}\ \bibnamefont {Abanin}},\ }\bibfield  {title} {\bibinfo {title} {Scaling of temporal entanglement in proximity to integrability},\ }\href {https://doi.org/doi.org/10.1103/PhysRevB.104.035137} {\bibfield  {journal} {\bibinfo  {journal} {Physical Review B}\ }\textbf {\bibinfo {volume} {104}},\ \bibinfo {pages} {035137} (\bibinfo {year} {2021})}\BibitemShut {NoStop}%
\bibitem [{\citenamefont {Ye}\ and\ \citenamefont {Chan}(2021)}]{ye2021constructing}%
  \BibitemOpen
  \bibfield  {author} {\bibinfo {author} {\bibfnamefont {E.}~\bibnamefont {Ye}}\ and\ \bibinfo {author} {\bibfnamefont {G.~K.}\ \bibnamefont {Chan}},\ }\bibfield  {title} {\bibinfo {title} {Constructing tensor network influence functionals for general quantum dynamics},\ }\bibfield  {journal} {\bibinfo  {journal} {The Journal of Chemical Physics}\ }\textbf {\bibinfo {volume} {155}},\ \href {https://doi.org/10.1063/5.0047260} {10.1063/5.0047260} (\bibinfo {year} {2021})\BibitemShut {NoStop}%
\bibitem [{\citenamefont {Park}\ \emph {et~al.}(2025)\citenamefont {Park}, \citenamefont {Gray},\ and\ \citenamefont {Chan}}]{park2025simulating}%
  \BibitemOpen
  \bibfield  {author} {\bibinfo {author} {\bibfnamefont {G.}~\bibnamefont {Park}}, \bibinfo {author} {\bibfnamefont {J.}~\bibnamefont {Gray}},\ and\ \bibinfo {author} {\bibfnamefont {G.~K.-L.}\ \bibnamefont {Chan}},\ }\bibfield  {title} {\bibinfo {title} {Simulating quantum dynamics in two-dimensional lattices with tensor network influence functional belief propagation},\ }\href {https://doi.org/10.1103/7jzt-xhn6} {\bibfield  {journal} {\bibinfo  {journal} {Physical Review B}\ }\textbf {\bibinfo {volume} {112}},\ \bibinfo {pages} {174310} (\bibinfo {year} {2025})}\BibitemShut {NoStop}%
\bibitem [{\citenamefont {Lubasch}\ \emph {et~al.}(2014)\citenamefont {Lubasch}, \citenamefont {Cirac},\ and\ \citenamefont {Banuls}}]{lubasch2014unifying}%
  \BibitemOpen
  \bibfield  {author} {\bibinfo {author} {\bibfnamefont {M.}~\bibnamefont {Lubasch}}, \bibinfo {author} {\bibfnamefont {J.~I.}\ \bibnamefont {Cirac}},\ and\ \bibinfo {author} {\bibfnamefont {M.-C.}\ \bibnamefont {Banuls}},\ }\bibfield  {title} {\bibinfo {title} {Unifying projected entangled pair state contractions},\ }\href {https://doi.org/10.1088/1367-2630/16/3/033014} {\bibfield  {journal} {\bibinfo  {journal} {New Journal of Physics}\ }\textbf {\bibinfo {volume} {16}},\ \bibinfo {pages} {033014} (\bibinfo {year} {2014})}\BibitemShut {NoStop}%
\bibitem [{Note10()}]{Note10}%
  \BibitemOpen
  \bibinfo {note} {We use the Cauchy-Schwarz inequality, but one may also use the more general H{\"{o}}lder inequality and Schatten norms.}\BibitemShut {Stop}%
\bibitem [{\citenamefont {Gigena}\ and\ \citenamefont {Rossignoli}(2016)}]{gigena2016one}%
  \BibitemOpen
  \bibfield  {author} {\bibinfo {author} {\bibfnamefont {N.}~\bibnamefont {Gigena}}\ and\ \bibinfo {author} {\bibfnamefont {R.}~\bibnamefont {Rossignoli}},\ }\bibfield  {title} {\bibinfo {title} {One-body information loss in fermion systems},\ }\href {https://doi.org/PhysRevA.94.042315} {\bibfield  {journal} {\bibinfo  {journal} {Physical Review A}\ }\textbf {\bibinfo {volume} {94}},\ \bibinfo {pages} {042315} (\bibinfo {year} {2016})}\BibitemShut {NoStop}%
\bibitem [{\citenamefont {Vidal}\ \emph {et~al.}(2021)\citenamefont {Vidal}, \citenamefont {Bera}, \citenamefont {Riera}, \citenamefont {Lewenstein},\ and\ \citenamefont {Bera}}]{vidal2021quantum}%
  \BibitemOpen
  \bibfield  {author} {\bibinfo {author} {\bibfnamefont {N.~T.}\ \bibnamefont {Vidal}}, \bibinfo {author} {\bibfnamefont {M.~L.}\ \bibnamefont {Bera}}, \bibinfo {author} {\bibfnamefont {A.}~\bibnamefont {Riera}}, \bibinfo {author} {\bibfnamefont {M.}~\bibnamefont {Lewenstein}},\ and\ \bibinfo {author} {\bibfnamefont {M.~N.}\ \bibnamefont {Bera}},\ }\bibfield  {title} {\bibinfo {title} {Quantum operations in an information theory for fermions},\ }\href {https://doi.org/10.1103/PhysRevA.104.032411} {\bibfield  {journal} {\bibinfo  {journal} {Physical Review A}\ }\textbf {\bibinfo {volume} {104}},\ \bibinfo {pages} {032411} (\bibinfo {year} {2021})}\BibitemShut {NoStop}%
\bibitem [{Note11()}]{Note11}%
  \BibitemOpen
  \bibinfo {note} {Let us recall that while for discrete spacetime the Stone-Von Neummann theorem guarantees that the Hilbert space is well-defined by the canonical algebra (up to unitary equivalence), the continuum limit carries the usual mathematical subtleties. Nonetheless, these subtleties are the standard ones, as the algebra of Eq.\ \protect \eqref {eq:spacetimefieldalg} is isomorphic to the equal-time algebra of standard QFT with an additional dimension. As such, they may be tackled by proper regularization.}\BibitemShut {Stop}%
\bibitem [{\citenamefont {Haag}\ and\ \citenamefont {Kastler}(1964)}]{haag1964algebraic}%
  \BibitemOpen
  \bibfield  {author} {\bibinfo {author} {\bibfnamefont {R.}~\bibnamefont {Haag}}\ and\ \bibinfo {author} {\bibfnamefont {D.}~\bibnamefont {Kastler}},\ }\bibfield  {title} {\bibinfo {title} {An algebraic approach to quantum field theory},\ }\href {https://doi.org/10.1063/1.1704187} {\bibfield  {journal} {\bibinfo  {journal} {Journal of Mathematical Physics}\ }\textbf {\bibinfo {volume} {5}},\ \bibinfo {pages} {848} (\bibinfo {year} {1964})}\BibitemShut {NoStop}%
\bibitem [{\citenamefont {Aharonov}\ and\ \citenamefont {Kaufherr}(1984)}]{aharonov1984quantum}%
  \BibitemOpen
  \bibfield  {author} {\bibinfo {author} {\bibfnamefont {Y.}~\bibnamefont {Aharonov}}\ and\ \bibinfo {author} {\bibfnamefont {T.}~\bibnamefont {Kaufherr}},\ }\bibfield  {title} {\bibinfo {title} {Quantum frames of reference},\ }\href {https://doi.org/https://doi.org/10.1103/PhysRevD.30.368} {\bibfield  {journal} {\bibinfo  {journal} {Physical Review D}\ }\textbf {\bibinfo {volume} {30}},\ \bibinfo {pages} {368} (\bibinfo {year} {1984})}\BibitemShut {NoStop}%
\bibitem [{\citenamefont {Giacomini}\ \emph {et~al.}(2019{\natexlab{b}})\citenamefont {Giacomini}, \citenamefont {Castro-Ruiz},\ and\ \citenamefont {Brukner}}]{giacomini2019relativistic}%
  \BibitemOpen
  \bibfield  {author} {\bibinfo {author} {\bibfnamefont {F.}~\bibnamefont {Giacomini}}, \bibinfo {author} {\bibfnamefont {E.}~\bibnamefont {Castro-Ruiz}},\ and\ \bibinfo {author} {\bibfnamefont {{\v{C}}.}~\bibnamefont {Brukner}},\ }\bibfield  {title} {\bibinfo {title} {Relativistic quantum reference frames: the operational meaning of spin},\ }\href {https://doi.org/10.1103/PhysRevLett.123.090404} {\bibfield  {journal} {\bibinfo  {journal} {Physical review letters}\ }\textbf {\bibinfo {volume} {123}},\ \bibinfo {pages} {090404} (\bibinfo {year} {2019}{\natexlab{b}})}\BibitemShut {NoStop}%
\bibitem [{\citenamefont {Apadula}\ \emph {et~al.}(2024)\citenamefont {Apadula}, \citenamefont {Castro-Ruiz},\ and\ \citenamefont {Brukner}}]{apadula2022quantum}%
  \BibitemOpen
  \bibfield  {author} {\bibinfo {author} {\bibfnamefont {L.}~\bibnamefont {Apadula}}, \bibinfo {author} {\bibfnamefont {E.}~\bibnamefont {Castro-Ruiz}},\ and\ \bibinfo {author} {\bibfnamefont {{\v{C}}.}~\bibnamefont {Brukner}},\ }\bibfield  {title} {\bibinfo {title} {Quantum reference frames for lorentz symmetry},\ }\href {https://doi.org/10.22331/q-2024-08-14-1440} {\bibfield  {journal} {\bibinfo  {journal} {Quantum}\ }\textbf {\bibinfo {volume} {8}},\ \bibinfo {pages} {1440} (\bibinfo {year} {2024})}\BibitemShut {NoStop}%
\bibitem [{Note12()}]{Note12}%
  \BibitemOpen
  \bibinfo {note} {We note that the literature also includes resource-theoretic approaches to QRFs; see, e.g., \cite {bartlett2007reference,marvian2014modes}.}\BibitemShut {Stop}%
\bibitem [{\citenamefont {Dirac}(1950)}]{dirac1950generalized}%
  \BibitemOpen
  \bibfield  {author} {\bibinfo {author} {\bibfnamefont {P.~A.~M.}\ \bibnamefont {Dirac}},\ }\bibfield  {title} {\bibinfo {title} {Generalized hamiltonian dynamics},\ }\href {https://doi.org/10.4153/CJM-1950-012-1} {\bibfield  {journal} {\bibinfo  {journal} {Canadian journal of mathematics}\ }\textbf {\bibinfo {volume} {2}},\ \bibinfo {pages} {129} (\bibinfo {year} {1950})}\BibitemShut {NoStop}%
\bibitem [{\citenamefont {Fishman}\ \emph {et~al.}(2022)\citenamefont {Fishman}, \citenamefont {White},\ and\ \citenamefont {Stoudenmire}}]{fishman2022itensor}%
  \BibitemOpen
  \bibfield  {author} {\bibinfo {author} {\bibfnamefont {M.}~\bibnamefont {Fishman}}, \bibinfo {author} {\bibfnamefont {S.}~\bibnamefont {White}},\ and\ \bibinfo {author} {\bibfnamefont {E.}~\bibnamefont {Stoudenmire}},\ }\bibfield  {title} {\bibinfo {title} {The itensor software library for tensor network calculations},\ }\href {https://doi.org/10.21468/SciPostPhysLectNotes.54} {\bibfield  {journal} {\bibinfo  {journal} {SciPost Physics Codebases}\ ,\ \bibinfo {pages} {004}} (\bibinfo {year} {2022})}\BibitemShut {NoStop}%
\bibitem [{\citenamefont {Bezanson}\ \emph {et~al.}(2017)\citenamefont {Bezanson}, \citenamefont {Edelman}, \citenamefont {Karpinski},\ and\ \citenamefont {Shah}}]{bezanson2017julia}%
  \BibitemOpen
  \bibfield  {author} {\bibinfo {author} {\bibfnamefont {J.}~\bibnamefont {Bezanson}}, \bibinfo {author} {\bibfnamefont {A.}~\bibnamefont {Edelman}}, \bibinfo {author} {\bibfnamefont {S.}~\bibnamefont {Karpinski}},\ and\ \bibinfo {author} {\bibfnamefont {V.~B.}\ \bibnamefont {Shah}},\ }\bibfield  {title} {\bibinfo {title} {Julia: A fresh approach to numerical computing},\ }\href {https://doi.org/10.1137/141000671} {\bibfield  {journal} {\bibinfo  {journal} {SIAM {R}eview}\ }\textbf {\bibinfo {volume} {59}},\ \bibinfo {pages} {65} (\bibinfo {year} {2017})}\BibitemShut {NoStop}%
\bibitem [{\citenamefont {Wootters}(1987)}]{wootters1987wigner}%
  \BibitemOpen
  \bibfield  {author} {\bibinfo {author} {\bibfnamefont {W.~K.}\ \bibnamefont {Wootters}},\ }\bibfield  {title} {\bibinfo {title} {A wigner-function formulation of finite-state quantum mechanics},\ }\href {https://doi.org/https://doi.org/10.1016/0003-4916(87)90176-X} {\bibfield  {journal} {\bibinfo  {journal} {Annals of Physics}\ }\textbf {\bibinfo {volume} {176}},\ \bibinfo {pages} {1} (\bibinfo {year} {1987})}\BibitemShut {NoStop}%
\bibitem [{\citenamefont {Renes}\ \emph {et~al.}(2004)\citenamefont {Renes}, \citenamefont {Blume-Kohout}, \citenamefont {Scott},\ and\ \citenamefont {Caves}}]{renes2004symmetric}%
  \BibitemOpen
  \bibfield  {author} {\bibinfo {author} {\bibfnamefont {J.~M.}\ \bibnamefont {Renes}}, \bibinfo {author} {\bibfnamefont {R.}~\bibnamefont {Blume-Kohout}}, \bibinfo {author} {\bibfnamefont {A.~J.}\ \bibnamefont {Scott}},\ and\ \bibinfo {author} {\bibfnamefont {C.~M.}\ \bibnamefont {Caves}},\ }\bibfield  {title} {\bibinfo {title} {Symmetric informationally complete quantum measurements},\ }\href {https://doi.org/10.1063/1.1737053} {\bibfield  {journal} {\bibinfo  {journal} {Journal of Mathematical Physics}\ }\textbf {\bibinfo {volume} {45}},\ \bibinfo {pages} {2171} (\bibinfo {year} {2004})}\BibitemShut {NoStop}%
\bibitem [{\citenamefont {Kamenev}(2005)}]{kamenev2005course}%
  \BibitemOpen
  \bibfield  {author} {\bibinfo {author} {\bibfnamefont {A.}~\bibnamefont {Kamenev}},\ }\bibfield  {title} {\bibinfo {title} {Course 3 many-body theory of non-equilibrium systems},\ }in\ \href {https://doi.org/10.1016/S0924-8099(05)80045-9} {\emph {\bibinfo {booktitle} {Les Houches}}},\ Vol.~\bibinfo {volume} {81}\ (\bibinfo  {publisher} {Elsevier},\ \bibinfo {year} {2005})\ pp.\ \bibinfo {pages} {177--246}\BibitemShut {NoStop}%
\bibitem [{\citenamefont {Bartlett}\ \emph {et~al.}(2007)\citenamefont {Bartlett}, \citenamefont {Rudolph},\ and\ \citenamefont {Spekkens}}]{bartlett2007reference}%
  \BibitemOpen
  \bibfield  {author} {\bibinfo {author} {\bibfnamefont {S.~D.}\ \bibnamefont {Bartlett}}, \bibinfo {author} {\bibfnamefont {T.}~\bibnamefont {Rudolph}},\ and\ \bibinfo {author} {\bibfnamefont {R.~W.}\ \bibnamefont {Spekkens}},\ }\bibfield  {title} {\bibinfo {title} {Reference frames, superselection rules, and quantum information},\ }\href {https://doi.org/10.1103/RevModPhys.79.555} {\bibfield  {journal} {\bibinfo  {journal} {Reviews of Modern Physics}\ }\textbf {\bibinfo {volume} {79}},\ \bibinfo {pages} {555} (\bibinfo {year} {2007})}\BibitemShut {NoStop}%
\bibitem [{\citenamefont {Marvian}\ and\ \citenamefont {Spekkens}(2014)}]{marvian2014modes}%
  \BibitemOpen
  \bibfield  {author} {\bibinfo {author} {\bibfnamefont {I.}~\bibnamefont {Marvian}}\ and\ \bibinfo {author} {\bibfnamefont {R.~W.}\ \bibnamefont {Spekkens}},\ }\bibfield  {title} {\bibinfo {title} {Modes of asymmetry: the application of harmonic analysis to symmetric quantum dynamics and quantum reference frames},\ }\href {https://doi.org/10.1103/PhysRevA.90.062110} {\bibfield  {journal} {\bibinfo  {journal} {Physical Review A}\ }\textbf {\bibinfo {volume} {90}},\ \bibinfo {pages} {062110} (\bibinfo {year} {2014})}\BibitemShut {NoStop}%
\end{thebibliography}
\end{document}